\documentclass[11pt]{report}

\def\draft{1}  

\usepackage{amsmath,amsthm,amssymb}
\usepackage{newtxtext,newtxmath}

\usepackage{color}
\usepackage{eucal}
\usepackage[multiple,symbol]{footmisc}

\usepackage{fullpage}

\usepackage{titlesec}

\usepackage{pgf,tikz}
\usepackage{mathrsfs}
\usetikzlibrary{arrows}
\usepackage{xspace}
\usepackage{mdframed}

\usepackage{subfig}

\usepackage[noend,noline]{algorithm2e}
\SetEndCharOfAlgoLine{}
\SetArgSty{}
\SetKwBlock{Repeat}{repeat}{}

\mdfsetup{frametitlealignment=\centering}
\mdfdefinestyle{offset}{backgroundcolor=white,linecolor=black,innerrightmargin=15pt,innermargin=10pt,outermargin=10pt,innertopmargin=.5\baselineskip,innerbottommargin=.5\baselineskip}

\usepackage{tablefootnote} 
\makeatletter 
\AfterEndEnvironment{mdframed}{%
 \tfn@tablefootnoteprintout%
 \gdef\tfn@fnt{0}%
}
\makeatother

\usepackage[numbers]{natbib}
\usepackage[nottoc]{tocbibind}

\definecolor{weborange}{rgb}{.8,.3,.3}
\definecolor{webblue}{rgb}{0,0,.8}
\definecolor{internallinkcolor}{rgb}{0,.5,0}
\definecolor{externallinkcolor}{rgb}{0,0,.5}

\usepackage[%
  pagebackref,
  hyperfootnotes=false,
  colorlinks=true,
  urlcolor=externallinkcolor,
  linkcolor=internallinkcolor,
  filecolor=externallinkcolor,
  citecolor=internallinkcolor,
  breaklinks=true,
  pdfstartview=FitH,
  pdfpagelayout=OneColumn]{hyperref}

\theoremstyle{plain}
\newtheorem{theorem}{Theorem}[chapter]

\newtheorem{lemma}[theorem]{Lemma}
\newtheorem{corollary}[theorem]{Corollary}
\newtheorem{claim}[theorem]{Claim}
\newtheorem{proposition}[theorem]{Proposition}
\newtheorem{folklore}{Folklore Belief}

\theoremstyle{definition}
\newtheorem{fact}[theorem]{Fact}

\newtheorem{definition}[theorem]{Definition}
\newtheorem{example}[theorem]{Example}
\newtheorem{remark}[theorem]{Remark}

\newcommand{\R}{{\mathbb R}}

\newcommand{\poly}{\mathrm{poly}}

\newcommand{\wm}[1][k]{\textsc{Welfare-Maximization}($#1$)}
\newcommand\mdisj{\textsc{Multi-Disjointness}\xspace}

\newcommand{\NE}{\mathsf{NE}}
\newcommand{\eNE}{\epsilon\text{-}\mathsf{NE}}
\newcommand{\NP}{\mathsf{NP}}
\newcommand{\E}{\mathsf{E}}
\newcommand{\RP}{\mathsf{RP}}
\newcommand{\BPP}{\mathsf{BPP}}
\newcommand{\ptime}{\mathsf{P}}
\newcommand{\pspace}{\mathsf{PSPACE}}
\newcommand{\coNP}{\mathsf{co}\mbox{-}\mathsf{NP}}
\newcommand{\FNP}{\mathsf{FNP}}
\newcommand{\TFNP}{\mathsf{TFNP}}
\newcommand{\PLS}{\mathsf{PLS}}
\newcommand{\PPA}{\mathsf{PPA}}
\newcommand{\PPP}{\mathsf{PPP}}
\newcommand{\PPADS}{\mathsf{PPADS}}
\newcommand{\PPAD}{\mathsf{PPAD}}

\newcommand{\brouwer}{{\sc Brouwer}}
\newcommand{\cclass}[1]{\mathsf{#1}}

\newcommand\eol{{\sc EoL}\xspace}
\newcommand\cceol{{\sc 2EoL}\xspace}
\newcommand\disj{{\sc Disjointness}\xspace}
\newcommand\ind{{\sc Index}\xspace}
\newcommand\bfp{{\sc $\eps$-BFP}\xspace}
\newcommand\ccbfp{$\eps$-{\sc 2BFP}\xspace}
  
\newcommand{\PH}{\cclass{PH}}

\newcommand{\sharpP}{\cclass{\#P}}

\ifnum\draft=1
\newcommand{\Snote}[1]{\textbf{[Salil's Note: #1]}}
\else
\newcommand{\Snote}[1]{\textbf{}}
\fi

\newcommand{\prob}[2][]{\mathbf{Pr}\ifthenelse{\not\equal{}{#1}}{_{#1}}{}\!\left[#2\right]}
\newcommand{\expect}[2][]{\mathbf{E}\ifthenelse{\not\equal{}{#1}}{_{#1}}{}\!\left[#2\right]}

\newcommand{\eps}{\epsilon}

\newcommand{\tO}{\tilde{O}}
\newcommand{\zo}{\{0,1\}}
\newcommand{\xhat}{\hat{x}}
\newcommand{\yhat}{\hat{y}}
\newcommand{\mne}{(\xhat,\yhat)}
\newcommand{\xbar}{\bar{x}}
\newcommand{\ybar}{\bar{y}}
\newcommand{\zbar}{\bar{z}}
\DeclareMathOperator{\argmax}{argmax}

\newcommand{\strat}{s}

\def\discH{H_{\eps}}
\def\p{p}

\def\r{r}

\def\bfz{\mathbf{z}}

\def\x{x}
\def\bfx{\mathbf{x}}
\def\boldp{\mathbf{p}}
\def\minputs{(\bfx_1,\ldots,\bfx_k)}
\def\w{w}
\def\y{y}

\def\Alg{\mathcal{A}}
\def\A{\mathcal{A}}
\renewcommand{\P}{\mathcal{P}}
\def\V{\mathcal{V}}
\def\F{\mathcal{F}}

\def\sse{\subseteq}

\newcommand{\n}[1]{
{\| {#1} \|}
}
\newcommand{\mixed}{x}

\newpage
\titleformat{\chapter}[display]
{\normalfont\Large\filcenter\rmfamily\scshape}
{\titlerule[1pt]%
 \vspace{1pt}%
 \titlerule
 \vspace{1ex}%
 \LARGE{Solar Lecture} \thechapter
}
{.5ex}
{\normalfont\large\slshape}

\newcommand{\lecture}[1]{%
  \chapter{#1}%
\titlerule
}

\begin{document}

\thispagestyle{empty}
\begin{center}
  {\Huge\textsc{Barbados Lectures on Complexity Theory, Game Theory, and Economics}\par}
  \bigskip
  \bigskip
  {\Large Tim Roughgarden\\Columbia University\par}
  \bigskip
  \bigskip
  {\LARGE 29th McGill Invitational\\ Workshop on Computational Complexity\par}
  \bigskip
  \bigskip
  {\Large Bellairs Institute\\ Holetown, Barbados\par}
  \vfill
\end{center}

\newpage
\addcontentsline{toc}{section}{Foreward}

\renewcommand{\thefootnote}{\fnsymbol{footnote}}

\begin{center}
  {\Large \textsl{Foreward}\par}
\end{center}
\vspace{1.25\baselineskip}

This monograph is based on lecture notes from my mini-course
``Complexity Theory, Game Theory, and Economics,'' taught at
the Bellairs Research Institute of McGill University, Holetown,
Barbados, February 19--23,
2017, as the 29th  McGill Invitational Workshop on Computational
Complexity.

The goal of this monograph is twofold:
\begin{itemize}

\item [(i)] to explain how complexity
theory has helped illuminate several barriers in economics and game
theory; and 

\item [(ii)] to illustrate how game-theoretic questions have led
to new and interesting complexity theory, including several
very recent breakthroughs.

\end{itemize}
It consists of two five-lecture sequences: the {\em Solar Lectures,}
focusing on the communication and
computational complexity of computing equilibria; and the {\em Lunar
Lectures,} focusing on applications of complexity theory in game
theory and economics.\footnote{Cris Moore: ``So when are the
  {\em stellar} lectures?''}
No background in game theory is assumed.

Thanks are due to many people: Denis Therien and Anil Ada for
organizing the workshop and for inviting me to lecture; Omri
Weinstein, for giving a guest lecture on simulation theorems in
communication complexity; Alex Russell, for coordinating the scribe
notes; the scribes\footnote{Anil Ada,
Amey Bhangale,
Shant Boodaghians,
Sumegha Garg,
Valentine Kabanets, 
Antonina Kolokolova,
Michal Kouck\'y,
Cristopher Moore,
Pavel Pudl\'ak,
Dana Randall,
Jacobo Tor\'an,
Salil Vadhan,
Joshua R. Wang, and
Omri Weinstein.}, for putting together a terrific first draft;
and all of the workshop attendees, for making the experience so
unforgettable (if intense!).  I also thank 
Yakov Babichenko,
Mika G\"o\"os,
Aviad Rubinstein,
Eylon Yogev,
and an anonymous reviewer for
numerous helpful comments on earlier drafts of this monograph.

The writing of this monograph was supported in part by NSF award
CCF-1524062, a Google Faculty Research Award, and a Guggenheim Fellowship.
I would be very happy to receive any comments or
corrections from readers.

\vspace{2\baselineskip}

\noindent Tim Roughgarden\\
\noindent Bracciano, Italy\\ 
December 2017\\
(Revised December 2019)

\newpage
\tableofcontents

\newpage


\newpage
\addcontentsline{toc}{section}{Overview}
\setcounter{footnote}{0}

\begin{center}
  {\Large \textsl{Overview}\par}
\end{center}
\vspace{1.25\baselineskip}

There are 5 solar lectures and 5 lunar lectures.  The solar
lectures focus on the communication and computational complexity of
computing an (approximate) Nash equilibrium.  The lunar lectures are less
technically intense and 
meant to be understandable even after consuming a rum punch; they
focus on applications of computational complexity theory to game
theory and economics.

\subsection*{The Solar Lectures: Complexity of Equilibria}

\paragraph{Lecture 1: Introduction and wish list.}
The goal of the first lecture is to get the lay of the land.
We'll focus on the types of positive results about equilibria that we
want, like fast algorithms and quickly converging distributed
processes.  Such positive results are possible in special cases (like
zero-sum games), and the challenge for complexity theory is to prove
that they cannot be extended to the general case.
The topics in this lecture are mostly classical.

\paragraph{Lectures 2 and 3: The communication complexity of Nash
  equilibria.} 
These two lectures cover the main ideas in the 
recent paper of \citet{BR17}, which proves strong communication complexity
lower bounds for computing an approximate Nash equilibrium.
Discussing the proof also gives us an excuse to talk about
``simulation theorems'' in the spirit of
\citet{DBLP:journals/combinatorica/RazM99}, which lift query
complexity lower bounds to communication complexity lower bounds and
have recently found a number of exciting applications.

\paragraph{Lecture 4: $\TFNP$, $\PPAD$, and all that.}
In this lecture we begin our study of the {\em computational} complexity of
computing a Nash equilibrium, where we want conditional but
super-polynomial lower bounds.  Proving analogs of
$\NP$-completeness results requires developing
customized complexity classes appropriate for the study of equilibrium
computation.\footnote{Why can't we use the tried-and-true theory of
  $\NP$-completeness?
Because the guaranteed existence
  (Theorem~\ref{t:nash}) and
  efficient verifiability of a Nash equilibrium imply that computing
  one is an easier task than solving an $\NP$-complete problem, under appropriate
  complexity assumptions (see Theorem~\ref{t:mp91}).}
This lecture also discusses the existing evidence for the intractability of
these complexity classes, including some very recent developments.

\paragraph{Lecture 5: The computational complexity of computing an
  approximate Nash equilibrium of a bimatrix game.}
The goal of this lecture is to give a high-level overview of 
Rubinstein's recent breakthrough result~\cite{R16} that an ETH-type
assumption for $\PPAD$ implies a quasi-polynomial-time lower bound for
the problem of computing an approximate Nash equilibrium (which is
tight, by Corollary~\ref{cor:lmm2}).

\subsection*{The Lunar Lectures: Complexity-Theoretic Barriers in Economics}

Most of the lunar lectures have the flavor of ``applied complexity
theory.''\footnote{Not an oxymoron!}  While the solar lectures build
on each other to some extent, the lunar lectures are episodic and
can be read independently of each other.

\paragraph{Lecture 1: The 2016 FCC Incentive Auction.}
The recent FCC Incentive Auction is a great case study of how computer
science has influenced real-world auction design.  This lecture
provides our first broader glimpse of the vibrant field called {\em
  algorithmic game theory}, at most 10\% of which concerns the
complexity of computing equilibria.

\paragraph{Lecture 2: Barriers to near-optimal equilibria.}
This lecture concerns the ``price of anarchy,'' meaning the extent to
which the Nash equilibria of a game approximate an optimal outcome.  It
turns out that nondeterministic communication complexity lower bounds
can be translated, in black-box fashion, to lower bounds on the price
of anarchy.  We'll see how this translation enables a theory of
``optimal simple auctions.''

\paragraph{Lecture 3: Barriers in markets.}
You've surely heard of the idea of ``market-clearing prices,'' which
are prices in a market such that supply equals demand.  When the goods
are divisible (milk, wheat, etc.), market-clearing prices exist under
relatively mild technical assumptions.  With indivisible goods (houses,
spectrum licenses, etc.), market-clearing prices may or may not
exist.  It turns out that complexity considerations can be used to explain
when such prices exist and when they do not.
This is cool and surprising because the issue of equilibrium existence
seems to have nothing to do with computation (in contrast to the Solar
Lectures, where the questions studied are explicitly about
computation).

\paragraph{Lecture 4: The borders of Border's theorem.}
Border's theorem is a famous result in auction theory from~1991, about
single-item auctions.  Despite its fame, no one has been able to
extend it to significantly more general settings.  We'll see that
complexity theory explains this mystery: significantly generalizing
Border's theorem would imply that the polynomial hierarchy collapses!

\paragraph{Lecture 5: Tractable relaxations of Nash equilibria.}
With the other lectures focused largely on negative results for
computing Nash equilibria, for an epilogue we'll conclude with
positive algorithmic results for relaxations of Nash equilibria, such
as correlated equilibria.

\part{Solar Lectures}

\renewcommand{\thefootnote}{\arabic{footnote}}

\setcounter{chapter}{0}

\newpage
\lecture{Introduction, Wish List, and Two-Player Zero-Sum Games}

\vspace{1cm}



\section{Nash Equilibria in Two-Player Zero-Sum Games}

\subsection{Preamble}

To an algorithms person (like the author), complexity theory is the
science of why you can't get what you want.  So what is it we want?
Let's start with some cool positive results for a very special class
of games---two-player zero-sum games---and then we can study whether
or not they extend to more general games.  For the first positive
result, we'll review the famous Minimax theorem, and see how it leads
to a polynomial-time algorithm for computing a Nash equilibrium of a
two-player zero-sum game.  Then we'll show that there are natural
``dynamics'' (basically, a distributed algorithm) that converge
rapidly to an approximate Nash equilibrium.

\subsection{Rock-Paper-Scissors}

Recall the game of rock-paper-scissors (or roshambo, if you
like)\footnote{\url{https://en.wikipedia.org/wiki/Rock-paper-scissors}}:
there are two players, each simultaneously picks a strategy from
$\{ \text{rock}, \text{paper}, \text{scissors} \}$.  If both players
choose the same strategy then the game is a draw; otherwise, rock
beats scissors, scissors beats paper, and paper beats rock.\footnote{Here are some
  fun facts
  about rock-paper-scissors.  There's a World Series of RPS every
  year, with a top prize of at least \$50K.  If you watch some videos
  from the event, you will see pure psychological warfare.  Maybe this
  explains why some of the same players seem to end up in the later
  rounds of the tournament every year.

There's also a robot hand, built at the University of Tokyo, that
plays rock-paper-scissors with a winning probability of 100\% (check
out the video).  No surprise, a very high-speed camera is involved.}

Here's an idea: how about we play rock-paper-scissors, and you go
first?  This is clearly unfair---no matter what strategy you choose, I
have a response that guarantees victory.  But what if you only have to
commit to a
{\em probability distribution} over your three strategies (called a
\emph{mixed strategy})?  To be clear, the order of operations is: (i)
you pick a distribution; (ii) I pick a response; (iii) nature flips
coins to sample a strategy from your distribution.  Now you can
protect yourself---by picking a strategy uniformly at random, no
matter what I do, you have an equal chance of a win, a loss, or a
draw.

The {\em Minimax theorem} states that, in any game of ``pure
competition'' like rock-paper-scissors, a player can always protect
herself with a suitable randomized strategy---there is no disadvantage
of having to move first.  The proof of the Minimax
theorem also gives as a byproduct a polynomial-time algorithm for
computing a Nash equilibrium (by linear programming).

\subsection{Formalism}

We specify a two-player zero-sum game with an $m \times n$ payoff
matrix $A$ of numbers. The rows correspond to the possible choices of
Alice (the ``row player'') and the columns correspond to possible
choices for Bob (the ``column player''). Entry $A_{ij}$ contains
Alice's payoff when Alice chooses row $i$ and Bob chooses column
$j$. In a zero-sum game, Bob's corresponding payoff is automatically
defined to be $-A_{ij}$.
Throughout the solar lectures, we normalize the payoff matrix so
that $|A_{ij}| \leq 1$ for all $i$ and $j$.\footnote{This is without loss of
generality, by scaling.}

For example, the payoff matrix corresponding to rock-paper-scissors is:
\\
\begin{center}
\begin{tabular}{ r|c|c|c| }
\multicolumn{1}{r}{}
 &  \multicolumn{1}{c}{R}
 & \multicolumn{1}{c}{P}
 & \multicolumn{1}{c}{S} \\
\cline{2-4}
R & 0 & -1 & 1 \\
\cline{2-4}
P & 1 & 0 & -1 \\
\cline{2-4}
S & -1 & 1 & 0 \\
\cline{2-4}
\end{tabular}
\end{center}
\

Mixed strategies for Alice and Bob correspond to probability
distributions $x$ and $y$ over rows and columns,
respectively.\footnote{A {\em pure strategy} is the special case of a
  mixed strategy that is deterministic (i.e., allots all its
  probability to a single strategy).}

When
speaking about Nash equilibria, one always assumes that players
randomize independently.  For a two-player zero-sum game $A$ and mixed
strategies $x,y$, we can write Alice's expected payoff as
\[
    x^{\top} A y = \sum_{i,j} A_{ij} x_i y_j\,.
\]
Bob's expected payoff is the negative of this quantity, so 
his goal is to minimize the expression above.

\subsection{The Minimax Theorem}

The question that the Minimax theorem addresses is the following: 
\begin{quote}
  If two players make choices \emph{sequentially} in a zero-sum
  game, is it better to go first or second?
\end{quote}
In a zero-sum game, there can only be a first-mover disadvantage.
Going second gives a player the opportunity to adapt to what the other
player does first.  And the second player always has the option of
choosing whatever mixed strategy she would have chosen had she gone
first.  But does going second ever strictly help?  The Minimax theorem
gives an amazing answer to the question above: {\em it doesn't matter!}

\begin{theorem}[Minimax Theorem] \label{t:minmax}
Let $A$ be the payoff matrix of a two-player zero-sum game. Then
\begin{equation} \label{eq:min-max}
    \max_x \left ( \min_y  \; x^{\top} A y \right ) = \min_y \left ( \max_x \; x^{\top} A y \right )\,,
\end{equation}
where $x$ and $y$ range over probability distributions over the rows
and columns of $A$, respectively.
\end{theorem}

On the left-hand side of~\eqref{eq:min-max}, the row player moves first and
the column player second.  The column player plays optimally given the
strategy chosen by the row player, and the row player plays optimally
anticipating the column player's response.
On the right-hand side of
\eqref{eq:min-max}, the roles of the two players are reversed.
The Minimax theorem asserts that, under optimal play, the expected
payoff of each player is the same in both scenarios.

The first proof of the Minimax theorem was due to
von Neumann~\cite{vN28} and used fixed-point-type arguments (which
we'll have much more to say about later).
von Neumann and Morgenstern~\cite{vNM44}, inspired by Ville~\cite{V38}, 
later realized that the Minimax theorem can be deduced from strong
linear programming duality.\footnote{Dantzig \cite[p.5]{D81} describes
  meeting John von Neumann on October 3, 1947: ``In under a minute I
  slapped the geometric and the algebraic version of the [linear
  programming] problem on the blackboard.  Von Neumann stood up and
  said `Oh that!'  Then for the next hour and a half, he proceeded to
  give me a lecture on the mathematical theory of linear programs.

``At one point seeing me sitting there with my eyes popping and my
mouth open (after all I had searched the literature and found
nothing), von Neumann said: `I don't want you to think I am pulling
all this out of my sleeve on the spur of the moment like a magician.  I
have just recently completed a book with Oskar Morgenstern on the
Theory of Games.  What I am doing is conjecturing that the two problems
are equivalent.''

This equivalence between strong linear programming duality and the
Minimax theorem is made precise in Dantzig~\cite{D51}, Gale et al.~\cite{GKT51}, and Adler
\cite{A13}.}

\begin{proof}
The idea is to formulate the problem faced by the first player as a
linear program.  The theorem will then follow from linear programming
duality.

First, the player who moves second always has an optimal pure (i.e., deterministic)
strategy---given the probability distribution chosen by the first
player, the second player can simply play the strategy with the highest
expected payoff.  This means the inner $\min$ and $\max$
in~\eqref{eq:min-max} may as well range over columns and rows,
respectively, rather than over all probability distributions.  The
expression on the left-hand side of~\eqref{eq:min-max} then translates
to the following linear program:
\begin{align*}
&&\max_{x,v}  &\quad v  \\
&&\text{s.t.} &\quad v \le \quad\sum_{i=1}^m A_{ij} x_i \quad  \text{for
                all columns $j$}, \\
&& & \quad x \text{ is a probability distribution over rows.}
\end{align*}
If the optimal point is $(v^*, x^*)$, then $v^*$ equals the
left-hand-side of~\eqref{eq:min-max} and $x^*$ belongs to the
corresponding arg-max.  In plain terms, $x^*$ is what Alice should
play if
she has to move first, and $v^*$ is the consequent expected payoff 
(assuming Bob responds optimally).

Similarly, we can write a second linear program that computes the
optimal point $(w^*, y^*)$ from Bob's perspective, where $w^*$ equals
the right-hand-side of~\eqref{eq:min-max} and $y^*$ is in the
corresponding arg-min:
\begin{align*}
&& \min_{y,w} &\quad w  \\
&& \text{s.t.} &\quad w \ge \sum_{j=1}^n A_{ij} y_j \quad  \text{for
                 all rows $i$}, \\
&& &\quad y \text{ is a probability distribution over columns.}
\end{align*}
It is straightforward to verify that these two linear programs are in
fact duals of each other (left to the reader, or see
Chv\'atal~\cite{chvatal}).  By strong linear programming duality, we
know that the two linear programs have equal optimal objective
function values and hence $v^* = w^*$.  This means that the payoff that
Alice can guarantee herself when she goes first is the same as 
when Bob goes first (and plays optimally), completing the proof.
\end{proof}

\begin{definition}[Values and Min-Max Pairs]
Let $A$ be the payoff matrix of a two-player zero-sum game. The
\emph{value} of the game is defined as the common value of 
\[
\max_x \left ( \min_y  \; x^{\top} A y \right) 
\text{~~and~~} \min_y \left ( \max_x \; x^{\top} A y \right ).
\]
A {\em min-max strategy} is a strategy $x^*$ in the arg-max of the left-hand
side or a strategy $y^*$ in the arg-min of the right-hand side.
A \emph{min-max pair} is a pair $(x^*,y^*)$ where $x^*$ and $y^*$ are
both min-max strategies.
\end{definition}
For example,
the value of the rock-paper-scissors game is $0$ and $(u,u)$ is its
unique min-max pair, where~$u$ denotes the uniform probability
distribution.

The min-max pairs are the optimal solutions of the two linear programs
in the proof of Theorem~\ref{t:minmax}.  Because the optimal solution
of a linear program can be computed in polynomial time, so can a min-max
pair.

\subsection{Nash Equilibrium}

In zero-sum games, a min-max pair is closely related to the notion of
a Nash equilibrium, defined next.\footnote{If you think you learned this
  definition from the movie \emph{A Beautiful Mind}, it's time to
  learn the correct definition!}

\begin{definition}[Nash Equilibrium in a Two-Player Zero-Sum
  Game]\label{d:ne}
Let $A$ be the payoff matrix of a two-player zero-sum game. The pair
$(\hat{x}, \hat{y})$ is a \emph{Nash equilibrium} if:
\begin{itemize}

\item[(i)] $\hat{x}^{\top} A \hat{y} \geq x^{\top} A \hat{y} \;$ for all $x$
  (given that Bob plays $\hat{y}$, Alice cannot increase her expected
  payoff by deviating unilaterally to a strategy different from
  $\hat{x}$,
  i.e., $\hat{x}$ is optimal given $\hat{y}$);

\item[(ii)] $\hat{x}^{\top} A \hat{y} \leq \hat{x}^{\top} A y \;$ for all $y$
(given $\hat{x}$,   $\hat{y}$ is an optimal strategy for Bob).

\end{itemize}
\end{definition}
The pairs in Definition~\ref{d:ne} are sometimes called \emph{mixed}
Nash equilibria, to stress that players are allowed to randomize.  (As
opposed to a {\em pure} Nash equilibrium, where both players
play deterministically.)  Unless otherwise noted, we will
always be concerned with mixed Nash equilibria.

\begin{proposition}[Equivalence of Nash Equilibria and Min-Max Pairs]\label{claim:ne}
  In a two-player zero-sum game, a pair $(x^*,y^*)$ is a min-max pair
  if and only if it is a Nash equilibrium.
\end{proposition}

\begin{proof}
  Suppose $(x^*,y^*)$ is a min-max pair, and so Alice's expected
  payoff is~$v^*$,
  the value of the game.  Because Alice plays her min-max strategy, Bob
  cannot make her payoff smaller than~$v^*$ via some other strategy.
  Because Bob plays his min-max strategy, Alice cannot make her payoff
larger than~$v^*$.  Neither player can do better with a
unilateral deviation, and so $(x^*,y^*)$ is a Nash equilibrium.

Conversely, suppose $(x^*,y^*)$ is not a min-max pair with, say, Alice
not playing a min-max strategy.  If Alice's expected payoff is less
than $v^*$, then $(x^*,y^*)$ is not a Nash equilibrium (she could do
better by deviating to a min-max strategy).
Otherwise, because $x^*$ is not a min-max strategy,
Bob has a response~$y$ such that Alice's expected payoff would be
strictly less than $v^*$.  Here, Bob could do better by deviating
unilaterally to $y$.
In any case, $(x^*,y^*)$ is not a Nash equilibrium.
\end{proof}

There are several interesting consequences of 
Theorem~\ref{t:minmax} and Proposition~\ref{claim:ne}:
\begin{enumerate}

    \item The set of all Nash equilibria of a two-player zero-sum game
     is convex, as the
      optimal solutions of a linear program form a convex set.

    \item All Nash equilibria $(x,y)$ of a two-player zero-sum game lead to the same value of $x^{\top} A
      y$.  That is, each player receives the same expected payoff
      across all Nash equilibria.

    \item Most importantly, because the proof of
      Theorem~\ref{t:minmax} provides
      a polynomial-time algorithm to compute a min-max pair
      $(x^*, y^*)$, we have a polynomial-time algorithm to compute a
      Nash equilibrium of a two-player zero-sum game.

\end{enumerate}

\begin{corollary}\label{cor:zerosum}
A Nash equilibrium of a two-player zero-sum game can be computed in polynomial time.
\end{corollary}

\subsection{Beyond Zero-Sum Games (Computational Complexity)}\label{ss:beyond}

Can we generalize Corollary~\ref{cor:zerosum} to more general classes
of games?  After all, while
two-player zero-sum games are important---von Neumann was
largely focused on them, with applications ranging from poker to
war---most game-theoretic situations are not purely
zero-sum.\footnote{Games can even have a collaborative aspect, for example if
  you and I want to meet at some intersection in Manhattan.  Our
  strategies are intersections, and either we both get a high payoff
  (if we choose the same strategy) or we both get a low payoff
  (otherwise).}  
For example, what about {\em
  bimatrix games}, in which there are still two players but the game is
not necessarily zero-sum?\footnote{Notice that three-player zero-sum
  games are already more general than bimatrix games---to turn one of
  the latter into one of the former, add a dummy third player with
  only one strategy whose payoff is the negative of the combined payoff of
  the original two players.  Thus the most compelling negative results
  would be for the case of bimatrix games.}
Solar Lectures~4 and~5 are devoted to this question, and provide
evidence that there is no polynomial-time algorithm for computing a Nash
equilibrium (even an approximate one) of a bimatrix game.



\subsection{Who Cares?}\label{ss:whocares}

Before proceeding to our second cool fact about two-player zero-sum
games, let's take a step back and be clear about what we're trying to
accomplish.  
Why do we care about computing equilibria of games, anyway?

\begin{enumerate}

    \item We might want fast algorithms to use in practice.
      The demand for equilibrium computation algorithms is
      significantly less than that for, say, linear programming
      solvers, but the author regularly meets researchers who would
      make good use of better off-the-shelf solvers for computing an
      equilibrium of a game.

\item Perhaps most relevant for this monograph's audience, the study of
  equilibrium computation naturally leads to interesting and new
  complexity theory (e.g., definitions of new complexity classes, such
  as $\PPAD$).  We will see that the most celebrated results in the
  area are quite deep and draw on ideas from all across theoretical
  computer science.

\item Complexity considerations can be used to support or critique the
  practical relevance of an equilibrium concept such as the Nash
  equilibrium.  It is tempting to interpret a polynomial-time
  algorithm for computing an equilibrium as a plausibility argument
  that players can figure one out quickly, and an intractability
  result as evidence that players will not generally reach an
  equilibrium in a reasonable amount of time.

  Of course, the real story is more complex.  First, 
computational intractability is
  not necessarily first on the list of the Nash equilibrium's issues.
  For example, its non-uniqueness in non-zero-sum games already limits
  its predictive power.\footnote{Recall our ``meeting
    in Manhattan'' example---every intersection is a Nash
    equilibrium!}  

Second, it's not particularly helpful to critique a definition without
suggesting an alternative.  Lunar Lecture~5 partially addresses this
issue by discussing two tractable equilibrium concepts, correlated
equilibria and coarse correlated equilibria.

Third, does an arbitrary polynomial-time algorithm, such as one based
on solving a non-trivial linear program, really suggest that
independent play by strategic players will actually converge to an
equilibrium?  Algorithms for linear programming do not resemble how
players typically make decisions in games.  A stronger positive result
would involve a behaviorally plausible distributed algorithm that
players can use to efficiently converge to a Nash equilibrium through
repeated play over time.  We discuss such a result for two-player
zero-sum games next.

\end{enumerate}


\section{Uncoupled Dynamics}\label{s:uncoupled}

In the first half of the lecture, we saw that a Nash equilibrium
of a two-player zero-sum game  can be computed in polynomial time
using linear programming.
%
It would be more compelling, however, to come up with a definition of
a plausible process by which players can learn a Nash equilibrium.
Such a result requires a behavioral model for what players do when not
at equilibrium.  
The goal is then to investigate whether or not the process
converges to a Nash equilibrium (for an appropriate notion of
convergence), and if so, how quickly.


\subsection{The Setup}

{\em Uncoupled dynamics} refers to a class of processes 
with the properties mentioned above.
The idea is that each player initially knows only her own payoffs (and
not those of the other players), \`{a} la the number-in-hand model in
communication complexity.\footnote{If a player knows the game is
  zero-sum and also her own payoff matrix, then she automatically knows the
  other player's payoff matrix. 
  Nonetheless, it is non-trivial and illuminating to investigate the
  convergence properties of general-purpose uncoupled dynamics 
in the zero-sum case, thereby identifying 
an aspiration point
for the analysis of general games.}
The game is then played repeatedly,
with each player picking a strategy in each time step as a function
only of her own payoffs and what transpired in the past.


\begin{mdframed}[style=offset,frametitle={Uncoupled Dynamics
    (Two-Player Version)},nobreak=true]
At each time step $t=1,2,3,\ldots$:
\begin{enumerate}
\item Alice chooses a strategy $x^t$ as a function only of her own
  payoffs and the previously chosen strategies $x^1,\ldots,x^{t-1}$
and $y^1,\ldots,y^{t-1}$.
\item Bob simultaneously chooses a strategy $y^t$ as a function only of his own
  payoffs and the previously chosen strategies $x^1,\ldots,x^{t-1}$
and $y^1,\ldots,y^{t-1}$.
\item Alice learns $y^t$ and Bob learns $x^t$.\tablefootnote{When
    Alice and Bob use mixed strategies, there are two
    natural feedback models, one where each
    player learns the actual mixed strategy chosen by the other
    player, and one where each learns
only a sample (a pure strategy) from the other player's chosen
    distribution.  It's generally easier to prove results in the first
    model, but such proofs usually can be extended with some
    additional work to hold
(with high probability over the strategy realizations) in the second model
as well.\label{foot:feedback}}
\end{enumerate}
\end{mdframed}
Uncoupled dynamics have been studied 
at length in both the game theory and computer science literatures
(often under different names).
Specifying such dynamics boils down to a definition of how Alice and
Bob choose strategies as a function of their payoffs and the joint
history of play.  Let's look at some famous examples.

\subsection{Fictitious Play}


One natural idea is to best respond to the observed behavior of your
opponent. 
\begin{example}[Fictitious Play] 
In \emph{fictitious play}, each player
  assumes that the other player will mix according to the relative
  frequencies of their past actions (i.e., the empirical distribution
  of their past play), and plays a best response.\footnote{In
the first time step, Alice and Bob both choose a default strategy, such
as the     uniform distribution.}
\begin{mdframed}[style=offset,frametitle={Fictitious Play
    (Two-Player Version)}] 
At each time step $t=1,2,3,\ldots$:
\begin{enumerate}
\item Alice chooses a strategy $x^t$ that is a best response against
$\yhat^{t-1} = \tfrac{1}{t-1} \sum_{s=1}^{t-1} y^s$, the past actions of Bob
(breaking ties arbitrarily).

\item Bob simultaneously chooses a strategy $y^t$ that is a best
  response  against
$\xhat^{t-1} = \tfrac{1}{t-1} \sum_{s=1}^{t-1} x^s$, the past actions of Alice
(breaking ties arbitrarily).

\item Alice learns $y^t$ and Bob learns $x^t$.

\end{enumerate}
\end{mdframed}
Note that each
 player picks a pure strategy in each time step (modulo tie-breaking
  in the case of multiple best responses).
One way to interpret fictitious play is to imagine that
each player assumes that the other is using the same mixed strategy
every time step, and estimates this time-invariant mixed strategy with
the empirical distribution of the strategies chosen in the past.
%
\end{example}

Fictitious play has an interesting history:
\begin{enumerate}

\item It was first proposed by G.\ W.\ Brown in 1949 (published in
  1951~\cite{B51}) as a computer algorithm to compute a Nash
  equilibrium of a two-player zero-sum game.  This is not so long
  after the birth of either game theory or computers!

\item In 1951, Julia Robinson (better known for her contributions to
  the resolution of Hilbert's tenth problem about Diophantine
  equations) proved that, in
  two-player zero-sum games, the time-averaged payoffs of the players
  converge to the value of the game~\cite{Rob51}.  
Robinson's proof gives only an
  exponential (in the number of strategies) bound on the number of
  iterations required for convergence.  In 1959, 
  Karlin~\cite{K59} conjectured that a polynomial bound should be
  possible (for two-player zero-sum games).  Fast forward to~2014, and
  \citet{DP14} refuted Karlin's conjecture and proved an exponential
  lower bound for the case of adversarial (and not necessarily
  consistent) tie-breaking.

\item It is still an open question whether or not fictitious play
  converges quickly in two-player zero-sum games
for natural (or even just consistent) tie-breaking
  rules!  The goal here would be to show that 
$\poly(n,1/\epsilon)$ time steps suffice for the time-averaged payoffs
to be within $\epsilon$ of the value of the game (where $n$ is the
total number of rows and columns).

\item The situation for non-zero-sum games was murky until 1964, when
  Lloyd Shapley discovered a $3 \times 3$ game (a non-zero-sum
  variation on rock-paper-scissors) where fictitious play never
  converges to a Nash equilibrium~\cite{S64}.  Shapley's counterexample
foreshadowed future separations between the tractability of zero-sum
and non-zero-sum games.

\end{enumerate}
Next we'll look at a different choice of dynamics with better
convergence properties.



\subsection{Smooth Fictitious Play}

Fictitious play is ``all-or-nothing''---even if two strategies have
almost the same expected payoff against the opponent's empirical
distribution, the slightly worse one is completely ignored in favor of
the slightly better one.  
A more stable approach, and perhaps a more behaviorally plausible one,
is to assume that players randomize, biasing their decision toward the
strategies with the highest expected payoffs (again, against the
empirical distribution of the opponent).  In other words, each player
plays a ``noisy best response'' against the observed play of the other
player.

For example, already in 1957 \citet{Han57} considered dynamics where
each player chooses a strategy with probability proportional to her
expected payoff (against the empirical distribution of the other
player's past play), and proved polynomial convergence to the Nash
equilibrium payoffs in two-player zero-sum games.
Even better convergence properties are possible if poorly performing
strategies are abandoned more aggressively, corresponding to
a ``softmax'' version of fictitious play.


\begin{example}[Smooth Fictitious Play] 
%
  In time $t$ of \emph{smooth fictitious play}, a player (Alice, say)
  computes the empirical distribution
  $\yhat^{t-1} = \sum_{s=1}^{t-1} y^s$ of the other player's past
  play, computes the expected payoff $\pi^t_i$ of each pure strategy
  $i$ under the assumption that Bob plays $\yhat^{t-1}$, and chooses $x^t$
  by playing each strategy with probability proportional to
  $e^{\eta^t \pi^t_i}$.  (When~$t=1$, interpret the $\pi^t_i$'s as~0
  and hence the player chooses the uniform distribution.)
Here $\eta^t$ is a tunable parameter that interpolates between always playing
uniformly at random (when $\eta = 0$) and fictitious play with random
tie-breaking (when $\eta = +\infty$).
The choice $\eta^t \approx \sqrt{t}$ is often the best one for proving
convergence results.
\end{example}

\begin{mdframed}[style=offset,frametitle={Smooth Fictitious Play
    (Two-Player Version)}] 
\textbf{Given: } parameter family $\{ \eta^t \in [0,\infty) \,:\,
t=1,2,3,\ldots\}$.

\vspace{\baselineskip}

At each time step $t=1,2,3,\ldots$:
\begin{enumerate}
\item Alice chooses a strategy $x^t$ by playing each strategy $i$ with
  probability proportional to $e^{\eta^t\pi^t_i}$, where $\pi^t_i$
  denotes the expected payoff of
  strategy $i$ when Bob plays the mixed strategy
$\yhat^{t-1} = \tfrac{1}{t-1} \sum_{s=1}^{t-1} y^s$.

\item Bob simultaneously chooses a strategy $y^t$ by playing each
  strategy $j$ with probability proportional to $e^{\eta^t\pi^t_j}$,
  where $\pi^t_j$ is the expected payoff of strategy $j$ when Alice
  plays the mixed strategy
  $\xhat^{t-1} = \tfrac{1}{t-1} \sum_{s=1}^{t-1} x^s$.

\item Alice learns $y^t$ and Bob learns $x^t$.
\end{enumerate}
\end{mdframed}
Versions of smooth fictitious play have been studied independently in
the game theory literature (beginning with \citet{FL95}) and the
computer science
literature (beginning with~\citet{FS99}).  
It converges extremely quickly. 
\begin{theorem}[Fast Convergence of Smooth Fictitious Play~\cite{FL95,FS99}]\label{t:sfp}
For a zero-sum two-player game with $m$ rows and $n$ columns and a
parameter $\epsilon > 0$, 
after $T=O(\log(n+m)/\epsilon^2)$ time steps of smooth fictitious play
with $\eta^t=\Theta(\sqrt{t})$ for each $t$, the empirical
distributions $\xhat = \tfrac{1}{T} \sum_{t=1}^T x^t$ and
$\yhat = \tfrac{1}{T} \sum_{t=1}^T y^t$ constitute an
$\epsilon$-approximate Nash equilibrium.
\end{theorem}
The $\eps$-approximate Nash equilibrium condition in Theorem~\ref{t:sfp} is
exactly what it sounds like:
neither player can improve their expected payoff by more than
$\epsilon$ via a unilateral deviation (see also Definition
\ref{d:ene}, below).\footnote{Recall our assumption that payoffs have been
scaled to lie in $[-1,1]$.}

There are two steps in the proof of Theorem~\ref{t:sfp}: (i) the noisy
best response in smooth fictitious play is equivalent to the
``Exponential Weights'' algorithm, which has ``vanishing regret'';
and (ii) in a two-player zero-sum game, vanishing-regret guarantees
translate to (approximate) Nash equilibrium convergence.  
The optional Sections \ref{ss:sfp1}--\ref{ss:sfp3} provide more details for the
interested reader.

\subsection{Beyond Zero-Sum Games (Communication Complexity)}\label{ss:implycomm}

Theorem \ref{t:sfp} implies that smooth fictitious play can be used to
define a randomized $O(\log^2 (n+m)/\epsilon^2)$-bit communication
protocol for computing an $\epsilon$-$\NE$ of a two-player zero sum game.\footnote{This
  communication bound applies to
 the variant of smooth fictitious play where Alice (respectively,
  Bob) learns only a random sample from $y^t$ (respectively,
  $x^t$); see footnote~\ref{foot:feedback}.  Each such sample can be  communicated to the other
  player in $\log (n+m)$ bits.  Theorem \ref{t:sfp} continues to hold
  (with  high
  probability over the samples) for this variant of smooth fictitious
  play \cite{FL95,FS99}.\label{foot:uncoupled}}
The goal of Solar Lectures~2 and~3 is to prove that there is no
analogously efficient communication protocol for computing an approximate
Nash equilibrium of a general bimatrix game.\footnote{The communication
  complexity of computing anything about a
  two-player zero-sum game is zero---Alice knows the entire
  game at the beginning (as Bob's payoff is the negative of hers)
  and can unilaterally compute whatever she wants.  But it still makes
  sense to ask if the communication bound implied by smooth
  fictitious play can be replicated in non-zero-games (where Alice and
  Bob initially know only their own payoff matrices).}
Ruling out low-communication protocols will in particular rule out any
type of quickly converging uncoupled dynamics.\footnote{The relevance
  of communication complexity to fast learning in games was
  first pointed out by \citet{CS04}.}

\subsection{Proof of Theorem~\ref{t:sfp}, Part 1: Exponential
  Weights (Optional)}\label{ss:sfp1}

To elaborate on the first step of the proof of Theorem~\ref{t:sfp}, we
need to explain the standard setup
for online decision-making.
\begin{mdframed}[style=offset,frametitle={Online Decision-Making}]
\begin{algorithm}[H]
At each time step $t=1,2,\ldots,T$:\;
\vspace{.25\baselineskip}
\Indp
\hangindent=.5\skiptext\hangafter=1
a decision-maker picks a probability distribution  $\p^t$
over her actions $\Lambda$\;
\vspace{.25\baselineskip}
an adversary picks a reward vector $\r^t:\Lambda \rightarrow
  [-1,1]$\;
\vspace{.25\baselineskip}
\hangindent=.5\skiptext\hangafter=1
an action $a^t$ is chosen according to the distribution $\p^t$,
  and the decision-maker receives reward $r^t(a^t)$\;
\vspace{.25\baselineskip}
\hangindent=.5\skiptext\hangafter=1
the decision-maker learns $\r^t$, the entire reward vector
\end{algorithm}
\end{mdframed}
In smooth fictitious play, each of Alice and Bob are in effect solving
the online decision-making problem (with actions corresponding to the
game's strategies).  For Alice, the reward vector $\r^t$ is induced by
Bob's action at time step $t$ (if Bob plays strategy $j$, then $r^t$
is the $j$th column of the game matrix $A$), and similarly
for Bob (with reward vector equal to the $i$th row multiplied by
$-1$).  Next we interpret
Alice's and Bob's behavior in smooth fictitious play as algorithms for
online decision-making.

An {\em online decision-making algorithm} specifies for each~$t$ the
probability distribution~$\p^t$, as a function of the reward
vectors~$\r^1,\ldots,\r^{t-1}$ and realized
actions~$a^1,\ldots,a^{t-1}$ of the first $t-1$ time steps.
An {\em adversary} for such an algorithm $\Alg$ specifies for each~$t$ the
reward vector~$\r^t$, as a function of
the probability distributions $\p^1,\ldots,\p^t$ used by $\Alg$ on
the first $t$ days and the realized actions $a^1,\ldots,a^{t-1}$ of the first
  $t-1$ days.  

Here is a famous online decision-making algorithm, the
``Exponential Weights (EW)'' algorithm
(see~\cite{LW94,FS97}).\footnote{Also known as the ``Hedge''
  algorithm.  The closely related ``Multiplicative Weights'' algorithm
  uses the update rule $w^{t+1}(a) = w^t(a) \cdot (1 + \eta^t r^t(a))$ 
instead of $w^{t+1}(a) = w^t(a) \cdot (e^{\eta^t r^t(a)})$~\cite{CBMS07}.}
\begin{mdframed}[style=offset,frametitle={Exponential Weights (EW)
    Algorithm}]
\begin{algorithm}[H]
initialize $w^1(a) = 1$ for every $a \in \Lambda$\;
\vspace{.25\baselineskip}
\For{each time step $t=1,2,3,\ldots$}{
\vspace{.25\baselineskip}
\hangindent=.5\skiptext\hangafter=1
use the distribution $\p^t :=
  \w^t/\Gamma^t$ over actions, where $\Gamma^t = \sum_{a \in \Lambda} w^t(a)$
  is the sum of the actions' current weights\;
\vspace{.25\baselineskip}
\hangindent=.5\skiptext\hangafter=1
given the reward vector $\r^t$, 
update the weight of each action $a \in \Lambda$
using the formula
  $w^{t+1}(a) = w^t(a) \cdot (e^{\eta^t r^t(a)})$ 
(where $\eta^t$ is a
  parameter, canonically $\approx \sqrt{t}$)
}
\end{algorithm}
\end{mdframed}
The EW algorithm maintains a
weight, intuitively a ``credibility,'' for each action.  At each
time step the algorithm chooses an action with
probability proportional to its current weight.  The weight of each
action evolves over time according to the action's past performance.

Inspecting the descriptions of smooth fictitious play and the EW
algorithm, we see that we can rephrase the former as follows:
\begin{mdframed}[style=offset,frametitle={Smooth Fictitious Play
    (Rephrased)}] 
\textbf{Given: } parameter family $\{ \eta^t \in [0,\infty) \,:\,
t=1,2,3,\ldots\}$.

\vspace{\baselineskip}

At each time step $t=1,2,3,\ldots$:
\begin{enumerate}
\item Alice uses an instantiation of the EW algorithm to choose a mixed strategy $x^t$.

\item Bob uses a different instantiation of the EW algorithm to choose a mixed strategy $y^t$.

\item Alice learns $y^t$ and Bob learns $x^t$.

\item Alice feeds her EW algorithm a reward vector $r^t$ with $r^t(i)$
  equal to the expected payoff of playing row $i$, given Bob's mixed
  strategy $y^t$ over columns; and similarly for Bob.

\end{enumerate}
\end{mdframed}

How should we assess the performance of an online decision-making
algorithm like the EW algorithm, and do guarantees for the algorithm
have any implications for smooth fictitious play?

\subsection{Proof of Theorem~\ref{t:sfp}, Part 2: Vanishing Regret (Optional)}\label{ss:sfp2}

One of the big ideas in online learning is to compare the 
time-averaged reward earned by an online algorithm with that
earned by the best {\em fixed action} in hindsight.\footnote{There
  is no hope of competing with the best action {\em sequence} in
  hindsight:  consider two
  actions and an
  adversary that flips a coin each time step to choose between the
  reward vectors $(1,0)$ and $(0,1)$.}
\begin{definition}[(Time-Averaged) Regret]\label{d:regreta}
Fix reward vectors~$\r^1,\ldots,\r^T$.
The {\em regret}
of the action sequence
$a^1,\ldots,a^T$
is
\begin{equation}\label{eq:regreta}
\underbrace{\frac{1}{T} \max_{a \in \Lambda} \sum_{t=1}^T r^t(a)}_{\text{best fixed action}}-
 \underbrace{\frac{1}{T} \sum_{t=1}^T r^t(a^t)}_{\text{our algorithm}}.
\end{equation}
\end{definition}
Note that, by linearity, there is no difference between considering
the best fixed action  and the best fixed distribution over actions
(there is always an optimal pure action in hindsight).

We aspire to an online decision-making algorithm that achieves low
regret, as close to 0 as possible.
Because rewards lie in $[-1,1]$, the regret can never be larger than~2.
We think of regret $\Omega(1)$ (as $T \rightarrow \infty$)
as an epic fail for an algorithm.

It turns out that the EW algorithm has the best-possible worst-case
regret guarantee (up to constant factors).\footnote{For the matching
  lower bound, with
  $n$ actions, consider an adversary that sets the reward of each
  action uniformly at random from $\{-1,1\}$ at each time step.  Every
  online algorithm earns expected cumulative reward~0,
  while the expected cumulative reward of the best action in hindsight is
  $\Theta(\sqrt{T} \cdot \sqrt{\log n})$.}

\begin{theorem}[Regret Bound for the EW Algorithm]\label{t:noregret}
For every adversary, the EW algorithm has
expected regret $O(\sqrt{(\log n)/T})$, where $n=|\Lambda|$.
\end{theorem}
See e.g.~the book of \citet{CBL06} for a proof of Theorem
\ref{t:noregret}, which is not overly difficult.

An immediate corollary is that the number of time steps needed
to drive the expected regret down to a small constant is
only logarithmic in the number of actions---this is surprisingly fast!
\begin{corollary}\label{cor:noregret}
There is an online decision-making algorithm that, for every
adversary and $\eps > 0$, has expected regret at most
$\eps$ after $O((\log n)/\eps^2)$ time steps, where $n=|\Lambda|$.
\end{corollary}


\subsection{Proof of Theorem~\ref{t:sfp}, Part 3: Vanishing Regret Implies
  Convergence (Optional)}\label{ss:sfp3}



Consider a zero-sum game $A$ with payoffs in $[-1,1]$
and some $\eps > 0$.
Let $n$ denote the number of rows or the number of columns, whichever
is larger, and set $T = \Theta((\log n)/\eps^2)$ so that the guarantee
in Corollary \ref{cor:noregret} holds with error $\eps/2$.
Let $x^1,\ldots,x^T$ and $y^1,\ldots,y^T$ be the mixed strategies used
by Alice and Bob throughout~$T$ steps of smooth fictitious play.
Let $\hat{\x} = \tfrac{1}{T} \sum_{t=1}^T \x^t$
and
$\hat{\y} = \tfrac{1}{T} \sum_{t=1}^T \y^t$
denote the time-averaged strategies of Alice and Bob, respectively.
We claim that $(\xhat,\yhat)$ is an $\eNE$.

In proof, let 
\[
v = \frac{1}{T} \sum_{t=1}^T (\x^t)^{\top}A\y^t
\]
denote Alice's time-averaged payoff.  
Alice and Bob both used (in effect) the EW algorithm to choose
their strategies, so we can apply the vanishing regret guarantee in
Corollary \ref{cor:noregret} once for each player 
and use linearity to obtain
\begin{equation}\label{eq:noregret1}
v \ge 
\left( \max_{\x} \frac{1}{T} \sum_{t=1}^T \x^{\top}A\y^t \right) - \frac{\eps}{2}
=
\left( \max_{\x} \x^{\top}A\hat{\y} \right) - \frac{\eps}{2}
\end{equation}
and
\begin{equation}\label{eq:noregret2}
v 
\le \left( \min_{\y} \frac{1}{T} \sum_{t=1}^T ({\x}^t)^{\top}A\y
\right) +  \frac{\eps}{2}
= \left( \min_{\y} \hat{\x}^{\top}A\y \right) + \frac{\eps}{2}.
\end{equation}
In particular, taking 
$\x = \xhat$ in~\eqref{eq:noregret1}
and $\y = \yhat$ in~\eqref{eq:noregret2} shows that
\begin{equation}\label{eq:noregret3}
\xhat^{\top}A\yhat \in \left[ v - \frac{\eps}{2},
v + \frac{\eps}{2} \right].
\end{equation}
Now consider a (pure) deviation from $(\xhat,\yhat)$, say by Alice to
the row~$i$.  Denote this deviation by $e_i$.  By 
inequality~\eqref{eq:noregret1} (with $\x = e_i$) we have
\begin{equation}\label{eq:noregret4}
e_i^{\top}A\yhat 
\le v+\frac{\eps}{2}.
\end{equation}
Because Alice receives expected payoff at least $v-\tfrac{\eps}{2}$ in
$(\xhat,\yhat)$ (by~\eqref{eq:noregret3}) and at most
$v+\tfrac{\eps}{2}$ from any deviation (by~\eqref{eq:noregret4}), her
$\eNE$ conditions are
satisfied.  A symmetric argument applies to Bob, completing the proof.

%
%
%

\section{General Bimatrix Games}\label{s:bimatrix}

A general bimatrix game is defined by two independent
payoff matrices, an $m \times n$ matrix $A$ for Alice and an $m \times
n$ matrix $B$ for Bob.
(In a zero-sum game, $B=-A$.)
The definition of an
(approximate) Nash equilibrium is what you'd think it would be:
\begin{definition}[$\eps$-Approximate Nash Equilibrium]\label{d:ene}
For a bimatrix game $(A,B)$, row and column mixed strategies $\hat{x}$
and  $\hat{y}$ constitute an
$\epsilon$-$\NE$ if
\begin{align}
	\hat{x}^{\top} A\hat{y}\ &\geq\  x^{\top} A\hat{y}- \epsilon
                                   \qquad \forall x \,, \text{ and }\\
	\hat{x}^{\top} B\hat{y}\ &\geq\ \hat{x}^{\top} By - \epsilon\qquad \forall y\,.
\end{align}
\end{definition}

It has long been known that many of the nice properties of zero-sum
games break down in general bimatrix games.\footnote{We already mentioned
 Shapley's 1964 example showing that fictitious play need not
  converge \cite{S64}.}
\begin{example}[Strange Bimatrix Behavior]\label{ex:bimatrix}
  Suppose two friends, Alice and Bob, want to go for dinner, and are
  trying to agree on a restaurant. Alice prefers Italian
  over Thai, and Bob prefers Thai over Italian, but both would rather
  eat together than eat alone.\footnote{In older game theory texts,
    this example is called the ``Battle of the Sexes.''}
Supposing the rows and columns are
  indexed by Italian and Thai, in that order, and Alice is the
  row player, we get the following payoff matrices:
	\[
		A=\left[\begin{matrix}2&0\\0&1\end{matrix}\right], \qquad
		B=\left[\begin{matrix}1&0\\0&2\end{matrix}\right], \qquad
		\text{or, in shorthand, }\quad
		(A,B)=\left[\begin{matrix}(2,1)&(0,0)\\(0,0)&(1,2)\end{matrix}\right]\enspace .
	\]
There are two obvious Nash equilibria, both pure: either Alice and Bob
go to the Italian restaurant, or they both go to the Thai restaurant.
But there's a third Nash equilibrium, a mixed one\footnote{Fun fact:
  outside of degenerate cases, every game has an {\em odd} number of
  Nash equilibria (see also Solar Lecture 4).}:
Alice chooses
        Italian over Thai with probability $\tfrac23$, and Bob chooses
        Thai over Italian with probability $\tfrac23$.  This is an
        undesirable Nash equilibrium, with Alice and Bob eating alone
        more than half the time.
\end{example}
Example \ref{ex:bimatrix} shows that, unlike in zero-sum games,
different Nash equilibria can result in different expected player
payoffs.  Similarly, the Nash equilibria of a bimatrix game do not
generally form a convex set (unlike in the zero-sum case).

Nash equilibria of bimatrix games are not completely devoid of nice
properties, however.  For starters, we have guaranteed
existence.
\begin{theorem}[Nash's Theorem~\cite{Nas51,Nas50}]\label{t:nash} 
  Every bimatrix game has at least one (mixed) Nash equilibrium. 
\end{theorem}
The proof is a fixed-point argument that we will have more to say
about in Solar Lecture~2.\footnote{Von Neumann's alleged reaction when
  Nash 
  told him his theorem \cite[P.94]{nasar}: ``That's trivial, you
  know. That's  just a
  fixed point theorem.''}  
Nash's theorem holds more generally for
games with any finite number of players and strategies.

Nash equilibria of bimatrix games have nicer structure than
those in games with three or more players.  First, in bimatrix games
with 
integer payoffs, there is a Nash equilibrium in which all
probabilities are 
rational numbers with bit complexity polynomial in that of the
game.\footnote{Exercise: prove this by showing that, after you've
  guessed the two support sets of a Nash equilibrium, you can recover
  the exact probabilities using two linear programs.}  Second, there is a
simplex-type pivoting algorithm, called the {\em Lemke-Howson
  algorithm} \cite{LH64}, which computes a Nash equilibrium of a bimatrix
game in a finite number of steps (see~\citet{vS07} for a survey).  
Like the simplex method, the
Lemke-Howson algorithm takes an exponential number of steps in the
worst case \cite{M94,SvS04}.
The similarities between Nash equilibria of bimatrix
games and optimal solutions of linear programs 
initially led to some optimism that computing the former might be as
easy as computing the latter (i.e., might be a polynomial-time
solvable problem).
Alas, as we'll
see, this does not seem to be the case.


\section{Approximate Nash Equilibria in Bimatrix Games}

The last topic of this lecture is
some semi-positive results about {\em approximate}
Nash equilibria in general bimatrix games.  While simple, these
results are important and will show up repeatedly in the rest of the
lectures.

\subsection{Sparse Approximate Nash Equilibria}

Here is a crucial result for us: there are always {\em sparse}
approximate Nash equilibria.\footnote{\citet{A94} and \citet{LY94}
  independently proved a precursor to this result in the special case
  of zero-sum games.  The focus of the latter paper is applications in
  complexity theory (like ``anticheckers'').}\footnote{Exercise: there
  are arbitrarily large games where every exact Nash equilibrium has
  full support.  Hint: generalize rock-paper-scissors.  Alternatively,
  see Section~\ref{ss:althofer} of Solar Lecture~5.}
\begin{theorem}[Existence of Sparse Approximate Nash Equilibria (\citet{LMM03})]\label{t:lmm}
For every $\eps > 0$ and
every $n\times n$ bimatrix game,
there exists an $\epsilon$-$\NE$ in which each
player randomizes uniformly 
over a multi-set of $O((\log n)/\epsilon^2)$ pure
strategies.\footnote{By a padding argument, there is no loss of
  generality in assuming that Alice and Bob have the same number of strategies.}
\end{theorem}

\begin{proof}[Proof idea.] 
Fix an $n \times n$ bimatrix game $(A,B)$.
\begin{enumerate}

\item Let $(x^*,y^*)$ be an exact Nash equilibrium of $(A,B)$.  (One
  exists, by Theorem \ref{t:nash}.)

\item As a thought experiment, sample $\Theta((\log n)/\eps^2)$ pure
  strategies for Alice i.i.d.\ (with replacement) from~$x^*$, and
  similarly for Bob i.i.d.\ from $y^*$.  

\item Let $\xhat,\yhat$ denote the empirical distributions of the
  samples (with probabilities equal to frequencies in the
  sample)---equivalently, the uniform distributions over the two
  multi-sets of pure strategies.

\item Use Chernoff bounds to argue that $(\xhat,\yhat)$ is an $\eNE$
  (with high probability).  Specifically, because of our choice of the
  number of samples,
the expected payoff of each
  row strategy w.r.t.\ $\yhat$ differs from that w.r.t.\ $y^*$ by at
  most $\eps/2$ (w.h.p.).  Because every strategy played with non-zero
  probability in $x^*$ is an exact best response to $y^*$, every
  strategy played with non-zero probability in $\xhat$ is within $\eps$
  of a best response to $\yhat$.  (The same argument applies with the roles of
  $\xhat$ and $\yhat$ reversed.)  This is a sufficient condition for
  being an $\eNE$.\footnote{This sufficient condition has its own name: a {\em
    well-supported $\eNE$.}}

\end{enumerate}
\end{proof}

\subsection{Implications for Communication Complexity}

Theorem~\ref{t:lmm} immediately implies the existence of an $\eNE$ of an
$n \times n$ bimatrix game with description length
$O((\log^2 n)/\eps^2)$, with $\approx \log n$ bits used to describe each of the
$O((\log n)/\eps^2)$ pure strategies in the multi-sets promised by the
theorem.
Moreover, if an all-powerful prover writes down an alleged such
description on a publicly observable blackboard, then Alice and Bob
can privately verify that the described pair of mixed strategies is
indeed an $\eNE$.  For example, Alice can use the (publicly viewable)
description of Bob's mixed strategy to compute the expected payoff of 
her best response and check that it is at most $\eps$ more than her
expected payoff when playing the mixed strategy suggested by the
prover.  Summarizing:
\begin{corollary}[Polylogarithmic Nondeterministic
  Communication Complexity]\label{cor:lmm1}
The nondeterministic communication complexity of computing an $\eNE$ of
an $n \times n$ bimatrix game is $O((\log^2 n)/\eps^2)$.
\end{corollary}

Thus, if there {\em is} a polynomial lower bound on the deterministic
or randomized communication complexity of computing an approximate
Nash equilibrium, the only way to prove it is via techniques that
don't automatically apply also to the problem's nondeterministic
communication complexity.  This observation rules out many of the most
common lower bound techniques.  In Solar Lectures~2 and~3, we'll see how
to thread the needle using a {\em simulation theorem}, which lifts
a deterministic or random query (i.e., decision tree) lower bound to
an analogous communication complexity lower bound.

\subsection{Implications for Computational Complexity}


The second important consequence of Theorem~\ref{t:lmm} is a limit on
the strongest-possible computational hardness we could hope to prove for
the problem of
computing an approximate Nash equilibrium of a bimatrix game: at
worst, the problem is quasi-polynomial-hard.

\begin{corollary}[Quasi-Polynomial Computational Complexity]\label{cor:lmm2}
There is an algorithm that, given as input a description of an $n
\times n$ bimatrix game and a parameter $\eps$, outputs an $\eNE$ in
$n^{O((\log n)/\eps^2)}$ time.
\end{corollary}

\begin{proof}
The algorithm enumerates all $n^{O((\log n)/\eps^2)}$ possible choices
for the multi-sets promised by Theorem~\ref{t:lmm}.  It is easy to
check whether or not the mixed strategies induced by such a choice
constitute an $\eNE$---just compute the expected payoffs of each
strategy and of the players' best responses, as in the proof of Corollary~\ref{cor:lmm1}.
\end{proof}

Because of the apparent paucity of natural problems with
quasi-polynomial complexity,
the quasi-polynomial-time approximation scheme (QPTAS) in
Corollary~\ref{cor:lmm2} initially led to optimism that there should
be a PTAS for the problem.
Also, if there {\em were} a reduction showing quasi-polynomial-time
hardness for 
computing an approximate Nash equilibrium, what would be the
appropriate complexity assumption, and what would the reduction look
like?  Solar Lectures~4 and~5 answer this question.


\lecture{Communication Complexity Lower Bound for Computing an
  Approximate Nash Equilibrium of a Bimatrix Game (Part I)}

\vspace{1cm}


This lecture and the next consider the communication complexity of
computing an approximate Nash equilibrium, culminating with a proof of
the recent breakthrough polynomial lower bound of \citet{BR17}.
This lower bound rules out the possibility of quickly converging
uncoupled dynamics in general bimatrix games (see
Section~\ref{s:uncoupled}).


\section{Preamble}\label{s:ccpreamble}

Recall the setup: there are two players, Alice and Bob, each with
their own payoff matrices $A$ and $B$.  Without loss of generality (by
padding), the
two players have the same number~$N$ of strategies.  
We consider a two-party model where,
initially, Alice knows only~$A$ and Bob knows only~$B$.
The goal is then for Alice and Bob to
compute an approximate Nash equilibrium (Definition~\ref{d:ene}) 
with as little communication as possible.

This lecture and the next explain all of the main ideas behind the
following result:
\begin{theorem}[\citet{BR17}]\label{t:br17}
There is a constant $c > 0$ such that, for all sufficiently small
constants $\eps > 0$ and sufficiently large~$N$, the randomized
communication complexity of computing an $\eNE$ is $\Omega(N^c)$.\footnote{This $\Omega(N^c)$ lower bound was recently
  improved to $\Omega(N^{2-o(1)})$ by G\"o\"os and
  Rubinstein~\cite{GR18} (for constant $\eps > 0$ and $N
  \rightarrow \infty$).  The proof follows the same
high-level road
  map 
 used here (see Section~\ref{s:map}), with a number of additional optimizations.} 
\end{theorem}
For our purposes, a randomized protocol with communication cost~$b$
always uses at most $b$ bits of communication, and 
terminates with
at least one player knowing an $\eNE$ of the game 
with probability at
least $\tfrac{1}{2}$ (over the protocol's coin flips).




Thus, while there are lots of obstacles to players reaching an
equilibrium of a game (see also Section~\ref{ss:whocares}),
communication alone is already a significant bottleneck.
A corollary of Theorem~\ref{t:br17} is that there can be no uncoupled
dynamics (Section~\ref{s:uncoupled}) that converge to an approximate
Nash equilibrium in a sub-polynomial number of rounds 
in general bimatrix games (cf.,
the guarantee in Theorem~\ref{t:sfp} for smooth fictitious play in
zero-sum games).  This is because uncoupled dynamics can be simulated
by a randomized communication protocol with logarithmic overhead (to
communicate which strategy gets played each round).\footnote{See also
  footnote~\ref{foot:uncoupled} in Solar Lecture~1.}
This corollary should be regarded as a fundamental contribution to
pure game theory and economics.

The goal of this and the next lecture is to sketch a full proof of the
lower bound in Theorem~\ref{t:br17} for deterministic communication
protocols.
We do really care about randomized protocols, however, as these are
the types of protocols induced by uncoupled dynamics (see
Section~\ref{ss:implycomm}).  The good news is that the argument for the
deterministic case will already showcase all of the 
conceptual ideas in the proof of Theorem~\ref{t:br17}.  Extending the proof to
randomized protocols requires substituting a simulation theorem for
randomized protocols (we'll use only a simulation theorem for
deterministic protocols, see
Theorem~\ref{t:rm}) and a few other minor
tweaks.\footnote{When \citet{BR17} first proved their result (in late
  2016), the state-of-the-art in simultaneous theorems for randomized
  protocols was   much more primitive than for deterministic protocols.  This
  forced \citet{BR17} to use a relatively weak simulation theorem for
  the randomized case
  (by \citet{G+16}), which led to a number of additional technical
  details in the proof.  Amazingly, a full-blown randomized simulation
  theorem  was published shortly thereafter~\cite{A+17,GPW17}!
With this in hand, 
extending the argument here for deterministic protocols to randomized
protocols is relatively straightforward.}

\section{Naive Approach: Reduction From \disj}\label{s:naive}

To illustrate the difficulty of proving a result like
Theorem~\ref{t:br17}, consider a
naive attempt that tries to reduce, say, the \disj problem to
the problem of computing an $\epsilon$-$\NE$, with
YES-instances mapped to games in which
all equilibria have some property $\Pi$, and NO-instances mapped
to games in which no equilibrium has property~$\Pi$
(Figure~\ref{f:naive}).\footnote{Recall the \disj function: Alice and
  Bob have input strings $a,b \in \{0,1\}^n$, and the output of the
  function is ``0'' if there is a coordinate $i \in \{1,2,\ldots,n\}$
  with $a_i = b_i = 1$ and ``1'' otherwise.  One of the first things
  you learn in communication complexity is that the nondeterministic
  communication complexity of \disj (for certifying 1-inputs) is
  $n$ (see e.g.~\cite{KN96,w15}).  And of course one of the
  most famous and useful results in communication complexity is that
  the function's randomized communication complexity (with two-sided
  error) is $\Omega(n)$~\cite{KS92,R92}.}
For the reduction to be useful, $\Pi$ needs to be some property that
can be checked with little to no communication, such as ``Alice plays
her first strategy with positive probability'' or ``Bob's strategy has
full support.''
The only problem is that {\em this is impossible!}
The reason is that the problem of computing an approximate Nash
equilibrium has polylogarithmic {\em nondeterministic}
communication  complexity
(because of the existence of sparse approximate equilibria, see
Theorem~\ref{t:lmm} and Corollary~\ref{cor:lmm1}), while the \disj
function does not (for 1-inputs).  A reduction of the proposed form
would translate a nondeterministic lower bound for the latter problem
to one for the former, and hence cannot exist.\footnote{Mika G\"o\"os
  (personal communication, January 2018) points out that there are more
  clever reductions from \disj, starting with Raz and
  Wigderson~\cite{RW90}, that {\em can} imply strong lower bounds on the
  randomized communication complexity of certain problems with low
  nondeterministic communication complexity;
and that it is plausible that a Raz-Wigderson-style proof, such as that for
search problems in G\"o\"os and Pitassi~\cite{GP18}, could be adapted
to give an alternative proof of Theorem~\ref{t:br17}.}

\begin{figure}
\centering
\includegraphics[width=.6\textwidth]{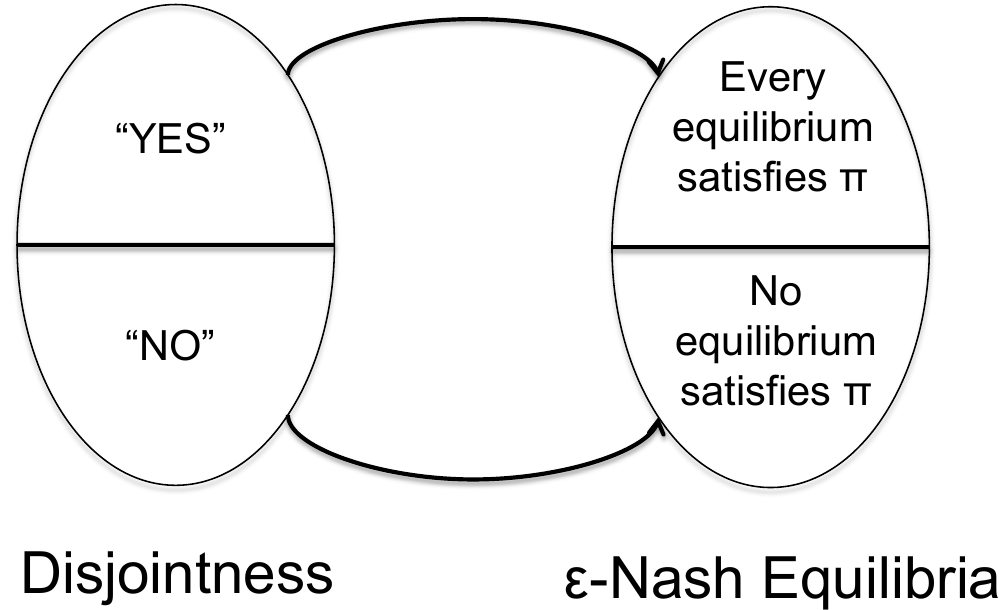}
\caption{A naive attempt to reduce the \disj problem to the
  problem of computing an approximate Nash equilibrium.}
\label{f:naive}
\end{figure}

Our failed reduction highlights two different challenges.  The first
is to resolve the typechecking error that we encountered between a
standard decision problem, where there might or might not be a
witness (like \disj, where a witness is an element in the intersection),
and a total search problem where there is always a witness
(like computing an approximate Nash equilibrium, which is guaranteed
to exist by Nash's theorem).
The second challenge is to figure out how to prove a strong lower
bound on the
deterministic or randomized communication complexity of computing an
approximate Nash equilibrium without inadvertently proving the same
(non-existent) lower bound for nondeterministic protocols.  To
resolve the second challenge, we'll make use of simulation theorems
that lift query complexity lower bounds to communication complexity
lower bounds (see Section~\ref{s:cceol}); these are tailored to a
specific computational model, like deterministic or randomized
protocols.  For the first challenge, we need to identify a total
search problem with high communication complexity.  That is, for total
search problems, which should be the analog of {\sc 3SAT} or \disj?  
The correct answer turns out to be {\em fixed-point computation}.


\section{Finding Brouwer Fixed Points (The \bfp Problem)}\label{s:bfp}

This section and the next describe reductions from computing Nash
equilibria to computing fixed points, and from computing fixed points
to a path-following problem.  These reductions are classical.  The
ingredients of the proof in Theorem~\ref{t:br17} are {\em reductions in
  the opposite direction}; these are discussed in Solar Lecture~3.

\subsection{Brouwer's Fixed-Point Theorem}

{\em Brouwer's fixed-point theorem} states that whenever you stir your 
coffee, there will be a point that ends up exactly where it began.
Or if you prefer a more formal statement:


\begin{theorem}[Brouwer's Fixed-Point Theorem (1910)]\label{t:bfp}
If $C$ is a compact convex subset of~$\R^d$, and $f\colon C\to C$ is
continuous, then 
there exists a {\em fixed point}: a point $x\in C$ with $f(x)=x$.
\end{theorem}
All of the hypotheses are necessary.\footnote{If convexity is dropped,
  consider rotating an annulus centered at the origin.  If boundedness is
  dropped, consider $x \mapsto x+1$ on $\R$.  If closedness is dropped,
  consider $x \mapsto \tfrac{x}{2}$ on $(0,1]$.  If continuity is
  dropped, consider $x \mapsto (x+ \tfrac{1}{2}) \bmod 1$ on $[0,1]$.
Many more general fixed-point theorems are known, and find
applications in economics and elsewhere; see e.g.~\cite{B85,M15}.}
We will be interested in a computational version of Brouwer's
fixed-point theorem, the {\em \bfp problem}: 
\begin{mdframed}[style=offset,frametitle={The \bfp Problem (Generic Version)}]
given a description of a compact
convex set $C \subseteq \R^d$ and a continuous function $f:C \rightarrow C$,
output an {\em $\eps$-approximate fixed point}, meaning a point $x \in
C$ such that $\n{f(x)-x} < \eps$.
\end{mdframed}
The \bfp problem, in its many different forms, plays a starring
role in the study of equilibrium computation.
The set $C$ is typically fixed in advance, for example to the
$d$-dimensional hypercube.  
While much of the work on the \bfp problem has
focused on the $\ell_{\infty}$ norm (e.g.~\cite{HPV89}), one
innovation in the proof of Theorem~\ref{t:br17} is to instead use a
normalized version of the~$\ell_2$ norm (following \citet{R16}).  

Nailing down the problem precisely
requires committing to a family of succinctly described continuous
functions~$f$.  The description of the family used in the proof of
Theorem~\ref{t:br17} is technical and best left to
Section~\ref{s:ccbfp}.  
Often (and in these lectures),
the family of functions considered contains only $O(1)$-Lipschitz
functions.\footnote{Recall that a function $f$ mapping a metric
  space $(X,d)$ to itself is {\em $\lambda$-Lipschitz} if
  $d(f(x),f(y)) \le \lambda \cdot d(x,y)$ for all $x,y \in X$.  That
  is, the function can only amplify distances between points by a
  $\lambda$ factor.}  In particular, this guarantees the existence of
an $\eps$-approximate fixed point with description length polynomial in the
dimension and $\log \tfrac{1}{\eps}$ (by rounding an
exact fixed point to its nearest neighbor on a suitably defined grid).



\subsection{From Brouwer to Nash}\label{ss:nashpf}

Fixed-point theorems have long been used to prove
equilibrium existence results, including the original proofs of
the Minimax theorem (Theorem~\ref{t:minmax}) and Nash's theorem
(Theorem~\ref{t:nash}).\footnote{In fact, the story behind von
  Neumann's original proof of
  the Minimax theorem is a little more complicated and nuanced; see
  \citet{K01} for a fascinating and detailed discussion.}
Analogously, algorithms for computing
(approximate) fixed points can be used to compute (approximate) Nash
equilibria.

\begin{fact}
Existence/computation of $\epsilon$-$\NE$ reduces to that of
$\epsilon$-BFP.
\end{fact}

To provide further details, let's sketch why
Nash's theorem (Theorem~\ref{t:nash}) reduces to Brouwer's
fixed-point theorem (Theorem~\ref{t:bfp}), following the version of
the argument in \citet{G03}.\footnote{This discussion is
  borrowed from \cite[Lecture 20]{f13}.}
Consider a bimatrix game $(A,B)$ and let $S_1,S_2$ denote the strategy
sets of Alice and Bob (i.e., the rows and columns).
The relevant convex compact
set is $C = \Delta_1 \times \Delta_2$, where $\Delta_i$
is the simplex representing the mixed strategies over~$S_i$.  We want
to define a continuous function $f:C \rightarrow C$, from mixed
strategy profiles to mixed strategy profiles, such that
the fixed points of~$f$ are the Nash equilibria of this game.  
We define $f$ separately for each component $f_i:C \rightarrow
\Delta_i$ for $i=1,2$.
A natural idea is to set~$f_i$ to be a best response of player~$i$ to
the mixed strategy of the other player.  This does
not lead to a continuous, or even well defined, function.
We can instead use a ``regularized'' version of this idea, defining
\begin{align}\label{eq:nash1}
f_1(\mixed_1,\mixed_2) = 
\underset{\mixed'_1 \in \Delta_1}{\argmax} \,\,
g_1(\mixed'_1,x_2),
\end{align}
where
\begin{align}\label{eq:nash2}
g_1(\mixed'_1,x_2) = 
\underbrace{%
(x'_1)^{\top}Ax_2}_{\text{linear in $\mixed'_1$}}
- \underbrace{\|\mixed'_1 - \mixed_1 \|^2_2}_{\text{strictly convex}},
\end{align}
and similarly for $f_2$ and $g_2$ (with Bob's payoff matrix $B$).
The first term of the function $g_i$ encourages a best response while
the second ``penalty term'' discourages big changes to player $i$'s mixed
strategy.  Because the function $g_i$ is strictly concave in
$\mixed'_i$, $f_i$ is well defined.  The function $f=(f_1,f_2)$ is
continuous (as you should check).  By definition, every Nash
equilibrium of the given game is a fixed point of $f$.  For the
converse, suppose that $(x_1,x_2)$ is not a Nash equilibrium, with
Alice (say) able to
increase her expected payoff by deviating unilaterally from $\mixed_1$
to $\mixed'_1$.  A simple computation shows that, for sufficiently
small $\eps > 0$,
$g_1((1-\eps)\mixed_1 + \eps \mixed'_1,\mixed_2) >
g_1(\mixed_1,\mixed_2)$,
and hence $(x_1,x_2)$ is not a fixed point of $f$ (as you should check).

Summarizing, an oracle for computing a Brouwer fixed point immediately
gives an oracle for computing a Nash equilibrium of a bimatrix game.
The same argument applies to games with any (finite) number of
players.  The same argument also shows that an oracle for computing an
$\eps$-approximate fixed point in the~$\ell_{\infty}$ norm can be used
to compute an $O(\eps)$-approximate Nash equilibrium of a
game.
The first high-level goal of the proof of
Theorem~\ref{t:br17} is to reverse the direction of the reduction---to
show that the problem of computing an approximate Nash equilibrium is
as general as computing an approximate fixed point, rather than merely
being a special case.
\begin{mdframed}[style=offset,frametitle={Goal \#1}]
\centering
\bfp $\leq$ $\epsilon$-$\NE$
\end{mdframed}
This goal follows in the tradition of a sequence of celebrated
computational hardness results last decade for computing an exact Nash
equilibrium (or an $\eps$-approximate Nash equilibrium with $\eps$ 
polynomial in $\tfrac{1}{n}$)~\cite{DGP09,CDT09}.

There are a couple of immediate issues.  First, it's not clear how to
meaningfully define the \bfp problem in a two-party communication
model---what are Alice's and Bob's inputs?  We'll address this issue
in Section~\ref{s:ccbfp}.  Second, even if we figure out how to define
the \bfp problem and implement goal \#1, so that the $\eNE$ problem is
at least as hard as the \bfp problem, what makes us so sure that the
latter is hard?  This brings us to our next topic---a
``generic'' total search problem that is hard almost by definition and
can be used to transfer hardness to other problems (like \bfp) via
reductions.\footnote{For an analogy, a ``generic'' hard decision problem for
  the complexity class $\NP$ is: given a description of a
  polynomial-time verifier, does there exist a witness (i.e., an
  input accepted by   the verifier)?}


\section{The End-of-the-Line (\eol) Problem}\label{s:eol}

\subsection{Problem Definition}

For equilibrium and fixed-point computation problems, it turns out
that the appropriate ``generic'' problem involves following a path in
a large graph; see also Figure~\ref{f:ppad}.

\begin{figure}
\centering
\includegraphics[width=.85\textwidth]{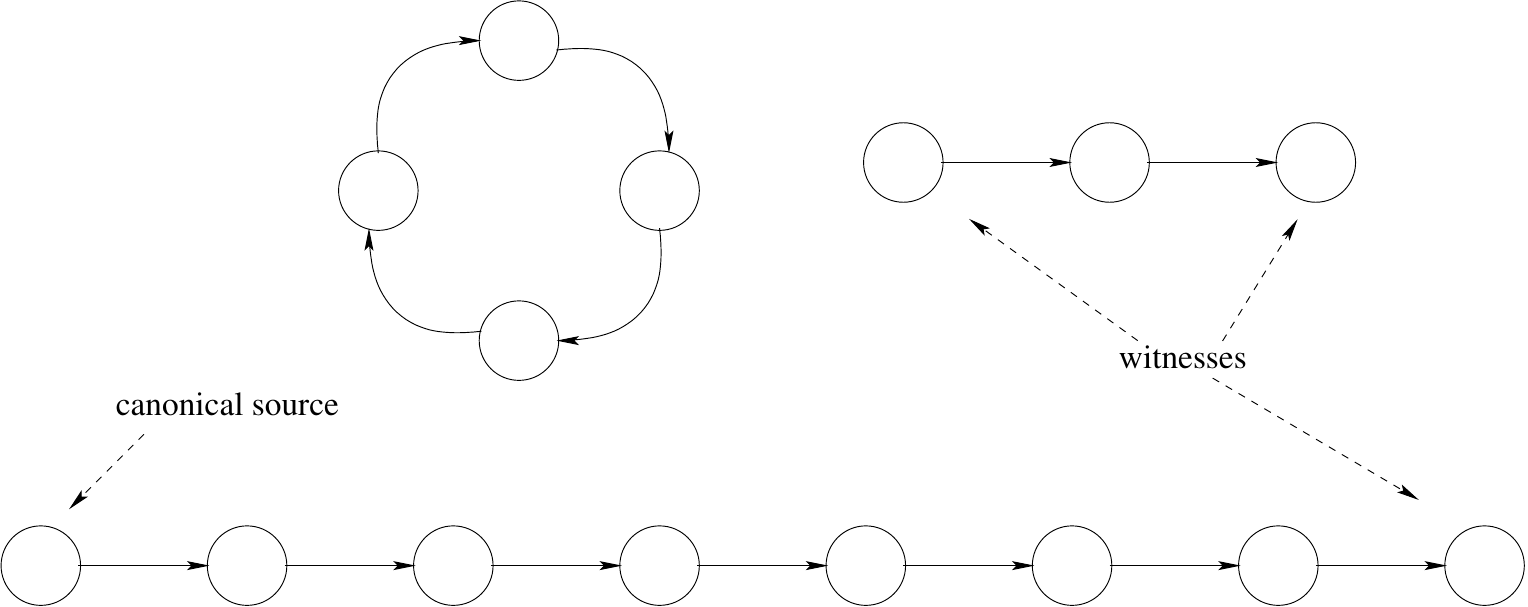}
\caption{An instance of the \eol problem
  corresponds to a directed graph
  with all in- and out-degrees at most 1.  Solutions correspond to
  sink vertices and source vertices other than the given one.}\label{f:ppad}
\end{figure}

\begin{mdframed}[style=offset,frametitle={The \eol Problem (Generic Version)}]
given a description of a directed graph~$G$ with maximum in- and
out-degree~1, and a source vertex $s$ of $G$, find either a sink
vertex of $G$ or a source vertex other than~$s$.
\end{mdframed}
The restriction on the in- and out-degrees forces the graph $G$ to consist
of vertex-disjoint paths and cycles, with at least one path (starting
at the source $s$).  The \eol problem is a total search
problem---there is always a solution, if nothing else the other end of
the path that starts at~$s$.  Thus an instance of \eol can always be
solved by rotely following the path from $s$; the question is whether
or not there is a more clever algorithm that always avoids searching
the entire graph.

It should be plausible that the \eol problem is hard, in the sense
that there is no algorithm that always improves over rote
path-following; see also Section~\ref{s:eollb}.
But what does it have to do with the \bfp problem?  A
lot, it turns out.


\begin{fact}
The problem of computing an approximate Brouwer fixed point reduces to
the \eol problem (i.e., $\epsilon$-BFP $\leq$ EoL).
\end{fact}

\subsection{From \eol to Sperner's Lemma}

The basic reason that fixed-point computation reduces to
path-following is {\em Sperner's lemma}, which we recall next (again
borrowing from~\cite[Lecture 20]{f13}).
Consider a subdivided triangle in the plane
(Figure~\ref{f:sperner}).  A {\em legal coloring} of its vertices colors
the top corner vertex red, the left corner vertex green, and the right
corner vertex blue.  A vertex on the boundary must have one of the two
colors of the endpoints of its side.  Internal vertices are allowed to
possess any
of the three colors.  
A small triangle is {\em trichromatic}
if all three colors are represented at its vertices.
\begin{figure}
\centering
\includegraphics[width=.3\textwidth]{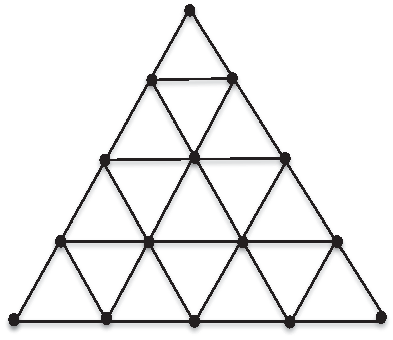}
\caption[Sperner's lemma]{A subdivided triangle in the plane.}
\label{f:sperner}
\end{figure}
Sperner's lemma then asserts that for every legal coloring, there is 
at least one trichromatic triangle.\footnote{The same result and
    proof extend by  induction to higher 
  dimensions.  Every subdivided simplex in $\R^n$ with vertices
  legally colored with $n+1$ colors has an odd number of panchromatic
  subsimplices, with a different color at each vertex.\label{foot:sperner}}
\begin{theorem}[Sperner's Lemma~\cite{sperner}]\label{t:sperner}
For every legal coloring of a subdivided triangle, there is an odd
number of trichromatic triangles.
\end{theorem}

\begin{proof}
The proof is constructive.  
Define an undirected graph $G$ that has one vertex corresponding to each
small triangle, plus a 
source vertex that corresponds to the region
outside the big triangle.  The graph~$G$ has
one edge for each pair of small triangles that share a side with
one red and one green endpoint.
Every trichromatic small triangle corresponds to a
degree-one vertex of $G$.  Every small triangle with one green and two red
corners or two green and one red corners corresponds to a vertex with
degree two in $G$.  The source vertex of $G$ has degree equal to the
number of red-green segments on the left side of the big triangle,
which is an
odd number.  Because every undirected graph has an even number of
vertices with odd degree, there is an odd number of trichromatic
triangles.
\end{proof}

The proof of Sperner's lemma shows that following a path
from a canonical source vertex in a suitable graph leads to a trichromatic
triangle.  Thus, computing a trichromatic triangle of a legally
colored subdivided triangle reduces to the \eol
problem.%
%
%
%
\footnote{We're glossing over some details.  The graph in an
  instance of \eol is directed, while the graph $G$ defined in the
  proof of Theorem~\ref{t:sperner} is
  undirected.  There is, however, a canonical way to direct the edges of the
  graph~$G$.  
Also, the canonical source vertex in an \eol instance has out-degree
1, while the source of the graph $G$ has degree $2k-1$ for some
positive integer~$k$.  This can be rectified by splitting the source
vertex of $G$ into $k$ vertices, a 
source vertex with out-degree 1 and $k-1$ vertices with in- and
out-degree 1.}

\subsection{From Sperner to Brouwer}

Next we'll use Sperner's lemma to prove Brouwer's fixed-point theorem for a
2-dimensional simplex $\Delta$; higher-dimensional versions of
Sperner's lemma (see footnote~\ref{foot:sperner}) similarly imply
Brouwer's fixed-point theorem for simplices of arbitrary
dimension.\footnote{Every compact convex subset of
finite-dimensional Euclidean space is homeomorphic to a simplex of the
same dimension (by scaling and radial projection, essentially), and
homeomorphisms preserve fixed points, so
Brouwer's fixed-point theorem carries over from simplices to all
compact convex subsets of Euclidean space.\label{foot:radial}}
Let $f:\Delta \rightarrow \Delta$ be
a $\lambda$-Lipschitz function (with respect to the $\ell_2$ norm, say).
\begin{enumerate}

\item Subdivide $\Delta$ into sub-triangles with side length at most $\eps/\lambda$.  Think of the points of
  $\Delta$ as parameterized by three coordinates $(x,y,z)$, with
  $x,y,z \ge 0$ and $x+y+z=1$.

\item Associate each of the three coordinates with a distinct color.
%
To color a point $(x,y,z)$, consider its image $(x',y',z')$ under $f$
and choose the color of a coordinate that strictly decreased (if there
are none, then $(x,y,z)$ is a fixed point and we're done).
Note that the conditions of Sperner's lemma are satisfied.

\item We claim that the center $(\xbar,\ybar,\zbar)$ of a trichromatic
  triangle must be an $O(\epsilon)$-fixed point (in the
  $\ell_{\infty}$ norm).  Because some corner of the triangle has its
  $x$-coordinate go down under $f$, $(\xbar,\ybar,\zbar)$ is at
  distance at most $\eps/\lambda$ from this corner, and $f$ is
  $\lambda$-Lipschitz, the
  $x$-coordinate of $f(\xbar,\ybar,\zbar)$ is at most $\xbar+O(\eps)$.
  The same argument applies to $\ybar$ and $\zbar$, which implies that
  each of the coordinates of $f(\xbar,\ybar,\zbar)$ is within
  $\pm O(\eps)$ of the corresponding coordinate of
  $(\xbar,\ybar,\zbar)$.

\end{enumerate}
Brouwer's fixed-point theorem now follows by taking the limit
$\eps \rightarrow 0$ and using the continuity of~$f$.



The second high-level goal of the proof of
Theorem~\ref{t:br17} is to reverse the direction of the above reduction
from \bfp to \eol.  That is, we would like to show that
the problem of computing an approximate Brouwer fixed point is 
as general as every path-following problem (of the form in \eol),
rather than merely being a special case.
\begin{mdframed}[style=offset,frametitle={Goal \#2}]
\centering
\eol $\leq$ \bfp
\end{mdframed}
If we succeed in implementing goals \#1 and \#2, and also prove
directly that the \eol problem is hard, then we'll have proven
hardness for the problem of computing an approximate Nash equilibrium.

\section{Road Map for the Proof of Theorem~\ref{t:br17}}\label{s:map}

The high-level plan for the proof in the rest of this and the next
lecture is to show that
\[
\text{a low-cost communication protocol for $\eNE$}
\]
implies
\[
\text{a low-cost communication protocol for \ccbfp},
\]
where \ccbfp is a two-party version of the problem of computing a
fixed point (to be defined), which then implies
\[
\text{a low-cost communication protocol for \cceol},
\]
where \cceol is a two-party version of the \eol
problem (to be defined), which then implies
\[
\text{a low-query algorithm for \eol}.
\]
Finally, we'll prove directly that the \eol problem does not admit a
low-query algorithm.  This gives us four things to prove (hardness of
\eol and the three implications); we'll tackle them one by one in
reverse order:
\begin{mdframed}[style=offset,frametitle={Road Map}]
\begin{itemize}

\item [] {\bf Step 1:} Query lower bound for \eol.

\item [] {\bf Step 2:} Communication complexity lower bound for \cceol
  via a simulation theorem.

\item [] {\bf Step 3:} \cceol reduces to \ccbfp.

\item [] {\bf Step 4:} \ccbfp reduces to $\eNE$.

\end{itemize}
\end{mdframed}
The first step (Section~\ref{s:eollb}) is easy.  The
second step (Section~\ref{s:cceol}) follows directly from one of the
simulation theorems alluded to in Section~\ref{s:ccpreamble}.  The
last two steps, which correspond to goals \#2 and \#1, respectively,
are harder and deferred to Solar Lecture~3.

Most of the ingredients in this road map were already present in a
paper by Roughgarden and Weinstein~\cite{RW16}, which was the
first paper to define and study two-party versions of fixed-point
computation problems, and to propose the use of simulation theorems
in the context of equilibrium computation.
One major innovation in \citet{BR17} is the use of the generic \eol
problem as the base of the reduction, thereby eluding the tricky
interactions in~\cite{RW16} between simulation theorems (which seem
inherently combinatorial) and fixed-point problems (which seem
inherently geometric).
\citet{RW16} applied a simulation
theorem directly to a fixed-point problem (relying on strong query
complexity lower bounds for finding fixed points~\cite{HPV89,B16}),
which yielded a hard but unwieldy version of a two-party fixed-point
problem.  It is not clear how to reduce this version to the problem of
computing an approximate Nash equilibrium.  \citet{BR17} instead apply
a simulation theorem directly to the \eol problem, which results in a
reasonably natural two-party version of the problem (see
Section~\ref{s:cceol}).  There is significant flexibility in how to
interpret this problem as a two-party fixed-point problem, and the
interpretation in \citet{BR17} (see Section~\ref{s:ccbfp}) yields a
version of the problem that is hard and yet structured enough to be
solved using approximate Nash equilibrium computation.  A second
innovation in~\cite{BR17} is the reduction from \ccbfp to $\eNE$
(see Section~\ref{s:mt06}) which, while not difficult, is both new and
clever.\footnote{Very recently, \citet{GKP19} showed how to implement
  directly the road map of \citet{RW16}, thereby giving an alternative
  proof of Theorem~\ref{t:br17}.}


\section{Step 1: Query Lower Bound for \eol}\label{s:eollb}

We consider the following ``oracle'' version of the \eol problem.
The vertex set $V$ is fixed to be $\zo^n$.
Let $N = |V| = 2^n$.  Algorithms are allowed to access the graph only
through vertex queries.  A query to the vertex~$v$ reveals its
alleged predecessor $pred(v)$ (if any, otherwise $pred(v)$ is NULL) and its
alleged successor $succ(v)$ (or NULL if it has no successor).
The interpretation is that the directed edge $(v,w)$ belongs to the
implicitly defined directed graph $G=(V,E)$ if and only if both
$succ(v)=w$ and $pred(w)=v$.  These semantics guarantee that the graph
has in- and out-degree at most~1.\footnote{For the proof of
  Theorem~\ref{t:br17}, we could restrict attention to instances that
  are consistent in the sense that $succ(v)=w$ if and only if
  $pred(w)=v$.  The computational hardness results in
  Solar Lectures~4 and~5 require the general (non-promise) version of
  the problem stated here.}
We also assume that $pred(0^n)=NULL$,
and interpret the vertex $0^n$ as the a priori known source vertex of
the graph.

The version of the \eol problem for this oracle model is:
\begin{mdframed}[style=offset,frametitle={The \eol Problem (Query Version)}]
given an oracle as above, find a vertex~$v \in V$ that satisfies one
of the following:
\begin{itemize}

\item [(i)] $succ(v)$ is NULL;

\item [(ii)] $pred(v)$ is NULL and $v \neq 0^n$;

\item [(iii)] $v \neq pred(succ(v))$; or

\item [(iv)] $v \neq succ(pred(v))$ and $v \neq 0^n$.

\end{itemize}
\end{mdframed}
According to our semantics, cases~(iii) and~(iv) imply that $v$ is a
sink and source vertex, respectively.  A solution is guaranteed to
exist---if nothing else, the other end of the path of~$G$ that
originates with the vertex~$0^n$.

It will sometimes be convenient to restrict
ourselves to a ``promise'' version of the \eol problem (which can only
be easier), where the graph $G$ is guaranteed to be a single
Hamiltonian path.
Even in this special case, because every vertex query reveals
information about at most three vertices, we have the following.
\begin{proposition}[Query Lower Bound for \eol]\label{c:eol}
Every deterministic algorithm that solves the \eol problem requires
$\Omega(N)$ queries in the worst case, even for instances that consist
of a single Hamiltonian path.
\end{proposition}
Slightly more formally, consider an adversary that always responds with
values of $succ(v)$ and $pred(v)$ that are never-before-seen vertices
(except as necessary to maintain the consistency of all of the adversary's
answers, so that cases~(iii) and~(iv) never occur).  
After only $o(N)$ queries, the known parts of~$G$ constitute a bunch
of vertex-disjoint paths, and $G$ could be
any Hamiltonian path of~$V$ consistent with these.  The end of this
Hamiltonian path could be any of $\Omega(N)$ different vertices, and
the algorithm has no way of knowing which one.\footnote{A similar
  argument, based on choosing a Hamiltonian path of $V$ at random,
  implies an $\Omega(N)$ lower bound for the randomized query
  complexity as well.}




\section{Step 2: Communication Complexity Lower Bound for \cceol via
 a Simulation Theorem}\label{s:cceol}

Our next step is to use a ``simulation theorem'' to transfer our query
lower bound for the \eol problem to a communication lower
bound for a two-party version of the problem, \cceol.\footnote{This
  monograph does not reflect a beautiful lecture given by Omri
  Weinstein at the associated workshop on
  the history and applications of simulation theorems (e.g., to the
  first non-trivial lower bounds for the clique vs.\ independent set
  problem~\cite{G15}).  Contact him for his slides!}  The exact
definition of the \cceol problem will be determined by the output of
the simulation theorem.

\subsection{The Query Model}

Consider an arbitrary function $f:\Sigma^N \rightarrow \Sigma$, where
$\Sigma$ denotes a finite alphabet.  There is an input $\bfz =
(z_1,\ldots,z_N) \in \Sigma^N$, initially unknown to an algorithm.  The
algorithm can query the input $\bfz$ adaptively, with each query
revealing $z_i$ for a coordinate~$i$ of the algorithm's choosing.
It is trivial to evaluate $f(\bfz)$ using $N$ queries; the question
is whether or not there is an algorithm that always does better
(for some function~$f$ of interest).
For example, the query version of the \eol problem 
in Proposition~\ref{c:eol} can be viewed as a
special case of this model, with $\Sigma = \zo^n \times \zo^n$ (to
encode $pred(v)$ and $succ(v)$) and $f(\bfz)$ encoding the (unique)
vertex at the end of the Hamiltonian path.

\subsection{Simulation Theorems}

We now describe how a function~$f:\Sigma^N \rightarrow \Sigma$ as
above induces a two-party communication problem.
The idea is to ``factor'' the input $\bfz=(z_1,\ldots,z_N)$ to the
query version of the problem between Alice and Bob, so that neither
player can unilaterally figure out any coordinate of $\bfz$.  We
use an \ind gadget for this purpose, as follows.  
(See also Figure~\ref{f:rm}.)
%
\begin{mdframed}[style=offset,frametitle={Two-Party Problem Induced
    by~$f:\Sigma^N \rightarrow \Sigma$}]
\textbf{Alice's input:} $N$ ``blocks'' $A_1,\dots,A_N$. Each block
has $M=\poly(N)$ entries (with each entry in $\Sigma$).  (Say, $M=N^{20}$.)

\vspace{.1in}
\noindent
\textbf{Bob's input:} $N$ indices $y_1,\dots,y_N\in [M]$.

\vspace{.1in}
\noindent
\textbf{Communication problem:} compute $f(A_1[y_1],\dots,A_N[y_N])$.
\end{mdframed}

\begin{figure}
\centering
\includegraphics[width=.3\textwidth]{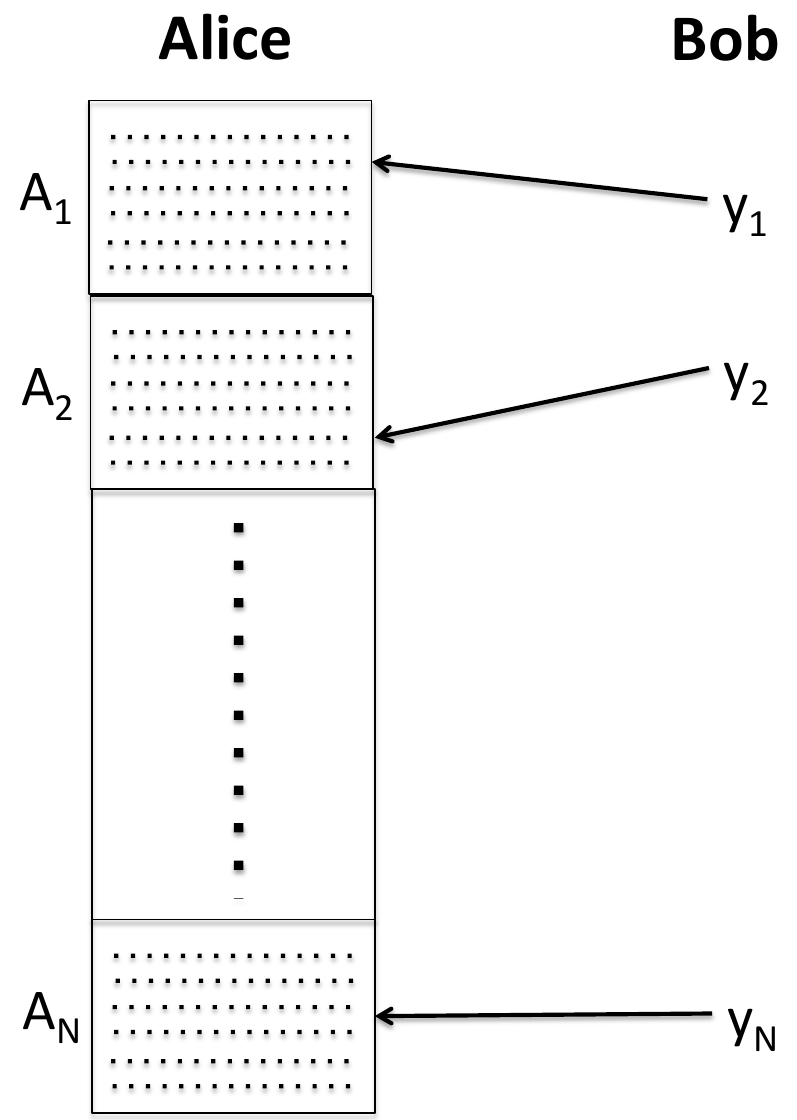}
\caption{A query problem induces a two-party communication problem.
Alice receives $N$ blocks, each containing a list of possible values
for a given coordinate of the input.  Bob receives $N$ indices,
specifying where in Alice's blocks the actual vales of the input
reside.}
\label{f:rm}
\end{figure}

Note that the $y_i$th entry of $A_i$---Bob's index into Alice's
block---is playing the role of $z_i$ in the original problem.  Thus
each block~$A_i$ of Alice's input can be thought of as a ``bag of
garbage,'' which tells Alice a huge number of possible values for the
$i$th coordinate of the input without any clue about which is the
real one.  Meanwhile, Bob's indices tell him the locations of the real
values, without any clues about what these values are.

If~$f$ can be evaluated with a query algorithm that always uses at
most $q$ queries, then the induced two-party problem can be solved
using $O(q \log N)$ bits of communication.  For Alice can just
simulate the query algorithm; whenever it
needs to query the $i$th coordinate of the input, Alice asks Bob for
his index~$y_i$ and supplies the query algorithm with $A_i[y_i]$.
Each of the at most $q$ questions posed by Alice can be communicated with
$\approx \log N$ bits,
and each answer from Bob with $\approx \log M = O(\log N)$ bits.

There could also be communication protocols for the two-party problem
that look nothing like such a straightforward simulation.  For example,
Alice and Bob could send each other the exclusive-or of all of their
input bits.  It's unclear why this would be useful, but it's
equally unclear how to prove that it {\em can't} be useful.
The remarkable {\em Raz-McKenzie simulation theorem} asserts that
there are no communication protocols for the two-party problem that
improve over the straightforward simulation of a query algorithm.

\begin{theorem}[Raz-McKenzie Simulation Theorem~\cite{DBLP:journals/combinatorica/RazM99,GPW18}]\label{t:rm}
If every deterministic query algorithm for~$f$ requires at least $q$
queries in the worst case, then every deterministic communication
protocol for the induced two-party problem has cost $\Omega(q \log N)$.
\end{theorem}
The proof, which is not easy but also not unreadable, shows how to
extract a good query algorithm from an arbitrary low-cost
communication protocol (essentially by a potential function argument).

The original Raz-McKenzie
theorem~\cite{DBLP:journals/combinatorica/RazM99} and the streamlined
version by \citet{GPW18} are both restricted to
deterministic algorithms and protocols, and this is the version we'll
use in this monograph.
Recently, \citet{GPW17} and \citet{A+17}
proved the analog of Theorem~\ref{t:rm}
for randomized query algorithms and randomized communication protocols
(with two-sided error).\footnote{Open question: prove a
  simulation theorem for quantum computation.}
This randomized simulation theorem
simplifies the original proof of Theorem~\ref{t:br17}
(which pre-dated~\cite{GPW17,A+17})
to the point that it's almost the same as the argument given here for the
deterministic case.\footnote{For typechecking reasons, the argument
  for randomized protocols needs to work with
a decision
version of the \eol problem, such as ``is the least significant bit of the
vertex at the end of the   Hamiltonian path equal to 1?''}

The Raz-McKenzie theorem provides a generic way to generate a hard
communication problem from a hard query problem.  We can apply it in
particular to the \eol problem, and we call the induced two-party
problem \cceol.\footnote{\citet{DBLP:journals/combinatorica/RazM99}
  stated their result for the binary alphabet and for total
  functions.  \citet{GPW18} note that it applies
  more generally to arbitrary alphabets and partial functions, which
  is important for its application here.  For further proof details of
  these extensions, see \citet{RW16}.}
\begin{mdframed}[style=offset,frametitle={The \eol Problem (Two-Party
    Version, \cceol)}]
\begin{itemize}

\item Let $V=\zo^n$ and $N=|V|=2^n$.

\item Alice's input consists of $N$ blocks, one for each vertex of $V$,
and each block $A_v$ contains $M$ entries, each encoding a possible
predecessor-successor pair for $v$.  

\item Bob's input consists of one index $y_v \in \{1,2,\ldots,M\}$ for
  each vertex $v \in V$, encoding the entry of the corresponding block 
holding the ``real''
predecessor-successor pair for $v$.  

\item The goal is to identify a vertex $v \in V$ that satisfies one
of the following:
\begin{itemize}

\item [(i)] the successor in $A_v[y_v]$ is NULL;

\item [(ii)] the predecessor in $A_v[y_v]$ is NULL and $v \neq 0^n$;

\item [(iii)] $A_v[y_v]$ encodes the successor $w$ but $A_w[y_w]$ does
  not encode the predecessor $v$; or

\item [(iv)] $A_v[y_v]$ encodes the predecessor $u$ but $A_u[y_u]$ does
  not encode the successor $v$, and $v \neq 0^n$.

\end{itemize}
\end{itemize}

\end{mdframed}

The next statement
is an immediate consequence of Proposition~\ref{c:eol}
and Theorem~\ref{t:rm}.


\begin{corollary}[Communication Lower Bound for \cceol]\label{cor:cceol}
The deterministic communication complexity of the \cceol problem is
$\Omega(N \log N)$, even for instances that consist of a single
Hamiltonian path.
\end{corollary}
A matching upper bound of $O(N \log N)$ is trivial, as Bob always has
the option of
sending Alice his entire input.

Corollary~\ref{cor:cceol} concludes the second step of the proof of
Theorem~\ref{t:br17} and furnishes a generic hard total search
problem.  The next order of business is to transfer this communication
complexity lower bound to the more natural \bfp and $\eNE$ problems via reductions.

\lecture{Communication Complexity Lower Bound for Computing an
  Approximate Nash Equilibrium of a Bimatrix Game (Part II)}

\vspace{1cm}

This lecture completes the proof of Theorem~\ref{t:br17}.
As a reminder, this result states that if Alice's and Bob's
private inputs are the two payoff matrices of an $N \times N$ bimatrix game,
and $\eps$ is a sufficiently small constant, then 
$N^{\Omega(1)}$ communication is required to compute an
$\eps$-approximate Nash equilibrium (Definition~\ref{d:ene}),
even when randomization is allowed.
In terms of the proof road map in Section~\ref{s:map}, it remains to
complete steps~3 and~4.  This corresponds
to implementing Goals \#1 and \#2 introduced in the last
lecture---reversing the direction of the classical reductions from the
\bfp problem to path-following and from the $\eNE$ problem to
the \bfp problem.

\section{Step 3: \cceol $\leq$ \ccbfp}\label{s:ccbfp}

\subsection{Preliminaries}

We know from Corollary~\ref{cor:cceol} that \cceol, the two-party
version of the {\sc End-of-the-Line} problem defined in
Section~\ref{s:cceol}, has large communication
complexity.  This section transfers this lower bound to a two-party
version of an approximate fixed point problem, by reducing 
the \cceol problem to it.

We next define our two-party version of the \bfp problem, the \ccbfp
problem.  The problem is parameterized by the dimension~$d$ and an
approximation parameter $\eps$.  The latter should be thought of as a
sufficiently small constant (independent of~$d$). 
\begin{mdframed}[style=offset,frametitle={The \bfp Problem (Informal
    Two-Party Version)}]
\begin{itemize}

\item Let $H = [0,1]^d$ denote the $d$-dimensional hypercube.

\item Alice and Bob possess private inputs that, taken together, implicitly
  define a continuous function $f:H \rightarrow H$.

\item The goal is to identify an {\em $\eps$-approximate fixed point},
  meaning a point $x \in H$ such that $\n{f(x)-x} < \eps$, where
  $\|\cdot\|$ denotes the normalized $\ell_2$ norm:
\[
\n{a} = \sqrt{\frac{1}{d} \sum_{i=1}^d a_i^2}.
\]
\end{itemize}
\end{mdframed}


The normalized $\ell_2$ norm of a point in the hypercube (or the
difference between two such points) is always
between 0 and 1.  If a point $x \in H$ is {\em not} an
$\eps$-approximate fixed point with respect to this norm,
then $f(x)$ and $x$ differ by a constant amount in a constant fraction
of the coordinates.  This version of the problem can only be easier
than the more traditional version, which uses the $\ell_{\infty}$
norm.

To finish the description of the \ccbfp problem, we need to explain
how Alice and Bob interpret their inputs as jointly defining a
continuous function.

\subsection{Geometric Intuition}

Our reduction from \cceol to \ccbfp will use no communication---Alice
and Bob will simply reinterpret their \cceol inputs as \ccbfp inputs
in a specific way, and a solution to the \cceol instance will be easy
to recover from any approximate fixed point.

Figure~\ref{f:hpv} shows the key intuition: graphs of paths and cycles
naturally lead to continuous functions, where the gradient of the
function ``follows the line'' and fixed points correspond to sources
and sinks of the graph.  Following the line (i.e., ``gradient
ascent'') guarantees discovery of an approximate fixed point; the goal
will be to show that no cleverer algorithm is possible.

\begin{figure}[h]
\centering
\includegraphics[width=.7\textwidth]{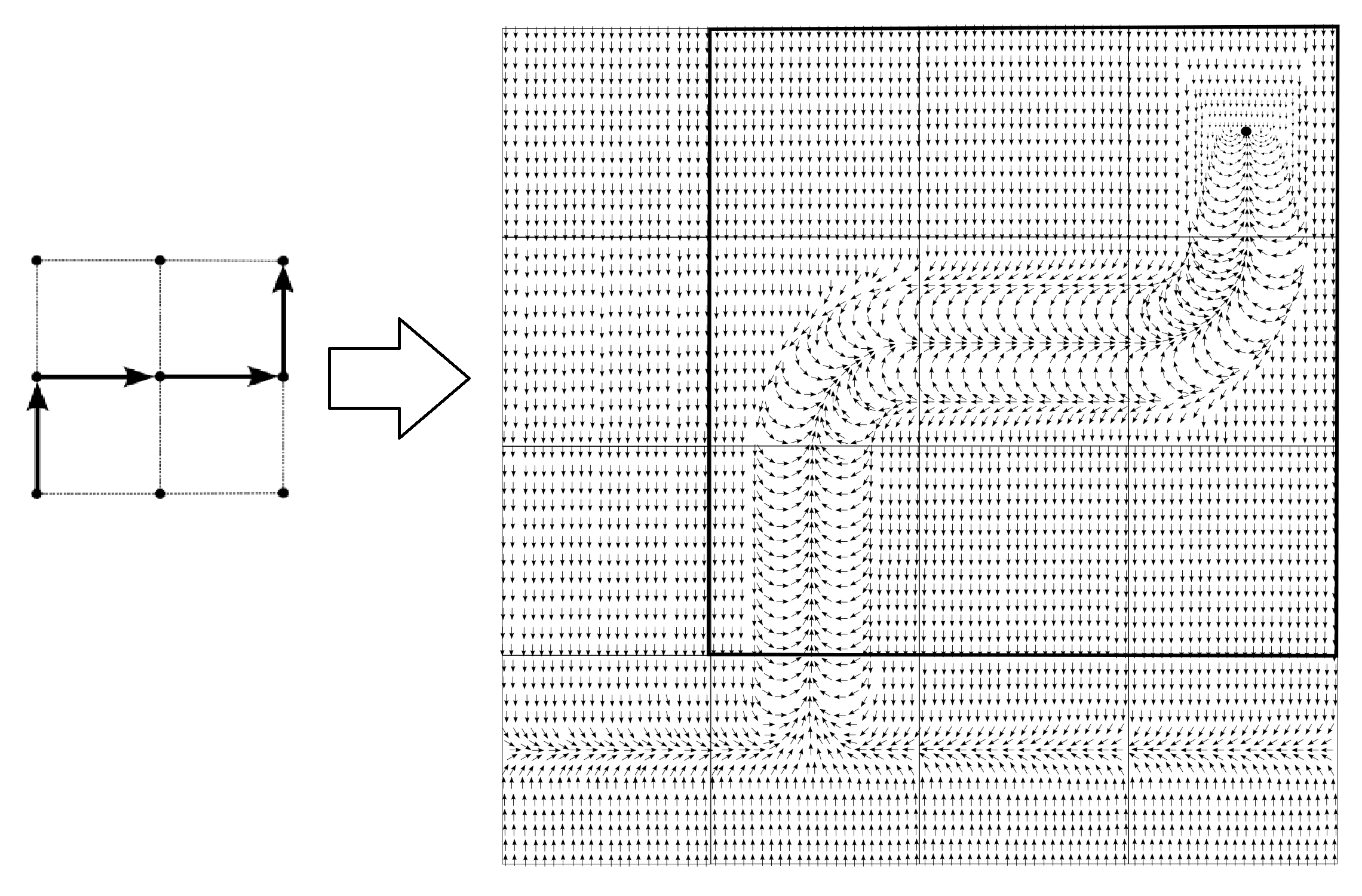}
\caption{Directed paths and cycles can be interpreted as a continuous
  function whose gradient ``follows the line.''  Points far from the path are
  moved by $f$ in some canonical direction.  (Figure courtesy of Yakov
Babichenko.)}
\label{f:hpv}
\end{figure}

This idea originates in
\citet{HPV89}, who considered approximate fixed points in the
$\ell_{\infty}$ norm.  \citet{R16} showed how to modify the
construction so that it works even for the normalized $\ell_2$ norm.
\citet{BR17} used the construction from~\cite{R16} in
their proof of Theorem~\ref{t:br17}; our treatment here includes some  simplifications.

\subsection{Embedding a Graph in the Hypercube}\label{ss:embed1}

Before explaining exactly how to interpret graphs as continuous
functions, we need to set up an embedding of every possible
graph on a given vertex set into the hypercube.

Let $V = \zo^n$ and $N=|V|=2^n$.  Let $K$ denote the complete
undirected graph with vertex set $V$---all edges that could
conceivably be present in an \eol instance (ignoring their
orientations).  Decide once and for all on an embedding $\sigma$ of
$K$ into $H=[0,1]^d$, where $d = \Theta(n) = \Theta(\log N)$, with two
properties:\footnote{By an {\em embedding}, we mean a
  function~$\sigma$ that 
  maps each edge $(v,w)$ of $K$ to a continuous path in $H$ with
  endpoints $\sigma(v)$ and $\sigma(w)$.}
\begin{itemize}

\item [(P1)] The images of the vertices are well separated:
for every
  $v,w \in V$ (with $v \neq w$), $\n{\sigma(v)-\sigma(w)}$ is at least some constant
  (say $\tfrac{1}{10}$).

\item [(P2)] The images of the edges are well separated.  More precisely,
a point $x \in H$ is close (within distance~$10^{-3}$,
  say) to the images $\sigma(e)$ and $\sigma(e')$ of distinct edges
  $e$ and $e'$ only if $x$ is close to the image of a shared endpoint
  of $e$ and $e'$.  (In particular, if $e$ and $e'$ have no endpoints
  in common, then no $x \in H$ is close to both $\sigma(e)$ and
  $\sigma(e')$.)

\end{itemize}
Property~(P1) asserts that the images of two different vertices differ
by a constant amount in a constant fraction of their
coordinates.\footnote{In the original construction of \citet{HPV89},
vertices of~$K$ could potentially be mapped to points of~$H$ that differ
significantly in only one coordinate.  This construction is good
enough to prevent spurious approximate fixed points in the
$\ell_{\infty}$ norm, but not in the normalized $\ell_2$ norm.}
One natural way to achieve this property is via an error-correcting code
with constant rate.  The simplest way to achieve both properties is to
take a random straight-line embedding.  Each vertex $v \in V$ is
mapped to a point in $\{ \tfrac{1}{4}, \tfrac{3}{4} \}^d$, with each
coordinate set to $\tfrac{1}{4}$ or $\tfrac{3}{4}$ independently with
50/50 probability.\footnote{For reasons related to the omitted
  technical details, it's convenient to
  have a ``buffer zone'' between the embedding of the graph and the
  boundary of the hypercube.}  Each edge is mapped to a straight line
between the images of its endpoints.  Provided $d=cn$ for a sufficiently
large constant $c$, properties~(P1) and~(P2) both hold with high
probability.\footnote{In the two-party communication model, we need
  not be concerned about efficiently constructing such an
  embedding.  Because Alice and Bob have unbounded computational power,
  they can both compute the lexicographically first such embedding in
  advance of the protocol.  When we consider computational lower
  bounds in Solar Lecture~5, we'll need an efficient construction.}

The point of properties~(P1) and~(P2) is to classify the points of $H$
into three categories: (i) those close to the image of a (unique)
vertex of~$K$; (ii) those not close to the image of any vertex but
close to the image of a (unique) edge of $K$; and (iii) points not
close to the image of any vertex or edge of~$K$.  Accordingly, each
point $x \in H$ can be ``decoded'' to a unique vertex $v$ of $K$, a
unique edge $(v,w)$ of $K$, or~$\bot$.  Don't forget that this
classification of points of $H$ is made in advance of receiving any
particular \cceol input.  In the \ccbfp problem, because Alice and Bob
both know the embedding in advance, they can decode points at will
without any communication.\footnote{As suggested by
  Figure~\ref{f:hpv}, in the final construction it's important to use
  a more nuanced classification that ``interpolates'' between points
  in the three different categories.  It will still be the case that
  Alice and Bob can classify any point $x \in H$ appropriately
  without any communication.}

\subsection{Interpreting Paths as Continuous Functions}\label{ss:embed2}

Given the embedding above, we can now describe how to interpret a
directed graph $G=(V,E)$ induced by an instance of \eol as a continuous
function on the hypercube, with approximate fixed points of the
function corresponding only to sources and sinks of~$G$.  Write a
function $f:H \rightarrow H$ as $f(x) = x + g(x)$ for the
``displacement function'' $g:H \rightarrow [-1,1]^d$.  
(The final construction will take care to define $g$ so that $x+g(x)
\in H$ for every $x \in H$.)  An $\eps$-approximate
fixed point is a point $x$ with $\n{g(x)} < \eps$, so it's crucial for
our reduction that our definition of $g$ satisfies $\n{g(x)} \ge \eps$
whenever $x$ is not close to the image of a source or sink of $G$.

\bigskip
Consider for simplicity a directed graph~$G=(V,E)$ of an \eol instance
that has no 2-cycles and no isolated vertices.\footnote{Recall from
  Corollary~\ref{cor:cceol} that the \cceol problem is already hard in
the special case where the encoded graph~$G$ is guaranteed to be a
Hamiltonian path.}
For a (directed) edge $(u,v) \in E$, define
\[
\gamma_{uv} = 
\frac{\sigma(v)-\sigma(u)}{\n{\sigma(v)-\sigma(u)}}
\]
as the unit vector with the same direction as the embedding of the
corresponding undirected edge of $K$, oriented from~$u$ toward $v$.
A rough description of the displacement function $g(x)$ corresponding to~$G$
is as follows, where $\delta > 0$ is a parameter (cf., Figure~\ref{f:hpv}):
\begin{enumerate}
\item For $x$ close to the embedding $\sigma(e)$ of the (undirected)
  edge $e \in K$ with endpoints $u$ and $v$, but not close to
  $\sigma(u)$ or $\sigma(v)$, define
\[
g(x)=\delta \cdot \left\{
\begin{array}{ll}
\gamma_{uv}
& \text{if edge $(u,v) \in E$}\\
\gamma_{vu}
& \text{if edge $(v,u) \in E$}\\
\text{some default direction} & \text{otherwise}.
\end{array}
\right.
\]

\item For $x$ close to $\sigma(v)$ for some $v \in V$,

 \begin{enumerate}

 \item if $v$ has an
   incoming edge $(u,v) \in E$ and an outgoing edge $(v,w) \in E$, 
then define $g(x)$ by
   interpolating between $\delta \cdot \gamma_{uv}$ and
$\delta \cdot \gamma_{vw}$
(i.e., ``turn slowly'' as in Figure~\ref{f:hpv});

 \item otherwise (i.e., $v$ is a source or sink of~$G$), 
   define $g(x)$ by interpolating between the all-zero vector and the
   displacement vector (as defined in case~1)
associated with
$v$'s (unique) incoming or outgoing edge in~$G$.

 \end{enumerate}

\item For $x$ that are not close to any $\sigma(v)$ or $\sigma(e)$,
  define $g(x)$ as $\delta$ times the default direction. 

\end{enumerate}
For points~$x$ ``in between'' the three cases (e.g., almost but not quite
close enough to the image $\sigma(v)$ of a vertex $v \in V$),
$g(x)$ is defined by interpolation
(e.g., a weighted average of the displacement vector associated with
$v$ in case~2
and $\delta$ times the default direction, with the weights determined by
$x$'s proximity to $\sigma(v)$).

The default direction can be implemented by doubling the number of
dimensions to $2d$, and defining the displacement direction
as the vector $(0,0,\ldots,0,1,1,\ldots,1)$.
Special handling (not detailed here) is then required at points $x$
with value close to 1 in one of these extra coordinates, to ensure that
$x+g(x)$ remains in~$H$ while also not introducing any unwanted
approximate fixed points.
Similarly, special handling is required for the source vertex $0^n$, to
prevent $\sigma(0^n)$ from being a fixed point.  Roughly, this can be
implemented by mapping the vertex $0^n$ to one corner of the hypercube
and defining $g$ to point in the opposite direction.
The parameter $\delta$ is a constant, bigger than $\eps$ by a constant
factor.  (%
For example, one can assume that $\eps \le
10^{-12}$ and take $\delta \approx 10^{-6}$.)  
This ensures that whenever the normalized~$\ell_2$ norm of a direction
vector $y$ is at least a sufficiently large constant, $\delta \cdot y$ 
has norm larger than $\eps$.
This completes our sketch of how to interpret an instance of \eol as a
continuous function on the hypercube.  

\subsection{Properties of the Construction}\label{ss:props}

Properly implemented, the construction in Sections~\ref{ss:embed1}
and~\ref{ss:embed2} has the following properties:
\begin{enumerate}

\item Provided $\eps$ is at most a sufficiently small constant,
a point $x \in H$ satisfies $\n{g(x)} < \eps$ only if it is
  close to the image
  of a source or sink of $G$ different from the canonical source
  $0^n$.  (Intuitively, this should be true by construction.)


\item There is a constant $\lambda$, independent of $d$,
  such that the
  function $f(x)=x+g(x)$ is $\lambda$-Lipschitz.
In particular, $f$ is continuous.  (Intuitively, this is because we
take care to linearly interpolate between regions of $H$ with
different displacement vectors.)

\end{enumerate}
Sections~\ref{ss:embed1} and~\ref{ss:embed2},
together with Figure~\ref{f:hpv}, provide a
plausibility argument that a construction with these two
properties is possible along the proposed lines.  Readers interested
in further details should start with the carefully written
two-dimensional construction in \citet[Section 4]{HPV89}---where many of
these ideas originate---before proceeding to the general case
in~\cite[Section 5]{HPV89} for the $\ell_{\infty}$ norm and
finally~\citet{BR17} for the version tailored to the normalized
$\ell_2$ norm (which is needed here).

\subsection{The \ccbfp Problem and Its Communication Complexity}

We can now formally define the
two-party version of the \bfp problem that we consider, denoted
\ccbfp.  The problem is parameterized by a positive integer~$n$ and a
constant $\eps > 0$.
\begin{mdframed}[style=offset,frametitle={The \ccbfp Problem}]
\begin{itemize}

\item Alice and Bob begin with private inputs to the \cceol problem: Alice with
$N=2^n$ ``blocks'' $A_1,\dots,A_N$, each with $M=\poly(N)$ entries from the
alphabet $\Sigma = \zo^n \times \zo^n$, and Bob
with $N$ indices $y_1,\dots,y_N\in [M]$.

\item Let~$G$ be the graph induced by these inputs (with $V = \zo^n$
  and $A_v[y_v]$ encoding\\ $(pred(v),succ(v))$.

\item Let~$f$ denote the continuous function $f:H \rightarrow H$
  induced by~$G$, as per the construction in Sections~\ref{ss:embed1}
  and~\ref{ss:embed2}, where $H = [0,1]^d$ is the $d$-dimensional
  hypercube with $d = \Theta(n)$.

\item The goal is to compute a point $x \in H$ such that $\n{f(x)-x} <
  \eps$, where $\n{\cdot}$ denotes the normalized $\ell_2$ norm.

\end{itemize}
\end{mdframed}



The first property in Section~\ref{ss:props} implies a communication
complexity lower bound for the
\ccbfp problem, which implements step~3 of the road map in
Section~\ref{s:map}.  (The second property is important for
implementing step~4 of the road map in the next section.)
\begin{theorem}[\citet{BR17}]\label{t:ccbfp}
For every sufficiently small constant $\eps > 0$,
the deterministic communication complexity of the \ccbfp problem is
$\Omega(N \log N)$.
\end{theorem}

\begin{proof}
  If there is a deterministic communication protocol with cost $c$ for
  the \ccbfp problem, then there is also one for the \cceol problem:
  Alice and Bob interpret their \cceol inputs as inputs to the \ccbfp
  problem, run the assumed protocol to compute an $\eps$-approximate
  fixed point~$x \in H$ of the corresponding function~$f$, and (using
  no communication) decode~$x$ to a source or sink vertex of $G$
  (that is different from $0^n$).  The theorem follows immediately
  from Corollary~\ref{cor:cceol}.
\end{proof}

\subsection{Local Decodability of \ccbfp Functions}\label{ss:local}

There is one more important property of the functions~$f$ constructed
in Sections~\ref{ss:embed1} and~\ref{ss:embed2}: they are {\em locally
decodable} in a certain sense.  Suppose Alice and Bob want to compute
the value of $f(x)$ at some commonly known point $x \in H$.  If $x$
decodes to $\bot$ (i.e., is not close to the image of any vertex or
edge of the complete graph~$K$ on vertex set~$V$), then Alice and Bob
know the value of~$f(x)$ without any communication whatsoever: $f(x)$
is $x$ plus $\delta$ times the default direction (or a known
customized displacement if $x$ is too close to certain boundaries
of~$H$).  If $x$ decodes to the edge~$e=(u,v)$ of the complete
graph~$K$, then Alice and Bob can compute $f(x)$ as soon as they know
whether or not edge~$e$ belongs to the directed graph~$G$ induced by
their inputs, along with its orientation.  This requires Alice and Bob
to exchange predecessor-successor information about only two vertices
($u$ and $v$).
Analogously, if~$x$ decodes to the vertex~$v$ of~$K$, then Alice and
Bob can compute~$f(x)$ after exchanging information about at most
three vertices ($v$, $pred(v)$, and $succ(v)$).




\section{Step 4: \ccbfp $\le \eNE$}\label{s:mt06}

This section completes the proof of Theorem~\ref{t:br17} by reducing
the \ccbfp problem to the $\eNE$ problem, where $\eps$ is a
sufficiently small constant.

\subsection{The McLennan-Tourky Analytic Reduction}\label{ss:mt06}

The starting point for our reduction is a purely analytic reduction of
\citet{MT06},
which reduces the existence of (exact) Brouwer fixed points to the
existence of (exact) Nash equilibria.\footnote{This reduction was
  popularized in a Leisure of the Theory Class blog post by Eran
  Shmaya (\url{https://theoryclass.wordpress.com/2012/01/05/brouwer-implies-nash-implies-brouwer/}), who heard about the   result from Rida Laraki.}
Subsequent sections explain the additional
ideas needed to implement this reduction for approximate fixed points and
Nash equilibria in the two-party communication model.


\begin{theorem}[\citet{MT06}]\label{t:mt06}
Nash's theorem (Theorem~\ref{t:nash}) implies Brouwer's fixed-point
theorem (Theorem~\ref{t:bfp}).
\end{theorem}

\begin{proof}
Consider an arbitrary continuous function~$f:H \rightarrow H$, where
$H = [0,1]^d$ is the $d$-dimensional hypercube (for some positive
integer~$d$).\footnote{If fixed points are guaranteed for hypercubes
  in every dimension, then they are also guaranteed for all compact
  convex   subsets of finite-dimensional Euclidean space; see
  footnote~\ref{foot:radial} in Solar Lecture~2.}
Define a two-player game as follows.
The pure strategies of Alice and Bob both correspond to points of
$H$.  For pure strategies $x,z \in H$, Alice's payoff is defined as
\begin{equation}\label{eq:apayoff}
1 - \n{x-z}^2 = 
1 - \frac{1}{d} \sum_{i=1}^d (x_i-z_i)^2
\end{equation}
and Bob's payoff as
\begin{equation}\label{eq:bpayoff}
1 - \n{z-f(x)}^2 = 
1 - \frac{1}{d} \sum_{i=1}^d (z_i-f(x)_i)^2.
\end{equation}
Thus Alice wants to imitate Bob's strategy, while Bob wants to imitate
the image of Alice's strategy under the function~$f$.






For any mixed strategy $\sigma$ of Bob (i.e., a distribution over
points of the hypercube),
Alice's unique best response
is the corresponding center of gravity $\expect[z \sim \sigma]{z}$ (as
you should check).
Thus, in any Nash equilibrium, Alice plays a pure strategy $x$.  Bob's
unique best response to such a pure strategy is the pure strategy $z =
f(x)$.  That is, every Nash equilibrium is pure, with $x = z = f(x)$ a
fixed point of $f$.
Because a Nash equilibrium exists, so does a fixed
point of~$f$.\footnote{Strictly speaking, we're assuming a more
  general form of Nash's theorem that asserts the existence of a pure
  Nash   equilibrium
whenever every player has a convex
  compact strategy set (like~$H$) and a continuous concave payoff
  function (like~\eqref{eq:apayoff} and~\eqref{eq:bpayoff}).  
(The version in Theorem~\ref{t:nash} corresponds to the special case 
where each strategy set corresponds to a finite-dimensional simplex of mixed
  strategies, and where all payoff functions are linear.)
Most proofs of Nash's theorem---including the one outlined in
Section~\ref{ss:nashpf}---are straightforward to generalize in this way.}
%
\end{proof}

An extension of the argument above shows that, for $\lambda$-Lipschitz
functions~$f$, an $\eps'$-approximate fixed point (in the normalized
$\ell_2$ norm) can be extracted easily from any $\eps$-approximate
Nash equilibrium, where~$\eps'$ is a function of $\eps$ and $\lambda$
only.\footnote{It is not clear how to easily extract an approximate
  fixed point in the $\ell_{\infty}$ norm from an approximate Nash
  equilibrium without losing a super-constant factor in the
  parameters.  The culprit is the ``$\tfrac{1}{d}$'' factor
  in~\eqref{eq:apayoff} and~\eqref{eq:bpayoff}---needed to ensure that
  payoffs are bounded---which allows each player to behave in an
  arbitrarily crazy way in a few coordinates without violating the
  $\eps$-approximate Nash equilibrium conditions.  (Recall $\eps > 0$
  is constant while $d \rightarrow \infty$.)  This is one of the
  primary reasons why \citet{R16} and \citet{BR17} needed to modify
  the construction in \citet{HPV89} to obtain their results.}

\subsection{The Two-Party Reduction: A Naive Attempt}

We now discuss how to translate the McLennan-Tourky analytic reduction
to an analogous reduction in the two-party model.  First, we need to
discretize the hypercube.  Define $\discH$ as the set of
$\approx \left(\tfrac{1}{\eps}\right)^d$ points of $[0,1]^d$ for which
all coordinates are multiples of $\eps$.  Every $O(1)$-Lipschitz
function~$f$---including every function arising in a \ccbfp instance
(Section~\ref{ss:props})---is guaranteed to have an
$O(\eps)$-approximate fixed point at some point of this discretized
hypercube (by rounding an exact fixed point to its nearest neighbor in
$H_{\eps}$).  This also means that the corresponding game (with
payoffs defined as in~\eqref{eq:apayoff} and~\eqref{eq:bpayoff}) has
an $O(\eps)$-approximate Nash equilibrium in which each player
deterministically chooses a point of~$\discH$.


The obvious attempt at a two-party version of the McLennan-Tourky
reduction is: 
\begin{enumerate}

\item Alice and Bob start with inputs to the \ccbfp problem.

\item The players interpret
these inputs as a two-player game, with strategies corresponding to
points of the discretized hypercube~$\discH$, and with Alice's payoffs given
by~\eqref{eq:apayoff} and Bob's payoffs by~\eqref{eq:bpayoff}.

\item The players run
the assumed low-cost communication protocol for computing an
approximate Nash equilibrium.

\item The players extract an approximate fixed
point of the \ccbfp function from the approximate Nash equilibrium.  

\end{enumerate}
Just one problem: {\em this doesn't make
  sense}.  The issue is that Bob needs to be able to compute $f(x)$ to
evaluate his payoff function in~\eqref{eq:bpayoff}, and his \ccbfp
input (a bunch of indices into Alice's blocks) does not provide
sufficient information to do this.  Thus, the proposed reduction does
not produce a well-defined input to the $\eNE$ problem.



\subsection{Description of the Two-Party Reduction}\label{ss:step4}

The consolation prize is that Bob can compute the function~$f$ at a
point $x$ after a brief conversation with Alice.  Recall from
Section~\ref{ss:local} that computing~$f$ at a point $x \in H$
requires information about at most three vertices of the \cceol input
that underlies the \ccbfp input (in addition to~$x$).  
Alice can send
$x$ to Bob, who can then send the relevant indices to Alice (after
decoding $x$ to some vertex or edge of $K$), and Alice can respond
with the corresponding predecessor-successor pairs.
%
%
This requires
$O(\log N)$ bits of communication, where $N=2^n$ is the number of
vertices in the underlying \cceol instance.  (We are suppressing the
dependence on the constant~$\eps$ in the big-O notation.)  Denote this
communication protocol by~$P$.

At this point, it's convenient to restrict the problem to the hard
instances of \cceol used to prove Corollary~\ref{cor:cceol}, where in
particular, $succ(v) = w$ if and only if $v = pred(w)$.  (I.e.,
cases~(iii) and~(iv) in the definition of the \cceol problem in
Section~\ref{s:cceol} never come up.)  For this special case, $P$ can
be implemented as a two-round protocol where Alice and Bob exchange
information about one relevant vertex~$v$ (if $x$ decodes to $v$) or
two relevant vertices~$u$ and~$v$ (if $x$ decodes to the edge
$(u,v)$).\footnote{If $x$ decodes to the edge $(u,v)$, then Alice and
  Bob exchange information about $u$ and $v$ in two rounds.  If $x$
  decodes to the vertex $v$, they exchange information about $v$ in
  two rounds.  This reveals $v$'s opinion of its predecessor $u$ and
  successor $w$.  In the general case, Alice and Bob would still need
  to exchange information about $u$ and $w$ using two more rounds of
  communication to confirm that $succ(u)=pred(w)=v$.  (Recall our
  semantics: directed edge $(v,w)$ belongs to $G$ if and only if both
  $succ(v)=w$ and $pred(w)=v$.)  In the special case of instances
  where $succ(v) = w$ if and only if $v = pred(w)$, these two extra
  rounds of communication are redundant.}


How can we exploit the local decodability of \ccbfp functions?  The
idea is to enlarge the strategy sets of Alice and Bob, beyond the
discretized hypercube~$\discH$, so that the players' strategies at
equilibrium effectively simulate the protocol~$P$.  Alice's pure
strategies are the pairs $(x,\alpha)$, where $x \in \discH$ is a point
of the discretized hypercube and $\alpha$ is a possible transcript of
Alice's communication in the protocol~$P$.  Thus $\alpha$ consists of
at most two predecessor-successor pairs.  Bob's pure strategies are
the pairs $(z,\beta)$, where $z \in \discH$ and $\beta$ is a
transcript that could be generated by Bob in~$P$---a specification of
at most two different vertices and his corresponding indices for
them.\footnote{In the protocol~$P$, Bob does not need to communicate
  the names of any vertices---Alice can decode~$x$ privately.  But
  it's convenient for the reduction to include the names of the
  vertices relevant for~$x$ in the $\beta$ component of Bob's strategy.}
Crucially, because the protocol~$P$ has cost $O(\log N)$, there are
only $N^{O(1)}$ possible $\alpha$'s and $\beta$'s.  There are also
only $N^{O(1)}$ possible choices of $x$ and $z$---since $\eps$ is a
constant and $d=\Theta(n)$ in the \ccbfp problem,
$|\discH|\approx \left(\tfrac{1}{\eps}\right)^d$ is polynomial
in~$N=2^n$.  We conclude that the size of the resulting game is
polynomial in the length of the given \ccbfp (or \cceol) inputs.


We still need to define the payoffs of the game.  Let
$A_1,\ldots,A_N$ and $y_1,\ldots,y_N$ denote Alice's and Bob's
private inputs in the given \ccbfp (equivalently, \cceol) instance and
$f$ the corresponding function.
Call an outcome $(x,\alpha,z,\beta)$ \emph{consistent} if $\alpha$ and
$\beta$ are the transcripts generated by Alice and Bob when they
honestly follow the protocol~$P$ to compute $f(x)$.
Precisely, a consistent outcome is one that meets the following two conditions:
\begin{itemize}

\item[(i)] 
for each of the (zero, one, or two) vertices~$v$ 
and corresponding indices~$\yhat_v$ announced
by Bob in $\beta$,
$\alpha$ contains the correct response~$A_v[\yhat_v]$;

\item[(ii)] 
$\beta$ specifies the names of the vertices relevant for Alice's
announced point $x \in H_{\eps}$, and for each such vertex~$v$,
$\beta$ specifies the correct index~$y_v$.

\end{itemize}
Observe that Alice can privately check if condition~(i) holds (using
her private input $A_1,\ldots,A_N$
and the vertex names and indices in
Bob's announced strategy $\beta$), and Bob can privately check
condition~(ii) (using his private input $y_1,\ldots,y_N$ and the
point~$x$ announced by Alice).


For an outcome $(x,\alpha,z,\beta)$, we define Alice's payoffs by
\begin{equation}\label{eq:apayoff2}
\left\{ \begin{array}{cl}
-1-\frac{1}{d} \sum_{i=1}^d (x_i-z_i)^2 &\mbox{ if (i) fails}\\
1-\frac{1}{d} \sum_{i=1}^d (x_i-z_i)^2 &\mbox{ otherwise.}
\end{array}\right.
\end{equation}
(Compare~\eqref{eq:apayoff2} with~\eqref{eq:apayoff}.)
This definition makes sense because Alice can privately check whether
or not~(i) holds and hence can privately compute her
payoff.\footnote{If you want to be a stickler and insist on payoffs in
  $[0,1]$, then shift and scale the payoffs in~\eqref{eq:apayoff2} appropriately.}

For Bob's payoffs, we need a preliminary definition.
Let $f_{\alpha}(x)$ denote the value that the induced function~$f$
would take on if $\alpha$ was consistent with $x$ and
with Alice's and Bob's private inputs.
That is, to compute $f_{\alpha}(x)$:
\begin{enumerate}

\item Decode $x$ to a vertex or an edge (or $\bot$).

\item Interpret $\alpha$ as the predecessor-successor pairs for the
  vertices relevant for evaluating~$f$ at $x$.

\item Output $x$ plus the displacement $g_{\alpha}(x)$
defined as in Sections~\ref{ss:embed1} and~\ref{ss:embed2} (with
$\alpha$ supplying any predecessor-successor
pairs that are necessary).

\end{enumerate}

To review, $f$ is the \ccbfp function that Alice and Bob want to find a
fixed point of, and $f(x)$ generally depends on the private inputs
$A_1,\ldots,A_N$ and $y_1,\ldots,y_N$ of
both Alice and Bob.  The function~$f_{\alpha}$ is a speculative
version of~$f$, predicated on Alice's announced predecessor-successor
pairs in her strategy~$\alpha$.  Crucially, the definition of
$f_{\alpha}$ does not depend at all on Alice's private input,
only on Alice's {\em announced strategy}.
Thus given $\alpha$, Bob can
privately execute the three steps above and evaluate $f_{\alpha}(x)$
for any $x \in \discH$.
The other crucial property of $f_{\alpha}$ is that, if $\alpha$
happens to be the actual predecessor-successor pairs $\{ A_v[y_v] \}$
for the vertices relevant for~$x$ (given Alice's and Bob's private
inputs), then $f_{\alpha}(x)$ agrees with the value~$f(x)$ of the true
\ccbfp function.

We can now define Bob's payoffs as follows (compare
with~\eqref{eq:bpayoff}):
\begin{equation}\label{eq:bpayoff2}
\left\{ \begin{array}{cl}
-1&\mbox{ if (ii) fails}\\
1- \frac{1}{d} \sum_{i=1}^d (z_i-f_{\alpha}(x)_i)^2
&\mbox{ otherwise.}
\end{array}\right.
\end{equation}
Because Bob can privately check condition~(ii) and compute
$f_{\alpha}(x)$ (given $x$ and $\alpha$), Bob can privately compute
his payoff.  This completes the description of the reduction from
the \ccbfp problem to the $\eNE$ problem.  

Alice and Bob can carry out this reduction with no communication---by
construction, their \ccbfp inputs fully determine their payoff
matrices.  As noted earlier, because $\eps$ is a constant,
the sizes of the produced $\eNE$ inputs are polynomial in those of the
\ccbfp inputs.


\subsection{Analysis of the Two-Party Reduction}

Finally, we need to show that the reduction ``works,'' meaning that
Alice and Bob can recover an approximate fixed point of the \ccbfp
function~$f$ from any approximate Nash equilibrium of the game
produced by the reduction.

For intuition, let's think first about the case where Alice's and Bob's
strategies are points of the hypercube~$H$ (rather than the
discretized hypercube $\discH$) and the case of exact fixed points and
Nash equilibria.  (Cf., Theorem~\ref{t:mt06}.)
What could a Nash equilibrium of the game look like?  Consider mixed
strategies by Alice and Bob.  
\begin{enumerate}

\item
Alice's payoff in~\eqref{eq:apayoff2}
includes a term $-\frac{1}{d} \sum_{i=1}^d (x_i-z_i)^2$ that is
independent of her choice of $\alpha$ or Bob's choice of $\beta$, and
the other term (either~1 or~-1) is independent of her choice of~$x$
(since condition~(i) depends only on~$\alpha$ and~$\beta$).
Thus, analogous to the proof of Theorem~\ref{t:mt06}, in every one of
Alice's best responses,
she deterministically chooses $x = \expect[z \sim \sigma]{z}$, where
$\sigma$ denotes the marginal distribution of $z$ in Bob's mixed
strategy.  

\item
Given that Alice is playing deterministically in her
$x$-coordinate, in every one of Bob's best responses,
he deterministically chooses~$\beta$ to name the vertices relevant for
Alice's announced
point~$x$ and his indices for these vertices (to land in the second
case of~\eqref{eq:bpayoff2} with probability~1).

\item Given that Bob is playing deterministically in his
  $\beta$-coordinate, 
Alice's unique best response is to choose~$x$ as before and also
deterministically choose the (unique) message $\alpha$ that satisfies
condition~(i), so that she will be in the more favorable second case
of~\eqref{eq:apayoff2} with probability~1.

\item Given that Alice is playing deterministically in both her $x$- and
  $\alpha$-coordinates, Bob's unique best response is to
  choose~$\beta$ as before and set $z=f_{\alpha}(x)$ (to maximize his
  payoff in the second case of~\eqref{eq:bpayoff2}).

\end{enumerate}
These four steps imply that every (exact) Nash
equilibrium~$(x,\alpha,z,\beta)$ of the game is pure, with $\alpha$
and $\beta$ consistent with $x$ and Alice's and Bob's private
information about the corresponding relevant vertices, and with
$x=z=f_{\alpha}(x) = f(x)$ a fixed point of~$f$.

As with Theorem~\ref{t:mt06}, a more technical version of the same
argument implies that an approximate fixed point---a
 point~$x$ satisfying $\n{f(x)-x} < \eps'$ with respect to the
 normalized $\ell_2$ norm---can
be easily extracted by Alice and Bob from any $\eps$-approximate Nash
equilibrium, where $\eps'$ depends only on $\eps$ (e.g.,
$\eps' = O(\eps^{1/4})$ suffices).  For example, the first step of
the proof becomes: in an $\eps$-approximate Nash equilibrium, Alice
must choose a point $x \in \discH$ that is close to $\expect{z}$
except with small probability (otherwise she could increase her
expected payoff 
by more
than $\eps$ by switching to the point of $\discH$ closest to
$\expect{z}$).  
And so on.  Carrying out approximate
versions of all four steps above, while keeping careful track of the
epsilons, completes the proof of
Theorem~\ref{t:br17}.\footnote{The fact that~$f$ is~$O(1)$-Lipschitz 
is important for carrying out the last of these steps.}

We conclude that computing an approximate Nash equilibrium of a
general bimatrix game requires a polynomial amount of communication, and
in particular there are no uncoupled dynamics guaranteed to converge
to such an equilibrium in a polylogarithmic number of iterations.






\lecture{$\TFNP$, $\PPAD$, \& All That}

\vspace{1cm}


Having resolved the communication complexity of computing an
approximate Nash equilibrium of a bimatrix game, we turn our attention
to the {\em computational} complexity of the problem.  Here, the goal
will be to prove a super-polynomial lower bound on the amount of
computation required, under appropriate complexity assumptions.  The
techniques developed in the last two lectures for our communication
complexity lower bound will again prove useful for this goal, but we
will also need several additional ideas.

This lecture identifies the appropriate complexity class for
characterizing the computational complexity of computing an exact or
approximate Nash equilibrium of a bimatrix game, namely $\PPAD$.  Solar
Lecture~5 sketches some of the ideas in Rubinstein's recent
proof~\cite{R16} of a quasi-polynomial-time lower bound for the problem,
assuming an analog of the Exponential Time Hypothesis for $\PPAD$.

Section~\ref{s:preamble} explains why customized complexity classes are
needed to reason about equilibrium computation and other total search
problems.  
Section~\ref{s:tfnp} defines the class $\TFNP$ and some of its syntactic
subclasses, including $\PPAD$.\footnote{Some of the discussion in
  these two sections is drawn from~\cite[Lecture~20]{f13}.}
Section~\ref{s:ppad} reviews a number of $\PPAD$-complete problems.
Section~\ref{s:evidence} discusses the existing evidence that $\TFNP$
and its important subclasses are hard, and proves that the class
$\TFNP$ is hard on average assuming that~$\NP$ is hard on average.

\section{Preamble}\label{s:preamble}

We consider two-player (bimatrix) games, where each player has (at most)~$n$
strategies. The $n \times n$ payoff matrices for Alice and Bob~$A$
and~$B$ are described explicitly, with $A_{ij}$ and
$B_{ij}$ indicating Alice's and Bob's payoffs when Alice plays her
$i$th strategy and Bob his $j$th strategy.
%
Recall from Definition~\ref{d:ene} that an
$\epsilon$-$\NE$ is a pair $\xhat,\yhat$ of mixed strategies such that
neither player can increase their payoff with a unilateral deviation
by more than $\eps$.



What do we know about the complexity of computing an $\epsilon$-$\NE$
of a bimatrix game?  Let's start with the exact case ($\eps=0$),
where no subexponential-time (let alone polynomial-time) algorithm is
known for the problem.  (This contrasts with the zero-sum case, see
Corollary~\ref{cor:zerosum}.)
It is tempting to speculate that no such algorithm exists.  How 
would we amass evidence that the problem is intractable?  As we're
interested in super-polynomial lower bounds, communication complexity
is of no direct help.

Could the problem be $\NP$-complete?\footnote{Technically, we're
  referring to the {\em search} version of $\NP$ (sometimes called
  $\FNP$, where the ``F'' stands for ``functional''), where the goal is to either exhibit a witness or correctly
  deduce that no witness exists.}
The following theorem by \citet{MP91} rules out this
possibility (unless $\NP = \coNP$).  

\begin{theorem}[\citet{MP91}]\label{t:mp91}
The problem of computing a Nash equilibrium of a bimatrix game is
$\NP$-hard only if $\NP = \coNP$. 
\end{theorem}

\begin{figure}
\centering
\includegraphics[width=\textwidth]{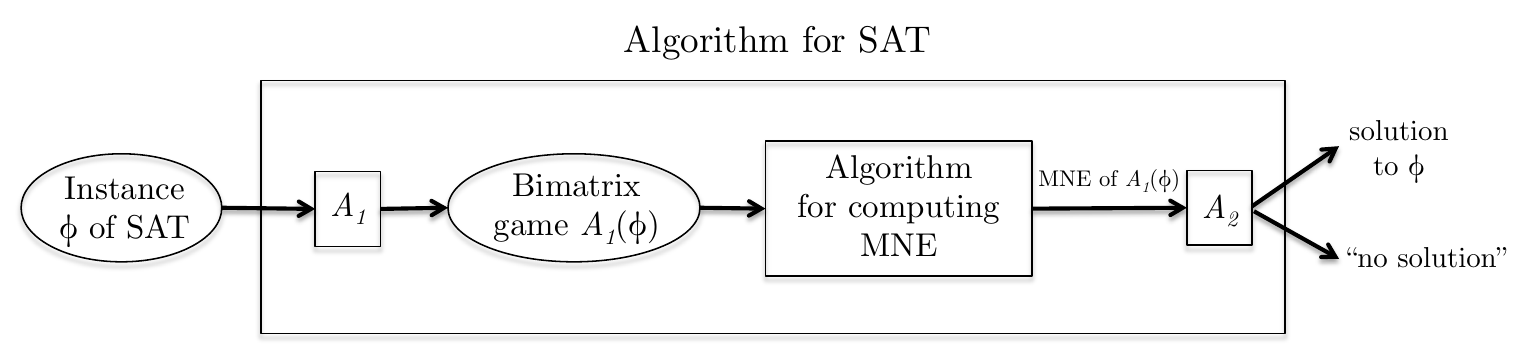}
\caption{A reduction
  from the search version of the SAT problem to the problem of
  computing a Nash equilibrium of a bimatrix game
would yield a polynomial-time verifier for the
  unsatisfiability problem.}
\label{f:mp}
\end{figure}


\begin{proof}
The proof is short but a bit of a mind-bender, analogous to the
argument back in Section~\ref{s:naive}.
Suppose there is
a reduction from, say, (the search version of) satisfiability
to the problem of computing a Nash equilibrium of a bimatrix game.  By
definition, the reduction comprises two algorithms:
\begin{enumerate}

\item A polynomial-time algorithm $\A_1$ that maps every SAT
  formula~$\phi$ to a bimatrix game $\A_1(\phi)$.

\item A polynomial-time algorithm $\A_2$ that maps every Nash
  equilibrium $\mne$ of a game $\A_1(\phi)$ to a satisfying assignment
  $\A_2\mne$ of $\phi$, if one exists, and to the string ``no''
  otherwise.

\end{enumerate}
We claim that the existence of these algorithms $\A_1$ and $\A_2$ imply that
$\NP = \coNP$ (see also Figure~\ref{f:mp}).
In proof, consider an unsatisfiable SAT formula
$\phi$, and an arbitrary Nash equilibrium $\mne$ of the game
$\A_1(\phi)$.\footnote{Crucially, $\A_1(\phi)$ has at least one Nash equilibrium, including one whose description
length is polynomial in that of the game (see 
Theorem~\ref{t:nash} and the subsequent discussion).}
We claim that $\mne$ is a short, efficiently verifiable proof of the
unsatisfiability of $\phi$, implying that $\NP = \coNP$.
Given an alleged certificate $\mne$ that
$\phi$ is unsatisfiable,  the verifier performs two checks: (1)
compute the game
$\A_1(\phi)$ using algorithm $\A_1$ and 
verify that $\mne$ is a Nash equilibrium of $\A_1(\phi)$; (2) use the
algorithm $\A_2$ to verify that $\A_2\mne$ is the string ``no.''  
This verifier runs in time polynomial in the description lengths of
$\phi$ and $\mne$. 
If $\mne$ passes both of these tests, then correctness of the
algorithms $\A_1$ and $\A_2$ implies that $\phi$ is 
unsatisfiable.
\end{proof}

\section{$\TFNP$ and Its Subclasses}\label{s:tfnp}

\subsection{$\TFNP$}

What's really going on in the proof of Theorem~\ref{t:mp91} is a
mismatch between the search version of an $\NP$-complete problem like
SAT, where an instance may or may not have a witness, and a problem
like computing a Nash equilibrium, where every instance has at least
one witness.
While the correct answer to a SAT instance might well be ``no,'' a
correct answer to an instance of Nash equilibrium computation is
always a Nash equilibrium.  It seems that if the problem of computing
a Nash equilibrium is going to be complete for some complexity class,
it must be a class smaller than~$\NP$.

The subset of $\NP$ (search) problems for
which every instance has at least one witness is called $\TFNP$, for
``total functional $\NP$.''  The proof of Theorem~\ref{t:mp91} shows more
generally that if {\em any} $\TFNP$ problem is $\NP$-complete, then
$\NP = \coNP$.  
Thus a fundamental barrier to $\NP$-completeness is the guaranteed
existence of a witness.

Since computing a Nash equilibrium does not seem to be $\NP$-complete,
the sensible refined goal is
to prove that the problem is $\TFNP$-complete---as hard as any other
$\NP$ problem with a guaranteed witness.

\subsection{Syntactic vs.\ Semantic Complexity Classes}

Unfortunately, $\TFNP$-completeness is also
too ambitious a goal.  The reason is that $\TFNP$ does not seem to have
complete problems.  Think about the complexity classes 
that {\em are} known to have complete problems---$\NP$ of course,
and also classes like $\ptime$ and $\pspace$.  What do these
complexity classes have in
common?  They are ``syntactic,'' meaning that membership can be
characterized via acceptance by some concrete computational model,
such as polynomial-time or polynomial-space deterministic or
nondeterministic Turing machines.
In this sense, there is a generic reason for membership in these
complexity classes.

Syntactically defined complexity classes always have a ``generic''
complete problem, where the input is a description of a problem in
terms of the accepting machine and an instance of the problem, and
the goal is to solve the given instance of the given problem.  For
example, the generic $\NP$-complete problem takes as input a
description of a verifier, a polynomial time bound,
and an encoding of an instance, and the goal is to decide whether
or not there is a witness, meaning a string that causes the given
verifier to accept the given instance in at most the given number of
steps.

$\TFNP$ has no obvious generic reason for membership, and as such is
called a ``semantic'' class.\footnote{There are many other interesting
  examples of classes that appear to be semantic in this sense, such
  as $\RP$ and $\NP \cap \coNP$.}  For example, the problem of
computing a Nash equilibrium of a bimatrix game belongs to $\TFNP$
because of the topological arguments that guarantee the existence of a
Nash equilibrium (see Section~\ref{s:bfp}).  Another problem in $\TFNP$ is
factoring: given a positive integer, output its factorization.  Here,
membership in $\TFNP$ has a number-theoretic
explanation.\footnote{There are many other natural examples of $\TFNP$
  problems, including computing a local minimum of a function,
  computing an approximate Brouwer fixed point, and inverting a
  one-way permutation.}  Can the guaranteed existence of a Nash
equilibrium of a
game and of a factorization of an integer be regarded as separate
instantiations of some ``generic'' $\TFNP$ argument?  No one knows the
answer.




\subsection{Syntactic Subclasses of $\TFNP$}

Given that the problem of computing a Nash equilibrium appears too
specific to be complete for $\TFNP$, we must refine our goal again,
and try to prove that the problem is complete for a still smaller
complexity class.  \citet{P94} initiated the search for syntactic
subclasses of $\TFNP$ that contain interesting problems not known to
belong to $\ptime$.  His proposal was to categorize $\TFNP$ problems
according to the type of mathematical proof used to 
guaranteed the existence of a witness.  Interesting subclasses include the
following: 
%
%
\begin{itemize}
 
\item $\PPAD$ (for polynomial parity argument, directed version):
  Problems that can be solved by path-following in a
  (exponential-size) directed graph with in- and out-degree at most~1
  and a known source vertex (specifically, the problem of identifying
  a sink or source vertex other than the given one).

   \item $\PPA$ (for polynomial parity argument, undirected version):
     Problems that can be solved by path-following in an undirected
     graph (specifically, given an odd-degree vertex, the problem of
     identifying a different odd-degree vertex).

    \item $\PLS$ (for polynomial local search): Problems that can be
      solved by path-following in a directed acyclic graph
      (specifically, given such a graph, the problem of identifying a sink
      vertex).\footnote{$\PLS$ was actually defined prior to $\TFNP$, by
        \citet{JPY88}.}

    \item $\PPP$ (for polynomial pigeonhole principle): Problems that
      reduce to the following: given a
      function~$f$ mapping $\{1,2,\ldots,n\}$ to $\{1,2,\ldots,n-1\}$,
      find $i \neq j$ such that $f(i)=f(j)$.

\end{itemize}
All of these complexity classes can be viewed as intermediate to
$\ptime$ and $\NP$.  The conjecture, supported by oracle
separations~\cite{B+98}, is that all four of these classes are distinct
(Figure~\ref{fig:belief}).
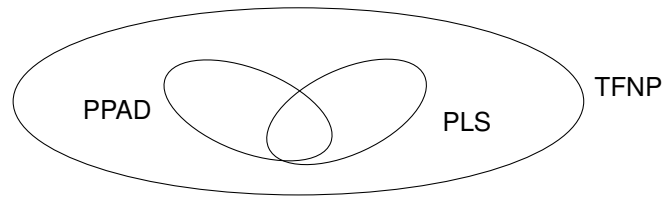
\begin{figure}[h!]
\center
\begin{tikzpicture}[line cap=round,line join=round,>=triangle 45,x=1.0cm,y=1.0cm]
\clip(-3, 1.4) rectangle (18.3,5);
\draw [rotate around={0.:(5.5,3.)}] (5.5,3.) ellipse (3.792864479119171cm and 1.2450269068279611cm);
\draw [rotate around={-21.392890351981443:(4.83,2.88)}] (4.83,2.88) ellipse (1.1797724164164307cm and 0.5536812752270697cm);
\draw [rotate around={25.01689347810002:(6.14,2.86)}] (6.14,2.86) ellipse (1.1411849820365318cm and 0.5620526338571116cm);
\draw (9.3,3.46) node[anchor=north west] {$\TFNP$};
\draw (7.28,3.) node[anchor=north west] {$\PLS$};
\draw (2.5,3.16) node[anchor=north west] {$\PPAD$};
\end{tikzpicture}
\caption{Oracle separations suggest that the different well-studied
  syntactic subclasses of $\TFNP$ are in fact distinct (from each
  other and from $\TFNP$).}
\label{fig:belief}
\end{figure}

 


Section~\ref{s:bfp} outlined the argument that the guaranteed
existence of Nash equilibria reduces to the guaranteed existence of
Brouwer fixed points, and Section~\ref{s:eol} showed (via Sperner's
lemma) that Brouwer's fixed-point theorem reduces to path-following in
a directed graph with in- and out-degrees at most~1.  Thus, $\PPAD$
would seem to be the subclass of $\TFNP$ with the best chance of
capturing the complexity of computing a Nash equilibrium.

\section{$\PPAD$ and Its Complete Problems}\label{s:ppad}

\subsection{\eol: The Generic Problem for $\PPAD$}\label{ss:seol}

We can formally define the class $\PPAD$ by defining its generic
problem.  
(A problem is then in $\PPAD$ if it reduces in polynomial time to the
generic problem.)  Just as the {\sc End-of-the-Line (\eol)} problem
served as the starting point of our communication complexity lower
bound (see Section~\ref{s:eol}), a succinct version of the problem
will be the basis for our computational hardness results.

\begin{mdframed}[style=offset,frametitle={The \eol Problem (Succinct
    Version)}]
Given two circuits $S$ and $P$ (for ``successor'' and
``predecessor''), each mapping $\zo^n$ to $\zo^n \cup \{
NULL\}$ and with size polynomial in~$n$, and with $P(0^n) = NULL$,
find an input~$v \in \zo^n$ that satisfies one of the following:
\begin{itemize}

\item [(i)] $S(v)$ is NULL;

\item [(ii)] $P(v)$ is NULL and $v \neq 0^n$;

\item [(iii)] $v \neq P(S(v))$; or

\item [(iv)] $v \neq S(P(v))$ and $v \neq 0^n$.

\end{itemize}
\end{mdframed}
Analogous to Section~\ref{s:eol}, we can view the circuits $S$ and $P$
as defining a graph~$G$ with in- and out-degrees at most~1 (with edge $(v,w)$
in $G$ if and only if $S(v) = w$ and $P(w) = v$), and with a given
source vertex~$0^n$.  The \eol problem then corresponds to identifying
either a sink vertex of~$G$ or a source vertex other
than~$0^n$.\footnote{The undirected version of the problem can be used
  to define the class $\PPA$.
The version of the problem where only sink
  vertices count as witnesses seems to give rise to a different (larger)
  complexity class called $\PPADS$.}
A solution is guaranteed to exist---if nothing else, the other end of
the path of~$G$ that originates with the vertex~$0^n$.  Thus \eol
does indeed belong to $\TFNP$, and $\PPAD \subseteq \TFNP$.  Note also
that the class is syntactic and by definition has a complete problem,
namely the \eol problem.





\subsection{Problems in $\PPAD$}

The class $\PPAD$ contains several natural problems (in addition to
the \eol problem).  For example, it contains a computational version
of Sperner's lemma---given a succinct description (e.g.,
polynomial-size circuits) of a legal coloring of an exponentially
large triangulation of a simplex, find a sub-simplex such that its
vertices showcase all possible colors.  This problem can be regarded
as a special case of the \eol problem (see Section~\ref{s:eol}), and
hence belongs to $\PPAD$.

Another example is the problem of computing an approximate fixed
point.  Here the input is a succinct description of a
$\lambda$-Lipschitz function~$f$ (on the hypercube in~$d$ dimensions,
say) and a parameter $\eps$, and the goal is to compute a point $x$
with $\n{f(x)-x} < \eps$ (with respect to some norm).  The description
length of~$x$ should be polynomial in that of
the function~$f$.  Such a point is guaranteed to exist provided $\eps$ is
not too small relative to~$\lambda$.%
\footnote{For example, for the $\ell_{\infty}$ norm, existence of such
  a point is guaranteed with $\eps$ as small as
  $\tfrac{(\lambda+1)}{2^n}$, where~$n$ is the description length
  of~$f$.  This follows from rounding each coordinate of an exact
  fixed point to its nearest multiple of $2^{-n}$.}  The reduction
from Brouwer's fixed-point theorem to Sperner's lemma (with colors
corresponding to directions of movement, see Section~\ref{s:bfp})
shows that computing an approximate fixed point can also be regarded
as a special case of the \eol problem, and hence belongs to $\PPAD$.

The problem of computing an exact or approximate Nash equilibrium of a
bimatrix game also belongs to $\PPAD$.  For the problem of computing
an $\eps$-approximate Nash equilibrium (with $\eps$ no smaller than
inverse exponential in~$n$), this follows from the proof of
Nash's theorem outlined in Section~\ref{ss:nashpf}.
That proof shows that
computing an $\eNE$ is a special case of computing an
approximate fixed point (of the regularized best-response
function defined in~\eqref{eq:nash1} and~\eqref{eq:nash2}), and hence
the problem belongs to $\PPAD$.
The same argument shows that this is true more generally with any
finite number of players (i.e.,  not only for bimatrix games).

The problem of computing an exact Nash equilibrium ($\eps=0$) also
belongs to $\PPAD$ in the case of two-player (bimatrix)
games.\footnote{\citet{EY07} proved that, with 3 or more players, the
  problem of computing an exact Nash equilibrium of a game appears to
  be strictly harder than any problem in $\PPAD$.}  One way to prove
this is via the Lemke-Howson algorithm~\cite{LH64} (see also
Section~\ref{s:bimatrix}), which reduces the computation of an (exact)
Nash equilibrium of a bimatrix game to a path-following problem, much
in the way that the simplex method reduces computing an optimal
solution of a linear program to following a path of improving edges
along the boundary of the feasible region.  
The proof of the
Lemke-Howson algorithm's inevitable convergence uses parity arguments
akin to the one in the proof of Sperner's lemma,
and shows that the problem of computing a Nash equilibrium of a
bimatrix game belongs to $\PPAD$.


\subsection{$\PPAD$-Complete Fixed-Point Problems}\label{ss:ppadbfp}

The \eol problem is $\PPAD$-complete by construction.  What about
``more natural'' problems?  Papadimitriou~\cite{P94} built evidence
that $\PPAD$ is a fundamental complexity class by showing that
fixed-point problems are complete for it.
%
%
%

To be precise, let \brouwer$(\n{\cdot}, d, \F, \epsilon)$ denote the
following problem: given a (succinct description of a) function $f \in
\F$, with $f:[0,1]^d \rightarrow [0,1]^d$, compute a point $x \in
[0,1]^d$ such that $\n{f(x)-x} < \eps$.  The original hardness result
from~\cite{P94} is the following.


\begin{theorem}[\citet{P94}]\label{t:brouwer1}
The \brouwer$(\n{\cdot}, d, \F, \epsilon)$ problem is
$\PPAD$-complete, even when $d=3$, the functions in~$\F$ are
$O(1)$-Lipschitz, $\n{\cdot}$ is the $\ell_{\infty}$ norm, and $\eps$
is exponentially small in the description
length~$n$ of a function~$f \in \F$.
\end{theorem}
The high-level idea of the proof is similar to the construction in
Section~\ref{s:ccbfp} that 
shows how to interpret \eol instances as
implicitly defined Lipschitz functions on the hypercube.  Given
descriptions of the circuits $S$ and $P$ in an instance of the generic
\eol problem, it is possible to define an (efficiently computable)
function whose gradient ``follows the line'' of an embedding of the induced
directed graph into the hypercube.  Three dimensions are needed in the
construction in~\cite{P94} to ensure that the images of different
edges do not intersect (except at a shared endpoint).
Some time later, \citet{CD09} used a somewhat
different approach to prove that Theorem~\ref{t:brouwer1} holds even
when $d=2$.\footnote{The one-dimensional case can be solved in
  polynomial time, essentially by binary search.}

Much more recently, with an eye toward hardness results for
$\eps$-approximate Nash equilibria with constant~$\eps$ (see Solar Lecture~5),
\citet{R16} proved the following.\footnote{Theorem~\ref{t:brouwer1}
  proves hardness in the regime where $d$ and $\eps$ are both small,
  Theorem~\ref{t:brouwer2} when both are large.  This is not an
  accident; if $d$ is small (i.e., constant) and $\eps$ is large
  (i.e., constant), the problem can be solved in polynomial time by
  exhaustively checking a constant number 
of evenly spaced grid points.}
\begin{theorem}[\citet{R16}]\label{t:brouwer2}
  The \brouwer$(\n{\cdot}, d, \F, \epsilon)$ problem is
  $\PPAD$-complete even when 
the functions in~$\F$ are $O(1)$-Lipschitz functions,
$d$ is linear in the description length~$n$ of a function in~$\F$, 
$\n{\cdot}$ is the normalized $\ell_2$ norm
  (with $\n{x} = \sqrt{\tfrac{1}{d} \sum_{i=1}^d x_i^2}$), and $\eps$
  is a sufficiently small constant.
\end{theorem}

The proof of Theorem~\ref{t:brouwer2} is closely related to the third
step of our communication complexity lower bound
(Section~\ref{s:ccbfp}), and in particular makes use of a similar
embedding of graphs into the hypercube with the properties~(P1)
and~(P2) described in Section~\ref{ss:embed1}.%
\footnote{We have reversed the chronology;
  Theorem~\ref{t:br17} was proved after Theorem~\ref{t:brouwer2} and
  used the construction in~\cite{R16} more or less as a black box.}
One major difference is that our proof of existence of the embedding
in Section~\ref{s:ccbfp} used the probabilistic method and hence is
not constructive (which is not an issue in the two-party communication
model), while the computational lower bound in
Theorem~\ref{t:brouwer2} requires an efficiently computable embedding.
In particular, the
reduction from \eol to \brouwer$(\n{\cdot}, d, \F, \epsilon)$ must
efficiently produce a succinct description of the function~$f$ induced
by an instance of \eol, and it should be possible to efficiently
evaluate~$f$, presumably while using the given \eol circuits~$S$ and~$P$
only as black boxes.
For example, it should be possible to efficiently decode points of
the hypercube (to a vertex, edge, or~$\bot$, see
Section~\ref{ss:embed1}).

Conceptually, the fixes for these problems are relatively simple.
First, rather than mapping the vertices 
randomly into the hypercube, the reduction in the proof of
Theorem~\ref{t:brouwer2} embeds the vertices using an error-correcting
code (with constant rate and efficient encoding and decoding
algorithms).  This enforces property~(P1) of Section~\ref{ss:embed1}.
Second, rather than using a straight-line embedding, the reduction is
more proactive about making the images of different edges stay far
apart (except for at shared endpoints).  Specifically, an edge of the
directed graph induced by the given \eol instance is now mapped to~4
straight line segments, and along each line segment, two-thirds of the
coordinates stay fixed.  (This requires blowing up the number of
dimensions by a constant factor.)  For example, the directed
edge~$(u,v)$ can be mapped to the path
\[
(\sigma(u),\sigma(u),\mathbf{\tfrac{1}{4}}) \mapsto
(\sigma(u),\sigma(v),\mathbf{\tfrac{1}{4}}) \mapsto
(\sigma(u),\sigma(v),\mathbf{\tfrac{3}{4}}) \mapsto
(\sigma(v),\sigma(v),\mathbf{\tfrac{3}{4}}) \mapsto
(\sigma(v),\sigma(v),\mathbf{\tfrac{1}{4}}),
\]
where~$\sigma$ denotes the error-correcting code used to map the
vertices to the hypercube and the boldface $\mathbf{\tfrac{1}{4}}$ and
$\mathbf{\tfrac{3}{4}}$ indicate the value of the last third of the
coordinates.  This maneuver enforces property~(P2) of
Section~\ref{ss:embed1}.  It also ensures that it is easy to decode
points of the hypercube that are close to the image of an edge of the
graph---at least one of the edge's endpoints can be recovered from the values
of the frozen coordinates, and the other endpoint can be recovered using
the given predecessor and successor circuits.\footnote{This embedding
  is defined only for the directed
  edges that are present in the given \eol instance, rather than for all
  possible edges (in contrast to
the embedding in Sections~\ref{ss:embed1}   and~\ref{ss:embed2}).}

\subsection{$\PPAD$-Complete Equilibrium Computation Problems}




\citet{P94} defined the class $\PPAD$ in large part to capture the
complexity of computing a Nash equilibrium, conjecturing that the
problem is in fact $\PPAD$-complete.  Over a decade later, a flurry of
papers confirmed this conjecture.  First, Daskalakis, Goldberg, and
Papadimitriou~\cite{DGP06,GP06} proved that computing an $\eNE$ of a
four-player game, with $\eps$ inverse exponential in the size of the
game, is $\PPAD$-complete.  This approach was quickly
refined~\cite{CD05,DP05}, culminating in the proof of Chen and
Deng~\cite{CD06} that computing a Nash equilibrium (or even an $\eNE$
with exponentially small $\eps$) of a bimatrix game is
$\PPAD$-complete.  Thus the nice properties possessed by Nash
equilibria of bimatrix games (see Section~\ref{s:bimatrix}) are not
enough to elude computational intractability.  \citet{CDT06}
strengthened this result to hold even for values of $\eps$ that are
only inverse polynomial in the size of the game.\footnote{In
  particular, under standard complexity assumptions, this rules out an
  algorithm for computing an exact Nash equilibrium of a bimatrix game
  that has smoothed polynomial complexity in the sense of
  \citet{ST04}.  Thus the parallels between the simplex method and the
  Lemke-Howson algorithm (see Section~\ref{s:bimatrix}) only go so
  far.}  The papers by \citet{DGP09} and \citet{CDT09} give a full
account of this breakthrough sequence of results.
\begin{theorem}[\citet{DGP09,CDT09}]\label{t:cdt}
The problem of computing an $\eNE$ of an $n \times n$ bimatrix game is
$\PPAD$-complete, even when $\eps = 1/\poly(n)$.
\end{theorem}
The proof of Theorem~\ref{t:cdt}, which is a tour de force, is also
outlined in the surveys by \citet{J07}, \citet{P07}, \citet{DGPcacm},
and \citet{et}.  Fundamentally, the proof shows how to define a
bimatrix game so that every Nash equilibrium effectively performs a
gate-by-gate simulation of the circuits of a given \eol instance.

Theorem~\ref{t:cdt} left open the possibility that,
for every constant $\eps > 0$,
an $\eNE$ of a bimatrix game can be computed in polynomial time.
(Recall from Corollary~\ref{cor:lmm2} that one can be computed in {\em
  quasi-polynomial} time.)
A decade later,  \citet{R16} ruled out this possibility (under
suitable complexity assumptions) by proving a
quasi-polynomial-time hardness result for the problem 
when~$\eps$ is a sufficiently small constant.
We will have much more to say about this result in Solar Lecture~5.





\section{Are $\TFNP$ Problems Hard?}\label{s:evidence}

It's all fine and good to prove that a problem is as hard as any other
problem in $\PPAD$, but what makes us so sure that $\PPAD$ problems
(or even $\TFNP$ problems) can be computationally difficult?  

\subsection{Basing the Hardness of $\TFNP$ on Cryptographic
  Assumptions}

The first evidence of hardness of problems in $\TFNP$ came in the form of
exponential lower bounds for functions given as ``black
boxes,'' or equivalently query complexity lower bounds, as in
Proposition~\ref{c:eol} for the \eol problem or \citet{HPV89} for the
\brouwer~problem.

Can we relate the hardness of $\TFNP$ and its subclasses to other
standard complexity assumptions?  Theorem~\ref{t:mp91} implies that
we can't base hardness of $\TFNP$  on the assumption that $\ptime \neq
\NP$, unless $\NP = \coNP$.  What about cryptographic assumptions?
After all, the problem of inverting a one-way permutation belongs to
$\TFNP$ (and even the subclass $\PPP$).  Thus, sufficiently strong
cryptographic assumptions imply hardness of $\TFNP$.

Can we prove hardness also for all of the other interesting subclasses of
$\TFNP$, or can we establish the hardness of $\TFNP$ under weaker
assumptions (like the existence of one-way functions)?  
Along the former lines, a recent sequence of papers (not discussed
here) show that sufficiently strong cryptographic assumptions 
imply that $\PPAD$ is hard~\cite{BPR15,GPS16,RSS17,HY17,C+19}.
The rest of this lecture covers a recent result in the second
direction by \citet{HNY17},
who show that the average-case hardness of $\TFNP$ can be based on the
average-case hardness of $\NP$.  (Even though the worst-case hardness
of $\TFNP$ {\em cannot} be based on that of $\NP$, unless
$\NP=\coNP$!)
Note that assuming that $\NP$ is hard on average is only weaker than
assuming
the existence of one-way functions.


\begin{theorem}[\citet{HNY17}]
\label{theorem:average_hard}
If there exists a hard-on-average
language in $\NP$, then there
exists a hard-on-average search problem in $\TFNP$.
\end{theorem}
There is some fine print in the precise statement of the result (see
Remarks~\ref{rem:public} and~\ref{rem:uniform}), but the statement in
Theorem~\ref{theorem:average_hard} is the gist of it.\footnote{The
  amount of fine print was reduced very recently by \citet{PV19}.}

\subsection{Proof Sketch of Theorem~\ref{theorem:average_hard}}




Let $L$ be a language in $\NP$ that is hard on average w.r.t.\ some
family of distributions $D_n$ on input strings of length $n$. 
Average-case hardness of $(L, D_n)$ means that there is no
polynomial-time algorithm with an advantage of $1/\poly(n)$ over
random guessing when the input is sampled according to $D_n$ (for any
polynomial).  Each~$D_n$ should be efficiently sampleable, so that
hardness cannot be baked into the input distribution.
%
Can we convert such a problem into one that is
total while retaining its average-case hardness?



Here's an initial attempt:
\begin{mdframed}[style=offset,frametitle={Attempt \#1}]
{\bf Input:} $l$ independent samples $x_1, x_2, \ldots, x_l$ from
$D_n$.\\ 
{\bf Output:} a witness for some $x_i \in L$.
\end{mdframed}
For sufficiently large $l$, this problem is 
``almost total.''  Because $(L,D_n)$ is hard-on-average, random
instances are nearly equally likely to be ``yes'' or ``no'' instances
(otherwise a constant response would beat random guessing).
Thus, except with probability $\approx 2^{-l}$, at least one of the
sampled instances~$x_i$ is a ``yes'' instance and has a witness.
Taking~$l$ polynomial in~$n$, we get a problem that is total except
with exponentially small probability.  How can we make it ``totally
total?''


The idea is to sample the $x_i$'s in a correlated way, using a random
shifting trick reminiscent of Lautemann's proof that $\BPP \subseteq
\Sigma_2 \cap \Pi_2$~\cite{L83}.  This will give a non-uniform version
of Theorem~\ref{theorem:average_hard}; Remark~\ref{rem:uniform}
sketches the changes necessary to get a uniform version.

Fix~$n$.  Let $D_n(r)$ denote the output of the sampling algorithm
for~$D_n$, given the
random seed $r \in \zo^n$.  (By padding, we can assume that the
input length and the random seed length both equal~$n$.)
%
%
%
Call a set containing the strings $s_1, s_2, \ldots, s_l \in \zo^n$
{\em good} if
for every seed $r\in \{0,1\}^n$ there exists an index $i\in [l]$ such
that $D(r\oplus s_i) \in L$.
We can think of the $s_i$'s as masks; goodness then means that there
is always a mask whose application yields a ``yes'' instance.


\begin{claim}
\label{claim:goodness}
If $s_1, s_2, \ldots, s_{2n} \sim \{0,1\}^n$ are sampled uniformly and
independently, then $\{ s_1,\ldots,s_{2n} \}$ is good 
except with exponentially small probability.
\end{claim}

\begin{proof}
  Fix a seed $r\in \{0,1\}^n$.  The distribution of $r \oplus s_i$
  (over $s_i$) is uniform, so $D_n(r \oplus s_i)$ has a roughly 50\%
  chance of being a ``yes'' instance (since $(L,D_n)$ is hard on
  average).  Thus the probability (over $s_1,\ldots,s_{2n}$) that
  $D_n(r \oplus s_i)$ is a ``no'' instance for {\em every} $s_i$ is
  $\approx 2^{-2n}$.  Taking a union bound over the $2^n$ choices
  for~$r$ completes the proof.
%
\end{proof}

Consider now the following reduction, from the assumed hard-on-average
$\NP$ problem $(L,D_n)$ to a hopefully hard-on-average $\TFNP$ problem.
\begin{mdframed}[style=offset,frametitle={\bf Attempt 2} (non-uniform)]
{\bf Chosen in advance:} A good set of strings $\{s_1, s_2, \ldots,
s_{2n}\}$.\\
 {\bf Input:} an instance $x$ of $(L,D_n)$, in the form of the random
 seed $\hat{r}$ used to generate $x = D_n(\hat{r})$.\\
 {\bf Output:} a witness for one of the instances
$D(\hat{r}\oplus s_1),\ldots,D(\hat{r}\oplus s_{2n})$.
\end{mdframed}
By the definition of a good set of strings, there is always at least
one witness of the desired form, and so the output of this reduction
is a $\TFNP$ problem (or more accurately, a $\TFNP/\poly$ problem,
with $s_1,\ldots,s_{2n}$ given as advice).
Let $D'$ denote the distribution over instances of this problem
induced by the uniform distribution over $\hat{r}$.  It remains to
show how a (non-uniform) algorithm that solves this
$\TFNP/\poly$ problem (with respect to~$D'$)
can be used to beat random guessing (with inverse polynomial
advantage) for $(L,D_n)$ in a comparable amount of time.
%
%
Given an algorithm~$A$ for the former problem
(and the corresponding good set of strings), 
consider the following algorithm $B$ for~$(L,D_n)$.

\vspace{10pt}
\begin{mdframed}[style=offset,frametitle={Algorithm $B_{s_1, s_2,
      \ldots, s_{2n}}$}]
{\bf Input:} A random instance $x$ of $(L,D_n)$ and the random
seed~$\hat{r}$ that generated it (so $x = D_n(\hat{r})$).
\begin{enumerate}
\item Choose $i\in [2n]$ uniformly at random.
\item Set $r^\star = \hat{r} \oplus s_i$.
\item Use the algorithm $A$ to generate a witness~$w$ for one
  of the instances
$$D(r^\star\oplus s_1), D(r^\star\oplus s_2), \ldots, D(r^\star\oplus s_{2n}).$$
(Note that the $i$th problem is precisely the one we want to solve.)
\item If $w$ is a witness for $D(r^\star\oplus s_i)$,
then output  ``yes.'' 
\item 
Otherwise, randomly answer ``yes'' or ``no'' (with 50/50
  probability).
\end{enumerate}
\end{mdframed}

Consider a ``yes'' instance~$D_n(\hat{r})$ of~$L$.  If algorithm~$A$
happens to output a witness to the $i$th instance $D_n(r^\star\oplus s_i) =
D_n(\hat{r})$, then algorithm~$B$ correctly decides the problem.  The
worry is that the algorithm~$A$ somehow conspires
to always output a witness for an instance other than the ``real''
one.  

Suppose algorithm~$A$, when presented with the instances
$D(r^\star\oplus s_1), D(r^\star\oplus s_2), \ldots, D(r^\star\oplus
s_{2n})$, exhibits a witness for the $j$th instance $D(r^\star\oplus
s_j)$.  This collection of instances could have been produced by the
reduction in exactly~$2n$ different ways: 
with $i=1$ and $\hat{r}=r^{\star} \oplus s_1$,
with $i=2$ and $\hat{r}=r^{\star} \oplus s_2$, and so on.
Since~$i$ and~$\hat{r}$ were chosen independently and uniformly at
random, each of these~$2n$ outcomes is equally likely, and
algorithm~$A$ has no way of distinguishing between them.  Thus
whatever~$j$ is, $A$'s witness has at least a $1/2n$ chance of being a
witness for the true problem $D_n(\hat{r})$ (where the probability is
over both~$\hat{r}$ and~$i$).  We conclude that, for ``yes'' instances
of~$L$, algorithm~$B$ has advantage~$\tfrac{1}{2n}$ over random
guessing.  Since roughly 50\% of the instances $D_n(\hat{r})$ are
``yes'' instances (since $(L,D_n)$ is average-case hard),
algorithm~$B$ has advantage roughly~$\tfrac{1}{4n}$ over random
guessing for~$(L,D_n)$.  This contradicts our assumption that
$(L,D_n)$ is hard on average.




We have completed the proof of Theorem~\ref{theorem:average_hard},
modulo two caveats.


\begin{remark}[Public vs.\ Private Coins]\label{rem:public}
  The algorithm~$B$ used in the reduction above beats random guessing
  for $(L,D_n)$, provided the algorithm receives as input the random
  seed $\hat{r}$ used to generate an instance of~$(L,D_n)$.  That is,
  our current proof of Theorem~\ref{theorem:average_hard} assumes
  that~$(L,D_n)$ is hard on average {\em even with public coins}.
  While there are problems in $\NP$ conjectured to be average-case
  hard in this sense (like random SAT near the phase transition),
  it would be preferable to have a
  version of Theorem~\ref{theorem:average_hard} that allows for
  private coins.
Happily, \citet{HNY17} prove
  that there exists a private-coin average-case hard problem in $\NP$
  only if there is also a public-coin such problem.  This implies that
  Theorem~\ref{theorem:average_hard} holds also in the private-coin
  case.
\end{remark}

\begin{remark}[Uniform vs.\ Non-Uniform]\label{rem:uniform}
Our proof of Theorem~\ref{theorem:average_hard} only proves hardness
for the non-uniform class $\TFNP/\poly$.  (The good set $\{
s_1,\ldots,s_{2n} \}$ of strings is given as ``advice'' separately
for each~$n$.)  It is possible to extend the argument to (uniform)
$\TFNP$, under some additional (reasonably standard) complexity assumptions.
The idea is to use techniques from derandomization.  We already know
from Claim~\ref{claim:goodness}
that almost all sets of $2n$ strings from $\zo^n$ are good.  Also,
the problem of checking whether or not a set of strings is good is a
$\Pi_2$ problem (for all $r \in \zo^n$ there exists $i \in [2n]$ such
that $D_n(r \oplus s_i)$ has a witness).  Assuming that there is a
problem in $\E$ with exponential-size $\Pi_2$ circuit complexity, it
is possible to derandomize the probabilistic argument and efficiently
compute a good set $\{ s_1,\ldots,s_l\}$ of strings (with $l$ larger
than $2n$ but still polynomial in~$n$), \`{a} la \citet{IW97}.
\end{remark}

An important open research direction is to extend
Theorem~\ref{theorem:average_hard} to subclasses of $\TFNP$, such as
$\PPAD$.


\vspace{.5\baselineskip}

\noindent{\bf Open Problem:} Does an analogous
average-case hardness result hold for $\PPAD$?

\lecture{The Computational Complexity
  of Computing an Approximate Nash Equilibrium}

\vspace{1cm}



\section{Introduction}



Last lecture we stated without proof the result by \citet{DGP09} and
\citet{CDT09} that computing an $\eps$-approximate Nash equilibrium of
a bimatrix game is $\PPAD$-complete, even when $\eps$ is an inverse
polynomial function of the game size (Theorem~\ref{t:cdt}).  Thus, it
would be surprising if there were a polynomial-time (or even
subexponential-time) algorithm for this problem.
%
%
%
Recall from Corollary~\ref{cor:lmm2} in Solar Lecture~1 that the story is
different for constant values of $\eps$, where an $\eps$-approximate
Nash equilibrium can be computed in quasi-polynomial (i.e., $n^{O(\log
  n)}$) time.



The Pavlovian response of a theoretical computer scientist to a
quasi-polynomial-time algorithm is to conjecture that a
polynomial-time algorithm must also exist.  (There are only a few
known natural problems that appear to have inherently quasi-polynomial
time complexity.)  But recall that the algorithm in the proof of
Corollary~\ref{cor:lmm2} is just exhaustive search over all
probability distributions that are uniform over a multi-set of
logarithmically many strategies (which is good enough, by
Theorem~\ref{t:lmm}).  Thus the algorithm reveals no structure of the
problem other than the fact that the natural search space for it has
quasi-polynomial size.  It is easy to imagine that there are no
``shortcuts'' to searching this space, in which case a
quasi-polynomial amount of time would indeed be necessary.  How would
we ever prove such a result?  Presumably by a non-standard
super-polynomial reduction from some $\PPAD$-complete problem like
succinct \eol (defined in Section~\ref{ss:seol}).  This
might seem hard to come by, but in a recent breakthrough, \citet{R16}
provided just such a reduction!
\begin{theorem}[\cite{R16}] \label{thm:Aviad} For all sufficiently
  small constants $\eps>0$, for every constant $\delta > 0$,
there is no $n^{O(\log^{1-\delta} n)}$-time
algorithm for computing an $\eps$-approximate Nash equilibrium of a
bimatrix game, unless the succinct \eol problem has a
$2^{O(n^{1-\delta'})}$-time algorithm for some constant $\delta' > 0$.
\end{theorem}
In other words, assuming an analog of the Exponential Time
Hypothesis (ETH)~\cite{IPZ01} for $\PPAD$, the quasi-polynomial-time
algorithm in Corollary~\ref{cor:lmm2} is essentially
optimal!\footnote{To obtain a quantitative lower bound like the
  conclusion of Theorem~\ref{thm:Aviad}, it is necessary to make a
  quantitative complexity assumption (like an analog of ETH).  This
  approach belongs to the tradition of ``fine-grained'' complexity
  theory.}\footnote{How plausible is the assumption that the ETH holds
  for $\PPAD$, even after assuming that the ETH holds for $\NP$ and
 that $\PPAD$ has no polynomial-time algorithms?  The answer is far from
  clear, although there are exponential query lower bounds for $\PPAD$
  problems (e.g.~\cite{HPV89}) and no known techniques that show
  promise for a subexponential-time algorithm for the succinct \eol
  problem.}  

Three previous papers that used an ETH assumption (for
$\NP$) along with PCP machinery to prove quasi-polynomial-time lower
bounds for $\NP$ problems are:
\begin{enumerate}

\item
\citet{AaronsonImMo14}, for the problem of computing the value of free
games (i.e., two-prover proof systems with stochastically independent
questions), up to additive error $\eps$;

\item
\citet{BravermanYoWe15}, for the problem of computing the
$\eps$-approximate Nash equilibrium with the highest expected
sum of player payoffs; and

\item \citet{BKRW17}, for the problem of distinguishing graphs with a
$k$-clique from those that only have $k$-vertex subgraphs with density
at most $1-\eps$.

\end{enumerate}
In all three cases, the hardness results apply when $\eps > 0$ is a
sufficiently small constant.  Quasi-polynomial-time algorithms are
known for all three problems.


The main goal of this lecture is to convey some of the ideas in the
proof of Theorem~\ref{thm:Aviad}.  The proof is a tour de force and
the paper~\cite{R16} is~57 pages long, so our treatment will
necessarily be impressionistic.  
We hope to explain the following:
\begin{enumerate}

\item What the reduction in Theorem~\ref{thm:Aviad} must look like.
  (Answer: a blow-up from size~$n$ to size $\approx 2^{\sqrt{n}}$.)

\item How a $n \mapsto \approx 2^{\sqrt{n}}$-type blowup
  can naturally arise in a reduction to the problem of computing an
  approximate Nash equilibrium.

\item Some of the tricks used in the reduction.

\item Why these tricks naturally lead to the development and
  application of PCP machinery.

\end{enumerate}

%


\section{Proof of Theorem~\ref{thm:Aviad}: An Impressionistic Treatment}

\subsection{The Necessary Blow-Up}

The goal is to reduce length-$n$ instances of the succinct \eol problem
to length-$f(n)$ instances of the problem of computing an
$\eps$-approximate Nash equilibrium with constant $\eps$, so that a
sub-quasi-polynomial-time algorithm for the latter implies a
subexponential-time algorithm for the former.  Thus the mapping $n
\mapsto f(n)$ should satisfy $2^n \approx f(n)^{\log f(n)}$
and hence $f(n) \approx 2^{\sqrt{n}}$.  That is, we should be looking
to encode a length-$n$ instance of succinct \eol as a $2^{\sqrt{n}}
\times 2^{\sqrt{n}}$ bimatrix game.
%
%
The $\sqrt{n}$ will essentially come from the ``birthday paradox,''
with random subsets of~$[n]$ of size~$s$ likely to intersect once
$s$ exceeds~$\sqrt{n}$.
The blow-up from $n$ to
$2^{\sqrt{n}}$ will come from PCP-like machinery, as well as a
game-theoretic gadget (``Alth\"ofer games,'' see Section~\ref{ss:althofer})
that forces players to randomize nearly uniformly over size-$\sqrt{n}$
subsets of~$[n]$ in every approximate Nash equilibrium.

\subsection{The Starting Point: \bfp}

The starting point of the reduction is the $\PPAD$-complete version of
the \bfp problem in Theorem~\ref{t:brouwer2}.  We restate that result
here.
\begin{theorem}[\citet{R16}]\label{t:brouwer3}
  The \brouwer$(\n{\cdot}, d, \F, \epsilon)$ problem is
  $\PPAD$-complete when 
the functions in~$\F$ are $O(1)$-Lipschitz functions from
 the $d$-dimensional hypercube $H=[0,1]^d$ to itself, 
$d$ is linear in the description length~$n$ of a function in~$\F$, 
$\n{\cdot}$ is the normalized $\ell_2$ norm
  (with $\n{x} = \sqrt{\tfrac{1}{d} \sum_{i=1}^d x_i^2}$), and $\eps$
  is a sufficiently small constant.
\end{theorem}
The proof is closely related to the reduction from \cceol to \ccbfp
outlined in Section~\ref{s:ccbfp}, and Section~\ref{ss:ppadbfp}
describes the additional ideas needed to prove
Theorem~\ref{t:brouwer3}.  As long as the error-correcting code used
to embed vertices into the hypercube (see Section~\ref{ss:ppadbfp})
has linear-time encoding and decoding algorithms (as in~\cite{S97},
for example), the reduction can be implemented in linear time.  In
particular, our assumption that the succinct \eol problem has no
subexponential-time algorithms automatically carries over to this
version of the \bfp problem.  In addition to the properties of the
functions in~$\F$ that are listed in the statement of
Theorem~\ref{t:brouwer3}, the proof of Theorem~\ref{thm:Aviad}
crucially uses the ``locally decodable'' properties of these functions
(see Section~\ref{ss:local}).



\subsection{\bfp $\le \eNE$ (Attempt \#1): Discretize McLennan-Tourky}

One natural starting point for a reduction from \bfp to $\eNE$ is the
McLennan-Tourky analytic reduction in Section~\ref{ss:mt06}.  Given a
description of an~$O(1)$-Lipschitz
function~$f:[0,1]^d \rightarrow [0,1]^d$, with~$d$ linear in the
length~$n$ of the function's description, the simplest reduction would
proceed as follows.  Alice and Bob each have a strategy set
corresponding to the discretized hypercube~$H_{\eps}$ (points of
$[0,1]^d$ such that every coordinate is a multiple of $\eps$).
Alice's and Bob's payoffs are defined as in the proof of
Theorem~\ref{t:mt06}: for strategies $x,y \in H_{\eps}$, Alice's
payoff is
\begin{equation}\label{eq:apayoff3}
1 - \n{x-y}^2 = 
1 - \frac{1}{d} \sum_{i=1}^d (x_i-y_i)^2
\end{equation}
and Bob's payoff is
\begin{equation}\label{eq:bpayoff3}
1 - \n{y-f(x)}^2 = 
1 - \frac{1}{d} \sum_{j=1}^d (y_j-f(x)_j)^2.
\end{equation}
(Here $\n{\cdot}$ denotes the normalized $\ell_2$ norm.)
Thus Alice wants to imitate Bob's strategy, while Bob wants to imitate
the image of Alice's strategy under the function~$f$.



This reduction is correct in that in every $\eps$-approximate Nash equilibrium of
this game, Alice's and Bob's strategies are concentrated around an
$O(\eps)$-approximate fixed point 
of the given function~$f$
(in the normalized $\ell_2$ norm).  See also the discussion in Section~\ref{ss:mt06}.

The issue is that the reduction is not efficient enough.  
Alice and Bob each have $\Theta((1/\eps)^d)$ pure strategies; since $d
= \Theta(n)$, this is exponential in the size~$n$ of the given \bfp
instance, rather than exponential in $\sqrt{n}$.
This exponential blow-up in size means that this reduction has no
implications for the problem of computing an approximate Nash equilibrium.

\subsection{Separable Functions}\label{ss:separable}

How can we achieve a blow-up exponential in~$\sqrt{n}$ rather than
in~$n$?  We might guess that the birthday paradox is somehow
involved.  To build up our intuition, we'll discuss at length a
trivial special case of the \bfp problem.  It
turns out that the hard functions used in Theorem~\ref{t:brouwer3} are
in some sense surprisingly close to this trivial case.

For now, we consider only instances~$f$ of \bfp where $f$ is {\em
  separable}. 
That is, $f$ has the form
\begin{equation}\label{eq:separable}
f(x_1,\ldots,x_d) = (f_1(x_1),\ldots,f_d(x_d))
\end{equation}
for efficiently computable functions $f_1,\ldots,f_d :
[0,1]\rightarrow [0,1]$.
Separable functions enjoy the ultimate form of ``local decodability''---to
compute the $i$th coordinate of $f(x)$, you only need to know the
$i$th coordinate of~$x$.
Finding a fixed point of a separable function is easy: the problem decomposes
into $d$ one-dimensional fixed point problems (one per coordinate),
and each of these can be solved efficiently by a form of binary
search.  The hard functions used in Theorem~\ref{t:brouwer3} possess
a less extreme form of ``local decodability,'' in that each coordinate
of $f(x)$ can be computed using only a small amount of ``advice'' 
about~$f$ and~$x$ 
(cf., the \ccbfp $\le \eNE$ reduction in Section~\ref{ss:step4}).

\subsection{\bfp $\le \eNE$ (Attempt \#2): Coordinatewise Play}\label{ss:coordinatewise}

Can we at least compute fixed points of separable functions via
approximate Nash equilibria, using a reduction with only subexponential
blow-up?  The key idea is, instead of Alice and Bob each picking one
of the (exponentially many) points of the discretized hypercube
$H_{\eps}$, each will pick {\em only a single coordinate} of points
$x$ and $y$.  Thus a pure strategy of Alice comprises an index
$i \in [d]$ and a number $x_i \in [0,1]$ that is a multiple of $\eps$,
and similarly Bob chooses~$j \in [d]$ and $y_j \in [0,1]$.  Given
choices $(i,x_i)$ and $(j,y_j)$, Alice's payoff is defined as
\[
  \begin{cases}
   1 - (x_i-y_i)^2     & \quad \text{if $i=j$} \\
    0 & \quad \text{if $i \neq j$} \\
  \end{cases}
\]
and Bob's payoff is
\[
  \begin{cases}
   1 - (y_i-f_i(x_i))^2     & \quad \text{if $i=j$} \\
    0 & \quad \text{if $i \neq j$.} \\
  \end{cases}
\]
Thus Alice and Bob receive payoff~0 unless they ``interact,'' meaning
choose the same coordinate to play in, in which case their payoffs are
analogous to~\eqref{eq:apayoff3} and~\eqref{eq:bpayoff3}.  Note that
Bob's payoff is well defined only because we have assumed that~$f$ is
separable (Bob only knows the coordinate $x_i$ proposed by Alice, but
this is enough to compute the $i$th coordinate of the output of~$f$
and hence his payoff).  Each player has only $\approx \tfrac{d}{\eps}$
strategies, so this is a polynomial-time reduction, with no blow-up.
%
%


The good news is that (approximate) fixed points give rise to
(approximate) Nash equilibria of this game.  Specifically,
if $\hat{x}=\hat{y}=f(\hat{x})$ is a fixed point of $f$, then the
following is a Nash equilibrium (as you should check): Alice and Bob
pick their coordinates $i,j$ uniformly at random and set
$x_i=\hat{x}_i$ and $y_j=\hat{y}_j$.
The problem is that
the game also has equilibria other than the intended ones, for
example where Alice and Bob choose pure strategies with $i=j$ and
$x_i=y_i=f_i(x_i)$.

\subsection{\bfp $\le \eNE$ (Attempt \#3): Gluing Alth\"ofer Games}\label{ss:althofer}

Our second attempt failed because Alice and Bob were not forced to
randomize their play over all~$d$ coordinates.  We can address this
issue with a game-theoretic gadget called an {\em Alth\"ofer
  game}~\cite{A94}.\footnote{Similar ideas have been used previously, including
  in the proofs that computing an $\eps$-approximate Nash equilibrium
  with $\eps$ inverse polynomial in~$n$ is a $\PPAD$-complete problem~\cite{DGP09,CDT09}.}
For a positive and even integer~$k$, this $k \times \binom{k}{k/2}$ game
is defined as follows.
%
\begin{itemize}
\item Alice chooses an index $i \in [k]$.
\item Bob chooses a subset $S \sse [k]$ of size~$k/2$.
\item Alice's payoff is 1 if $i\in S$, and -1 otherwise.
\item Bob's payoff is -1 if $i\in S$, and 1 otherwise.
\end{itemize}
For example, here is the payoff matrix for the
$k=4$ case (with only Alice's payoffs shown):
\[
\left( 
\begin{array}{cccccc}
1 & 1 & 1 & -1 & -1 & -1\\
1 & -1 & -1 & -1 & 1 & 1\\
-1 & 1 & -1 & 1 & -1 & 1\\
-1 & -1 & 1 & 1 & 1 & -1\\
\end{array}
\right)
\]
Every Alth\"ofer game is a zero-sum game with value 0: for both
players, choosing a uniformly random strategy guarantees expected
payoff~0.  The following claim proves a robust converse for Alice's
play.  Intuitively, if Alice deviates much from the uniform
distribution, Bob is well-positioned to punish her.\footnote{The
  statement and proof here include a constant-factor improvement, 
due  to Salil Vadhan, 
over those in~\cite{R16}.}

\begin{claim} \label{claim:Althofer} In every $\eps$-approximate Nash
  equilibrium of an Alth\"ofer game, Alice's strategy is
  $\eps$-close to uniformly random in statistical distance (a.k.a.\
  total variation distance).
\end{claim}


\begin{proof}
  Suppose that Alice plays strategy $i \in [k]$ with probability~$p_i$.
  After sorting the coordinates so that
  $p_{i_1}\leq p_{i_2}\leq \cdots \leq p_{i_k}$, Bob's best response
  is to play the subset $S=\{i_1,i_2,\ldots,i_{k/2}\}$.  We must have
  either $p_{i_{k/2}} \leq 1/k$ or $p_{i_{k/2}+1}\geq 1/k$ (or both).
  Suppose that $p_{i_{k/2}} \leq 1/k$; the other case is similar.
Bob's expected payoff from playing $S$ is then:
\begin{eqnarray*}
\sum_{j>k/2} p_{i_j} - \sum_{j\leq k/2} p_{i_j}
&=& \sum_{j>k/2} (p_{i_j}-1/k) + \sum_{j\leq k/2} (1/k-p_{i_j})\\
&=& \sum_{j : p_{i_j}> 1/k} (p_{i_j}-1/k) + \sum_{j > k/2 : p_{i_j}\leq 1/k} (p_{i_j}-1/k) + \sum_{j\leq k/2} (1/k-p_{i_j})\\
&\geq& \sum_{j: p_{i_j}> 1/k} (p_{i_j}-1/k),
\end{eqnarray*} 
where the last inequality holds 
because the $p_{i_j}$'s are sorted in increasing order and
$p_{i_{k/2}} \leq 1/k$.  
The final expression above
equals the statistical distance between Alice's mixed strategy
$\vec{p}$ and the uniform distribution.  
The claim now follows from that fact that Bob
cannot achieve a payoff larger than $\eps$ in any $\eps$-approximate Nash
equilibrium (otherwise, Alice could 
increase her expected payoff by more than~$\eps$
by switching to the uniform distribution).
\end{proof}
In Claim~\ref{claim:Althofer}, it's important that the loss in statistical
distance (as a function of~$\eps$) is independent of the size~$k$ of
the game.  For example, straightforward generalizations of
rock-paper-scissors fail to achieve the guarantee in
Claim~\ref{claim:Althofer}.

\paragraph{Gluing Games.}
We incorporate Alth\"ofer games into our coordinatewise play game
as follows. Let
\begin{itemize}
\item $G_1 = \text{the $\tfrac{d}{\eps} \times \tfrac{d}{\eps}$
    coordinatewise game of Section~\ref{ss:coordinatewise}}$;
\item $G_2 = \text{a $d\times {d \choose d/2}$ Alth\"ofer game;}$ and
\item $G_3 = \text{a ${d \choose d/2} \times d$ Alth\"ofer game, with
    the roles of Alice and Bob reversed.}$
\end{itemize}
Consider the following game, where Alice and Bob effectively play all
three games simultaneously:
\begin{itemize}

\item A pure strategy of Alice comprises an index $i\in [d]$,
a multiple~$x_i$ of $\eps$ in $[0,1]$, and a set $T\subseteq [d]$ of
  size $d/2$.  The interpretation is that she
plays $(i,x_i)$ in $G_1$, $i$ in $G_2$, and $T$ in
  $G_3$.

\item A pure strategy of Bob comprises an index $j\in [d]$,
 a multiple $y_j$ of $\eps$ in $[0,1]$, and a set $S\subseteq [d]$ of
 size $d/2$, 
  interpreted as playing $(j,y_j)$ in $G_1$, $S$ in $G_2$, and $j$ in
  $G_3$.

\item Each player's payoff is a weighted average of their payoffs in the
  three games: $\tfrac{1}{100} \cdot G_1 + 
\tfrac{99}{200}\cdot G_2 + 
\tfrac{99}{200}\cdot G_3$.

\end{itemize}
The good news is that, in every exact Nash equilibrium of the combined
game, Alice and Bob mix uniformly over their choices of~$i$ and~$j$.
Intuitively, because deviating from the uniform strategy can be
punished by the other player at a rate linear in the deviation (Claim~\ref{claim:Althofer}), it is
never worth doing (no matter what happens in~$G_1$).  Given this,
\`a la the McLennan-Tourky reduction (Theorem~\ref{t:mt06}),
the $x_i$'s and $y_j$'s must correspond to a fixed point of~$f$ (for
each~$i$, Alice must set $x_i$ to the center of mass of Bob's
distribution over $y_i$'s, and then Bob must set $y_i = f_i(x_i)$).


The bad news is that this argument breaks down for $\eps$-approximate
Nash equilibria with constant~$\eps$.
The reason is that, even when the distributions of $i$ and $j$ are
perfectly uniform, the two players interact (i.e., choose $i=j$) only with
probability $1/d$.  This means that the contribution of the game $G_1$
to the expected payoffs is at most $1/d \ll \eps$, freeing the players
to choose their $x_i$'s and $y_j$'s arbitrarily.  Thus we need another
idea to force Alice and Bob to interact more frequently.

A second problem is that the sizes of the Alth\"ofer games are too
big---exponential in $d$ rather than in~$\sqrt{d}$.  

\subsection{\bfp $\le \eNE$ (Attempt \#4): Blockwise Play}\label{ss:blockwise}


To solve both of the problems with the third attempt, we force Alice
and Bob to play larger sets of coordinates at a time.  Specifically, we
view $[d]$ as a $\sqrt{d} \times \sqrt{d}$ grid, and any
$x,y\in [0,1]^d$ as $\sqrt{d}\times \sqrt{d}$ matrices.  Now Alice and
Bob will play a row and column of their matrices, respectively, and
their payoffs will be determined by the entry where the row and column
intersect.  That is, we replace the coordinatewise game of 
Section~\ref{ss:coordinatewise} with the following {\em blockwise game}:

\begin{itemize}
\item A pure strategy of Alice comprises
an index $i\in \left[\sqrt{d}\right]$ and a row
  $x_{i*}\in [0,1]^{\sqrt{d}}$.  (As usual, every $x_{ij}$ should be a
  multiple of $\eps$.)
\item A pure strategy of Bob comprises
 an index $j\in \left[\sqrt{d}\right]$ and a column $y_{*j}\in [0,1]^{\sqrt{d}}$.
\item Alice's payoff in the outcome $(x_{i*},y_{*j})$ 
is  
\[1-(x_{ij}-y_{ij})^2.
\]
\item Bob's payoff in the outcome $(x_{i*},y_{*j})$ is 
\begin{equation}\label{eq:bpayoff4}
1-(y_{ij}-f_{ij}(x_{ij}))^2.
\end{equation}
\end{itemize}
Now glue this game together with $k\times {k \choose k/2}$ and
$\binom{k}{k/2} \times k$ Alth\"ofer games with $k=\sqrt{d}$, as in
Section~\ref{ss:althofer}.  (For example, Alice's
index~$i \in \left[\sqrt{d}\right]$ is identified with a row in the
first Alth\"ofer game, and now Alice also picks a subset
$S \sse \left[\sqrt{d}\right]$ in the second Alth\"ofer game, in
addition to $i$ and $x_{i*}$.)  This construction yields exactly what
we want: a game of size $\exp(\tO(k)) = \exp(\tO(\sqrt{d}))$ in which
every $\eps$-approximate Nash equilibrium can be easily translated to
a $\delta$-approximate fixed point of $f$ (in the normalized $\ell_2$
norm), where~$\delta$ depends only on~$\eps$.\footnote{The
  $\tO(\cdot)$ notation suppresses logarithmic factors.}\footnote{In
  more detail, in every $\eps$-approximate Nash equilibrium of the
  game, Alice and Bob both randomize nearly uniformly over~$i$
  and~$j$; this is enforced by the Alth\"ofer games as in
  Section~\ref{ss:althofer}.  Now think of each player as choosing its
  strategy in two stages, first the index $i$ or $j$ and then the
  corresponding values $x_{i*}$ or $y_{*j}$ in the row or column.
  Whenever Alice plays~$i$, her best response (conditioned on~$i$) is
  to play $\expect{y_{ij}}$ in every column~$j$, where the expectation
  is over the distribution of~$y_{ij}$ conditioned on Bob choosing
  index~$j$.  In an $\eps$-approximate Nash equilibrium, in most
  coordinates, Alice must usually choose $x_{ij}$'s that are close to
  this best response.  Similarly, for most indices
  $j \in \left[\sqrt{d}\right]$, whenever Bob chooses~$j$, he must
  usually choose a value of $y_{ij}$ that is close to
  $\expect{f_{ij}(x_{ij})}$ (for each~$i$).  
It can be shown that these facts imply
  that Alice's strategy corresponds to a $\delta$-approximate fixed
  point (in the normalized $\ell_2$ norm), where $\delta$ is a
  function of $\eps$ only.}

\subsection{Beyond Separable Functions}\label{ss:nonseparable}

We now know how to use an $\eps$-approximate Nash equilibrium of a
subexponential-size game (with constant~$\eps$) to compute
a $\delta$-approximate fixed point of a function that is separable in
the sense of~\eqref{eq:separable}.
This is not immediately interesting, because a fixed point of a
separable function is easy to find by doing binary search
independently in each coordinate.
The hard Brouwer functions identified in Theorem~\ref{t:brouwer3} have
lots of nice properties, but they certainly aren't separable.  

Conceptually, the rest of the proof of Theorem~\ref{thm:Aviad}
involves pushing in two directions: first, identifying hard Brouwer
functions that are even ``closer to separable'' than the functions in
Theorem~\ref{t:brouwer3}; and second, extending the reduction in
Section~\ref{ss:blockwise} to accommodate ``close-to-separable''
functions.  We already have an intuitive feel for what the second step
looks like, from Step~4 of our communication complexity lower bound
(Section~\ref{ss:step4} in Solar Lecture~3), where we enlarged the
strategy sets of the players so that they could smuggle ``advice''
about how to decode a hard Brouwer function~$f$ at a given point.  We
conclude the lecture with one key idea for the further simplification
of the hard Brouwer functions in Theorem~\ref{t:brouwer3}.



\subsection{{\sc Local \eol}}

Recall the hard Brouwer functions constructed in our communication
complexity lower bound (see Section~\ref{s:ccbfp}), which ``follow the
line'' of an embedding of an \eol instance, as well as the additional
tweaks needed to prove Theorem~\ref{t:brouwer3} (see
Section~\ref{ss:ppadbfp}).  We are interested in the ``local
decodability'' properties of these functions.  That is, if Bob needs
to compute the $j$th coordinate of $f(x)$ (to evaluate the $j$th term
in his payoff in~\eqref{eq:bpayoff3}), how much does he need to know
about~$x$?
For a separable function~$f=(f_1,\ldots,f_d)$, he only needs to
know~$x_j$.  For the hard Brouwer functions
in Theorem~\ref{t:brouwer3}, Bob needs to know whether or not~$x$ is
close to an edge (of the embedding of the succinct \eol instance into
the hypercube) and, if so, which edge (or pair of edges, if $x$ is
close to a vertex).  Ultimately, this requires evaluating the
successor circuit~$S$ and predecessor circuit~$P$ of the succinct \eol
instance that defines the hard Brouwer function.
It is therefore in our interest to force~$S$ and~$P$ to be as simple
as possible, subject to the succinct \eol problem remaining
$\PPAD$-complete.  In a perfect world, minimal advice (say, $O(1)$
bits) would be enough to compute $S(v)$ and $P(v)$
from~$v$.\footnote{It is also important that minimal advice suffices
 to translate between points~$x$ of the hypercube and vertices~$v$ of
  the underlying succinct \eol instance (as $f$ is defined on the former,
  while~$S$ and~$P$ operate on the latter).  
  This can be achieved by using a state-of-the-art locally decodable
  error-correcting code (with query complexity~$d^{o(1)}$, similar to
  that in~\citet{KMRS17}) to embed the vertices into the hypercube (as
  described in Section~\ref{ss:ppadbfp}).  Incorporating the advice
  that corresponds to local decoding into the game produced by the
  reduction results in a further blow-up of $2^{d^{o(1)}}$.  This is
  effectively absorbed by the $2^{\sqrt{d}}$ blow-up that is already
  present in the reduction in Section~\ref{ss:blockwise}.}
%
%
%
The following lemma implements this idea.  It 
shows that a variant of the succinct \eol problem, called
{\sc Local \eol}, remains $\PPAD$-complete even when $S$ and~$P$ are
guaranteed to change only~$O(1)$ bits of the input, and 
when~$S$ and~$P$ are $\mathrm{NC}^0$ circuits (and hence each output
bit depends on only~$O(1)$ input bits).
%

\begin{lemma}[\citet{R16}]\label{l:local}
The following {\sc Local \eol} problem is $\PPAD$-complete:
\begin{enumerate}
\item the vertex set~$V$ is a subset of $\zo^n$, with membership
  in~$V$ specified by a given $\mathrm{AC}^0$ circuit;
\item the successor and predecessor circuits $S,P$ are computable in $\mathrm{NC}^0$;
\item for every vertex $v \in V$, $S(v)$ and $P(v)$ differ from $v$
  in $O(1)$ coordinates.
\end{enumerate}
\end{lemma}
The proof idea is to start from the original circuits~$S$ and~$P$ of a
succinct \eol instance and form circuits~$S'$ and $P'$
that operate on partial computation transcripts,
carrying out the computations performed by the circuits $S$ or $P$ one
gate/line at a time (with $O(1)$ bits changing in each step of the
computation).  The vertex set $V$ then corresponds to the set of valid
partial computation transcripts.  The full proof is not overly
difficult; see~\cite[Section 5]{R16} for the details.  This reduction
from succinct \eol to {\sc Local \eol}
can be implemented in linear time, so our assumption that the former
problem admits no subexponential-time algorithm carries over to the
latter problem.

In the standard succinct \eol problem, every $n$-bit string
$v \in \zo^n$ is a legitimate vertex.  In the {\sc Local \eol}
problem, only elements of $\zo^n$ that satisfy the given
$\mathrm{AC}^0$ circuit are legitimate vertices.  In our reduction, we
need to produce a game that also incorporates checking membership
in~$V$, also with only a $d^{o(1)}$ blow-up in how much of $x$ we need
to access.  This is the reason why \citet{R16} needs to develop
customized PCP machinery in his proof of Theorem~\ref{thm:Aviad}.
These PCP proofs can then be incorporated into the blockwise play game
(Section~\ref{ss:blockwise}),
analogous to how we incorporated a low-cost interactive protocol into the
game in our reduction from \cceol to $\eNE$ in
Section~\ref{ss:step4}.

\part{Lunar Lectures}


\titleformat{\chapter}[display]
{\normalfont\Large\filcenter\rmfamily\scshape}
{\titlerule[1pt]%
 \vspace{1pt}%
 \titlerule
 \vspace{1ex}%
 \LARGE{Lunar Lecture} \thechapter}
{.5ex}
{\normalfont\large\slshape}

\setcounter{chapter}{0}

\lecture{How Computer Science Has Influenced Real-World Auction Design.\\  Case Study: The 2016--2017 FCC Incentive Auction}

\vspace{1cm}


\section{Preamble}

Computer science is changing the way auctions are designed and
implemented.  For over 20 years, the US and other countries have used
{\em spectrum auctions} to sell licenses for wireless spectrum to the
highest bidder.  What's different this decade, and what necessitated a
new auction design, is that in the US the juiciest parts of the
spectrum for next-generation wireless applications are already
accounted for, owned by over-the-air television broadcasters.  This led
Congress to authorize the FCC in the fall of~2012 to design a novel
auction (the {\em FCC Incentive Auction}) that would repurpose
spectrum---procuring licenses from television broadcasters (a
relatively low-value activity) and selling them to parties that would put
them to better use (e.g., telecommunication companies who want to roll
out the next generation of wireless broadband services).
Thus the FCC Incentive Auction is really a {\em double auction}, comprising two
stages: a \emph{reverse
  auction}, where the government buys back licenses for spectrum from
their current owners; and then a \emph{forward
  auction}, where the government sells the procured licenses to the
highest bidder.  
Computer science techniques played a crucial role in the design of the
new reverse auction.  The main aspects of the forward auction have
been around a long time; here, theoretical computer science has
contributed on the analysis side, and to understanding when and why
such forward auctions work well.  Sections~\ref{s:reverse}
and~\ref{s:forward} give more details on the reverse and forward parts
of the auction, respectively.

The FCC Incentive Auction finished around the end of March 2017, and
so the numbers are in.  The government spent roughly 10 billion USD in
the reverse part of the auction buying back licenses from television
broadcasters, and earned roughly 20 billion USD of revenue in the forward
auction.  Most of the 10 billion USD profit was used to reduce the US
debt!\footnote{This was the plan all along, which is probably one
  of the reasons the bill didn't have trouble passing a
  notoriously partisan Congress.  Another reason might be the
  veto-proof title of the bill: ``The Middle Class Tax Relief and Job
  Creation Act.''}



\section{Reverse Auction}\label{s:reverse}

\subsection{Descending Clock Auctions}

The reverse auction is the part of the FCC Incentive Auction that was
totally new, and where computer science techniques played a crucial
role in the design.
%
%
The auction format, proposed by Milgrom and Segal~\cite{MS20}, is
what's called a {\it descending clock auction}.  By design, the
auction is very simple from the perspective of any one participant.
The auction is iterative, and operates in rounds.  In each round of
the auction, each remaining broadcaster is asked a question of the
form: ``Would you or would you not be willing to sell your license for
(say) 1 million dollars?''  The broadcaster is allowed to say ``no,'' with
the consequence of getting kicked out of the auction forevermore 
(the station will keep its license and remain on the air, and will
receive no compensation from the government).  The broadcaster is also
allowed to say ``yes'' and accept the buyout offer.  In the latter
case, the government will not necessarily buy the license for 1
million dollars---in the next round, the broadcaster might get asked the
same question, with a lower buyout price (e.g., 950,000 USD).  If a
broadcaster is still in the auction when it ends (more on how it ends
in a second), then the government does indeed buy their license, at the
most recent (and hence lowest) buyout offer.  
Thus all a broadcaster 
has to do is answer a sequence of ``yes/no'' questions for some decreasing
sequence of buyout offers.  The obvious strategy for a broadcaster is
to formulate the lowest acceptable offer for their license, and to
drop out of the auction once the buyout price drops below this
threshold.



The auction begins with very high buyout offers, so that every
broadcaster would be ecstatic to sell their license at the initial
price.  Intuitively, the auction then tries to reduce the buyout
prices as much as possible, subject to clearing a target amount of
spectrum.  Spectrum is divided into channels which are blocks of 6 MHz
each.  For example, one could target broadcasters assigned to channels
38--51, and insist on clearing 10 out of these~14 channels (60 MHz
overall).\footnote{The FCC Incentive Auction wound up clearing 84 MHz
  of spectrum (14 channels).}
By ``clearing a channel,'' we mean clearing it
{\em nationwide}.  Of course, in the descending clock auction, bidders
will drop out in an uncoordinated way---perhaps the first station to
drop out is channel~51 in Arizona, then channel~41 in western
Massachusetts, and so on.  To clear several channels nationwide without
buying out essentially everybody, it was essential for the government
to use its power to {\em reassign} the channels of the stations that
remain on the air.  Thus while a station that drops out of the auction
is guaranteed to retain its license, it is not guaranteed to retain
its channel---a station broadcasting on channel~51 before the auction
might be forced to broadcast on channel~41 after the auction.

The upshot is that the auction maintains the invariant that
the stations that have dropped out of the auction (and hence remain on
the air) can be assigned channels so that at most a target number of
channels are used (in our example, 4 channels).  This is called the
{\em repacking problem}.  Naturally, two stations with overlapping
broadcasting regions cannot be assigned the same channel (otherwise
they would interfere with each other).  See Figure~\ref{f:stations}.

\begin{figure}[h]
\centering
\includegraphics[width=.8\textwidth]{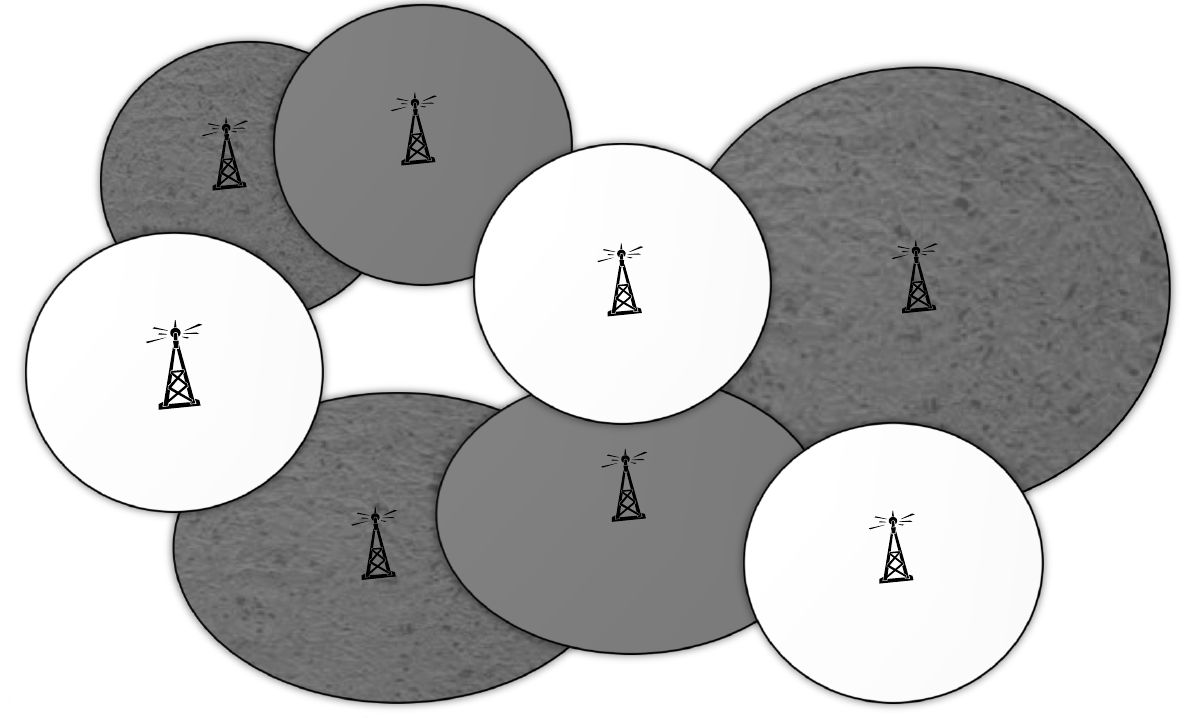}
\caption{Different TV stations with
  overlapping broadcasting areas must be assigned different channels
  (indicated by shades of gray).  Checking whether or not a given
  subset of stations can be assigned to a given number of channels
  without interference is an $\NP$-hard problem.}\label{f:stations}
\end{figure}



\subsection{Solving the Repacking Problem}

Any properly trained computer scientist will recognize the repacking
problem as the $\NP$-complete graph coloring problem in
disguise.\footnote{The actual repacking problem was more
  complicated---overlapping stations cannot even be assigned adjacent
  channels, and there are idiosyncratic constraints at the borders
  with Canada and Mexico.  See~\citet{LMS17} for more details.  But
  the essence of the repacking problem really is graph coloring.}  For
the proposed auction format to be practically viable, it must quickly
solve the repacking problem.  Actually, make that thousands of
repacking problems every round of the auction!\footnote{Before the
  auction makes a lower offer to some remaining broadcaster in the
  auction, it needs to check that it would be OK for the broadcaster
  to decline and drop out of the auction.  If a station's dropping out
  would render the repacking problem infeasible, then that station's
  buyout price remains frozen until the end of the auction.}

The responsibility of quickly solving repacking problems fell to a
team led by Kevin Leyton-Brown (see~\cite{FNL17,LMS17}).  The FCC gave
the team a budget of one minute per repacking problem, ideally with
most instances solved within one second.  The team's approach was to
build on state-of-the-art solvers for the satisfiability (SAT)
problem.  As you can imagine, it's straightforward to translate an
instance of the repacking problem into a SAT formula (even with the
idiosyncratic constraints).\footnote{A typical representative instance
  would have thousands of variables and tens of thousands of
  constraints.}  Off-the-shelf SAT solvers did pretty well, but still
timed out on too many representative instances.\footnote{Every time
  the repacking algorithm fails to find a repacking when one exists,
  money is left on the table---the auction has to conservatively leave
  the current station's buyout offer frozen, even though it could have
  safely lowered it.}  Leyton-Brown's team added several new
innovations, including taking advantage of 
problem structure specific to the application and implementing
a number of caching techniques (reusing work
done solving previous instances to quickly solve closely related
new instances).  In the end, they were able to solve more than 99\% of the
relevant repacking problems in under a minute.

Hopefully the high-level point is clear: 
\begin{quote}
without cutting-edge
techniques for solving $NP$-complete problems, {\em the FCC would have
  had to use a different auction format}.
\end{quote}



\subsection{Reverse Greedy Algorithms}

One final twist: the novel reverse auction format motivates some basic
algorithmic questions (and thus ideas flow from computer science to
auction theory and back).  We can think of the auction as an
algorithm, a heuristic that tries to maximize the value of the
stations that remain on the air, subject to clearing the target amount
of spectrum.  Milgrom and Segal~\cite{MS20} prove that, ranging over
all ways of implementing the auction (i.e., of choosing the sequences of
descending prices), the corresponding algorithms are exactly the {\em
  reverse greedy algorithms}.\footnote{For example, Kruskal's
  algorithm for the minimum spanning tree problem (start with the
  empty set, go through the edges of the graph from cheapest to most
  expensive, adding an edge as long as it doesn't create a cycle) is a
  standard (forward) greedy algorithm.  The reverse version is: start
  with the entire edge set, go through the edges in reverse sorted
  order, and remove an edge whenever it doesn't disconnect the graph.
  For the minimum spanning tree problem (and more generally for
  finding the minimum-weight basis of a matroid), the reverse greedy
  algorithm is just as optimal as the forward one.  In general (and
  even for e.g.~bipartite matching), the reverse version of a good
  forward greedy algorithm can be bad~\cite{DGR14}.}  This result
gives the first extrinsic reason to study the power and limitations of
reverse greedy algorithms, a research direction explored
by~\citet{DGR14} and~\citet{GMR17}.




\section{Forward Auction}\label{s:forward}

Computer science did not have an opportunity to influence the design
of the forward auction used in the FCC Incentive Auction, which
resembles the formats used over the past 20+ years.  Still, the
theoretical computer science toolbox turns out to be ideally suited
for explaining when and why these auctions work well.\footnote{Much of the discussion 
in Sections~\ref{ss:bad}--\ref{ss:exposure} is 
  from~\cite[Lecture~8]{f13}, which in turn takes inspiration from~\citet{Milgrom:2004jx}.}

\subsection{Bad Auction Formats Cost Billions}\label{ss:bad}

Spectrum auction design is stressful, because small mistakes can be
extremely costly.  One cautionary tale is provided by an auction run
by the New Zealand government in~1990 (before governments had much
experience with auctions).
For sale were~10 essentially identical
national licenses for television broadcasting.
For some reason, lost to the sands of time, the government decided to
sell these licenses by running 10 second-price auctions in parallel.
A {\em second-price} or {\em Vickrey} auction for a single
good is a sealed-bid auction that awards the item to the highest
bidder and charges her the highest bid by someone else (the
second-highest bid overall).  When selling a single item, the Vickrey
auction is often a good solution.  In particular, each bidder has a
dominant strategy (always at least as good as all alternatives), which
is to bid her true maximum willingness-to-pay.\footnote{Intuitively, a second-price auction shades your bid
  optimally after the fact, so there's no reason to try to game
  it.}\footnote{For a more formal treatment of single-item auctions,
  see Section~\ref{ss:basics} in Lunar Lecture~4.}

The nice properties of a second-price auction evaporate if many of
them are run simultaneously.  A bidder can now submit up to one bid in
each auction, with each license awarded to the highest bidder (on that
license) at a price equal to the second-highest bid (on that license).
With multiple simultaneous auctions, it is no longer clear how a
bidder should bid.  For example, imagine you want one of the licenses,
but only one.  How should you bid?  One legitimate strategy is to pick
one of the licenses---at random, say---and go for it.  Another
strategy is to bid less aggressively on multiple licenses, hoping that
you get one at a bargain price, and that you don't inadvertently win
extra licenses that you don't want.  The difficulty is trading off the
risk of winning too many licenses with the risk of winning too few.

The challenge of bidding intelligently in simultaneous
sealed-bid auctions makes the auction format prone to poor outcomes.
%
The revenue in the 1990 New Zealand auction was only \$36 million, a paltry
fraction of the projected \$250 million.  
On one license, the high bid was \$100,000 while the
second-highest bid (and selling price) was \$6!  On another, the high
bid was \$7 million and the second-highest was \$5,000.  To add insult
to injury, the winning bids were made available to the public, who could
then see just how much money was left on the table!  



\subsection{Simultaneous Ascending Auctions}

Modern spectrum auctions are based on {\em simultaneous ascending
  auctions (SAAs)}, following~1993 proposals by McAfee and by Milgrom
and Wilson.
You've seen---in the movies, at least---the call-and-response format
of an ascending single-item auction, where an auctioneer asks for
takers at successively higher prices.  Such an auction ends when
there's only one person left accepting the currently proposed price
(who then wins, at this price).
Conceptually, SAAs are like a bunch of
single-item English auctions being run in parallel in the same room,
with one auctioneer per item.

The primary reason that SAAs work better than sequential or sealed-bid
auctions is {\em price discovery}.  As a bidder acquires better
information about the likely selling prices of licenses, she can
implement mid-course corrections---abandoning licenses for which
competition is fiercer than anticipated, snapping up unexpected
bargains, and rethinking which packages of licenses to assemble.  The
format typically resolves the miscoordination problems that plague
simultaneous sealed-bid auctions.  



\subsection{Inefficiency in SAAs}\label{ss:exposure}

SAAs have two big vulnerabilities.  The first problem is {\em demand
  reduction}, and this is relevant even when items are substitutes.\footnote{Items are substitutes if they provide diminishing returns---having one item only makes others less valuable.  For two
items~$A$ and $B$, for example, the substitutes condition means that
a bidder's value for the bundle of $A$ and $B$ is at most the sum of
her values for $A$ and $B$ individually.
In a spectrum auction context, two licenses
for the same area with equal-sized frequency ranges are usually
substitute items.}
Demand reduction occurs when a bidder asks for fewer items
than she really wants, to lower competition and therefore the prices
paid for the items that it gets.

To illustrate, suppose there are two identical items and two bidders.
By the {\em valuation} of a bidder for a given bundle of items, we
mean her maximum willingness to pay for that bundle.
Suppose the first bidder has valuation 10
for one of the items and valuation 20 for both.  The second bidder has
valuation 8 for one of the items and does not want both (i.e., her
valuation remains 8 for both).  The socially optimal outcome is to
give both licenses to the first bidder.
Now consider how things play out in an SAA.  
The second bidder would be happy
to have either item at any price less than 8.  Thus, the second bidder
drops out only when the prices of both items exceed~8.  If the first bidder
stubbornly insists on winning both items, her utility is
$20-16=4$.  
An alternative strategy for the first bidder is to simply concede the
second item and never bid on it.  The second bidder takes the
second item and (because she only wants one license) withdraws interest in
the first, leaving it for the first bidder.  Both bidders get their
item essentially for free, and the utility of the first bidder has
jumped to~10.


The second big problem with SAAs is relevant when items can be
complements, and is called the {\em exposure problem}.\footnote{Items
  are complements if there are synergies between them,
  so that possessing one makes others more valuable.  With two items
  $A$ and $B$, this translates to a bidder's valuation for the bundle of $A$ and
  $B$ exceeding the sum of her valuations for $A$ and $B$ individually.
  Complements arise naturally in wireless spectrum auctions, as some
  bidders want a collection of licenses that are adjacent, either in
  their geographic areas or in their frequency ranges.}
As an example, consider two bidders and two
nonidentical items.  The first bidder only wants both items---they are
complementary items for the bidder---and her valuation is 100 for
them (and 0 for anything else).
The second bidder is willing to pay 75 for either
item but only wants one item.
The socially optimal outcome is to give both items to the first bidder.
But in an SAA, the second bidder will not drop out until the price of
both items reaches 75.  The first bidder is in a no-win situation: to
get both items she would have to pay 150, more than her value.  The
scenario of winning only one item for a nontrivial price could be even
worse.  
Thus the exposure problem leads to economically inefficient
allocations for two reasons.  First, an overly aggressive bidder might
acquire unwanted items.  Second, an overly tentative bidder might fail
to acquire items for which she has the highest valuation.

\subsection{When Do SAAs Work Well?}\label{ss:when}

If you ask experts who design or consult for bidders in
real-world SAAs, a rough consensus emerges about when they are likely
to work well.

\begin{folklore} Without strong complements, {\em SAAs work pretty
  well.  Demand reduction does happen, but it is not a deal-breaker because
  the loss of efficiency appears to be small.}
\end{folklore}

\begin{folklore}
  With strong complements, {\em simple auctions like SAAs are not good
    enough.  The exposure problem is a deal-breaker because it can
    lead to very poor outcomes (in terms of both economic efficiency
    and revenue).}
\end{folklore}

There are a number of beautiful and useful theoretical results about
spectrum auctions in the economics literature, but none map cleanly to
these two folklore beliefs.  A possible explanation: translating these
beliefs into theorems seems to fundamentally involve approximate
optimality guarantees, a topic that is largely avoided by
economists but right in the wheelhouse of theoretical
computer science.

In the standard model of {\em combinatorial auctions},
there are $n$ bidders (e.g., telecoms) and $m$ items (e.g.,
licenses).\footnote{This model is treated more thoroughly in the next
  lecture (see Section~\ref{s:camodel}).}
Bidder $i$ has a nonnegative valuation $v_i(S)$ for each subset
$S$ of items she might receive.  Note that, in general, 
describing a bidder's valuation function requires $2^m$
parameters.  Each bidder wants to maximize her utility, which is the
value of the items received minus the total price paid for them.
From a social perspective, we'd like to award bundles of items
$T_1,\ldots,T_n$ to the bidders to maximize the {\em social welfare}
$\sum_{i=1}^n v_i(T_i)$.

To make the first folklore belief precise, we need to commit to a
definition of ``without strong complements'' and to a specific auction
format.  We'll focus on simultaneous first-price auctions (S1As),
where each bidder submits a separate bid for each item, for each item
the winner is the highest bidder (on that item), and winning bidders
pay their bid on each item won.\footnote{Similar results hold for
  other auction formats, like simultaneous second-price auctions.
  Directly analyzing what happens in iterative auctions like SAAs when
  there are multiple items appears difficult.}  One relatively
permissive definition of ``complement-free'' is to restrict bidders to
have {\em subadditive valuations}.  This means what it sounds like: if
$A$ and $B$ are two bundles of items, then bidder $i$'s valuation
$v_i(A \cup B)$ for their union should be at most the sum
$v_i(A)+v_i(B)$ of her valuations for each bundle separately.  Observe
that subadditivity is violated in the exposure problem example in Section~\ref{ss:exposure}.

We also need to define what we mean by ``the outcome of an auction''
like S1As.  Remember that bidders are strategic, and will bid to
maximize their utility (value of items won minus the price paid).
Thus we should prove approximation guarantees for the {\em equilibria}
of auctions.  Happily, computer scientists have been working hard
since 1999 to prove approximation guarantees for game-theoretic
equilibria, also known as bounds on {\em the price of
  anarchy}~\cite{KP99,book,RT00}.\footnote{See Section~\ref{ss:poa} of
  the next lecture for a formal definition.}
In the early days, price-of-anarchy
bounds appeared somewhat ad hoc and problem-specific.  Fast forwarding
to the present, we now have a powerful and user-friendly theory for
proving price-of-anarchy bounds, which combine ``extension theorems''
and ``composition theorems'' to build up bounds for complex settings
(including S1As) from bounds for simple settings.\footnote{We will say
  more about this theory in Lunar Lecture~5.  See also \citet{RST17}
  for a recent survey.}  In particular, \citet{FFGL13} proved the
following translation of Folklore Belief \#1.\footnote{To better
  appreciate this result, we note that multi-item auctions like S1As
  are so strategically complex that they have historically been seen
  as unanalyzable.  For example, we have no idea what their equilibria
  look like in general.  Nevertheless, we can prove good approximation
  guarantees for them!}

\begin{theorem}[\citet{FFGL13}]\label{t:ffgl}
When every bidder has a subadditive valuation, every equilibrium of an
S1A has social welfare at least 50\% of the maximum possible.
\end{theorem}

One version of Theorem~\ref{t:ffgl} concerns (mixed) Nash equilibria
in the full-information model (in which bidders' valuations are common knowledge),
as studied in the Solar Lectures.  Even here, the bound in
Theorem~\ref{t:ffgl} is tight in the worst case~\cite{CKST16}.
The approximation guarantee in Theorem~\ref{t:ffgl}
holds more generally for {\em Bayes-Nash
  equilibria}, the standard equilibrium notion for games of incomplete
information.\footnote{In more detail, in this model there is a commonly known prior
  distribution over bidders' valuations.  In a Bayes-Nash equilibrium,
  every bidder bids to maximize her expected utility given her
  information at the time: her own valuation, her posterior belief
  about other bidders' valuations, and the bidding strategies (mapping
  valuations to bids) used by the other bidders.  Theorem~\ref{t:ffgl}
  continues to hold for every Bayes-Nash equilibrium of an S1A, as
  long as bidders' valuations are independently (and not necessarily
  identically) distributed.\label{foot:bn}}






Moving on to the second folklore belief, let's now drop the
subadditivity restriction.  S1As no longer work well.

\begin{theorem}[\citet{HKMN11}]\label{t:hkmn}
When bidders have arbitrary valuations, an S1A can have a mixed Nash
equilibrium with social welfare arbitrarily smaller than the maximum
possible.
\end{theorem}

Thus for S1As, the perspective of worst-case approximation confirms
the dichotomy between the cases of substitutes and complements.  But
the lower bound in Theorem~\ref{t:hkmn} applies only to one specific 
auction
format.  Could we do better with a different natural
auction format?
Folklore Belief \#2 asserts the stronger statement that
{\em no} ``simple'' auction works well with general valuations.  This
stronger statement can also be translated into a theorem (using
nondeterministic communication complexity), and this
will be the main subject of Lunar Lecture~2.


\begin{theorem}[\citet{R14}]\label{t:condpoa}
With general valuations, \emph{every} simple auction can have 
an equilibrium with social welfare arbitrarily smaller than the
maximum possible.
\end{theorem}
The definition of ``simple'' used in Theorem~\ref{t:condpoa} is quite
generous: it requires only that the number of strategies available to
each player is {\em sub-doubly-exponential} in the number of
items~$m$.  For example, running separate single-item auctions
provides each player with only an exponential (in~$m$) number of strategies
(assuming a bounded number of possible bid values for each item).  Thus
Theorem~\ref{t:condpoa} makes use of the theoretical computer science
toolbox to provide solid footing for Folklore Belief \#2.

\lecture{Communication Barriers to
  Near-Optimal Equilibria}

\vspace{1cm}


This lecture is about the 
communication complexity of the welfare-maximization problem in
combinatorial auctions and its implications for the price of anarchy
of simple auctions.  
Section~\ref{s:camodel} defines the model, Section~\ref{s:cclb} proves lower bounds for
nondeterministic communication protocols, and
Section~\ref{s:condpoa} gives a black-box translation of these lower
bounds to equilibria of simple auctions.  In particular,
Section~\ref{s:condpoa} provides the proof of Theorem~\ref{t:condpoa}
from last lecture.  
Section~\ref{s:open} concludes with
a juicy open problem on the topic.\footnote{Much of this lecture is drawn
  from~\cite[Lecture 7]{w15}.}

\section{Welfare Maximization in Combinatorial Auctions}\label{s:camodel}

Recall from Section~\ref{ss:when} the basic setup in the study of
combinatorial auctions.
\begin{enumerate}

\item There are $k$ players.  (In a spectrum auction, these are the telecoms.)

\item There is a set $M$ of $m$ items.  (In a spectrum auction, these
  are the licenses.)

\item Each player~$i$ has a {\em valuation} $v_i:2^M \rightarrow
  \R_+$.  The number $v_i(T)$ indicates $i$'s value, or willingness
  to pay, for the items~$T \subseteq M$.
The valuation is the private input of player $i$, meaning that $i$ knows
  $v_i$ but none of the other $v_j$'s.  (I.e., this is a
  number-in-hand model.)
We assume that $v_i(\emptyset) = 0$ and that the valuations are
{\em monotone}, meaning $v_i(S) \le v_i(T)$ whenever $S \subseteq T$.  (The
more items, the better.)
To avoid bit complexity issues, we'll also assume that all of the
$v_i(T)$'s are integers with description length polynomial in $k$ and
$m$.
We sometimes impose additional restrictions on the valuations to study
special cases of the general problem.

\end{enumerate}
Note that we may have more than two players---more than just Alice
and Bob.  (For example, you might want to think of $k$ as $\approx m^{1/3}$.)
Also note that the description length of a player's valuation is
exponential in the number of items $m$.  


In the {\em welfare-maximization problem}, the goal is to partition
the items $M$ into sets $T_1,\ldots,T_k$ to maximize, at least
approximately, the social welfare 
\begin{equation}\label{eq:welfare}
\sum_{i=1}^k v_i(T_i),
\end{equation}
using communication polynomial in $k$ and $m$.  Note this amount of
communication is logarithmic in the sizes of the private inputs.
Maximizing social welfare~\eqref{eq:welfare} is the most commonly studied
objective in combinatorial auctions, and it is the one we
will focus on in this lecture.



\section{Communication Lower Bounds for Approximate Welfare Maximization}\label{s:cclb}

This section studies the communication complexity of computing an
approximately welfare-maximizing allocation in a combinatorial
auction.  For reasons that will become clear in
Section~\ref{s:condpoa}, we are particularly interested in the
problem's nondeterministic communication complexity.\footnote{For
  basic background on nondeterministic multi-party communication
  protocols, see~\citet{KN96} or~\citet{w15}.}

\subsection{Lower Bound for General Valuations}

We begin with a result of Nisan \cite{Nis02} showing that, alas,
computing even a very weak approximation of the welfare-maximizing
allocation requires exponential communication.
To make this precise, it is convenient to turn the
optimization problem of welfare maximization into a decision problem.
In the \wm[k] problem, 
the goal is to correctly identify inputs that
fall into one of the following two cases:
\begin{itemize}

\item [(1)] Every partition $(T_1,\ldots,T_k)$ of the items has
  welfare at most~1.

\item [(0)] There exists a partition $(T_1,\ldots,T_k)$ of the items
  with welfare at least $k$.

\end{itemize}
Arbitrary behavior is permitted on inputs that fail to satisfy
either~(1) or~(0).
Clearly, communication lower bounds for \wm[k] apply 
to the more general
problem of obtaining a better-than-$k$-approximation of the
maximum welfare.\footnote{Achieving a $k$-approximation is trivial:
  every player communicates her value~$v_i(M)$ for the whole set of
  items, and the entire set of items is awarded to the bidder with the
  highest value for them.}

\begin{theorem}[\cite{Nis02}] \label{thm_nisan}
The nondeterministic communication complexity of \wm[k] is\\ $\exp \{
\Omega(m/k^2) \}$, where $k$ is the number of players and $m$ is the number of items.
\end{theorem}
This lower bound is exponential in~$m$, provided that $m =
\Omega(k^{2+\eps})$ for some $\eps > 0$.  Since communication
complexity lower bounds apply even to players who cooperate perfectly,
this impossibility result holds even when all of the (tricky)
incentive issues are ignored.

\subsection{The \mdisj Problem}

The plan for the proof of Theorem~\ref{thm_nisan} is to reduce a
multi-party version of the \disj problem to the \wm[k] problem.
There is some ambiguity about how to define a version of \disj for three or
more players.  For example, suppose there are three players, and
among the three possible pairings of them, two have disjoint sets
while the third have intersecting sets.  Should this count as a
``yes'' or ``no'' instance?  We'll skirt this issue by worrying only
about unambiguous inputs, that are either ``totally disjoint'' or
``totally intersecting.''

Formally, in the \mdisj problem, each of the $k$ players $i$ holds an
input $\bfx_i \in \zo^n$.  (Equivalently, a set $S_i \subseteq
\{1,2,\ldots,n\}$.)  The task is to correctly identify inputs that
fall into one of the following two cases:
\begin{itemize}

\item [(1)] ``Totally disjoint,'' with $S_i \cap S_{i'} = \emptyset$ for
every $i \neq i'$.

\item [(0)] ``Totally intersecting,'' with $\cap_{i=1}^k S_i \neq
  \emptyset$.

\end{itemize}
When $k=2$, this is the standard \disj problem.  When $k > 2$, there
are inputs that are neither 1-inputs nor 0-inputs.  We let protocols
off the hook on such ambiguous inputs---they can answer ``1'' or
``0'' with impunity.

The following communication complexity lower bound for \mdisj is
credited to Jaikumar Radhakrishnan and Venkatesh Srinivasan
in~\cite{Nis02}.  (The proof is elementary, and for completeness is
given in Section~\ref{s:mdisj}.)
\begin{theorem}\label{t:mdisj}
The nondeterministic communication complexity of \mdisj, with $k$
players with $n$-bit inputs, is $\Omega(n/k)$.
\end{theorem}
This nondeterministic lower bound is for verifying a 1-input.  (It is
easy to verify a 0-input---the prover just suggests the index of an
element $r$ in $\cap_{i=1}^k S_i$.)%
\footnote{In proving
  Theorem~\ref{thm_nisan},
we'll be interested in the
case where $k$ is much smaller than $n$, such as $k = \Theta(\log
n)$.  Intuition might suggest that the lower bound should be
$\Omega(n)$ rather than $\Omega(n/k)$, but this is incorrect---a
slightly non-trivial argument shows that Theorem~\ref{t:mdisj} is
tight for nondeterministic protocols
(for all small enough $k$, like $k = O(\sqrt{n})$).  
This factor-$k$ difference won't matter for our applications, however.}



\subsection{Proof of Theorem~\ref{thm_nisan}}

The proof of Theorem~\ref{thm_nisan} relies on Theorem~\ref{t:mdisj} and a
combinatorial gadget.  We construct this gadget using the
probabilistic method.  Consider $t$ random
partitions $P^1,\ldots,P^t$ of $M$, where $t$ is a parameter to be
defined later.  By a random partition $P^j = (P^j_1,\ldots,P^j_k)$, we
mean that each of the $m$ items is assigned
to exactly one of the
$k$ players, independently and uniformly at random.

We are interested in the probability that two classes of different
partitions intersect: for all $i \neq i'$ and $j \neq \ell$, 
because the probability that a given item is assigned to $i$ in
$P^j$ and also to $i'$ in $P^{\ell}$ is $\tfrac{1}{k^2}$, we
have
\[
\prob{P^j_i \cap P^{\ell}_{i'} = \emptyset} = 
\left( 1 - \frac{1}{k^2} \right)^m \le e^{-m/k^2}.
\]
Taking a Union Bound over the $k$ choices for $i$ and $i'$ and the $t$
choices for $j$ and $\ell$, we have
\begin{equation}\label{eq:int}
\prob{\exists i \neq i', j \neq \ell \text{ s.t.\ } P^j_i \cap
  P^{\ell}_{i'} = \emptyset} \le k^2t^2e^{-m/k^2}.
\end{equation}
Call $P^1,\ldots,P^t$ an {\em intersecting family} if $P^j_i \cap
P^{\ell}_{i'} \neq \emptyset$ whenever $i \neq i'$, $j \neq \ell$.
By~\eqref{eq:int}, the probability that our random experiment fails to
produce an intersecting family is less than~1 provided $t <
\tfrac{1}{k}e^{m/2k^2}$.  The following lemma is immediate.
\begin{lemma}\label{l:int}
For every $m,k \ge 1$,
there exists an intersecting family of partitions $P^1,\ldots,P^t$
with $t = \exp \{ \Omega(m/k^2) \}$.
\end{lemma}

A simple combination of Theorem~\ref{t:mdisj} and Lemma~\ref{l:int}
now proves Theorem~\ref{thm_nisan}.

\noindent
\begin{proof}
(of Theorem~\ref{thm_nisan})
The proof is a reduction from \mdisj.
Fix $k$ and $m$.  (To be interesting,~$m$ should be significantly
bigger than $k^2$.)  
Let $(S_1,\ldots,S_k)$ denote an input to \mdisj
with $t$-bit inputs, where $t = \exp \{ \Omega(m/k^2) \}$ is the same
value as in Lemma~\ref{l:int}.  We can assume that the players have
coordinated in advance on an intersecting family of $t$ partitions of a
set~$M$ of $m$ items.  Each player~$i$ uses this family and her input
$S_i$ to form the following valuation:
\[
v_i(T) = 
\left \{
\begin{array}{cl}
1 & \text{if $T \supseteq P^j_i$ for some $j \in S_i$}\\
0 & \text{otherwise.}
\end{array}
\right.
\]
That is, player~$i$ is either happy (value~1) or unhappy (value~0),
and is happy if and only if she receives all of the items in the
corresponding class $P^j_i$ of some partition $P^j$ with index~$j$
belonging to its input to \mdisj.
The valuations $v_1,\ldots,v_k$ define an input to \wm[k].
Forming this input requires no communication between the players.

Consider the case where the input to \mdisj is a 1-input, with $S_i
\cap S_{i'} = \emptyset$ for every $i \neq i'$.  We claim that the
induced input to \wm[k] is a 1-input, with maximum welfare at most~1.
To see this, consider a partition $(T_1,\ldots,T_k)$ in which some
player~$i$ is happy (with $v_i(T_i) = 1$).  For some $j \in S_i$,
player $i$ receives all the items in $P^j_i$.  Since $j \not\in S_{i'}$
for every $i' \neq i$, the only way to make a second player~$i'$ happy
is to give her all the items in $P^{\ell}_{i'}$ in some other partition
$P^{\ell}$ with $\ell \in S_{i'}$ (and hence $\ell \neq j$).  Since
$P^1,\ldots,P^t$ is an intersecting family, this is impossible ---
$P^j_i$ and $P^{\ell}_{i'}$ overlap for every $\ell \neq j$.

When the input to \mdisj is a 0-input, with an element $r$ in the
mutual intersection $\cap_{i=1}^k S_i$, we claim that the induced
input to \wm[k] is a 0-input, with maximum welfare at least~$k$.  This
is easy to see: for $i=1,2,\ldots,k$, assign the items of $P^r_i$ to
player~$i$.  Since $r \in S_i$ for every $i$, this makes all $k$
players happy.

This reduction shows that a (deterministic, nondeterministic, or
randomized) protocol for \wm[k] yields one for \mdisj (with $t$-bit
inputs) with the same communication.  We conclude that the
nondeterministic communication complexity of \wm[k] is
$\Omega(t/k) = \exp \{ \Omega(m/k^2) \}$.
\end{proof}

\subsection{Subadditive Valuations}

To an algorithms person, Theorem~\ref{thm_nisan} is depressing,
as it rules out any non-trivial positive results.
A natural idea is to seek positive results by imposing additional
structure on players' valuations.  Many such restrictions have been
studied.  We consider here the case of {\em subadditive} valuations
(see also Section~\ref{ss:when} of the preceding lecture),
where each $v_i$ satisfies $v_i(S \cup T) \le v_i(S) + v_i(T)$ for every
pair $S,T \subseteq M$.  

Our reduction in Theorem~\ref{thm_nisan} easily implies a weaker
inapproximability result for welfare maximization with subadditive
valuations.  Formally, 
define the \wm[2] problem as that of
identifying inputs that
fall into one of the following two cases:
\begin{itemize}

\item [(1)] Every partition $(T_1,\ldots,T_k)$ of the items has
  welfare at most~$k+1$.

\item [(0)] There exists a partition $(T_1,\ldots,T_k)$ of the items
  with welfare at least $2k$.

\end{itemize}
Communication lower bounds for \wm[2] apply also to the more general
problem of obtaining a better-than-$2$-approximation of the maximum
social welfare.

\begin{theorem}[\citet{DNS05}]\label{t:wm2}
The nondeterministic communication complexity of \wm[2] is $\exp \{
\Omega(m/k^2) \}$, even when all players have subadditive valuations.
\end{theorem}

This theorem follows from a modification of the proof of
Theorem~\ref{thm_nisan}.  The 0-1 valuations used in that proof are
not subadditive, but they can be made subadditive by adding~1 to each
bidder's valuation~$v_i(T)$ of each non-empty set~$T$.  The social
welfare obtained in inputs corresponding to 1- and 0-inputs of \mdisj
become $k+1$ and $2k$, respectively, and this completes the proof of
Theorem~\ref{t:wm2}.


There is also a quite non-trivial 
deterministic and polynomial-communication protocol
that guarantees a 2-approximation of the social welfare when bidders
have subadditive valuations~\citep{F06}.




\section{Lower Bounds on the Price of Anarchy of Simple Auctions}\label{s:condpoa}

The lower bounds of the previous section show that every protocol for
the welfare-maximization problem that interacts with the players and
then explicitly computes an allocation has either a bad approximation
ratio or high communication cost.  Over the past decade, many
researchers have considered shifting the work from the protocol to the
players, by analyzing the equilibria of simple auctions.  Can such
equilibria bypass the communication complexity lower bounds proved in
Section~\ref{s:cclb}?  The answer is not obvious, because equilibria
are defined non-constructively, and not through a low-cost
communication protocol.

\subsection{Auctions as Games}

What do we mean by a ``simple'' auction?  For example, recall the {\em
  simultaneous first-price auctions (S1As)} introduced in
Section~\ref{ss:when} of the preceding lecture.  Each player~$i$
chooses a strategy $b_{i1},\ldots,b_{im}$, with one bid per
item.\footnote{To keep the game finite, let's agree that each bid has
  to be an integer between 0 and some known upper bound $B$.}  Each
item is sold separately in parallel using a ``first-price
auction''---the item is awarded to the highest bidder on that item,
with the selling price equal to that bidder's bid.\footnote{In the
  preceding lecture we mentioned
 the {\em Vickrey} or {\em second-price} auction,
  where the winner does not pay their own bid, but rather the highest
  bid by someone else (the second-highest overall).  We'll stick with
  S1As for simplicity, but similar results are known for simultaneous
  second-price auctions, as well.}  The payoff of a player~$i$ in a given
outcome (i.e., given a choice of strategy for each player) is then her
utility:
\[
\underbrace{v_i(T_i)}_{\text{value of items won}} -
\underbrace{\sum_{j \in T_i} b_{ij}}_{\text{price paid for them}},
\]
where $T_i$ denotes the items on which $i$ is the highest bidder
(given the bids of the others).  

Bidders strategize already in a first-price auction for a single
item---a bidder certainly doesn't want to bid her actual valuation
(this would guarantee utility~0), and instead will ``shade'' her bid
down to a lower value.  (How much to shade is a tricky question, and
depends on what the other bidders are doing.)  Thus it makes sense to
assess the performance of an auction by its equilibria.  As usual, a
Nash equilibrium comprises a (randomized) strategy for each player, so
that no player can unilaterally increase her expected payoff through a
unilateral deviation to some other strategy (given how the other
players are bidding).




\subsection{The Price of Anarchy}\label{ss:poa}

So how good are the equilibria of various auction games, such as S1As?  To
answer this question, we use an analog of the approximation ratio,
adapted for equilibria.  Given a game~$G$ (like an S1A) and a nonnegative
maximization objective function~$f$ on the outcomes (like the social
welfare), \citet{KP99} defined
the {\em price of anarchy (POA)} of~$G$
as the ratio between the
objective function value of an optimal solution, and that of the worst
equilibrium:
\[ 
\mathsf{PoA}(G):= \frac{f(OPT(G))}{\min_{\text{$\rho$ is an
      equilibrium of $G$}} f(\rho)},
\] 
where $OPT(G)$ denotes the optimal outcome of~$G$ (with respect
to~$f$).\footnote{If~$\rho$ is a probability distribution over outcomes, as in
a mixed Nash equilibrium, then
$f(\rho)$ denotes the expected value of~$f$ w.r.t.~$\rho$.}
Thus the price of anarchy of a game quantifies the inefficiency of
selfish behavior.\footnote{Games generally have multiple
  equilibria.  Ideally, we'd like an approximation guarantee that
  applies to {\em all} equilibria, so that we don't need to worry
  about which one is reached---this is the point of the POA.}
The POA of a game and a maximization objective function is always at
least~1.  We can identify ``good performance'' of a system
with strategic participants as having a POA close to~1.\footnote{One
  caveat is that it's not
  always clear that a system will reach an equilibrium in a
  reasonable amount of time.  A natural way to resolve this issue
is to relax
  the notion of equilibrium enough so that it become relatively
  easy to reach an equilibrium.  See Lunar Lecture~5 for more on
  this point.}

The POA depends on the choice of
equilibrium concept.  For example, the POA  with respect
to approximate Nash equilibria can only be worse (i.e., bigger) than for
exact Nash equilibria (since there are only more of the
former).



\subsection{The Price of Anarchy of S1As}

As we saw in Theorem~\ref{t:ffgl} of the preceding lecture, the
equilibria of simple auctions like S1As can be surprisingly
good.\footnote{The first result of this type, for simultaneous
  second-price auctions and bidders with submodular valuations, is due
  to \citet{CKS16}.}
We restate that result here.\footnote{For a proof, see the original
  paper~\cite{FFGL13} or course notes by the author~\cite[Lecture
  17.5]{w14}.}
\begin{theorem}[\citet{FFGL13}]\label{t:ffgl2}
In every S1A with subadditive bidder valuations, the POA is at most~2.
\end{theorem}
This result is particularly impressive because achieving an
approximation factor of~2 for the welfare-maximization problem with
subadditive bidder valuations by any means (other than brute-force
search) is not easy (see~\cite{F06}).

As mentioned last lecture,
a recent result shows that the analysis of~\cite{FFGL13} is tight.
\begin{theorem}[\citet{CKST16}]\label{t:ckst14}
The worst-case POA of S1As with subadditive bidder valuations is at
least~2.
\end{theorem}
The proof of Theorem~\ref{t:ckst14} is an ingenious explicit
construction---the authors exhibit a choice of subadditive bidder
valuations and a Nash equilibrium of the corresponding S1A so that the
welfare of this equilibrium is only half of the maximum possible.
One reason that proving results like Theorem~\ref{t:ckst14} is
challenging is that it can be difficult to solve for a (bad) equilibrium
of a complex game like a S1A.




\subsection{Price-of-Anarchy Lower Bounds from Communication Complexity}

Theorem~\ref{t:ffgl2} motivates an obvious question: can we do better?
Theorem~\ref{t:ckst14} implies that the analysis in~\cite{FFGL13}
cannot be improved, but can we reduce the POA by considering a
different auction?  Ideally, the auction would still be ``reasonably
simple'' in some sense.  Alternatively, perhaps no ``simple'' auction
could be better than S1As?  If this is the case, it's not clear how to
prove it directly---proving lower bounds via explicit constructions
auction-by-auction does not seem feasible.

Perhaps it's a clue that the POA upper bound of~2 for S1As
(Theorem~\ref{t:ffgl2}) gets stuck at the same threshold for which
there is a lower bound for protocols that use polynomial communication
(Theorem~\ref{t:wm2}).  It's not clear, however, that a lower bound
for low-communication protocols has anything to do with equilibria.
Can we extract a low-communication protocol from an equilibrium?

\begin{theorem}[\citet{R14}]\label{t:condpoa2}
  Fix a class $\V$ of possible bidder valuations.  Suppose that, for
  some $\alpha \ge 1$, there is no nondeterministic protocol with
  subexponential (in $m$) communication for the 1-inputs of the
  following promise version of the welfare-maximization problem with
  bidder valuations in~$\V$:
\begin{itemize}

\item [(1)] Every allocation has welfare at most $W^*/\alpha$.

\item [(0)] There exists an allocation with welfare at least $W^*$.

\end{itemize}
Let $\eps$ be bounded below by some inverse polynomial function of $k$
and $m$.
Then, for every auction with sub-doubly-exponential (in $m$) strategies
per player, the worst-case POA of $\eps$-approximate Nash equilibria with bidder valuations in $\V$
is at least $\alpha$.
\end{theorem}
Theorem~\ref{t:condpoa2} says that lower bounds for nondeterministic
protocols carry over to all ``sufficiently simple'' auctions, where
``simplicity'' is measured by the number of strategies available
to each player.  These POA lower bounds follow automatically from
communication complexity lower bounds, and do not require any new
explicit constructions.

To get a feel for the simplicity constraint, note that
S1As with integral bids between 0 and~$B$ have $(B+1)^m$ strategies per
player---singly exponential in~$m$.  On the other hand, in a
``direct-revelation'' auction, where each bidder is allowed to
submit a bid on each bundle $S \subseteq M$ of items, each player has a
doubly-exponential (in $m$) number of strategies.\footnote{Equilibria can
  achieve the optimal welfare in a direct-revelation auction, so some
  bound on the number of strategies is
  necessary in Theorem~\ref{t:condpoa2}.}

The POA lower bound promised by Theorem~\ref{t:condpoa2} is only for
approximate Nash equilibria; since the POA is a worst-case measure and
the set of $\eNE$ is nondecreasing with $\eps$, this is weaker than a
lower bound for exact Nash equilibria.  It is an open question whether
or not Theorem~\ref{t:condpoa2} holds also for the POA of exact Nash
equilibria.\footnote{Arguably, Theorem~\ref{t:condpoa2} is good enough
  for all practical purposes---a POA upper bound that holds for exact
  Nash equilibria and does not hold (at least approximately) for
  approximate Nash equilibria with very small $\eps$ is too brittle to
  be meaningful.}

Theorem~\ref{t:condpoa2} has a number of interesting corollaries.
First, consider the case where $\V$ is the set of subadditive
valuations. 
Since S1As have only a singly-exponential (in $m$) 
number of strategies per
player, Theorem~\ref{t:condpoa2} applies to them.  Thus, 
combining it with Theorem~\ref{t:wm2} recovers the POA lower bound of
Theorem~\ref{t:ckst14}---modulo the exact vs.\ approximate Nash equilibria
issue---and shows the optimality of the upper bound in
Theorem~\ref{t:ffgl2} without an explicit construction.
Even more interestingly, this POA lower bound of~2 
applies not only to S1As, but more generally to all
auctions in which each player has a sub-doubly-exponential number of
strategies.  
Thus, S1As are in fact {\em optimal} among the class of
all such auctions when bidders have subadditive valuations
(w.r.t.\ the worst-case POA of $\eps$-approximate Nash equilibria).

We can also take $\V$ to be the set of all (monotone) valuations, and
then combine Theorem~\ref{t:condpoa2} with Theorem~\ref{thm_nisan} to
deduce that no ``simple'' auction gives a non-trivial (i.e.,
better-than-$k$) approximation for general bidder valuations.  We
conclude that with general valuations, complexity is essential to any
auction format that offers good equilibrium guarantees.  This
completes the proof of Theorem~\ref{t:condpoa} from the preceding
lecture and formalizes the second folklore belief in
Section~\ref{ss:when}; we restate that result here.
\begin{theorem}[\cite{R14}]
With general valuations, \emph{every} simple auction can have
equilibria with social welfare arbitrarily worse than the maximum
possible.
\end{theorem}

\subsection{Proof of Theorem~\ref{t:condpoa2}}

Presumably, the proof of Theorem~\ref{t:condpoa2} extracts a
low-communication protocol from a good POA bound.  The hypothesis of
Theorem~\ref{t:condpoa2} offers the clue that we should be looking to
construct a nondeterministic protocol.  So what could we use an
all-powerful prover for?  We'll see that a good role for the prover is
to suggest a Nash equilibrium to the players.

Unfortunately, it can be too expensive for the prover to write down
the description of a Nash equilibrium, even in S1As.  Recall that a mixed strategy is a
distribution over pure strategies, and that each player has an exponential (in
$m$) number of pure strategies available in a S1A.  Specifying a Nash equilibrium thus
requires an exponential number of probabilities.  To circumvent this
issue, we resort to approximate Nash equilibria, which are guaranteed
to exist even if we
restrict ourselves to distributions with small descriptions.  We proved
this for two-player games in Solar Lecture~1 (Theorem~\ref{t:lmm});
the same argument works for games with any number of players.
\begin{lemma}[\citet{LMM03}]\label{l:lmm}
For every $\eps > 0$ and every game with $k$ players with strategy sets
$A_1,\ldots,A_k$, there exists an $\eps$-approximate Nash equilibrium with description length
polynomial in $k$, $\log (\max_{i=1}^k |A_i|)$, and~$\tfrac{1}{\eps}$.
\end{lemma}
In particular, every game with a sub-doubly-exponential number of
strategies admits an approximate Nash equilibrium with subexponential
description length.

We now proceed to the proof of Theorem~\ref{t:condpoa2}.

\begin{proof}
  (of Theorem~\ref{t:condpoa2}) Fix an auction with at most $A$
  strategies per player, and a value for
  $\eps = \Omega(1/\poly(k,m))$.  Assume that, no matter what the
  bidder valuations $v_1,\ldots,v_k \in \V$ are, the POA of
  $\eps$-approximate Nash equilibria of the auction is at most
  $\rho < \alpha$.  We will show that $A$ must be doubly-exponential
  in $m$.

Consider the following nondeterministic protocol for verifying a
1-input of the welfare-maximization problem---for convincing the $k$
players that every allocation has welfare at most $W^*/\alpha$.  See
also Figure~\ref{f:condpoa}.  The
prover writes on a publicly visible blackboard an $\eps$-approximate
Nash equilibrium
$(\sigma_1,\ldots,\sigma_k)$ of the auction,
with description length polynomial in $k$, $\log A$, and
$\tfrac{1}{\eps} = O(\poly(k,m))$ as guaranteed by Lemma~\ref{l:lmm}.
The prover also writes down the expected welfare contribution
$\expect{v_i(S)}$ of each bidder~$i$ in this equilibrium.

\begin{figure}
\centering
\includegraphics[width=.6\textwidth]{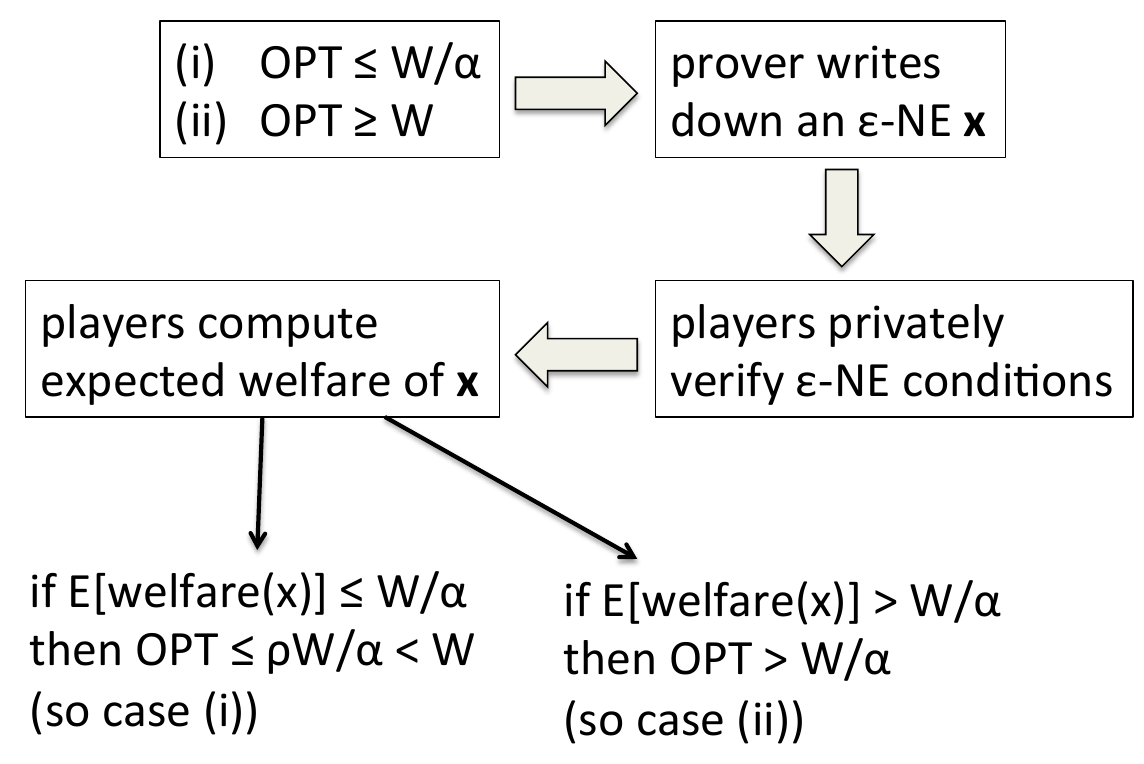}
\caption[Proof of Theorem~\ref{t:condpoa2}]{Proof of Theorem~\ref{t:condpoa2}.  How to extract a
  low-communication nondeterministic protocol from a good
  price-of-anarchy bound.}\label{f:condpoa}
\end{figure}

Given this advice, each player~$i$ verifies that $\sigma_i$ is indeed
an $\eps$-approximate best response to the other $\sigma_j$'s and that her
expected welfare is as claimed when all players play the mixed strategies
$\sigma_1,\ldots,\sigma_k$.  Crucially, player~$i$ is fully equipped
to perform both of these checks without any communication---she knows
her valuation $v_i$ (and hence her utility in each outcome of the
game) and the mixed strategies used by all players, and this is all
that is needed to verify her $\eps$-approximate Nash equilibrium
conditions and 
compute her expected contribution to the social welfare.\footnote{These
  computations may take a super-polynomial amount of time, but they
  do not contribute to the protocol's cost.}  Player~$i$
accepts if and only if the prover's advice passes these two tests, and
if the expected welfare of the equilibrium is at most~$W^*/\alpha$.

For the protocol correctness, consider first the case of a 1-input,
where every allocation has welfare at most $W^*/\alpha$.  If the prover
writes down the description of an arbitrary $\eps$-approximate Nash equilibrium and the appropriate
expected contributions to the social welfare, then all of the players
will accept (the expected welfare is obviously at most $W^*/\alpha$).
We also need to argue that, for the case of a 0-input---where
some allocation has welfare at least $W^*$---there is no proof that causes
all of the players to accept.  We can assume that the prover writes
down an $\eps$-approximate Nash equilibrium and its correct expected welfare~$W$, as otherwise at
least one player will reject.  Because the maximum-possible welfare is
at least $W^*$ and (by assumption) the POA of $\eps$-approximate Nash equilibria is at most $\rho <
\alpha$, the expected welfare of the given $\eps$-approximate Nash equilibrium must satisfy $W \ge
W^*/\rho > W^*/\alpha$.  The players will reject such a proof, so we
can conclude that the protocol is correct.  
Our assumption then implies that the protocol has communication cost
exponential in~$m$.
Since the cost of the protocol is polynomial in $k$, $m$, and $\log A$, 
$A$ must be doubly exponential in $m$.
\end{proof}
Conceptually, the proof of Theorem~\ref{t:condpoa2} argues
that, when the POA of $\eps$-approximate Nash equilibria is small,
every $\eps$-approximate Nash equilibrium  provides a
privately verifiable proof of a good upper bound
on the maximum-possible welfare.  When such upper bounds require large
communication, the equilibrium description length (and hence the
number of available strategies) must be large.

\section{An Open Question}\label{s:open}


While Theorems~\ref{t:wm2}, \ref{t:ffgl2}, and~\ref{t:condpoa2}
pin down the best-possible POA achievable by simple auctions with
subadditive bidder valuations, open questions remain for
other valuation classes.  For example, a valuation~$v_i$ is {\em
  submodular} if it satisfies
\[
v_i(T \cup \{j\}) - v_i(T) \le v_i(S \cup \{j\}) - v_i(S)
\]
for every $S \subseteq T \subset M$ and $j \notin T$.  This is a
``diminishing returns'' condition for set functions.  Every monotone
submodular function is also subadditive, so welfare-maximization with
the former valuations is only easier than with the latter.

The worst-case POA of S1As is exactly $\tfrac{e}{e-1} \approx 1.58$
when bidders have submodular valuations.  The upper bound was proved
by~\citet{ST13}, the lower bound by~\citet{CKST16}.  It is an open
question whether or not there is a simple auction with a smaller
worst-case POA.  The best lower bound known---for nondeterministic
protocols and hence, by Theorem~\ref{t:condpoa2}, for the POA of
$\eps$-approximate Nash equilibria of simple auctions---is
$\tfrac{2e}{2e-1} \approx 1.23$~\cite{DV13}.  Intriguingly, there is an upper
bound (very slightly) better than $\tfrac{e}{e-1}$ for
polynomial-communication protocols~\citep{FV06}---can this better
upper bound also be realized as the POA of a simple auction?  What is
the best-possible approximation guarantee, either for
polynomial-communication protocols or for the POA of simple auctions?
Resolving this question would require either a novel auction format
(better than S1As), a novel lower bound technique (better than
Theorem~\ref{t:condpoa2}), or both.

\section{Appendix: Proof of Theorem~\ref{t:mdisj}}\label{s:mdisj}

The proof of Theorem~\ref{t:mdisj} proceeds in three easy steps.

\vspace{.1in}
\noindent
\textbf{Step 1:}
{\em Every nondeterministic protocol with communication cost $c$ induces
a cover of the 1-inputs of $M(f)$ by at most $2^c$
monochromatic boxes.}  
By ``$M(f)$,'' we mean the $k$-dimensional array in which the $i$th
dimension is indexed by the possible inputs of player~$i$, and an
array entry contains the value of the function~$f$ on the
corresponding joint input.  By a ``box,'' we mean the $k$-dimensional
generalization of a rectangle---a subset of inputs that can be
written as a product $A_1 \times A_2 \times \cdots \times A_k$.
By ``monochromatic,'' we mean a box that does not contain both a
1-input and a 0-input.  (Recall that for the \mdisj problem there are
also inputs that are neither~1 nor~0---a monochromatic box can contain any
number of these.)
The proof of this step is the same as the standard one for the
two-party case (see e.g.~\cite{KN96}).



\vspace{.1in}
\noindent
\textbf{Step 2:}
{\em The number of $1$-inputs in $M(f)$ is $(k+1)^n$.}
In a $1$-input $\minputs$, for every 
coordinate $\ell$, at most one of the $k$ inputs has a 1 in the $\ell$th
coordinate.  This yields $k+1$ options for each of the $n$ coordinates,
thereby generating a total of $(k+1)^n$ 1-inputs.

\vspace{.1in}
\noindent
\textbf{Step 3:}
{\em The number of $1$-inputs in a monochromatic box is at most
  $k^n$.}  Let $B = A_1 \times A_2 \times \cdots \times A_k$ be a 1-box.
The key claim here is: for each coordinate $\ell=1,\ldots,n$, there is
a player $i \in \{1,\ldots,k\}$ such that, for every input $\bfx_i \in
A_i$, the $\ell$th coordinate of $\bfx_i$ is 0.  That is, to each
coordinate we can associate an ``ineligible player'' that, in this box,
never has a 1 in that coordinate.  This is easily seen by
contradiction: otherwise, there exists a coordinate $\ell$ such that,
for every player $i$, there is an input $\bfx_i \in A_i$ with a~1 in the
$\ell$th coordinate.  As a box, $B$ contains the input
$\minputs$.  But this is a 0-input, contradicting the assumption that
$B$ is a 1-box.

The claim implies the stated upper bound.  Every 1-input of $B$ can be
generated by choosing, for each coordinate $\ell$, an assignment of at
most one ``1'' in this coordinate to one of the $k-1$ eligible players
for this coordinate.  With only $k$ choices per coordinate, there are
at most $k^n$ 1-inputs in the box~$B$.

\vspace{.1in}
\noindent
\textbf{Conclusion:}
Steps~2 and~3 imply that covering the 1s of the $k$-dimensional
array of the \mdisj function requires at least $(1+\tfrac{1}{k})^n$
1-boxes.  By the discussion in Step~1, this implies a lower bound of
$n \log_2 (1 + \tfrac{1}{k}) = \Theta(n/k)$ on the nondeterministic
communication complexity of the \mdisj function (and output~1).  This
concludes the proof of Theorem~\ref{t:mdisj}.

\lecture{Why Prices Need Algorithms}

\vspace{1cm}



You've probably heard about ``market-clearing prices,'' which equate
the supply and demand in a market.  When are such prices
guaranteed to exist?  In the classical setting with divisible goods
(milk, wheat, etc.), market-clearing prices exist under reasonably
weak conditions~\cite{AD54}.  But with indivisible goods (houses,
spectrum licenses, etc.), such prices may or may not exist.  As you
can imagine, many papers in the economics and operations research
literatures study necessary and sufficient conditions for existence.
The punchline of today's lecture, based on joint work with Inbal
Talgam-Cohen~\cite{priceeq}, is that computational complexity
considerations in large part govern whether or not market-clearing
prices exist in a market of indivisible goods.  This is cool and
surprising because the question (existence of equilibria) seems to
have nothing to do with computation (cf., the questions studied in the
Solar Lectures).

\section{Markets with Indivisible Items}

The basic setup is the same as in the preceding lecture, when we were
studying price-of-anarchy bounds for simple combinatorial auctions (Section~\ref{s:camodel}).
To review, there are $k$ players, a set $M$ of $m$ items, and each
player~$i$ has a valuation $v_i:2^M \rightarrow \R_+$ describing her
maximum willingness to pay for each bundle of items.  For simplicity,
we also assume that $v_i(\emptyset)=0$ and that $v_i$ is monotone
(with $v_i(S) \le v_i(T)$ whenever $S \subseteq T$).  As in last
lecture, we will often vary the class~$\V$ of allowable valuations to
make the setting more or less complex.

\subsection{Walrasian Equilibria}

Next is the standard definition of ``market-clearing prices'' in a
market with multiple indivisible items.
\begin{definition}[\bf{Walrasian Equilibrium}]\label{d:we}
A {\em Walrasian equilibrium} is an allocation $S_1,\ldots,S_k$ of the
items of $M$ to the players and nonnegative prices $p_1,p_2,...,p_m$
for the items such that:
\begin{itemize}
\item [(W1)] All buyers are as happy as possible with their respective
  allocations, given   the prices:
for every~$i=1,2,\ldots,k$,
$S_i\in \text{argmax}_T \{ v_i(T)-\sum_{j\in T} p_j \}$.
\item [(W2)] Feasibility: $S_i \cap S_j = \emptyset$ for $i \neq j$.
\item [(W3)] The market clears: for every $j \in M$, $j \in S_i$ for
  some~$i$.\footnote{The  most common definition of a Walrasian
    equilibrium asserts instead that an item~$j$ is not awarded to any
    player only if~$p_j=0$.  With monotone valuations, there is no harm
    in insisting that every item is allocated.}
\end{itemize}
\end{definition}
Note that $S_i$ might be the empty set, if the prices are high enough
for~(W1) to hold for player~$i$.  Also, property~(W3) is crucial for
the definition to be non-trivial (otherwise set $p_j = +\infty$ for
every $j$).


Walrasian equilibria are remarkable: even though each player optimizes
independently (modulo tie-breaking) and gets exactly what she wants,
somehow the global feasibility constraint is respected.

%

\subsection{The First Welfare Theorem}

Recall from last lecture that the {\em social welfare} of an
allocation~$S_1,\ldots,S_k$ is defined as $\sum_{i=1}^k v_i(S_i)$.
Walrasian equilibria automatically maximize the social welfare, a
result known as the ``First Welfare Theorem.''

\begin{theorem}[\bf{First Welfare Theorem}]\label{t:fwt}
If the prices $p_1,p_2,\ldots,p_m$ and allocation
$S_1,S_2,\ldots,S_k$ of items constitute a Walrasian equilibrium, then
\[
(S_1,S_2,...,S_k)\in\textnormal{argmax}_{(T_1,T_2,...,T_k)}\sum_{i=1}^k
v_i(T_i),
\] 
where $(T_1,\ldots,T_k)$ ranges over all feasible allocations
(with $T_i \cap T_j = \emptyset$ for $i \neq j$).
\end{theorem}

If one thinks of a Walrasian equilibrium as the natural outcome of a
market, then
Theorem~\ref{t:fwt} can be interpreted as saying ``markets are
efficient.''\footnote{Needless to say, much blood and ink have been spilled
  over this interpretation over the past couple of centuries.}  There
are many versions of the ``First Welfare Theorem,'' and all have this flavor.

\noindent 
\begin{proof}
Let $(S^*_1,\ldots,S^*_k)$ denote a welfare-maximizing feasible
allocation.  We can apply property~(W1) of Walrasian equilibria to
obtain 
\[
v_i(S_i) - \sum_{j \in S_i} p_j  \ge                           
v_i(S^*_i) - \sum_{j \in S^*_i} p_j
\]
for each player~$i=1,2,\ldots,k$.  Summing over~$i$, we have
\begin{equation}\label{eq:we}
\sum_{i=1}^k v_i(S_i) - 
\sum_{i=1}^k \left( \sum_{j \in S_i} p_j \right) \ge    
\sum_{i=1}^k v_i(S^*_i) - \sum_{i=1}^k \left( \sum_{j \in S^*_i} p_j\right).
\end{equation}
Properties~(W2) and~(W3) imply that the second term on the left-hand
side of~\eqref{eq:we} equals the sum $\sum_{j=1}^m p_j$ of all the
item prices.  Since $(S^*_1,\ldots,S^*_n)$ is a feasible allocation,
each item is awarded at most once and hence the second term on the
right-hand side is at most $\sum_{j=1}^m p_j$.  Adding $\sum_{j=1}^m
p_j$ to both sides gives
\[
\sum_{i=1}^k v_i(S_i) \ge \sum_{i=1}^k v_i(S^*_i),
\]
which proves that the allocation $(S_1,\ldots,S_k)$ is also welfare-maximizing.
\end{proof}


\subsection{Existence of Walrasian Equilibria}

The First Welfare Theorem says that Walrasian equilibria are great
when they exist.  But when do they exist?
\begin{example}
Suppose $M$ contains only one item.  
Consider the allocation that awards the item to the player~$i$ with the
highest value for it, and a price that is between player $i$'s value and
the highest value of some other player (the second-highest overall).
This is a Walrasian equilibrium: the price is low enough that
bidder~$i$ prefers receiving the item to receiving nothing, and high
enough that all the other bidders prefer the opposite.  A simple case
analysis shows that these are all of the Walrasian equilibria.
\end{example}

\begin{example}\label{ex:nonex}
Consider a market with two items, $A$ and $B$.
Suppose the valuation of the first player is
\[ v_1(T) =
  \begin{cases}
   3     & \quad \text{for } T=\{A,B\} \\
    0 & \quad \text{otherwise} \\
  \end{cases}
\]
and that of the second player is
\[ v_2(T) =
  \begin{cases}
  2    & \quad \text{for } T \neq \emptyset\\
    0 & \quad \text{otherwise.} \\
  \end{cases}
\] 
The first bidder is called a ``single-minded'' or ``AND'' bidder, and
is happy only if she gets both items.  The second bidder is called a
``unit-demand'' or ``OR'' bidder, and effectively wants only one of
the items.\footnote{More formally, a {\em unit-demand} valuation $v$
  is characterized by nonnegative values $\{ \alpha_j \}_{j \in M}$, with
  $v(S) = \max_{j \in S} \alpha_j$ for each $S \sse M$.  Intuitively, a
  bidder with a unit-demand valuation throws away all her items except
her favorite.}

We claim that there is no Walrasian equilibrium in this market.  From
the First Welfare Theorem, we know that such an equilibrium must
allocate the items to maximize the social welfare, which in this case
means awarding both items to the first player.  For the second player
to be happy getting neither item, the price of each item must be at
least~2.  But then the first player pays~4 and has negative utility,
and would prefer to receive nothing.
\end{example}

These examples suggest a natural question: under what conditions is a
Walrasian equilibrium guaranteed to exist?  There is a well-known
literature on this question in economics (e.g.~\cite{KC82,GS99,M00});
here are the highlights.
\begin{enumerate}

\item If every player's valuation~$v_i$ satisfies the ``gross
  substitutes (GS)'' condition, then a Walrasian equilibrium is
  guaranteed to exist.  We won't need the precise definition of the~GS
  condition in this lecture.  GS valuations are closely related to
  weighted matroid rank functions, and hence are a subclass of the
  submodular valuations defined at the end of last lecture in
  Section~\ref{s:open}.\footnote{A weighted matroid rank function~$f$
    is defined using a matroid $(E,\mathcal{I})$ and nonnegative
    weights on the elements~$E$, with $f(S)$ defined as the maximum
    weight of an independent set (i.e., a member of $\mathcal{I}$)
    that lies entirely in the subset~$S \sse E$.}
A unit-demand (a.k.a.\ ``OR'') valuation, like that of the second player in
Example~\ref{ex:nonex}, satisfies the GS condition (corresponding to
the 1-uniform matroid).  It follows that single-minded (a.k.a.\
``AND'') valuations, like that of the first player in
Example~\ref{ex:nonex}, do not in general satisfy the GS condition
(otherwise the market in Example~\ref{ex:nonex} would have a Walrasian equilibrium).

\item If $\V$ is a class of valuations that contains all unit-demand
  valuations and also some valuation that violates the GS condition,
  then there is a market with valuations in $\V$ that does not possess
  a Walrasian equilibrium.

\end{enumerate}
These results imply that GS valuations are a maximal class of
valuations subject to the guaranteed existence of Walrasian
equilibria.  These results do, however, leave open the possibility of
guaranteed existence for classes $\V$ that contain non-GS valuations
but not all unit-demand valuations, and a number of recent papers in
economics and operations research have pursued this
direction (e.g.~\cite{BLN13,COP15,COP17,SY06}).  
All of the non-existence results in this line
of work use explicit constructions, like in Example~\ref{ex:nonex}.

\section{Complexity Separations Imply Non-Existence of Walrasian
  Equilibria}\label{s:priceeq}

\subsection{Statement of Main Result}

Next we describe a completely different approach to ruling out the
existence of Walrasian equilibria, based on complexity theory rather
than explicit constructions.  The main result is the following.
\begin{theorem}[\citet{priceeq}]\label{t:priceeq}
Let $\V$ denote a class of valuations.  Suppose the
welfare-maximization problem for $\V$ does not reduce to the
utility-maximization problem for $\V$.  Then, there exists a market
with all player valuations in $\V$ that has no Walrasian equilibrium.
\end{theorem}
In other words, a necessary condition for the guaranteed existence of
Walrasian equilibria is that welfare-maximization 
is no harder than utility-maximization.  This connects a
purely economic question (when do equilibria exist?) to a purely
algorithmic one.

To fill in some of the details in the statement of
Theorem~\ref{t:priceeq}, by ``does not reduce to,'' we mean that there
is no polynomial-time Turing reduction from the former problem to the
latter.  By ``the welfare-maximization problem for $\V$,'' we mean the
problem of, given player valuations~$v_1,\ldots,v_k \in \V$, computing
an allocation that maximizes the social welfare
$\sum_{i=1}^k v_i(S_i)$.\footnote{For concreteness, think about the case
  where every valuation~$v_i$ has a succinct description and can be
  evaluated in polynomial time.
Analogous
  results hold when an algorithm has only oracle access to the
  valuations.}  By ``the utility-maximization problem for $\V$,'' we
mean the problem of, given a valuation $v \in \V$ and nonnegative
prices $p_1,\ldots,p_m$, computing a utility-maximizing bundle
$S \in \textnormal{argmax}_{T \subseteq M} \{v(T) - \sum_{j \in T}
p_j\}$.

The utility-maximization problem, which involves only one player, can
generally only be easier than the multi-player welfare-maximization
problem.   Thus the two problems either have the same
computational complexity, or welfare-maximization is strictly harder.
Theorem~\ref{t:priceeq} asserts that whenever the second case holds,
Walrasian equilibria need not exist.  




\subsection{Examples}

Before proving Theorem~\ref{t:priceeq}, let's see how to apply it.
For most natural valuation classes~$\V$, a properly trained
theoretical computer scientist can identify the complexity of the
utility- and welfare-maximization problems in a matter of minutes.

\begin{example}[AND Valuations]\label{ex:and}
Let $\V_m$ denote the class of ``AND'' valuations for markets where
$|M|=m$.   That is, each $v \in \V_m$ has the following form, for some
$\alpha \ge 0$ and $T \subseteq M$:
\[v(S)= \begin{cases}
  \alpha   & \quad \text{if } S\supseteq T\\
    0 & \quad \text{otherwise.} \\
  \end{cases}\]
The utility-maximization problem for~$\V_m$ is trivial: for a single
player with an AND valuation with parameters $\alpha$ and $T$, the
better of $\emptyset$ or $T$ is a utility-maximizing bundle.  The
welfare-maximization problem for~$\V_m$ is essentially set packing and
is $\NP$-hard (with $m \rightarrow \infty$).\footnote{For example, given
  an instance $G=(V,E)$ of the {\sc Independent Set} problem, take $M=E$,
  make one player for each vertex $i \in V$, and give player~$i$ an
  AND valuation with parameters $\alpha = 1$ and $T$ equal to the
  edges that are incident to~$i$ in $G$.}  We conclude that the
welfare-maximization problem for~$\V$ does not reduce to the
utility-maximization problem for~$\V$ (unless $\ptime = \NP$).
Theorem~\ref{t:priceeq} then implies that, assuming $\ptime \neq \NP$,
there are markets with AND valuations that do not have any Walrasian
equilibria.\footnote{It probably seems weird to have a conditional
  result ruling out equilibrium existence.  A conditional
  non-existence result can of course be made unconditional through an
  explicit example.  A proof that the welfare-maximization problem for
  $\V$ is $\NP$-hard will generally suggest candidate markets to check
  for non-existence.

The following analogy may help: consider computationally tractable
linear programming relaxations of $\NP$-hard optimization problems.  
Conditional on $\ptime \neq \NP$, such relaxations cannot be exact
(i.e., have no integrality gap)
for all instances.  $\NP$-hardness
proofs generally suggest instances that can be used to prove directly
(and unconditionally) that a particular linear programming relaxation
has an integrality gap.\label{foot:gap}}
%
\end{example}

Of course, Example~\ref{ex:nonex} already shows, without any
complexity assumptions, that markets with AND bidders do not generally
have Walrasian equilibria.\footnote{Replacing the OR bidder in
  Example~\ref{ex:nonex} with an appropriate pair of AND bidders
  extends the example to markets with only AND bidders.}
Our next example addresses a class of valuations for which the
status of Walrasian equilibrium existence was not previously known.

\begin{example}[Capped Additive Valuations]
A {\em capped additive} valuation~$v$ is parameterized by $m+1$
numbers $c, \alpha_1,\alpha_2,\ldots,\alpha_m$ and is defined as
\[
v(S) = \min \left\{ c, \sum_{j \in S} \alpha_j \right\}.
\]
The $\alpha_j$'s indicate each item's value, and $c$ the ``cap'' on
the maximum value that can be attained.  Capped additive valuations
were proposed in~\citet{LLN06} as a natural subclass of submodular
valuations, and have been studied previously from a
welfare-maximization standpoint.

Let $\V_{m,d}$ denote the class of capped additive valuations in
markets with $|M|=m$ and with $c$ and $\alpha_1,\ldots,\alpha_m$
restricted to be positive integers between~1 and~$m^d$.  (Think of $d$
as fixed and $m \rightarrow \infty$.)  
A Knapsack-type dynamic programming algorithm shows that the
utility-maximization problem for~$\V_{m,d}$ can be solved in
polynomial time (using that $c$ and the $\alpha_j$'s are polynomially
bounded).  For~$d$ a sufficiently large constant, however, the
welfare-maximization problem for $\V_{m,d}$ is $\NP$-hard (it includes
the strongly $\NP$-hard Bin Packing problem).  Theorem~\ref{t:priceeq}
then implies that, assuming $\ptime \neq \NP$, there are markets with
valuations in~$\V_{m,d}$ with no Walrasian equilibrium.
\end{example}


%
%
%

\section{Proof of Theorem~\ref{t:priceeq}}\label{s:priceeqpf}

\subsection{The Plan}

Here's the plan for proving Theorem~\ref{t:priceeq}.
Fix a class $\V$ of valuations, and assume that a Walrasian
equilibrium exists in every market with player valuations in $\V$.  We
will show, in two steps, that the welfare-maximization problem for
$\V$ (polynomial-time Turing) reduces to the utility-maximization
problem for $\V$.

\vspace{.5\baselineskip}
\noindent
\textbf{Step 1: } The ``fractional'' version of the
welfare-maximization problem for $\V$ reduces to the
utility-maximization problem for $\V$.

\vspace{.5\baselineskip}
\noindent
\textbf{Step 2: } A market admits a Walrasian equilibrium if and only
if the 
fractional welfare-maximization problem has an optimal integral solution.
(We'll only need the ``only if'' direction.)

\vspace{.5\baselineskip}
Since every market with valuations in $\V$ admits a Walrasian
equilibrium (by assumption), these two steps imply that the integral
welfare-maximization problem reduces to utility-maximization.

\subsection{Step 1: Fractional Welfare-Maximization Reduces to
  Utility-Maximization}

This step is folklore, and appears for example in \citet{NS06}.
%
 %
%
Consider the following linear program (often called the {\em
  configuration LP}),
with one variable $x_{iS}$ for each player~$i$ and bundle~$S \subseteq 2^M$:
  \begin{align*}
    \max & \sum_{i=1}^k \sum_{S \subseteq M} v_i(S)x_{iS} \\
    \text{s.t.} & \sum_{i=1}^k \sum_{S \subseteq M \,:\, j \in S} x_{iS} \le
           1 \qquad \text{for $j=1,2,\ldots,m$} \\
      & \sum_{S \subseteq M} x_{iS} = 1 \qquad \text{for $i=1,2,\ldots,k$.} \\
  \end{align*}
The intended semantics are
\[
x_{iS} = \begin{cases}
  1   & \quad \text{if $i$ gets the bundle $S$}\\
    0 & \quad \text{otherwise.} \\
  \end{cases}\]
The first set of constraints enforces that each item is awarded only
once (perhaps fractionally), and the second set enforces that every
player receives one bundle (perhaps fractionally).  Every feasible
allocation induces a 0-1 feasible solution to this linear program
according to the intended semantics, and the objective function value
of this solution is exactly the social welfare of the allocation.

This linear program has an exponential (in $m$) number of variables.
The good news is that it has only a polynomial number of constraints.
This means that the dual linear program will have a polynomial number of
variables and an exponential number of constraints, which is right in
the wheelhouse of the ellipsoid method. 

Precisely, the dual linear program is:
  \begin{align*}
    \min &\sum_{i=1}^k u_i + \sum_{j=1}^m p_j \\
    \text{s.t.} \quad & u_i + \sum_{j \in S} p_j \ge v_i(S) \qquad \text{for
    all $i=1,2,\ldots,k$ and $S \subseteq M$}\\
      & p_j \ge 0 \qquad \text{for $j=1,2,\ldots,m$}, \\
  \end{align*}
where $u_i$ and $p_j$ correspond to the primal constraints that
bidder~$i$ receives one bundle and that item~$j$ is allocated at most once, respectively.

Recall that the ellipsoid method~\cite{K79} can solve a linear program
in time polynomial in the number of variables, as long as there is a
polynomial-time {\em separation oracle} that can verify whether or not
a given point is feasible and, if not, produce a violated constraint.
For the dual linear program above, this separation oracle boils down
to solving the following problem: for each player~$i=1,2,\ldots,k$,  
check that 
\[
u_i \ge \max_{S \subseteq M} \left[ v_i(S) - \sum_{j \in S} p_j
\right].
\]
But this reduces immediately to the utility-maximization problem for
$\V$!
Thus the ellipsoid method can be used to solve the dual linear program
to optimality, using a polynomial number of calls to a
utility-maximization oracle.  The optimal solution to the original
fractional welfare-maximization problem can then be efficiently
extracted from the optimal dual solution.\footnote{In more detail,
  consider the (polynomial number of) dual constraints generated by
  the ellipsoid method when solving the dual linear program.  Form a
  reduced version of the original primal problem, retaining only the
  (polynomial number of) variables that correspond to this subset of
  dual constraints.  Solve this polynomial-size reduced version of the primal
  linear program using your favorite polynomial-time linear
  programming algorithm.}
 
\subsection{Step 2: Walrasian Equilibria and Exact Linear Programming Relaxations}\label{ss:step2}

We now proceed with the second step, which is based on \citet{BM97} and
follows from strong linear programming duality.  Recall
from linear programming theory (see e.g.~\cite{chvatal}) that a pair
of primal and dual feasible solutions are both optimal if and only if
the ``complementary slackness'' conditions hold.\footnote{If you've
  never seen or have forgotten about complementary slackness, there's
  no need to be afraid.  To derive them, just write down the usual
  proof of weak LP duality (which is a chain of inequalities), and
  back out the conditions under which all the inequalities hold with
  equality.}
These conditions
assert that every non-zero decision variable in one of the linear
programs corresponds to a tight constraint in the other.  For our
primal-dual pair of linear programs, these conditions are:
\begin{itemize}

 \item [(i)] $x_{iS} > 0$ implies that $u_i = v_i(S) - \sum_{j \in S} p_j$
   (i.e., only utility-maximizing bundles are used);

 \item [(ii)] $p_j > 0$ implies that $\sum_i \sum_{S : j \in S} x_{iS} = 1$
   (i.e., item $j$ is not fully sold only if it is worthless).
\end{itemize}

Comparing the definition of Walrasian equilibria
(Definition~\ref{d:we}) with conditions~(i) and~(ii), we see that a
0-1 primal feasible solution $\bfx$ (corresponding to an allocation)
and a dual solution $\boldp$ (corresponding to item prices) constitute
a Walrasian equilibrium if and only if the complementary slackness
conditions hold (where~$u_i$ is understood to be set to
$\max_{S \subseteq M} v_i(S) - \sum_{j \in S} p_j$).  Thus a Walrasian
equilibrium exists if and only if there is a feasible 0-1 solution to
the primal linear program and a feasible solution to the dual linear
problem that satisfy the complementary slackness conditions, which in
turn holds if and only if the primal linear program has an optimal 0-1
feasible solution.\footnote{This argument re-proves the First Welfare
  Theorem (Theorem~\ref{t:fwt}).  It also proves the Second Welfare
  Theorem, which states that for every welfare-maximizing allocation,
  there exist prices that render it a Walrasian equilibrium---any
  optimal solution to the dual linear program furnishes such prices.}
We conclude that a Walrasian equilibrium exists if and only if the
fractional welfare-maximization problem has an optimal integral
solution.  This completes the proof of Theorem~\ref{t:priceeq}.


\section{Beyond Walrasian Equilibria}

For valuation classes~$\V$ that do not always possess Walrasian
equilibria, is it possible to define a more general notion of
``market-clearing prices'' so that existence is guaranteed?  For
example, what if we use prices that are more complex than item prices?
This section shows that complexity considerations provide an
explanation of why interesting generalizations of Walrasian equilibria
have been so hard to come by.

Consider a class~$\V$ of valuations, and a class~$\P$ of {\em pricing
  functions}.  A pricing function, just like a valuation, is a
function~$p:2^M \rightarrow \R_+$ from bundles to nonnegative
numbers.  The item prices $p_1,\ldots,p_m$ used to define Walrasian
equilibria correspond to additive pricing functions, with $p(S) =
\sum_{j \in S} p_j$.
The next definition articulates the appropriate generalization of
Walrasian equilibria to more general classes of pricing functions.
\begin{definition}[Price Equilibrium]\label{d:pe}
A {\em price equilibrium} (w.r.t.\ pricing functions $\P$) is an
allocation $S_1,\ldots,S_k$ of the
items of $M$ to the players and a pricing function $p \in \P$
such that:
\begin{itemize}
\item [(P1)] All buyers are as happy as possible with their respective
  allocations, given   the prices:
for every~$i=1,2,\ldots,k$,
$S_i\in \text{argmax}_T \{v_i(T)-p(T)\}$.
\item [(P2)] Feasibility: $S_i \cap S_j = \emptyset$ for $i \neq j$.
\item [(P3)] Revenue maximizing, given the prices: 
$(S_1,S_2,...,S_k)\in\textnormal{argmax}_{(T_1,T_2,...,T_k)}
\{ \sum_{i=1}^k p(T_i) \}$.
\end{itemize}
\end{definition}
Condition~(P3) is the analog of the market-clearing condition~(W3) in
Definition~\ref{d:we}.  It is not enough to assert that all items are
sold, because with a general pricing function, different ways of
selling all of the items can lead to different amounts of revenue.  
Under conditions~(P1)--(P3), the First Welfare Theorem
(Theorem~\ref{t:fwt}) still holds, with essentially the same proof,
and so every price equilibrium maximizes the social welfare.

For which choices of valuations~$\V$ and pricing functions~$\P$ is
Definition~\ref{d:pe} interesting?  
Ideally, the following properties should hold.
\begin{enumerate}

\item Guaranteed existence: for every set~$M$ of items and valuations
  $v_1,\ldots,v_k \in \V$, there exists a price equilibrium with
  respect to $\P$.

\item Efficient recognition: there is a polynomial-time algorithm
  for checking whether or not a given allocation
  and pricing function constitute a price equilibrium.  This boils down
  to assuming that utility-maximization (with respect to~$\V$
  and~$\P$) and revenue-maximization (with respect to~$\P$) are
  polynomial-time solvable problems (to check~(W1) and~(W3),
  respectively).

\item 
Markets with valuations in $\V$
  do not always have a Walrasian equilibrium.  (Otherwise, why bother
  generalizing item prices?)

\end{enumerate}

We can now see why there are no known natural choices of~$\V$ and~$\P$
that meet these three requirements.  The first two requirements imply
that the welfare-maximization problem belongs to $\NP \cap \coNP$.  To
certify a lower bound of~$W^*$ on the maximum social welfare, one can
exhibit an allocation with social welfare at least~$W^*$.  To certify
an upper bound of~$W^*$, one can exhibit a price equilibrium that has
welfare at most~$W^*$---this is well defined by the first condition,
efficiently verifiable by the second condition, and correct by the
First Welfare Theorem.

Problems in $(\NP \cap \coNP) \setminus \ptime$ appear to be rare,
especially in combinatorial optimization.
The preceding paragraph gives a heuristic argument that interesting
generalizations of Walrasian equilibria are possible only for
valuation classes for which welfare-maximization is polynomial-time
solvable.  
For every natural such class known, the linear programming relaxation
in Section~\ref{s:priceeqpf} has an optimal integral solution; in this
sense, solving the configuration LP appears to be a ``universal algorithm''
for polynomial-time welfare-maximization.  But the third requirement
asserts that a Walrasian equilibrium does not always exist in markets
with valuations in~$\V$ and so, by the second step of the proof of
Theorem~\ref{t:priceeq} (in Section~\ref{ss:step2}),
there are markets for which
the configuration LP sometimes has only fractional optimal solutions.

The upshot is that interesting generalizations of Walrasian equilibria
appear possible only for valuation classes  where a non-standard
algorithm is necessary and sufficient to solve the
welfare-maximization problem in polynomial time.
It is not clear if there are any natural valuation classes 
for which this algorithmic barrier can be overcome.%
\footnote{See~\cite[Section 5.3.2]{priceeq} for an unnatural
  such class.}

\lecture{The Borders of Border's Theorem}

\vspace{1cm}



Border's theorem~\cite{B91} is a famous result in auction theory about
the design space of single-item auctions, and it provides an explicit
linear description of the single-item auctions that are ``feasible''
in a certain sense.
Despite the theorem's fame, there have been few generalizations of it.
This lecture, based on joint work with Parikshit Gopalan and Noam
Nisan~\cite{GNR18}, uses complexity theory to explain why:
if there {\em were} significant generalizations of
Border's theorem, the polynomial hierarchy would collapse!


\section{Optimal Single-Item Auctions}

\subsection{The Basics of Single-Item Auctions}\label{ss:basics}

Single-item auctions have made brief appearances in previous lectures;
let's now study the classic model, due to \citet{V61}, in earnest.
There is a single seller of a single item.
There are $n$ bidders, and each bidder~$i$ has a valuation~$v_i$ for
the item (her maximum willingness to pay).  Valuations are {\em
  private}, meaning that~$v_i$ is known a priori to bidder~$i$ but not
to the seller or the other bidders.  Each bidder wants to maximize the
value obtained from the auction ($v_i$ if she
wins, 0 otherwise) minus the price she has to pay.  In the presence of
randomization (either in the input or internal to the auction), we
assume that bidders are risk-neutral, meaning they act to maximize
their expected utility.

This lecture is our only one on the classical {\em Bayesian}
model of auctions, which can be viewed as a form of average-case
analysis.  The key assumption is that each valuation~$v_i$ is drawn
from a distribution~$F_i$ that is known to the seller and possibly the
other bidders.  The actual realization~$v_i$ remains unknown to
everybody other than bidder~$i$.
For simplicity we'll work with
discrete distributions, and let~$V_i$ denote the support of~$F_i$ and~$f_i(v_i)$ the probability that
bidder~$i$'s valuation is~$v_i \in V_i$.  Typical examples include
(discretized versions of) the uniform distribution, the lognormal
distribution, the exponential distribution, and power-law
distributions.  We also assume that bidders' valuations are stochastically independent.

When economists speak of an ``optimal auction,'' they usually mean
the auction that maximizes the seller's expected revenue with respect
to a known prior distribution.\footnote{One advantage of assuming a
  distribution over inputs is that there is an unequivocal way to
  compare the performance of different auctions (by their expected
  revenues),   and hence an unequivocal way to define an optimal
  auction.  
One auction generally earns more revenue than another on some
inputs and less on others, so in the absence of a prior distribution,
it's not clear which one to prefer.}
Before identifying optimal auctions, we need to formally define the design space.
%
%
The auction designer needs to decide who wins and how much they pay.
Thus the designer must define two (possibly randomized) functions of
the bid vector $\vec{b}$: an \emph{allocation rule} $\vec{x}(\vec{b})$
which determines which bidder wins the item, where $x_i=1$ and if $i$
wins and $x_i=0$ otherwise, and a \emph{payment rule}
$\vec{p}(\vec{b})$ where $p_i$ is how much $i$ pays.
We impose the constraint that whenever bidder~$i$ bids~$b_i$, the expected
payment $\expect{p_i(\vec{b})}$ of the bidder is at most $b_i$ times the
probability $x_i(\vec{b})$ that she wins.  (The randomization is over
the bids by the other bidders and any randomness internal to the
auction.)  This participation constraint
ensures that a bidder who does not
overbid will obtain nonnegative expected utility from the auction.
(Without it, an auction could just charge $+\infty$ to every bidder.)
The {\em revenue} of an auction on the bid vector $\vec{b}$ is
$\sum_{i=1}^n p_i(\vec{b})$.

%

For example, in the {\em Vickrey} or {\em second-price auction}, the
allocation rule awards the item to the highest bidder, and the payment
rule charges the second-highest bid.  This auction is
\emph{(dominant-strategy) truthful},
meaning that for each bidder, truthful bidding (i.e., setting
$b_i=v_i$) is a \emph{dominant strategy}  that maximizes
her utility no matter what the other bidders do.  
With such a truthful auction, there is no need to assume that the
distributions~$F_1,\ldots,F_n$ are known to the bidders.
The beauty of
the Vickrey auction is that it delegates underbidding to the
auctioneer, who determines the optimal bid for the winner on their
behalf.

A {\em first-price auction} has the same allocation rule as a
second-price auction (give the item to the highest bidder), but the 
payment rule charges the winner her bid.
Bidding truthfully in a first-price auction guarantees zero utility,
so strategic bidders will underbid.  Because bidders do not have
dominant strategies---the optimal amount to underbid depends on the bids
of the others---it is non-trivial to reason about the outcome of
a first-price auction.  
The traditional solution is to assume that the distributions
$F_1,\ldots,F_n$ are known in advance to the bidders, and to consider
Bayes-Nash equilibria.  Formally, a {\em strategy} of a bidder~$i$ in
a first-price auction is a predetermined plan for bidding---a function
$b_i(\cdot)$ that maps a valuation~$v_i$ to a bid~$b_i(v_i)$ (or a
distribution over bids).  The semantics are: ``when my valuation is
$v_i$, I will bid $b_i(v_i)$.''  We assume that bidders'
strategies are common knowledge, with bidders' valuations (and hence
induced bids) private as usual.  A strategy profile
$b_1(\cdot),\cdots,b_n(\cdot)$ is a {\em Bayes-Nash equilibrium} if
every bidder always bids optimally given her information---if for
every bidder~$i$ and every valuation~$v_i$, the bid~$b_i(v_i)$
maximizes~$i$'s expected utility, where the expectation is with
respect to the distribution over the bids of other bidders induced
by~$F_1,\ldots,F_n$ and their bidding strategies.\footnote{Straightforward
  exercise: if there are~$n$ bidders with valuations drawn i.i.d.\
  from the uniform distribution on $[0,1]$, then setting
  $b_i(v_i) = \tfrac{n-1}{n} \cdot v_i$ for every~$i$ and~$v_i$ yields
  a Bayes-Nash equilibrium.}
Note that the set of Bayes-Nash equilibria of an auction generally
depends on the prior distributions $F_1,\ldots,F_n$.

An auction is called {\em Bayesian incentive compatible (BIC)} if truthful
bidding (with $b_i(v_i)=v_i$ for all~$i$ and $v_i$) is a Bayes-Nash
equilibrium.  That is, as a bidder, if all other bidders bid
truthfully, then you also want to bid truthfully.
A second-price auction is BIC, while a first-price
auction is not.\footnote{The second-price auction is in fact {\em
    dominant-strategy incentive compatible (DSIC)}---truthful bidding
    is a dominant strategy for every bidder, not merely a Bayes-Nash
    equilibrium.}
However, for every choice of $F_1,\ldots,F_n$,
there is a BIC auction that is equivalent to the first-price
auction.  
Specifically: given bids $a_1,\ldots,a_n$, implement the
outcome of the first-price auction with bids
$b_1(a_1),\ldots,b_n(a_n)$, where $b_1(\cdot),\ldots,b_n(\cdot)$
denotes a Bayes-Nash equilibrium of the first-price auction (with 
prior distributions $F_1,\ldots,F_n$).  
Intuitively, this auction makes the following pact with each bidder:
``you promise to tell me your true valuation, and I promise to bid
on your behalf as you would in a Bayes-Nash equilibrium.''
More generally, this simulation argument shows that for {\em every}
auction $A$, distributions $F_1,\ldots,F_n$, and Bayes-Nash equilibrium of
$A$ (w.r.t.\ $F_1,\ldots,F_n$), there is a BIC auction $A'$ whose
(truthful) outcome (and hence expected revenue) matches that of the
chosen Bayes-Nash equilibrium of $A$.  
This result is known as the {\em Revelation Principle.}
This principle implies that, to identify
an optimal auction, there is no loss of generality in restricting to
BIC auctions.\footnote{Of course, non-BIC auctions like first-price
  auctions are still useful in practice.  For example, the description
  of the first-price auction does not depend on bidders'
  valuation distributions $F_1,\ldots,F_n$ and can be deployed
  without knowledge of them.  This is not the case for the simulating
  auction.
}

\subsection{Optimal Auctions}\label{ss:m81}

In optimal auction design, the goal is to identify an expected
revenue-maximizing auction, as a function of the prior distributions
$F_1,\ldots,F_n$.  For example, suppose that $n=1$, and we restrict
attention to truthful auctions.
The only truthful auctions are take-it-or-leave-it offers (or a
randomization over such offers).  That is, the selling price must be
independent of the bidder's bid, as any dependence would result in
opportunities for the bidder to game the auction.
The optimal truthful auction is then the take-it-or-leave-it offer at
the price~$r$ that maximizes
\[
\underbrace{r}_{\text{revenue of a sale}} \cdot
\underbrace{(1-F(r))}_{\text{probability of a sale}},
\]
where $F$ denotes the bidder's valuation distribution.
Given a distribution~$F$, it is usually a simple matter to solve for the
best~$r$.  An optimal offer price is called a {\em
  monopoly price} of the distribution $F$.  
For example, if~$F$ is the
uniform distribution on $[0,1]$, then the monopoly price is
$\tfrac{1}{2}$.

Myerson~\cite{myerson} gave a complete solution to the optimal
single-item auction design problem, in the form of a generic compiler
that takes as input prior distributions~$F_1,\ldots,F_n$ and outputs a
closed-form description of the optimal auction for~$F_1,\ldots,F_n$.
The optimal auction is particularly easy to interpret in
the symmetric case, in which bidders' valuations
are drawn i.i.d.\ from a common distribution~$F$.  Here, the optimal
auction is simply a second-price auction with a reserve price~$r$ equal
to the monopoly price of~$F$ (i.e., an eBay auction with a suitably
chosen opening bid).\footnote{Intuitively, a reserve price of $r$ acts
as  an extra bid of~$r$ submitted by the seller.  In a second-price
auction with a reserve price, the winner is the highest
  bidder who clears the reserve (if any).  The winner (if any) pays
  either the reserve price or the second-highest bid, whichever is
  higher.}\footnote{Technically, this statement holds under a mild ``regularity''
  condition on the distribution~$F$, which holds for all of the most
  common parametric distributions.}
For example, with any number~$n$ of bidders
with valuations drawn i.i.d.\ from the uniform distribution on
$[0,1]$, the optimal single-item auction is a second-price auction
with a reserve price of $\tfrac{1}{2}$.
This is a pretty amazing confluence of theory and
practice---we optimized over the space of all imaginable auctions (which
includes some very strange specimens), and discovered that the
theoretically optimal auction format is one that is already in
widespread use!\footnote{In particular, 
there is always an optimal
  auction in which truthful bidding is a dominant strategy (as opposed to
  merely being a BIC auction).  This is
  also true in the asymmetric case.}

Myerson's theory of optimal auctions extends to the asymmetric case where
bidders have different distributions (where the optimal auction is
no longer so simple), and also well beyond single-item
auctions.\footnote{The theory applies more generally to
  ``single-parameter problems.'' 
These include problems
in which in each outcome a bidder is either a ``winner'' or a
``loser'' (with multiple winners allowed),
and each bidder~$i$ has a
private valuation~$v_i$ for winning (and value~0 for losing).\label{foot:sparam}}
The books by \citet{hartline} and the author~\cite[Lectures
3 and 5]{f13} describe this theory from a computer science perspective.


\section{Border's Theorem}

\subsection{Context}

Border's theorem identifies a tractable description of {\em all}
BIC single-item auctions, in the form of a polytope in polynomially
many variables.  (See Section~\ref{ss:basics} for the definition of a
BIC auction.)  This goal is in some sense more ambitious than
merely identifying the optimal auction; with this tractable
description in hand, one can efficiently compute the optimal auction
for any given set $F_1,\ldots,F_n$ of prior distributions.

Economists are interested in Border's theorem because it can be used
to extend the reach of Myerson's optimal auction theory
(Section~\ref{ss:m81}) to more general settings, such as the case of
risk-adverse bidders studied by \citet{MR84}.
\citet{M84} conjectured the precise result that was proved
by~\citet{B91}.
Computer scientists have used Border's theorem for orthogonal
extensions to Myerson's theory, like computationally tractable
descriptions of the expected-revenue maximizing auction in settings
with multiple non-identical items~\cite{A+19,CDW12}.  While there is no
hope of deriving a closed-form solution to the optimal auction design
problem with risk-adverse bidders or with multiple items, Border's
theorem at least enables an efficient algorithm for computing a
description of an optimal auction (given descriptions of the prior
distributions).

\subsection{An Exponential-Size Linear Program}

As a lead-in to Border's theorem, we show how to formulate the space
of BIC single-item auctions as an (extremely big) linear program.
The decision variables of the linear program encode the
allocation and payment rules of the auction (assuming truthful
bidding, as appropriate for BIC auctions).
There is one variable $x_i(\vec{v}) \in [0,1]$ that describes the
probability (over any randomization in the auction) that bidder~$i$
wins the item when bidders' valuations (and hence bids) are $\vec{v}$.
Similarly, $p_i(\vec{v}) \in \R_+$ denotes the expected payment made
by bidder~$i$ when bidders' valuations are $\vec{v}$.

Before describing the linear program, we need some odd but useful
notation (which is standard in game theory and microeconomics).
\begin{mdframed}[style=offset,frametitle={Some Notation}]
For an $n$-vector $\vec{z}$ and a coordinate~$i \in [n]$,
let~$\vec{z}_{-i}$ denote the $(n-1)$-vector obtained by removing the
$i$th component from $\vec{z}$.  We also identify~$(z_i,\vec{z}_{-i})$
with $\vec{z}$.
\end{mdframed}
Also, recall that $V_i$ denotes the possible valuations of bidder~$i$, and
that we assume that this set is finite.

Our linear program will have three sets of constraints.  The first set
enforces the property that truthful bidding is in fact a Bayes-Nash
equilibrium (as required for a BIC auction).  
For every bidder~$i$,
possible valuation $v_i \in V_i$ for~$i$, and possible false bid $v'_i
\in V_i$,
\begin{equation}\label{eq:bic1}
\underbrace{v_i \cdot \expect[\vec{v}_{-i} \sim \vec{F}_{-i}]{x_i(\vec{v})}
- \expect[\vec{v}_{-i} \sim
\vec{F}_{-i}]{p_i(\vec{v})}}_{\text{expected utility of truthful bid $v_i$}}
\ge
\underbrace{v_i \cdot \expect[\vec{v}_{-i} \sim \vec{F}_{-i}]{x_i(v'_i,\vec{v}_{-i})}
- \expect[\vec{v}_{-i} \sim
\vec{F}_{-i}]{p_i(v'_i,\vec{v}_{-i})}}_{\text{expected utility of
  false bid $v'_i$}}.
\end{equation}
The expectation is over both the randomness in $\vec{v}_{-i}$ and
internal to the auction.  Each of the expectations in~\eqref{eq:bic1}
expands to a sum over all possible $\vec{v}_{-i} \in \vec{V}_{-i}$,
weighted by the probability $\prod_{j \neq i} f_j(v_j)$.  Because all of
the $f_j(v_j)$'s are numbers known in advance, each of these
constraints is linear (in the $x_i(\vec{v})$'s and $p_i(\vec{v})$'s).

The second set of constraints encode the participation 
constraints from Section~\ref{ss:basics}, also known as the {\em
  interim individually rational (IIR)} constraints.
For every bidder~$i$ and
possible valuation~$v_i \in V_i$,
\begin{equation}\label{eq:iir1}
v_i \cdot \expect[\vec{v}_{-i} \sim \vec{F}_{-i}]{x_i(\vec{v})}
- \expect[\vec{v}_{-i} \sim
\vec{F}_{-i}]{p_i(\vec{v})} \ge 0.
\end{equation}
The final set of constraints assert that, with probability~1, the item is
sold to at most one bidder: for every $\vec{v} \in \vec{V}$,
\begin{equation}\label{eq:feas1}
\sum_{i=1}^n x_i(\vec{v}) \le 1.
\end{equation}

By construction, feasible solutions to the linear
system~\eqref{eq:bic1}--\eqref{eq:feas1} correspond to the allocation
and payment rules of BIC auctions with respect to the distributions
$F_1,\ldots,F_n$.  This linear program has an exponential number of
variables and constraints, and is not immediately useful.

\subsection{Reducing the Dimension with Interim Allocation Rules}

Is it possible to re-express the allocation and payment rules of BIC
auctions with a small number of decision variables?  Looking at the
constraints~\eqref{eq:bic1} and~\eqref{eq:iir1}, a natural idea is to use
only the decision variables 
$\{ y_i(v_i) \}_{i \in [n], v_i \in V_i}$
and
$\{ q_i(v_i) \}_{i \in [n], v_i \in V_i}$,
with the intended semantics that
\[
y_i(v_i) = \expect[\vec{v}_{-i}]{x_i(v_i,\vec{v}_{-i})}
\quad \text{and} \quad
q_i(v_i) = \expect[\vec{v}_{-i}]{p_i(v_i,\vec{v}_{-i})}.
\]
In other words, $y_i(v_i)$ is the probability that bidder $i$ 
wins when she bids $v_i$, and $q_i(v_i)$ is the expected amount that
she pays; these were the only quantities that actually mattered
in~\eqref{eq:bic1} and~\eqref{eq:iir1}.
(As usual, the expectation is over both the randomness
in $\vec{v}_{-i}$ and internal to the auction.)
In auction theory, the $y_i(v_i)$'s are called an {\em interim
  allocation rule}, the $q_i(v_i)$'s an {\em interim payment
  rule}.\footnote{Auction theory generally thinks about three
  informational scenarios: {\em ex ante}, where each bidder knows the prior
  distributions but not even her own valuation; {\em interim},
  where each bidder knows her own valuation but not those of the others;
  and {\em ex post}, where all of the bidders know everybody's
  valuation.  Bidders typically choose their bids at the interim stage.}

There are only $2 \sum_{i=1}^n |V_i|$ such decision variables, far
fewer than the $2 \prod_{i=1}^n |V_i|$ variables
in~\eqref{eq:bic1}--\eqref{eq:feas1}.
We'll think of the $|V_i|$'s (and hence the number of decision variables) 
as polynomially bounded.  For example,~$V_i$ could be the multiples of
some small $\eps$ that lie in some bounded range like $[0,1]$.

We can then express the BIC constraints~\eqref{eq:bic1} in terms of
this smaller set of variables by
\begin{equation}\label{eq:bic2}
\underbrace{v_i \cdot y_i(v_i)
- q_i(v_i)}_{\text{expected utility of truthful bid $v_i$}}
\ge
\underbrace{v_i \cdot y_i(v'_i)
- q_i(v'_i)}_{\text{expected utility of
  false bid $v'_i$}}
\end{equation}
for every bidder~$i$ and $v_i,v'_i \in V_i$.  Similarly, the IIR
constraints~\eqref{eq:iir1} become
\begin{equation}\label{eq:iir2}
v_i \cdot y_i(v_i) - q_i(v_i) \ge 0
\end{equation}
for every bidder~$i$ and $v_i \in V_i$.


Just one problem.  What about the feasibility
constraints~\eqref{eq:feas1}, which reference
the individual
$x_i(\vec{v})$'s and not merely their expectations?
The next definition articulates what feasibility means for an
interim allocation rule.

\begin{definition}[Feasible Interim Allocation Rule]
An interim allocation rule $\{ y_i(v_i) \}_{i \in [n], v_i \in V_i}$
is {\em feasible} if there exist nonnegative values for $\{
x_i(\vec{v}) \}_{i \in   [n], \vec{v} \in \vec{V}}$ such that
\[
\sum_{i=1}^n x_i(\vec{v}) \le 1
\] 
for every $\vec{v}$ (i.e., the
$x_i(\vec{v})$'s constitute a feasible allocation rule), and
\[
y_i(v_i) = \underbrace{\sum_{\vec{v}_{-i} \in \vec{V}_{-i}} \left(
    \prod_{j  \neq     i} f_j(v_j) \right) \cdot x_i(v_i,\vec{v}_{-i})}_{\expect[\vec{v}_{-i}]{x_i(v_i,\vec{v}_{-i})}}
\]
for every $i \in [n]$ and $v_i \in V_i$ (i.e., the intended semantics
are respected).
\end{definition}
In other words, the feasible interim allocation rules are exactly the
projections (onto the $y_i(v_i)$'s) of the feasible (ex post)
allocation rules.

The big question is: 
how can we translate interim feasibility 
into our new, more economical vocabulary?\footnote{In
  principle, we know this is possible.  The feasible (ex post)
  allocation rules form a polytope, the projection of a polytope is
  again a polytope, and every polytope can be described by a finite
  number of linear inequalities.  So the real question is
whether or not there's a
{\em computationally useful} description of interim feasibility.}
As we'll see, Border's theorem~\cite{B91} provides a crisp
and computationally useful solution.

\subsection{Examples}

To get a better feel for the issue of checking the feasibility of an
interim allocation rule, let's consider a couple of examples.
A necessary condition for interim feasibility is that the item is
awarded to at most one bidder in expectation (over the randomness in
the valuations and internal to the auction):
\begin{equation}\label{eq:nec}
\sum_{i=1}^n \underbrace{\sum_{v_i \in V_i} f_i(v_i)
  y_{i}(v_i) }_{\prob{\text{$i$ wins}}} \le 1.
\end{equation}
Could this also be a sufficient condition?  That is, is every
interim allocation rule $\{ y_i(v_i) \}_{i \in [n], v_i \in V_i}$ that
satisfies~\eqref{eq:nec} induced by a bona fide (ex post) allocation
rule?

\begin{example}\label{ex:ex1}
Suppose there are $n=2$ bidders.
Assume that $v_1,v_2$ are independent and each is equally likely
to be 1 or 2.
Consider the interim allocation rule given by 
\begin{equation}\label{eq:ex1}
y_1(1) = \tfrac{1}{2},
y_1(2) = \tfrac{7}{8},
y_2(1) = \tfrac{1}{8}, \text{ and }
y_2(2) = \tfrac{1}{2}.
\end{equation}
Since $f_i(v) = \tfrac{1}{2}$ for all $i=1,2$ and $v=1,2$,
the necessary condition in~\eqref{eq:nec} is satisfied.
Can you find an (ex post) allocation rule that induces this
interim rule?  Answering this question is
much like solving a Sudoku or KenKen puzzle---the goal is to fill in
the table entries in Table~\ref{t:blank} so that each row sums to at
most 1 (for feasibility) and that the
constraints~\eqref{eq:ex1} are satisfied.  For example, the average of
the top two entries in the first column of Table~\ref{t:blank} should
be $y_1(1) = \tfrac{1}{2}$.
In this example, there are a number of such solutions; one is shown in
Table~\ref{t:ex1}.  Thus, the given interim allocation rule is feasible.
\end{example}

\begin{table}
\centering
\begin{tabular}{|c|c|c|}
\hline
$(v_1, v_2)$ & $x_1(v_1, v_2)$ & $x_2(v_1, v_2)$\\
\hline
$(1,1)$ & & \\
$(1,2)$ & & \\
$(2,1)$ & & \\
$(2,2)$ & & \\
\hline
\end{tabular}
\caption{Certifying feasibility of an interim allocation rule is
  analogous to filling in the table entries while respecting 
constraints on the sums of certain subsets of entries.}
\label{t:blank}
\end{table}

\begin{table}
\centering
\begin{tabular}{|c|c|c|}
\hline
$(v_1, v_2)$ & $x_1(v_1, v_2)$ & $x_2(v_1, v_2)$\\
\hline
$(1,1)$ & 1& 0\\
$(1,2)$ & 0& 1\\
$(2,1)$ & 3/4& 1/4\\
$(2,2)$ & 1& 0\\
\hline
\end{tabular}
\caption{One solution to Example~\ref{ex:ex1}.}
\label{t:ex1}
\end{table}

\begin{example}\label{ex:ex2}
Suppose we change the interim allocation rule to
\begin{equation*}
y_1(1) = \tfrac{1}{4},
y_1(2) = \tfrac{7}{8},
y_2(1) = \tfrac{1}{8}, \text{ and }
y_2(2) = \tfrac{3}{4}.
\end{equation*}
The necessary condition~\eqref{eq:nec} remains satisfied.  Now,
however, the interim rule is not feasible.  One way to see this is to note
that $y_1(2) = \tfrac{7}{8}$ implies that $x_1(2,2) \ge \tfrac{3}{4}$ 
and hence $x_2(2,2) \le \tfrac{1}{4}$.  Similarly, $y_2(2) =
\tfrac{3}{4}$ implies that $x_2(2,2) \ge \tfrac{1}{2}$, a contradictory
constraint.
\end{example}

The first point of Examples~\ref{ex:ex1} and~\ref{ex:ex2} is that it
is not trivial to check whether or not a given interim allocation rule
is feasible---the problem corresponds to solving a big linear system of
equations and inequalities.  The second point is that~\eqref{eq:nec}
is not a sufficient condition for feasibility.  In hindsight, trying
to summarize the exponentially many ex post feasibility
constraints~\eqref{eq:feas1} with a single interim
constraint~\eqref{eq:nec} seems naive.  Is there a larger set of
linear constraints---possibly an exponential number---that
characterizes interim feasibility?


\subsection{Border's Theorem}

Border's theorem states that a collection of ``obvious'' necessary
conditions for interim feasibility are also sufficient.
To state these conditions,
assume for notational convenience
that the valuation sets $V_1,\ldots,V_n$ are disjoint.\footnote{This is
without loss of generality, since we can simply ``tag'' each valuation
$v_i \in V_i$ with the ``name'' $i$ (i.e., view each $v_i \in V_i$
as the set $\{ v_i,i\}$).}
Let $\{ x_i(\vec{v}) \}_{i \in [n], \vec{v} \in \vec{V}}$
be a feasible (ex post) allocation rule and $\{ y_i(v_i) \}_{i \in
  [n], v_i \in V_i}$ the induced (feasible) interim allocation rule.
Fix for each bidder $i$ a set $S_i \sse V_i$ of valuations.
Call the valuations $\cup_{i=1}^n S_i$ the {\em distinguished}
valuations. 
Consider first the probability, over the random valuation profile
$\vec{v} \sim \vec{F}$ and any coin flips of the ex post allocation rule,
that the winner of the auction (if any) has a distinguished valuation.
By linearity of expectations, this probability can be expressed in
terms of the interim allocation rule:
\begin{equation}\label{eq:lhs}
\sum_{i=1}^n \sum_{v_i \in S_i} f_i(v_i) y_i(v_i).
\end{equation}
The expression~\eqref{eq:lhs} is linear in the $y_i(v_i)$'s.

The second quantity we study is the probability, over $\vec{v} \sim
\vec{F}$, that there is a bidder with a distinguished valuation.  This has
nothing to do with the allocation rule, and is a function of the prior
distributions only:
\begin{equation}\label{eq:rhs}
1 - \prod_{i=1}^n \left( 1 - \sum_{v_i \in S_i} f_i(v_i)  \right).
\end{equation}
Because there can only be a winner with a distinguished valuation if there is
a bidder with a distinguished valuation, the quantity in~\eqref{eq:lhs} can
only be less than~\eqref{eq:rhs}.  Border's theorem asserts that these
conditions, ranging over all choices of $S_1 \sse V_1,\ldots,S_n \sse
V_n$, are also sufficient for the feasibility of an interim allocation
rule.

\begin{theorem}[Border's theorem~\cite{B91}]\label{t:border}
An interim allocation rule $\{ y_i(v_i) \}_{i \in [n], v_i \in V_i}$
is feasible if and only if for 
every choice $S_1 \sse V_1,\ldots,S_n \sse V_n$ of distinguished
valuations,
\begin{equation}\label{eq:border}
\sum_{i=1}^n \sum_{v_i \in S_i} f_i(v_i)y_i(v_i)
\le
1 - \prod_{i=1}^n \left( 1 - \sum_{v_i \in S_i} f_i(v_i) \right).
\end{equation}
\end{theorem}
Border's theorem can be derived from the max-flow/min-cut theorem
(following~\cite{B07,CKM13}); we include the proof in
Section~\ref{s:borderpf} for completeness.


Border's theorem yields an explicit description as a linear system of
the feasible interim allocation rules induced by BIC single-item auctions.
To review, this linear system is
\begin{align}
\label{eq:b1}
v_i \cdot y_i(v_i) - q_i(v_i)
&\ge
v_i \cdot y_i(v'_i) - q_i(v'_i)
& \forall i \text { and } v_i, v_i' \in V_i
\\
\label{eq:b2}
v_i \cdot y_i(v_i) - q_i(v_i)
&\ge 0
& \forall i \text { and } v_i \in V_i
\\
\label{eq:b3}
\sum_{i=1}^n \sum_{v_i \in S_i} f_i(v_i) y_i(v_i)
&\le
1 - \prod_{i=1}^n \left( 1 - \sum_{v_i \in S_i} f_i(v_i)
\right)
& \forall S_1 \sse V_1,\ldots,S_n \sse V_n.
\end{align}
For example, optimizing the objective function
\begin{equation}\label{eq:rev}
\max \sum_{i=1}^n f_i(v_i) \cdot q_i(v_i)
\end{equation}
over the linear system~\eqref{eq:b1}--\eqref{eq:b3} computes the
expected revenue of an optimal BIC single-item auction for the
distributions $F_1,\ldots,F_n$.


The linear system~\eqref{eq:b1}--\eqref{eq:b3} has only a polynomial
number of variables (assuming the $|V_i|$'s are polynomially bounded),
but it does have an exponential number of constraints of the form~\eqref{eq:b3}. 
One solution is to use the ellipsoid method, as the linear system does
admit a polynomial-time separation
oracle~\cite{A+19,CDW12}.\footnote{This
  is not immediately obvious, as the max-flow/min-cut argument in
  Section~\ref{s:borderpf} involves an exponential-size graph.}
Alternatively, \citet{A+19} provide a polynomial-size extended
formulation of the polytope of feasible interim allocation rules (with
a polynomial number of additional decision variables and only
polynomially many constraints).  In any case, we conclude that there
is a computationally tractable description of the feasible interim
allocation rules of BIC single-item auctions.

\section{Beyond Single-Item Auctions: A Complexity-Theoretic Barrier}

Myerson's theory of optimal auctions (Section~\ref{ss:m81}) extends
beyond single-item auctions to all ``single-parameter'' settings (see
footnote~\ref{foot:sparam} for discussion and Section~\ref{ss:pp} for two examples).  Can Border's
theorem be likewise extended?  There are analogs of Border's theorem
in settings modestly more general than single-item auctions, including
$k$-unit auctions with unit-demand bidders~\cite{A+19,CDW12,CKM13},
and approximate versions of Border's theorem exist fairly
generally~\cite{CDW12,CDW12b}.  Can this state-of-the-art be improved
upon?  
We next use complexity theory to develop evidence for a negative answer.
\begin{theorem}[\citet{GNR18}]\label{t:gnr15}
(Informal) There is no exact Border's-type theorem for settings
significantly more general than the known special cases (unless $\PH$
collapses). 
\end{theorem}
We proceed to defining what we mean by ``significantly more general''
and a ``Border's-type theorem.''

\subsection{Two Example Settings}\label{ss:pp}

The formal version of Theorem~\ref{t:gnr15} conditionally rules out
``Border's-type theorems'' for several specific settings that are
representative of what a more general version of Border's theorem
might cover.  We mention two of these here (more are in~\cite{GNR18}).

In a {\em public project} problem, there is a binary decision to make:
whether or not to undertake a costly project (like building a new
school).  
Each bidder~$i$ has a private valuation~$v_i$ for the outcome where
the project is built, and valuation~0 for the outcome where it is not.
If the project is built, then everyone can use it.
In this setting, feasibility means that all bidders receive the
same allocation: $x_1(\vec{v}) = x_2(\vec{v}) = \cdots = x_n(\vec{v}) \in
[0,1]$ for every valuation profile~$\vec{v}$.

In a {\em matching} problem, there is a set $M$ of items, and each
bidder is only interested in receiving a specific pair~$j,\ell \in M$
of items.  (Cf., the AND bidders of the preceding lecture.)  For each
bidder, the corresponding pair of items is common knowledge, while the
bidder's valuation for the pair is private as usual.  Feasible
outcomes correspond to (distributions over) matchings in the graph
with vertices~$M$ and edges given by bidders' desired pairs.

The public project and matching problems are both ``single-parameter''
problems (i.e., each bidder has only one private parameter).  As such,
Myerson's optimal auction theory (Section~\ref{ss:m81}) can be used to
characterize the expected revenue-maximizing auction.  Do these
settings also admit analogs of Border's theorem?

\subsection{Border's-Type Theorems}

What do we actually mean by a ``Border's-type theorem?''  
Because we aim to prove impossibility results, we should adopt a
definition that is as permissive as possible.
Border's theorem (Theorem~\ref{t:border}) gives a characterization of
the feasible interim allocation rules of a single-item auction as
the solutions to a finite system of linear inequalities.  This by
itself is not impressive---the set is a polytope, and 
as such is guaranteed to have such a characterization.
The appeal of Border's theorem is that the
characterization uses only the ``nice''
linear inequalities in~\eqref{eq:border}.
Our ``niceness'' requirement is that the characterization use only
linear inequalities that can be efficiently recognized and tested.
This is a weak necessary condition for such a characterization to be
computationally useful.

\begin{definition}[Border's-Type Theorem]\label{d:gbt}
A {\em
  Border's-type theorem} holds for an auction design setting if, for every
instance of the setting (specifying the number of bidders and
their prior distributions, etc.),
there is a system of linear inequalities such that the following
properties hold.
\begin{enumerate}

\item (Characterization) The feasible solutions of the linear system
  are precisely the feasible interim allocation rules of the instance.

\item (Efficient recognition) There is a polynomial-time algorithm
  that can decide whether or not a given linear inequality (described
  as a list of coefficients) belongs to
  the linear system.

\item (Efficient testing)
The bit complexity of each linear inequality is polynomial in the
description of the instance.
(The number of inequalities can be exponential.)


\end{enumerate}
\end{definition}
For example, consider the original Border's theorem, for single-item 
auctions (Theorem~\ref{t:border}).
The recognition problem is straightforward: the left-side
of~\eqref{eq:border} encodes the $S_i$'s, from which the right-hand
side can be computed and checked in polynomial time.
It is also evident that every inequality
in~\eqref{eq:border} has a polynomial-length description.%
\footnote{The characterization in Theorem~\ref{t:border} and the
  extensions in~\cite{A+19,CDW12,CKM13} have additional features not
  required or implied by Definition~\ref{d:gbt}, such as a
  polynomial-time separation oracle (and even a compact extended
  formulation in the single-item case~\cite{A+19}).  The impossibility
  results in Section~\ref{ss:imp} rule out analogs of Border's theorem
  that merely satisfy Definition~\ref{d:gbt}, let alone these stronger
  properties.}

\subsection{Consequences of a Border's-Type Theorem}

The high-level idea behind the proof of Theorem~\ref{t:gnr15} is to
show that a Border's-type theorem puts a certain computational problem
low in the polynomial hierarchy, and then to show that this problem is
$\sharpP$-hard for the public project and matching settings defined in
Section~\ref{ss:pp}.\footnote{Recall that Toda's theorem~\cite{toda}
  implies that a $\sharpP$-hard problem is contained in the polynomial
  hierarchy only if $\PH$ collapses.}  The computational problem is:
given a description of an instance (including the prior
distributions), compute the maximum-possible expected revenue that can
be obtained by a feasible and BIC auction.\footnote{Sanity check: this
  problem turns out to be polynomial-time solvable in the setting of  single-item auctions~\cite{GNR18}.}

What use is a Border's-type theorem?  For starters,
it implies that the problem of
testing the feasibility of an interim allocation rule is 
in $\coNP$.
To prove the infeasibility of such a rule,
one simply exhibits an inequality of the characterizing linear system
that the rule fails to satisfy.
Verifying this failure
reduces to the recognition and testing problems, which
by Definition~\ref{d:gbt} are polynomial-time
solvable.
\begin{proposition}\label{prop:conp}
If a Border's-type theorem holds for an auction design setting,
then the membership problem for the polytope of feasible interim
allocation rules belongs to $\coNP$.
\end{proposition}


Combining Proposition~\ref{prop:conp} with the ellipsoid method puts
the problem of computing the maximum-possible expected revenue in
$\ptime^{\NP}$.

\begin{theorem}\label{t:main}
  If a Border's-type theorem holds for an auction design setting, then
  the maximum expected revenue of a feasible BIC auction can be
  computed in $\ptime^{\NP}$.
\end{theorem}

\begin{proof}
  We compute the optimal expected revenue of a BIC auction via linear
  programming, as follows.  The decision variables are the same
  $y_i(v_i)$'s and $q_i(v_i)$'s as
  in~\eqref{eq:b1}--\eqref{eq:b3}, and we retain the BIC
  constraints~\eqref{eq:b1} and the IIR constraints~\eqref{eq:b2}.  By
  assumption, we can replace the single-item interim feasibility
  constraints~\eqref{eq:b3} with a linear system that satisfies the
  properties of Definition~\ref{d:gbt}.
  The maximum expected revenue of a feasible BIC auction can then be
  computed by optimizing a linear objective function (in the
  $q_i(v_i)$'s, as in~\eqref{eq:rev}) subject to these constraints.
  Using the ellipsoid method~\cite{K79}, this can be accomplished with
  a polynomial number of invocations of a separation oracle (which
  either verifies feasibility or exhibits a violated constraint).
  Proposition~\ref{prop:conp} implies that we can implement this
  separation oracle in $\coNP$, and thus compute the maximum expected
  revenue of a BIC auction in~$\ptime^{\NP}$.\footnote{One detail:
    Proposition~\ref{prop:conp} only promises solutions to the
    ``yes/no'' question of feasibility, while a separation oracle
    needs to produce a violated constraint when given an infeasible
    point.  But under mild conditions (easily satisfied here), an
    algorithm for the former problem can be used to solve the latter
    problem as well~\cite[P.189]{S86}.}
\end{proof}

\subsection{Impossibility Results from Computational Intractability}\label{ss:imp}

Theorem~\ref{t:main} concerns the problem of computing the maximum
expected revenue of a feasible BIC auction, given a description of an
instance.  It is easy to classify the complexity of this problem in
the public project and matching settings introduced in
Section~\ref{ss:pp} (and several other settings, see~\cite{GNR18}).

\begin{proposition}\label{prop:pp}
Computing the maximum expected revenue of a feasible BIC auction of a
public project instance is a $\sharpP$-hard problem.
\end{proposition}
Proposition~\ref{prop:pp} is a straightforward reduction from the
$\sharpP$-hard problem of computing the number of feasible solutions
to an instance of the \textsc{Knapsack} problem.\footnote{An aside for
  aficionados of the analysis of Boolean functions:
  Proposition~\ref{prop:pp} is essentially equivalent to the
  $\sharpP$-hardness of checking whether or not given Chow parameters
  can be realized by some bounded function on the hypercube.
  See~\cite{GNR18} for more details on the surprisingly strong
  correspondence between Myerson's optimal auction theory (in the
  context of public projects) and the analysis of Boolean functions.}

\begin{proposition}\label{prop:match}
Computing the maximum expected revenue of a feasible BIC auction of a
matching instance is a $\sharpP$-hard problem.
\end{proposition}
Proposition~\ref{prop:match} is a straightforward reduction from the
$\sharpP$-hard {\sc Permanent} problem.

We reiterate that Myerson's optimal auction theory applies to the
public project and matching settings, and in particular gives a
polynomial-time algorithm that outputs a description of an optimal
auction (for given prior distributions).  Moreover, the optimal
auction can be implemented as a polynomial-time algorithm.  Thus
it's not hard to figure out what the optimal auction is, nor to
implement it---what's hard is figuring out exactly how much 
revenue it makes on average!

Combining Theorem~\ref{t:main} with Propositions~\ref{prop:pp}
and~\ref{prop:match} gives the following corollaries, which indicate
that there is no Border's-type theorem significantly more general than
the ones already known.
\begin{corollary}\label{cor:pp}
If $\sharpP \not\subseteq \PH$, then there is no Border's-type theorem
for the setting of public projects.
\end{corollary}

\begin{corollary}\label{cor:match}
If $\sharpP \not\subseteq \PH$, then there is no Border's-type theorem
for the matching setting.
\end{corollary}

\section{Appendix: A Combinatorial Proof of Border's
  Theorem}\label{s:borderpf}

\begin{figure}
\centering
\subfloat[]{\includegraphics[scale=.9]{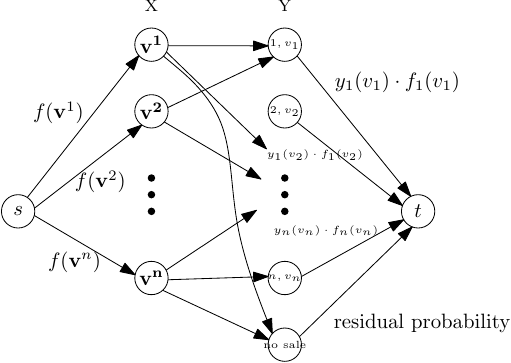}}
\qquad
\subfloat[]{\includegraphics[scale=.9]{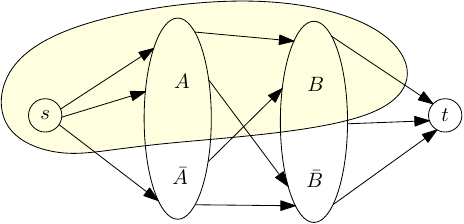}}
\caption{The max-flow/min-cut proof of Border's theorem.}
\label{f:border}
\end{figure}

\begin{proof}
(of Theorem~\ref{t:border})
We have already argued the ``only if'' direction, and now prove the
converse. The proof is by the max-flow/min-cut theorem---given the
statement of the theorem and this hint, the proof writes itself.

Suppose the interim
allocation rule $\{ y_i(v_i) \}_{i \in [n], v_i \in V_i}$ satisfies~\eqref{eq:border} for
every $S_1 \sse V_1,\ldots,S_n \sse V_n$.  
Form a
four-layer $s$-$t$ directed flow network $G$ as follows
(Figure~\ref{f:border}(a)). 
The first layer is the 
source $s$, the last the sink $t$.
In the second layer $X$, vertices correspond to valuation
profiles $\vec{v}$.  We abuse notation and refer to vertices of $X$ by the
corresponding valuation profiles.
There is an arc $(s,\vec{v})$ for every $\vec{v} \in X$, with capacity
$\prod_{i=1}^n f_i(v_i)$.  Note that the total capacity of these edges is~1.

In the third layer $Y$, vertices correspond to winner-valuation
pairs; there is also one additional ``no winner'' vertex.
We use $(i,v_i)$ to denote the vertex representing the event
that bidder $i$ wins the item and also has valuation $v_i$.  
For each $i$ and $v_i \in V_i$, there is an arc
$((i,v_i),t)$ with capacity $f_i(v_i) y_i(v_i)$.
There is also an arc from the ``no winner'' vertex to $t$, with
capacity $1 - \sum_{i=1}^n \sum_{v_i \in V_i} 
f_i(v_i) y_i(v_i) $.\footnote{If $\sum_{i=1}^n
  \sum_{v_i \in V_i} f_i(v_i) y_i(v_i) > 1$, then
 the interim allocation rule is clearly infeasible (recall~\eqref{eq:nec}).  
Alternatively, this would violate Border's
  condition for the choice $S_i = V_i$ for all $i$.}

Finally, each vertex $\vec{v} \in X$ has $n+1$ outgoing arcs, all with
infinite capacity, to the vertices
$(1,v_1),(2,v_2),$ $\ldots,$ $(n,v_n)$ of $Y$ and also to the ``no
winner'' vertex.

By construction, $s$-$t$ flows of $G$ with value~1 correspond to ex
post allocation rules with induced interim allocation rule
$\{ y_i(v_i) \}_{i \in [n], v_i \in V_i}$, with $x_i(\vec{v})$ equal
  to the amount of flow on the
arc $(\vec{v},(i,v_i))$ times $(\prod_{i=1}^n f_i(v_i))^{-1}$.  
%

To show that there exists a flow with value~1, it suffices to show
that every $s$-$t$ cut has value at least~1 (by the max-flow/min-cut
theorem).  So fix an $s$-$t$ cut.  Let this cut include the vertices
$A$ from $X$ and $B$ from $Y$.  Note that all arcs from $s$ to $X \setminus
A$ and from $B$ to $t$ are cut (Figure~\ref{f:border}(b)). 
For each bidder $i$, define $S_i \sse V_i$ as the possible valuations
of $i$ that are {\em not} represented among the valuation profiles
in $A$.   Then, for every valuation profile $\vec{v}$ containing at least
one distinguished valuation, the arc $(s,\vec{v})$ is cut.  The total
capacity of these arcs is the right-hand side~\eqref{eq:rhs} of
Border's condition.

Next, we can assume that every vertex of the form $(i,v_i)$ with
$v_i \notin S_i$ is in $B$, as otherwise an (infinite-capacity) arc
from $A$ to $Y \setminus B$ is cut.  Similarly, unless
$A = \emptyset$---in which case the cut has value at least~1 and we're
done---we can assume that the ``no winner'' vertex lies in $B$.  Thus,
the only edges of the form $((i,v_i),t)$ that are not cut involve a
distinguished valuation $v_i \in S_i$.  It follows that the total
capacity of the cut edges incident to $t$ is at least 1 minus the
left-hand size~\eqref{eq:lhs} of Border's condition.  Given our
assumption that~\eqref{eq:lhs} is at most~\eqref{eq:rhs}, this $s$-$t$
cut has value at least 1.  This completes the proof of Border's
theorem.
\end{proof}

\lecture{Tractable Relaxations of Nash Equilibria}

\vspace{1cm}

\section{Preamble}

Much of this monograph is about impossibility
results for the efficient computation of exact and approximate Nash
equilibria.  
How should we respond to such rampant
computational intractability?  What should be the message to
economists---should they change the way they do economic analysis in
some way?\footnote{Recall the discussion in Section~\ref{ss:whocares}
  of Solar Lecture~1: a critique of a widely used
  concept like the Nash equilibrium is not particularly helpful unless
  accompanied by a proposed alternative.}

One approach, familiar from coping with $\NP$-hard problems, is to
look for tractable special cases.  For example, Solar Lecture~1 proved
tractability results for two-player zero-sum games.
Some interesting tractable generalizations of zero-sum games have been
identified (see~\cite{CCDP} for a recent example), and polynomial-time
algorithms are also known for some relatively narrow classes of games
(see e.g.~\cite{graphicalgames}).  Still, for the lion's share of
games that we might care about, no polynomial-time algorithms for
computing exact or approximate Nash equilibria are known.

A different approach, which has been more fruitful, is to continue to
work with general games and look for an {\em equilibrium concept} that
is more computationally tractable than exact or approximate Nash
equilibria.  The equilibrium concepts that we'll consider---the
correlated equilibrium and the coarse correlated equilibrium---were
originally invented by game theorists, but computational complexity
considerations are now shining a much brighter spotlight on them.

Where do these alternative equilibrium concepts come from?  They
arise quite naturally from the study of uncoupled dynamics, which we
last saw in Solar Lecture~1.

\section{Uncoupled Dynamics Revisited}

Section~\ref{s:uncoupled} of Solar Lecture~1 introduced uncoupled
dynamics in the context of two-player games.  In this lecture we work
with the analogous setup for a general number~$k$ of players.
We use $S_i$ to denote the (pure) strategies of player~$i$, $s_i \in
S_i$ a specific strategy, $\sigma_i$ a mixed strategy, $\vec{s}$
and $\vec{\sigma}$ for profiles (i.e., $k$-vectors) of pure and mixed
strategies, and $u_i(\vec{s})$ for player~$i$'s payoff in the
outcome~$\vec{s}$.
\begin{mdframed}[style=offset,frametitle={Uncoupled Dynamics
    ($k$-Player Version)},nobreak=true]
At each time step $t=1,2,3,\ldots$:
\begin{enumerate}

\item Each player~$i=1,2,\ldots,k$ simultaneously chooses a mixed
  strategy $\sigma_i^t$ over~$S_i$ as a function only of her own
  payoffs and the strategies chosen by players in the first $t-1$ time steps.

\item Every player observes all of the strategies~$\vec{\sigma}^{t}$ chosen
  at time~$t$.

\end{enumerate}
\end{mdframed}
``Uncoupled'' refers to the fact that each player initially knows only
her own payoff function~$u_i(\cdot)$, while ``dynamics'' means a
process by which players learn how to play in a game.

One of the only positive algorithmic results that we've seen
concerned {\em smooth fictitious play (SFP)}.  The $k$-player
version of SFP is as follows.
\begin{mdframed}[style=offset,frametitle={Smooth Fictitious Play
    ($k$-Player Version)}] 
\textbf{Given: } parameter family $\{ \eta^t \in [0,\infty) \,:\,
t=1,2,3,\ldots\}$.

\vspace{\baselineskip}

At each time step $t=1,2,3,\ldots$:
\begin{enumerate}
\item Every player~$i$ simultaneously chooses the mixed strategy $\sigma_i^t$
by playing each strategy $s_i$ with
  probability proportional to $e^{\eta^t\pi^t_i}$, where $\pi^t_i$ is
  the time-averaged expected payoff player~$i$ would have earned by
  playing~$s_i$ at every previous time step.  Equivalently, $\pi^t_i$ is
  the expected payoff of
  strategy $s_i$ when the other players' strategies $\vec{s}_{-i}$ are
  drawn from the joint distribution $\tfrac{1}{t-1} \sum_{h=1}^{t-1}
  \vec{\sigma}^h_{-i}$.\tablefootnote{Recall from last lecture that
    for an $n$-vector $\vec{z}$ and a coordinate~$i \in [k]$,
$\vec{z}_{-i}$ denotes the $(k-1)$-vector obtained by removing the
$i$th component from $\vec{z}$, and we identify~$(z_i,\vec{z}_{-i})$
with $\vec{z}$.}

\item Every player observes all of the strategies~$\vec{\sigma}^{t}$ chosen
  at time~$t$.

\end{enumerate}
\end{mdframed}
A typical choice for the $\eta_t$'s is $\eta_t \approx \sqrt{t}$.

In Theorem~\ref{t:sfp} in Solar Lecture~1 we proved that, in an $m
\times n$
two-player zero-sum game, after $O(\log (m+n)/\eps^2)$ time steps, the
empirical distributions of the two players constitute an
$\eps$-approximate Nash equilibrium.\footnote{Recall the proof idea:
  smooth fictitious play corresponds to running the vanishing-regret
  ``exponential weights'' algorithm (with reward vectors induced by
  the play of others), and in a two-player zero-sum game, the
  vanishing-regret guarantee (i.e., with time-averaged payoff at least
  that   of the best fixed action in hindsight, up to~$o(1)$ error)
  implies the $\eps$-approximate Nash equilibrium condition.}
An obvious question is: what is the outcome of a logarithmic number of
rounds of smooth fictitious play in a non-zero-sum game?  Our
communication complexity lower bound in Solar Lectures~2 and~3 implies
that it cannot in general be an $\eps$-approximate Nash equilibrium.
Does it have some alternative economic meaning?  The answer to this
question turns out to be closely related to some classical
game-theoretic equilibrium concepts, which we discuss
next.


\section{Correlated and Coarse Correlated Equilibria}\label{s:ce}

\subsection{Correlated Equilibria}

The correlated equilibrium is a well-known equilibrium concept defined
by \citet{A74}.  We define it,
then explain the standard semantics, and then offer an
example.\footnote{This section
draws from~\cite[Lecture 13]{f13}.}
\begin{definition}[Correlated Equilibrium]\label{d:ce}
A joint distribution 
$\rho$ on the set $S_1 \times \cdots \times S_k$ of
outcomes of a game is a
{\em correlated equilibrium} if for every player $i \in
\{1,2,\ldots,k\}$, strategy $s_i \in S_i$, and 
deviation $s'_i \in S_i$,
\begin{equation}\label{eq:ce}
\expect[\vec{s} \sim \rho]{u_i(\vec{s}) \,|\, s_i} \ge 
\expect[\vec{s} \sim \rho]{u_i(s_i',\vec{s}_{-i}) \,|\, s_i}.
\end{equation}
\end{definition}
Importantly, the distribution $\rho$ in Definition~\ref{d:ce} need
not be a product distribution; in this sense, the strategies chosen by
the players are correlated.  The Nash equilibria of a game correspond to
the correlated equilibria that are product distributions. 

The usual interpretation of a correlated equilibrium
involves a trusted third party.  
The distribution $\rho$ over
outcomes is publicly known.  The trusted third party samples an outcome
$\vec{s}$ according to $\rho$.  For each player $i=1,2,\ldots,k$,
the trusted third party privately suggests the strategy $s_i$ to
$i$.  The player $i$ can follow the suggestion $s_i$, or not.  At the
time of decision making, a player~$i$ knows the distribution
$\rho$ and one component~$s_i$ of the realization $\vec{s}$, and
accordingly has a posterior distribution on others' suggested
strategies~$\vec{s}_{-i}$.  With these semantics, the correlated
equilibrium condition~\eqref{eq:ce} requires that every player
maximizes her expected payoff by playing the suggested strategy
$s_i$.  The expectation is conditioned on $i$'s
information---$\rho$ and $s_i$---and assumes that other players
play their recommended strategies $\vec{s}_{-i}$.

Definition~\ref{d:ce} is a bit of a mouthful.  But
you are intimately familiar with a good example of a
correlated equilibrium that is not a mixed Nash equilibrium---a
traffic light!
Consider the following two-player game, with each matrix entry listing
the payoffs of the row and column players in the corresponding outcome:
\begin{center}
\begin{tabular}{|c|c|c|}
  \hline
   & Stop & Go  \\ \hline
  Stop & 0,0 & 0,1 \\
  Go & 1,0 & -5,-5  \\
  \hline
\end{tabular}
\end{center}
This game has two pure Nash equilibria, the outcomes (Stop, Go) and
(Go, Stop).  Define $\rho$ by randomizing uniformly between
these two Nash equilibria.
This is not a product distribution over the game's four
outcomes, so it cannot correspond to a
Nash equilibrium of the game.  It is, however, a correlated
equilibrium.\footnote{For example, consider the row
player.  If the trusted third party (i.e., the traffic light) recommends
the strategy ``Go'' (i.e., is green), then the row player knows that
the column player was recommended ``Stop'' (i.e., has a red light).
Assuming the column player plays her recommended strategy and stops
at the red light, the best strategy for the row player is to follow her
recommendation and to go.}


\subsection{Coarse Correlated Equilibria}

The outcome of smooth fictitious play in non-zero-sum
games relates to a still more permissive equilibrium concept, the {\em
  coarse correlated equilibrium}, which was first studied by \citet{MV78}.
\begin{definition}[Coarse Correlated Equilibrium]\label{d:cce}
A joint distribution $\rho$ on the set $S_1 \times \cdots \times S_k$ of
outcomes of a game is a
{\em coarse correlated equilibrium} if for every player $i \in
\{1,2,\ldots,k\}$ and
every unilateral deviation $\strat'_i \in S_i$,
\begin{equation}\label{eq:cce}
\expect[\vec{s} \sim \rho]{u_i(\vec{s}) } \ge 
\expect[\vec{s} \sim \rho]{u_i(s_i',\vec{s}_{-i})}.
\end{equation}
\end{definition}
The condition~\eqref{eq:cce} is the same as that for
the Nash equilibrium (Definition~\ref{d:ne}),
except without the restriction that $\rho$ is a
product distribution.
In this condition,
when a player $i$ contemplates a deviation~$s_i'$, she knows only the
distribution~$\rho$ and {\em not} the component $s_i$ of the
realization.  That is, a coarse correlated equilibrium only protects
against unconditional unilateral deviations, as opposed to the
unilateral deviations conditioned on~$s_i$ that are addressed in
Definition~\ref{d:ce}.  It follows that every correlated equilibrium
is also a coarse correlated equilibrium (Figure~\ref{f:venn}).

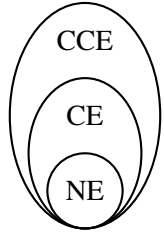
\begin{figure}
\begin{center}
  \begin{tikzpicture}[thick,scale=.5]
    \draw (0,1) ellipse (1cm and 1cm);
    \node at (0,1) {NE};
    \draw (0,2) ellipse (1.5cm and 2cm); 
    \node at (0,3) {CE};
    \draw (0,3) ellipse (2cm and 3cm); 
    \node at (0,5) {CCE};
  \end{tikzpicture}
\end{center}
\caption{The relationship between Nash equilibria (NE), correlated
  equilibria (CE), and coarse correlated equilibria (CCE).  Enlarging
the set of equilibria increases computational tractability but
decreases predictive power.}\label{f:venn}
\end{figure}

As you would expect, {\em$\eps$-approximate} correlated and coarse
correlated equilibria are defined by adding a ``$-\eps$'' to the
right-hand sides of~\eqref{eq:ce} and~\eqref{eq:cce}, respectively.
We can now answer the question about smooth fictitious play in
general games: the time-averaged history of
joint play under smooth fictitious play
converges to the set of coarse correlated equilibria.

\begin{proposition}[SFP Converges to CCE]\label{prop:cce}
  For every $k$-player game in which every player has at most $m$
  strategies, after $T=O((\log m)/\epsilon^2)$ time steps of
  smooth fictitious play, the time-averaged history of play
  $\tfrac{1}{T} \sum_{t=1}^T \vec{\sigma}^t$ is an
  $\epsilon$-approximate coarse correlated equilibrium.
\end{proposition}

Proposition~\ref{prop:cce} follows straightforwardly from the definition of
$\eps$-approximate coarse correlated equilibria and the vanishing
regret guarantee of smooth fictitious play that we proved in
Solar Lecture~1.  Precisely, by Corollary~\ref{cor:noregret} of that
lecture, after $O((\log m)/\epsilon^2)$ time steps of smooth
fictitious play, every player has at most~$\eps$ regret (with
respect to the best fixed strategy in hindsight, see
Definition~\ref{d:regreta} in Solar Lecture~1).  This regret guarantee
is equivalent to the conclusion of Proposition~\ref{prop:cce} (as you
should check).

What about correlated equilibria?  While the time-averaged history of
play in smooth fictitious play does not in general converge to the
set of correlated equilibria, \citet{FV97} and \citet{HM00} show that
the time-averaged play of
other reasonably simple types of uncoupled dynamics 
is guaranteed to be an $\eps$-correlated equilibrium after a
polynomial (rather than logarithmic) number of time steps.




\section{Computing an Exact Correlated or Coarse Correlated Equilibrium}\label{s:exact}

\subsection{Normal-Form Games}

Solar Lecture~1 showed that approximate Nash equilibria of two-player
zero-sum games can be learned (and hence computed) efficiently
(Theorem~\ref{t:sfp}).  Proposition~\ref{prop:cce} and the extensions
in~\cite{FV97,HM00} show analogs of this result for approximate
correlated and coarse correlated equilibria of general games.  
Solar Lecture~1 also
showed that an exact Nash equilibrium of a two-player zero-sum game
can be computed in polynomial time by linear programming
(Corollary~\ref{cor:zerosum}).  Is the same true for an exact
correlated or coarse correlated equilibrium of a general game?

Consider first the case of coarse correlated equilibria, and introduce
one decision variable~$x_{\vec{s}}$ per outcome~$\vec{s}$ of the game,
representing the probability assigned to $\vec{s}$ in a joint
distribution $\rho$.  The feasible solutions to the following linear
system are then precisely the coarse correlated equilibria of the
game:
\begin{align}\label{eq:cce1}
\sum_{\vec{s}} u_i(\vec{s})x_{\vec{s}} \ge
\sum_{\vec{s}} u_i(s_i',\vec{s}_{-i})x_{\vec{s}} & \quad\quad\text{for every $i \in
                                                        [k]$ and $s'_i
                                                       \in S_i$}\\
\label{eq:cce2}
\sum_{\vec{s} \in \vec{S}} x_{\vec{s}} = 1 &\\
\label{eq:cce3}
x_{\vec{s}} \ge 0 & \quad\quad\text{for every $\vec{s} \in \vec{S}$.}
\end{align}
Similarly, correlated equilibria are captured by the following linear
system:
\begin{align}\label{eq:ce1}
\sum_{\vec{s} \,:\, s_i = j} u_i(\vec{s})x_{\vec{s}} \ge
\sum_{\vec{s} \,:\, s_i = j} u_i(s_i',\vec{s}_{-i})x_{\vec{s}} & \quad\quad\text{for
                                                        every $i \in
                                                        [k]$ and $j,s_i'
                                                       \in S_i$}\\
\label{eq:ce2}
\sum_{\vec{s} \in \vec{S}} x_{\vec{s}} = 1 &\\
\label{eq:ce3}
x_{\vec{s}} \ge 0 & \quad\quad\text{for every $\vec{s} \in \vec{S}$.}
\end{align}

The following proposition is immediate.
\begin{proposition}[\citet{GZ89}]\label{prop:lp}
An exact correlated or coarse correlated equilibrium of a game can be
computed in time polynomial in the number of outcomes of the game.
\end{proposition}
More generally, any linear function 
(such as the sum of players' expected payoffs) 
can be optimized over the set of correlated or coarse correlated
equilibria in time polynomial in the number of outcomes.

For games described in {\em normal form}, with each player~$i$'s
payoffs $\{ u_i(\vec{s}) \}_{\vec{s} \in \vec{S}}$ given explicitly in
the input, Proposition~\ref{prop:lp} provides an algorithm with
running time polynomial in the input size.  However, the number of
outcomes of a game scales exponentially with the number~$k$ of
players.\footnote{This fact should provide newfound appreciation for
  the distributed learning algorithms that compute an approximate
  coarse correlated equilibrium (in Proposition~\ref{prop:cce}) and an
  approximate correlated equilibrium (in~\cite{FV97,HM00}), where the
  total amount of computation is only {\em polynomial} in $k$ (and in
  $m$ and $\tfrac{1}{\eps}$).}  The computationally interesting
multi-player games, and the multi-player games that naturally arise in
computer science applications, are those with a {\em succinct
  description}.  Can we compute an exact correlated or coarse
correlated equilibrium in time polynomial in the size of a game's
description?

\subsection{Succinctly Represented Games}

For concreteness, let's look at one well-studied example of a class of
succinctly represented games: {\em graphical games}~\cite{KLS01,KM03}.
A graphical game is described by an undirected graph~$G=(V,E)$, with
players corresponding to vertices, and a local payoff matrix for each
vertex.  The local payoff matrix for vertex~$i$ specifies~$i$'s payoff
for each possible choice of its strategy and the strategies chosen by
its neighbors in~$G$.  By definition, the payoff of a player is
independent of the strategies chosen by non-neighboring players.  When
the graph~$G$ has maximum degree~$\Delta$, the size of the game
description is exponential in~$\Delta$ but polynomial in the
number~$k$ of players.  The most interesting cases are when
$\Delta=O(1)$ or perhaps $\Delta=O(\log k)$.  In these cases, the 
number of outcomes (and hence the size of the game's normal-form description)
is exponential in the size of the succinct description of the game,
and solving the linear system~\eqref{eq:cce1}--\eqref{eq:cce3}
or~\eqref{eq:ce1}--\eqref{eq:ce3} does not result in a
polynomial-time algorithm.

We next state a result showing that, quite generally, an exact
correlated (and hence coarse correlated) equilibrium of a succinctly
represented game can be computed in polynomial time.  
The key assumption is that the following {\sc Expected Utility} problem
can be solved in time polynomial in the size of the game's
description.\footnote{Some kind of assumption is necessary to preclude
  baking an $\NP$-complete problem into the game's description.}
\begin{mdframed}[style=offset,frametitle={The {\sc Expected Utility} Problem}]
Given a succinct description of a player's 
payoff function~$u_i$ and mixed
strategies~$\sigma_1,\ldots,\sigma_k$ for all of the players, compute
the player's expected utility:
\[
\expect[\vec{s} \sim \vec{\sigma}]{u_i(\vec{s})}.
\]
\end{mdframed}
For most of the succinctly represented multi-player games that come up
in computer science applications, the {\sc Expected Utility} problem
can be solved in polynomial time.  For example, in a graphical game it
can be solved by brute force---summing over the entries in player~$i$'s
local payoff matrix, weighted by the probabilities in the given mixed
strategies.  This algorithm takes time exponential in~$\Delta$ but
polynomial in the size of the game's
succinct representation.

Tractability of solving the {\sc Expected Utility} problem is a
sufficient condition for the tractability of computing an exact
correlated equilibrium.
\begin{theorem}[\citet{PR08,JL15}]\label{t:pr}
There is a polynomial-time Turing reduction from the problem of
computing a correlated equilibrium of a succinctly described game to
the {\sc Expected Utility} problem.
\end{theorem}
Theorem~\ref{t:pr} applies to a long list of succinctly described
games that have been studied in the computer science literature, with
graphical games serving as one example.\footnote{For the specific case
  of graphical games, \citet{KKLO03} were the first to develop a
  polynomial-time algorithm for computing an exact correlated
  equilibrium.}

The starting point of the proof of Theorem~\ref{t:pr} is the
exponential-size linear system~\eqref{eq:ce1}--\eqref{eq:ce3}.  We
know that this linear system is feasible (by Nash's Theorem, since the
system includes all Nash equilibria).  With exponentially many
variables, however, it's not clear how to efficiently compute a
feasible solution.  The dual linear system, meanwhile, has a
polynomial number of variables (corresponding to the constraints
in~\eqref{eq:ce1}) and an exponential number of inequalities
(corresponding to game outcomes).  By Farkas's Lemma---or,
equivalently, strong linear programming duality (see
e.g.~\cite{chvatal})---we know that this dual linear system is
infeasible. 

The key idea is to run the ellipsoid algorithm~\cite{K79} on the
infeasible dual linear system---called the ``ellipsoid against hope''
in~\cite{PR08}.  A polynomial-time separation oracle must produce,
given an alleged solution (which we know is infeasible), a violated
inequality.  It turns out that this separation oracle reduces to
solving a polynomial number of instances of the {\sc Expected Utility}
problem (which is polynomial-time solvable by assumption) and
computing the stationary distribution of a polynomial number of
polynomial-size Markov chains (also polynomial-time solvable, e.g.~by
linear programming).  The ellipsoid against hope terminates after a
polynomial number of invocations of its separation oracle, necessarily
with a proof that the dual linear system is infeasible.  To recover a
primal feasible solution (i.e., a correlated equilibrium), one can
retain only the primal decision variables corresponding to the
(polynomial number of) dual constraints generated by the separation
oracle, and solve directly this polynomial-size reduced version of the
primal linear system.\footnote{As a bonus, this means that the
  algorithm will output a ``sparse'' correlated equilibrium, with
  support size polynomial in the size of the game description.}

\section{The Price of Anarchy of Coarse Correlated Equilibria}

\subsection{Balancing Computational Tractability with Predictive Power}

We now understand senses in which Nash equilibria are computationally
intractable (Solar Lectures~2--5) while correlated equilibria are
computationally tractable (Sections~\ref{s:ce} and~\ref{s:exact}).
From an economic perspective, these results suggest that it could be
prudent
to study the correlated equilibria of a game, rather 
than restricting attention only to its Nash equilibria.\footnote{This is not a totally
  unfamiliar idea to economists.
According to \citet{SV02}, Roger
  Myerson, winner of the 2007 Nobel Prize in Economics, asserted that 
``if there is intelligent life on other planets, in a majority of them,
they would have discovered correlated equilibrium before Nash
equilibrium.''}

Passing from Nash equilibria to the larger set of correlated
equilibria is a two-edged sword.  Computational tractability
increases, and with it the plausibility that actual game play will
conform to the equilibrium notion.  But whatever criticisms we had
about the Nash equilibrium's predictive power (recall
Section~\ref{ss:whocares} in Solar Lecture~1),
they are even more severe
for the correlated equilibrium (since there are only more of them).
The worry is that games typically have far too many correlated
equilibria to say anything interesting about them.  Our final order of
business is to dispel this worry, at least in the context of
price-of-anarchy analyses.

Recall from Lunar Lecture~2 that the {\em price of 
  anarchy (POA)} is defined as the ratio between the
objective function value of an optimal solution, and that of the worst
equilibrium:
\[ 
\mathsf{PoA}(G):= \frac{f(OPT(G))}{\min_{\text{$\rho$ is an
      equilibrium of $G$}} f(\rho)},
\] 
where~$G$ denotes a game, $f$ denotes a maximization objective
function (with $f(\rho) = \expect[\vec{s} \sim \rho]{f(\vec{s})}$ when
$\rho$ is a probability distribution), and $OPT(G)$ is the optimal
outcome of~$G$ with respect to~$f$.  Thus the POA of a game is always
at least~1, and the closer to~1, the better.

The POA of a game depends on the choice of equilibrium concept.
Because it is defined with respect to the worst equilibrium, the POA
only degrades as the set of equilibria grows larger.  Thus, the POA
with respect to coarse correlated equilibria is only worse (i.e.,
larger) than that with respect to correlated equilibria, which in
turn is only worse than the POA with respect to Nash equilibria
(recall Figure~\ref{f:venn}).

The hope is that there's a ``sweet spot'' equilibrium
concept---permissive enough to be computationally tractable, yet
stringent enough to allow good worse-case approximation guarantees.
Happily, the coarse correlated equilibrium is just such a sweet spot!

\subsection{Smooth Games and Extension Theorems}

After the first ten years of price-of-anarchy analyses (roughly
1999-2008), it was clear to researchers in the area that many such
analyses across different application domains share a common
architecture (in routing games, facility location games, scheduling
games, auctions, etc.).  The concept of ``proofs of POA bounds that
follow the standard template'' was made precise in the theory of smooth
games~\cite{robust}.\footnote{The formal definition is a bit technical, and we
won't need it here.  Roughly, it requires that the best-response
condition is invoked in an equilibrium-independent way and that a certain
restricted type of charging argument is used.}\footnote{There are
several important precursors to this theory, including~\citet{B+08},
\citet{CK05b}, and Vetta~\cite{V02}.  See~\cite{robust} for a detailed
history.}
One can then define the {\em robust
  price of anarchy} of a game as the best (i.e., smallest) bound on the
game's POA
that can be proved by following the standard template.

The proof template formalized by smooth games superficially appears 
relevant only for the POA with respect to {\em pure} Nash
equilibria, as the definition involves no randomness (let alone
correlation).  The good news is that the template's simplicity makes
it relatively easy to use.  One would expect the bad news to be that
bounds on the POA of more permissive equilibrium concepts require
different proof techniques, and that the corresponding POA bounds
would be much worse.  Happily, this is not the case---every POA bound proved
using the canonical template automatically applies not only to the
pure Nash equilibria of a game, but more generally to all of the
game's coarse correlated equilibria (and hence all of its correlated
and mixed Nash equilibria).\footnote{Smooth games and the ``extension
  theorem'' in Theorem~\ref{t:robust} are the starting point for the
  modular and user-friendly toolbox for proving POA bounds in complex
  settings mentioned in Section~\ref{ss:when} of Lunar Lecture~1.
  Generalizations of this theory to incomplete-information games (like
  auctions) and to the composition of smooth games (like
  simultaneous single-item auctions) lead to good
  POA bounds for simple auctions~\cite{ST13}.
(These generalizations also brought
  together two historically separate subfields of algorithmic game
  theory, namely algorithmic mechanism design and price-of-anarchy analyses.)
See   \cite{RST17} for a user's guide to    this toolbox.}
\begin{theorem}[\citet{robust}]\label{t:robust}
In every game, the POA with respect to coarse correlated equilibria is
bounded above by its robust POA.
\end{theorem}
For $\eps$-approximate coarse correlated equilibria---as guaranteed by
a logarithmic number of rounds of smooth fictitious play
(Proposition~\ref{prop:cce})---the POA bound in Theorem~\ref{t:robust}
degrades by an additive $O(\eps)$ term.








\newpage

\begin{thebibliography}{158}
\providecommand{\natexlab}[1]{#1}
\providecommand{\url}[1]{\texttt{#1}}
\expandafter\ifx\csname urlstyle\endcsname\relax
  \providecommand{\doi}[1]{doi: #1}\else
  \providecommand{\doi}{doi: \begingroup \urlstyle{rm}\Url}\fi

\bibitem[Aaronson et~al.(2014)Aaronson, Impagliazzo, and
  Moshkovitz]{AaronsonImMo14}
S.~Aaronson, R.~Impagliazzo, and D.~Moshkovitz.
\newblock A{M} with multiple {M}erlins.
\newblock In \emph{Proceedings of the 29th IEEE Conference on Computational
  Complexity (CCC)}, pages 44--55, 2014.

\bibitem[Adler(2013)]{A13}
I.~Adler.
\newblock The equivalence of linear programs and zero-sum games.
\newblock \emph{International Journal of Game Theory}, 42\penalty0
  (1):\penalty0 165--177, 2013.

\bibitem[Alaei et~al.(2019)Alaei, Fu, Haghpanah, Hartline, and Malekian]{A+19}
S.~Alaei, H.~Fu, N.~Haghpanah, J.~D. Hartline, and A.~Malekian.
\newblock Efficient computation of optimal auctions via reduced forms.
\newblock \emph{Mathematics of Operations Research}, 44\penalty0 (3):\penalty0
  1058--1086, 2019.

\bibitem[Alth{\"o}fer(1994)]{A94}
I.~Alth{\"o}fer.
\newblock On sparse approximations to randomized strategies and convex
  combinations.
\newblock \emph{Linear Algebra and Its Applications}, 199\penalty0
  (1):\penalty0 339--355, 1994.

\bibitem[Anshu et~al.(2017)Anshu, Goud, Jain, Kundu, and Mukhopadhyay]{A+17}
A.~Anshu, N.~Goud, R.~Jain, S.~Kundu, and P.~Mukhopadhyay.
\newblock Lifting randomized query complexity to randomized communication
  complexity.
\newblock Technical Report TR17-054, ECCC, 2017.

\bibitem[Arrow and Debreu(1954)]{AD54}
K.~J. Arrow and G.~Debreu.
\newblock Existence of an equilibrium for a competitive economy.
\newblock \emph{Econometrica}, 22:\penalty0 265--290, 1954.

\bibitem[Aumann(1974)]{A74}
R.~J. Aumann.
\newblock Subjectivity and correlation in randomized strategies.
\newblock \emph{Journal of Mathematical Economics}, 1\penalty0 (1):\penalty0
  67--96, 1974.

\bibitem[Babichenko(2016)]{B16}
Y.~Babichenko.
\newblock Query complexity of approximate {N}ash equilibria.
\newblock \emph{Journal of the ACM}, 63\penalty0 (4):\penalty0 36, 2016.

\bibitem[Babichenko and Rubinstein(2017)]{BR17}
Y.~Babichenko and A.~Rubinstein.
\newblock Communication complexity of approximate {Nash} equilibria.
\newblock In \emph{Proceedings of the 49th Annual {ACM} Symposium on Theory of
  Computing (STOC)}, pages 878--889, 2017.

\bibitem[Beame et~al.(1998)Beame, Cook, Edmonds, Impagliazzo, and
  Pitassi]{B+98}
P.~Beame, S.~Cook, J.~Edmonds, R.~Impagliazzo, and T.~Pitassi.
\newblock The relative complexity of {NP} search problems.
\newblock \emph{Journal of Computer and System Sciences}, 57\penalty0
  (1):\penalty0 3--19, 1998.

\bibitem[Ben-Zwi et~al.(2013)Ben-Zwi, Lavi, and Newman]{BLN13}
O.~Ben-Zwi, R.~Lavi, and I.~Newman.
\newblock Ascending auctions and {W}alrasian equilibrium.
\newblock Working paper, 2013.

\bibitem[Bikhchandani and Mamer(1997)]{BM97}
S.~Bikhchandani and J.~W. Mamer.
\newblock Competitive equilibrium in an exchange economy with indivisibilities.
\newblock \emph{Journal of Economic Theory}, 74:\penalty0 385--413, 1997.

\bibitem[Bitansky et~al.(2015)Bitansky, Paneth, and Rosen]{BPR15}
N.~Bitansky, O.~Paneth, and A.~Rosen.
\newblock On the cryptographic hardness of finding a {N}ash equilibrium.
\newblock In \emph{Proceedings of the 56th Annual Symposium on Foundations of
  Computer Science {(FOCS)}}, pages 1480--1498, 2015.

\bibitem[Blum et~al.(2008)Blum, Hajiaghayi, Ligett, and Roth]{B+08}
A.~Blum, M.~T. Hajiaghayi, K.~Ligett, and A.~Roth.
\newblock Regret minimization and the price of total anarchy.
\newblock In \emph{Proceedings of the 40th Annual {ACM} Symposium on Theory of
  Computing {(STOC)}}, pages 373--382, 2008.

\bibitem[Border(1985)]{B85}
K.~C. Border.
\newblock \emph{Fixed point theorems with applications to economics and game
  theory}.
\newblock Cambridge University Press, 1985.

\bibitem[Border(1991)]{B91}
K.~C. Border.
\newblock Implementation of reduced form auctions: A geometric approach.
\newblock \emph{Econometrica}, 59\penalty0 (4):\penalty0 1175--1187, 1991.

\bibitem[Border(2007)]{B07}
K.~C. Border.
\newblock Reduced form auctions revisited.
\newblock \emph{Economic Theory}, 31:\penalty0 167--181, 2007.

\bibitem[Braverman et~al.(2015)Braverman, Kun~Ko, and
  Weinstein]{BravermanYoWe15}
M.~Braverman, Y.~Kun~Ko, and O.~Weinstein.
\newblock Approximating the best {N}ash equilibrium in {$n^{o(\log n)}$}-time
  breaks the {E}xponential {T}ime {H}ypothesis.
\newblock In \emph{Proceedings of the 26th {A}nnual {ACM}-{SIAM} {S}ymposium on
  {D}iscrete {A}lgorithms (SODA)}, pages 970--982, 2015.

\bibitem[Braverman et~al.(2017)Braverman, Kun~Ko, Rubinstein, and
  Weinstein]{BKRW17}
M.~Braverman, Y.~Kun~Ko, A.~Rubinstein, and O.~Weinstein.
\newblock {ETH} hardness for densest-{$k$}-subgraph with perfect completeness.
\newblock In \emph{Proceedings of the 28th Annual ACM-SIAM Symposium on
  Discrete Algorithms (SODA)}, pages 1326--1341, 2017.

\bibitem[Brown(1951)]{B51}
G.~W. Brown.
\newblock Iterative solutions of games by fictitious play.
\newblock In T.~C. Koopmans, editor, \emph{Activity Analysis of Production and
  Allocation}, Cowles Commission Monograph No. 13, chapter XXIV, pages
  374--376. Wiley, 1951.

\bibitem[Cai et~al.(2012{\natexlab{a}})Cai, Daskalakis, and Weinberg]{CDW12}
Y.~Cai, C.~Daskalakis, and S.~M. Weinberg.
\newblock An algorithmic characterization of multi-dimensional mechanisms.
\newblock In \emph{Proceedings of the 44th Symposium on Theory of Computing
  (STOC)}, pages 459--478, 2012{\natexlab{a}}.

\bibitem[Cai et~al.(2012{\natexlab{b}})Cai, Daskalakis, and Weinberg]{CDW12b}
Y.~Cai, C.~Daskalakis, and S.~M. Weinberg.
\newblock Optimal multi-dimensional mechanism design: Reducing revenue to
  welfare maximization.
\newblock In \emph{Proceedings of the 53rd Annual Symposium on Foundations of
  Computer Science (FOCS)}, pages 130--139, 2012{\natexlab{b}}.

\bibitem[Cai et~al.(2016)Cai, Candogan, Daskalakis, and Papadimitriou]{CCDP}
Y.~Cai, O.~Candogan, C.~Daskalakis, and C.~H. Papadimitriou.
\newblock Zero-sum polymatrix games: A generalization of minmax.
\newblock \emph{Mathematics of Operations Research}, 41\penalty0 (2):\penalty0
  648--655, 2016.

\bibitem[Candogan et~al.(2015)Candogan, Ozdaglar, and Parrilo]{COP15}
O.~Candogan, A.~Ozdaglar, and P.~Parrilo.
\newblock Iterative auction design for tree valuations.
\newblock \emph{Operations Research}, 63\penalty0 (4):\penalty0 751--771, 2015.

\bibitem[Candogan et~al.(2018)Candogan, Ozdaglar, and Parrilo]{COP17}
O.~Candogan, A.~Ozdaglar, and P.~Parrilo.
\newblock Pricing equilibria and graphical valuations.
\newblock \emph{ACM Transactions on Economics and Computation}, 6\penalty0
  (1):\penalty0 2, 2018.

\bibitem[Cesa-Bianchi and Lugosi(2006)]{CBL06}
N.~Cesa-Bianchi and G.~Lugosi.
\newblock \emph{Prediction, Learning, and Games}.
\newblock Cambridge University Press, 2006.

\bibitem[Cesa-Bianchi et~al.(2007)Cesa-Bianchi, Mansour, and Stolz]{CBMS07}
N.~Cesa-Bianchi, Y.~Mansour, and G.~Stolz.
\newblock Improved second-order bounds for prediction with expert advice.
\newblock \emph{Machine Learning}, 66\penalty0 (2--3):\penalty0 321--352, 2007.

\bibitem[Che et~al.(2013)Che, Kim, and Mierendorff]{CKM13}
Y.-K. Che, J.~Kim, and K.~Mierendorff.
\newblock Generalized reduced form auctions: A network flow approach.
\newblock \emph{Econometrica}, 81:\penalty0 2487--2520, 2013.

\bibitem[Chen and Deng(2005)]{CD05}
X.~Chen and X.~Deng.
\newblock 3-{N}ash is {PPAD}-complete.
\newblock Technical Report TR05-134, ECCC, 2005.

\bibitem[Chen and Deng(2006)]{CD06}
X.~Chen and X.~Deng.
\newblock Settling the complexity of two-player {N}ash equilibrium.
\newblock In \emph{Proceedings of the 47th Annual Symposium on Foundations of
  Computer Science {(FOCS)}}, pages 261--270, 2006.

\bibitem[Chen and Deng(2009)]{CD09}
X.~Chen and X.~Deng.
\newblock On the complexity of {2D} discrete fixed point problem.
\newblock \emph{Theoretical Computer Science}, 410\penalty0 (44):\penalty0
  4448--4456, 2009.

\bibitem[Chen et~al.(2006{\natexlab{a}})Chen, Deng, and Teng]{CDT06}
X.~Chen, X.~Deng, and S.-H. Teng.
\newblock Computing {N}ash equilibria: Approximation and smoothed complexity.
\newblock In \emph{Proceedings of the 47th Annual Symposium on Foundations of
  Computer Science {(FOCS)}}, pages 603--612, 2006{\natexlab{a}}.

\bibitem[Chen et~al.(2006{\natexlab{b}})Chen, Deng, and Teng]{CDT06b}
X.~Chen, X.~Deng, and S.-H. Teng.
\newblock Sparse games are hard.
\newblock In \emph{Proceedings of the Second Annual International Workshop on
  Internet and Network Economics {(WINE)}}, pages 262--273, 2006{\natexlab{b}}.

\bibitem[Chen et~al.(2009)Chen, Deng, and Teng]{CDT09}
X.~Chen, X.~Deng, and S.-H. Teng.
\newblock Settling the complexity of computing two-player {Nash} equilibria.
\newblock \emph{Journal of the ACM}, 56\penalty0 (3):\penalty0 14, 2009.
\newblock Journal version of {\cite{CD05}}, {\cite{CD06}}, {\cite{CDT06}}, and
  {\cite{CDT06b}}.

\bibitem[Choudhuri et~al.(2019)Choudhuri, {Hub\'{a}\v{c}ek}, Kamath, Pietrzak,
  Rosen, and Rothblum]{C+19}
A.~R. Choudhuri, P.~{Hub\'{a}\v{c}ek}, C.~Kamath, K.~Pietrzak, A.~Rosen, and
  G.~N. Rothblum.
\newblock Finding a {N}ash equilibrium is no easier than breaking
  {F}iat-{S}hamir.
\newblock In \emph{Proceedings of the 51st Annual {ACM} Symposium on Theory of
  Computing (STOC)}, pages 1103--1114, 2019.

\bibitem[Christodoulou and Koutsoupias(2005)]{CK05b}
G.~Christodoulou and E.~Koutsoupias.
\newblock On the price of anarchy and stability of correlated equilibria of
  linear congestion games.
\newblock In \emph{Proceedings of the 13th Annual European Symposium on
  Algorithms {(ESA)}}, pages 59--70, 2005.

\bibitem[Christodoulou et~al.(2016{\natexlab{a}})Christodoulou, Kov\'{a}cs, and
  Schapira]{CKS16}
G.~Christodoulou, A.~Kov\'{a}cs, and M.~Schapira.
\newblock Bayesian combinatorial auctions.
\newblock \emph{Journal of the ACM}, 63\penalty0 (2):\penalty0 11,
  2016{\natexlab{a}}.

\bibitem[Christodoulou et~al.(2016{\natexlab{b}})Christodoulou, Kov{\'{a}}cs,
  Sgouritsa, and Tang]{CKST16}
G.~Christodoulou, A.~Kov{\'{a}}cs, A.~Sgouritsa, and B.~Tang.
\newblock Tight bounds for the price of anarchy of simultaneous first price
  auctions.
\newblock \emph{ACM Transactions on Economics and Computation}, 4\penalty0
  (2):\penalty0 9, 2016{\natexlab{b}}.

\bibitem[Chv{\'a}tal(1983)]{chvatal}
V.~Chv{\'a}tal.
\newblock \emph{Linear Programming}.
\newblock Freeman, 1983.

\bibitem[Conitzer and Sandholm(2004)]{CS04}
V.~Conitzer and T.~Sandholm.
\newblock Communication complexity as a lower bound for learning in games.
\newblock In \emph{Proceedings of the Twenty-first International Conference on
  Machine Learning (ICML)}, 2004.

\bibitem[Dantzig(1951)]{D51}
G.~B. Dantzig.
\newblock A proof of the equivalence of the programming problem and the game
  problem.
\newblock In T.~C. Koopmans, editor, \emph{Activity Analysis of Production and
  Allocation}, Cowles Commission Monograph No. 13, chapter~XX, pages 330--335.
  Wiley, 1951.

\bibitem[Dantzig(1981)]{D81}
G.~B. Dantzig.
\newblock Reminiscences about the origins of linear programming.
\newblock Technical Report SOL 81-5, Systems Optimization Laboratory,
  Department of Operations Research, Stanford University, 1981.

\bibitem[Daskalakis and Pan(2014)]{DP14}
C.~Daskalakis and Q.~Pan.
\newblock A counter-example to {Karlin's} strong conjecture for fictitious
  play.
\newblock In \emph{Proceedings of the 55th Annual Symposium on Foundations of
  Computer Science (FOCS)}, pages 11--20, 2014.

\bibitem[Daskalakis and Papadimitriou(2005)]{DP05}
C.~Daskalakis and C.~H. Papadimitriou.
\newblock Three-player games are hard.
\newblock Technical Report TR05-139, ECCC, 2005.

\bibitem[Daskalakis et~al.(2006)Daskalakis, Goldberg, and Papadimitriou]{DGP06}
C.~Daskalakis, P.~W. Goldberg, and C.~H. Papadimitriou.
\newblock The complexity of computing a {N}ash equilibrium.
\newblock In \emph{Proceedings of the 38th Annual {ACM} Symposium on Theory of
  Computing {(STOC)}}, pages 71--78, 2006.

\bibitem[Daskalakis et~al.(2009{\natexlab{a}})Daskalakis, Goldberg, and
  Papadimitriou]{DGP09}
C.~Daskalakis, P.~W. Goldberg, and C.~H. Papadimitriou.
\newblock The complexity of computing a {Nash} equilibrium.
\newblock \emph{SIAM Journal on Computing}, 39\penalty0 (1):\penalty0 195--259,
  2009{\natexlab{a}}.
\newblock Journal version of {\cite{DP05}}, {\cite{DGP06}}, and {\cite{GP06}}.

\bibitem[Daskalakis et~al.(2009{\natexlab{b}})Daskalakis, Goldberg, and
  Papadimitriou]{DGPcacm}
C.~Daskalakis, P.~W. Goldberg, and C.~H. Papadimitriou.
\newblock The complexity of computing a {N}ash equilibrium.
\newblock \emph{Communications of the ACM}, 52\penalty0 (2):\penalty0 89--97,
  2009{\natexlab{b}}.

\bibitem[Dobzinski and Vondr{\'a}k(2013)]{DV13}
S.~Dobzinski and J.~Vondr{\'a}k.
\newblock Communication complexity of combinatorial auctions with submodular
  valuations.
\newblock In \emph{Proceedings of the 24th Annual ACM-SIAM Symposium on
  Discrete Algorithms {(SODA)}}, pages 1205--1215, 2013.

\bibitem[Dobzinski et~al.(2010)Dobzinski, Nisan, and Schapira]{DNS05}
S.~Dobzinski, N.~Nisan, and M.~Schapira.
\newblock Approximation algorithms for combinatorial auctions with
  complement-free bidders.
\newblock \emph{Mathematics of Operations Research}, 35\penalty0 (1):\penalty0
  1--13, 2010.

\bibitem[D\"{u}tting et~al.(2017)D\"{u}tting, Gkatzelis, and
  Roughgarden]{DGR14}
P.~D\"{u}tting, V.~Gkatzelis, and T.~Roughgarden.
\newblock The performance of deferred-acceptance auctions.
\newblock \emph{Mathematics of Operations Research}, 42\penalty0 (4):\penalty0
  897--914, 2017.

\bibitem[Etessami and Yannakakis(2010)]{EY07}
K.~Etessami and M.~Yannakakis.
\newblock On the complexity of {N}ash equilibria and other fixed points.
\newblock \emph{SIAM Journal on Computing}, 39\penalty0 (6):\penalty0
  2531--2597, 2010.

\bibitem[Feige(2009)]{F06}
U.~Feige.
\newblock On maximizing welfare where the utility functions are subadditive.
\newblock \emph{SIAM Journal on Computing}, 39\penalty0 (1):\penalty0 122--142,
  2009.

\bibitem[Feige and Vondr{\'a}k(2010)]{FV06}
U.~Feige and J.~Vondr{\'a}k.
\newblock The submodular welfare problem with demand queries.
\newblock \emph{Theory of Computing}, 6\penalty0 (1):\penalty0 247--290, 2010.

\bibitem[Feldman et~al.(2013)Feldman, Fu, Gravin, and Lucier]{FFGL13}
M.~Feldman, H.~Fu, N.~Gravin, and B.~Lucier.
\newblock Simultaneous auctions are (almost) efficient.
\newblock In \emph{Proceedings of the Forty-fifth Annual ACM Symposium on
  Theory of Computing}, pages 201--210, 2013.

\bibitem[Foster and Vohra(1997)]{FV97}
D.~P. Foster and R.~Vohra.
\newblock Calibrated learning and correlated equilibrium.
\newblock \emph{Games and Economic Behavior}, 21\penalty0 (1--2):\penalty0
  40--55, 1997.

\bibitem[Fr{\'e}chette et~al.(2017)Fr{\'e}chette, Newman, and
  Leyton-Brown]{FNL17}
A.~Fr{\'e}chette, N.~Newman, and K.~Leyton-Brown.
\newblock Solving the station repacking problem.
\newblock In \emph{Handbook of Spectrum Auction Design}, chapter~38, pages
  813--827. Cambridge University Press, 2017.

\bibitem[Freund and Schapire(1997)]{FS97}
Y.~Freund and R.~E. Schapire.
\newblock A decision-theoretic generalization of on-line learning and an
  application to boosting.
\newblock \emph{Journal of Computer and System Sciences}, 55\penalty0
  (1):\penalty0 119--139, 1997.

\bibitem[Freund and Schapire(1999)]{FS99}
Y.~Freund and R.~E. Schapire.
\newblock Adaptive game playing using multiplicative weights.
\newblock \emph{Games and Economic Behavior}, 29\penalty0 (1--2):\penalty0
  79--103, 1999.

\bibitem[Fudenberg and Levine(1995)]{FL95}
D.~Fudenberg and D.~K. Levine.
\newblock Consistency and cautious fictitious play.
\newblock \emph{Journal of Economic Dynamics and Control}, 19\penalty0
  (5):\penalty0 1065--1089, 1995.

\bibitem[Gale et~al.(1951)Gale, Kuhn, and Tucker]{GKT51}
D.~Gale, H.~W. Kuhn, and A.~W. Tucker.
\newblock Linear programming and the theory of games.
\newblock In T.~C. Koopmans, editor, \emph{Activity Analysis of Production and
  Allocation}, Cowles Commission Monograph No. 13, chapter XIX, pages 317--329.
  Wiley, 1951.

\bibitem[Ganor et~al.(2019)Ganor, Karthik, and P{\'a}lv{\"o}lgyi]{GKP19}
A.~Ganor, C.~S. Karthik, and D.~P{\'a}lv{\"o}lgyi.
\newblock On communication complexity of fixed point computation.
\newblock arXiv:1909.10958, 2019.

\bibitem[Garg et~al.(2016)Garg, Pandey, and Srinivasan]{GPS16}
S.~Garg, O.~Pandey, and A.~Srinivasan.
\newblock Revisiting the cryptographic hardness of finding a {N}ash
  equilibrium.
\newblock In \emph{Proceedings of the 36th Annual International Cryptology
  Conference on Advances in Cryptology (CRYPTO)}, pages 579--604, 2016.

\bibitem[Geanakoplos(2003)]{G03}
J.~Geanakoplos.
\newblock {N}ash and {W}alras equilibrium via {B}rouwer.
\newblock \emph{Economic Theory}, 21\penalty0 (2/3):\penalty0 585--603, 2003.

\bibitem[Gilboa and Zemel(1989)]{GZ89}
I.~Gilboa and E.~Zemel.
\newblock Nash and correlated equilibria: Some complexity considerations.
\newblock \emph{Games and Economic Behavior}, 1\penalty0 (1):\penalty0 80--93,
  1989.

\bibitem[Gkatzelis et~al.(2017)Gkatzelis, Markakis, and Roughgarden]{GMR17}
V.~Gkatzelis, E.~Markakis, and T.~Roughgarden.
\newblock Deferred-acceptance auctions for multiple levels of service.
\newblock In \emph{Proceedings of the 18th Annual ACM Conference on Economics
  and Computation (EC)}, pages 21--38, 2017.

\bibitem[Goldberg and Papadimitriou(2006)]{GP06}
P.~W. Goldberg and C.~H. Papadimitriou.
\newblock Reducibility among equilibrium problems.
\newblock In \emph{Proceedings of the 38th Annual {ACM} Symposium on Theory of
  Computing {(STOC)}}, pages 61--70, 2006.

\bibitem[G{\"o\"o}s(2015)]{G15}
M.~G{\"o\"o}s.
\newblock Lower bounds for clique {vs.~independent} set.
\newblock In \emph{Proceedings of the 56th Annual Symposium on Foundations of
  Computer Science (FOCS)}, pages 1066--1076, 2015.

\bibitem[{G\"o\"os} and Pitassi(2018)]{GP18}
M.~{G\"o\"os} and T.~Pitassi.
\newblock Communication lower bounds via critical block sensitivity.
\newblock \emph{SIAM Journal on Computing}, 47\penalty0 (5):\penalty0
  1778--1806, 2018.

\bibitem[{G\"o\"os} and Rubinstein(2018)]{GR18}
M.~{G\"o\"os} and A.~Rubinstein.
\newblock Near-optimal communication lower bounds for approximate {N}ash
  equilibria.
\newblock In \emph{Proceedings of the 59th Annual Symposium on Foundations of
  Computer Science {(FOCS)}}, pages 397--403, 2018.

\bibitem[{G\"o\"os} et~al.(2016){G\"o\"os}, Lovett, Meka, Watson, and
  Zuckerman]{G+16}
M.~{G\"o\"os}, S.~Lovett, R.~Meka, T.~Watson, and D.~Zuckerman.
\newblock Rectangles are nonnegative juntas.
\newblock \emph{SIAM Journal on Computing}, 45\penalty0 (5):\penalty0
  1835--1869, 2016.

\bibitem[{G\"o\"os} et~al.(2017){G\"o\"os}, Pitassi, and Watson]{GPW17}
M.~{G\"o\"os}, T.~Pitassi, and T.~Watson.
\newblock Query-to-communication lifting for {BPP}.
\newblock In \emph{Proceedings of the 58th Annual IEEE Symposium on Foundations
  of Computer Science}, pages 132--143, 2017.

\bibitem[G{\"{o}}{\"{o}}s et~al.(2018)G{\"{o}}{\"{o}}s, Pitassi, and
  Watson]{GPW18}
M.~G{\"{o}}{\"{o}}s, T.~Pitassi, and T.~Watson.
\newblock Deterministic communication vs.\ partition number.
\newblock \emph{SIAM Journal on Computing}, 47\penalty0 (6):\penalty0
  2435--2450, 2018.

\bibitem[Gopalan et~al.(2018)Gopalan, Nisan, and Roughgarden]{GNR18}
P.~Gopalan, N.~Nisan, and T.~Roughgarden.
\newblock Public projects, {Boolean} functions, and the borders of {Border's}
  theorem.
\newblock \emph{ACM Transactions on Economic and Computation}, 6\penalty0
  (3-4):\penalty0 18, 2018.

\bibitem[Gul and Stacchetti(1999)]{GS99}
F.~Gul and E.~Stacchetti.
\newblock Walrasian equilibrium with gross substitutes.
\newblock \emph{Journal of Economic Theory}, 87:\penalty0 95--124, 1999.

\bibitem[Hannan(1957)]{Han57}
J.~Hannan.
\newblock Approximation to {B}ayes risk in repeated play.
\newblock In M.~Dresher, A.~W. Tucker, and P.~Wolfe, editors,
  \emph{Contributions to the Theory of Games}, volume~3, pages 97--139.
  Princeton University Press, 1957.

\bibitem[Hart and Mas-Colell(2000)]{HM00}
S.~Hart and A.~Mas-Colell.
\newblock A simple adaptive procedure leading to correlated equilibrium.
\newblock \emph{Econometrica}, 68\penalty0 (5):\penalty0 1127--1150, 2000.

\bibitem[Hartline(2017)]{hartline}
J.~D. Hartline.
\newblock Mechanism design and approximation.
\newblock Book draft, July 2017.

\bibitem[Hassidim et~al.(2011)Hassidim, Kaplan, Mansour, and Nisan]{HKMN11}
A.~Hassidim, H.~Kaplan, Y.~Mansour, and N.~Nisan.
\newblock Non-price equilibria in markets of discrete goods.
\newblock In \emph{Proceedings of the 12th Annual ACM Conference on Economics
  and Computation (EC)}, pages 295--296, 2011.

\bibitem[Hirsch et~al.(1989)Hirsch, Papadimitriou, and Vavasis]{HPV89}
M.~D. Hirsch, C.~H. Papadimitriou, and S.~A. Vavasis.
\newblock Exponential lower bounds for finding {B}rouwer fix points.
\newblock \emph{Journal of Complexity}, 5\penalty0 (4):\penalty0 379--416,
  1989.

\bibitem[{Hub\'{a}\v{c}ek} and Yogev(2017)]{HY17}
P.~{Hub\'{a}\v{c}ek} and E.~Yogev.
\newblock Hardness of continuous local search: Query complexity and
  cryptographic lower bounds.
\newblock In \emph{Proceedings of the 28th Annual ACM-SIAM Symposium on
  Discrete Algorithms (SODA)}, pages 1352--1371, 2017.

\bibitem[{Hub\'{a}\v{c}ek} et~al.(2017){Hub\'{a}\v{c}ek}, Naor, and
  Yogev]{HNY17}
P.~{Hub\'{a}\v{c}ek}, M.~Naor, and E.~Yogev.
\newblock The journey from {NP} to {TFNP} hardness.
\newblock In \emph{Proceedings of the 8th Conference on Innovations in
  Theoretical Computer Science (ITCS)}, 2017.
\newblock Article 60.

\bibitem[Impagliazzo and Wigderson(1997)]{IW97}
R.~Impagliazzo and A.~Wigderson.
\newblock {P = BPP if E requires exponential circuits: derandomizing the XOR
  lemma}.
\newblock In \emph{Proceedings of the 29th Annual {ACM} Symposium on Theory of
  Computing {(STOC)}}, pages 220--229, 1997.

\bibitem[Impagliazzo et~al.(2001)Impagliazzo, Paturi, and Zane]{IPZ01}
R.~Impagliazzo, R.~Paturi, and F.~Zane.
\newblock Which problems have strongly exponential complexity?
\newblock \emph{Journal of Computer and System Sciences}, 63\penalty0
  (4):\penalty0 512--530, 2001.

\bibitem[Jiang and Leyton-Brown(2015)]{JL15}
A.~X. Jiang and K.~Leyton-Brown.
\newblock Polynomial-time computation of exact correlated equilibrium in
  compact games.
\newblock \emph{Games and Economic Behavior}, 91:\penalty0 347--359, 2015.

\bibitem[Johnson(2007)]{J07}
D.~S. Johnson.
\newblock The {NP}-completeness column: Finding needles in haystacks.
\newblock \emph{ACM Transactions on Algorithms}, 3\penalty0 (2):\penalty0 24,
  2007.

\bibitem[Johnson et~al.(1988)Johnson, Papadimitriou, and Yannakakis]{JPY88}
D.~S. Johnson, C.~H. Papadimitriou, and M.~Yannakakis.
\newblock How easy is local search?
\newblock \emph{Journal of Computer and System Sciences}, 37\penalty0
  (1):\penalty0 79--100, 1988.

\bibitem[Kakade et~al.(2003)Kakade, Kearns, Langford, and Ortiz]{KKLO03}
S.~Kakade, M.~Kearns, J.~Langford, and L.~Ortiz.
\newblock Correlated equilibria in graphical games.
\newblock In \emph{Proceedings of the 4th {ACM} Conference on Electronic
  Commerce}, pages 42--47, 2003.

\bibitem[Kalyanasundaram and Schnitger(1992)]{KS92}
B.~Kalyanasundaram and G.~Schnitger.
\newblock The probabilistic communication complexity of set intersection.
\newblock \emph{SIAM Journal on Discrete Mathematics}, 5\penalty0 (4):\penalty0
  545--557, 1992.

\bibitem[Karlin(1959)]{K59}
S.~Karlin.
\newblock \emph{Mathematical Methods and Theory in Games, Programming, and
  Economics}.
\newblock Addison-Wesley, 1959.

\bibitem[Kearns(2007)]{graphicalgames}
M.~Kearns.
\newblock Graphical games.
\newblock In N.~Nisan, T.~Roughgarden, {\'E}.~Tardos, and V.~Vazirani, editors,
  \emph{Algorithmic Game Theory}, chapter~7, pages 159--180. Cambridge
  University Press, 2007.

\bibitem[Kearns et~al.(2001)Kearns, Littman, and Singh]{KLS01}
M.~Kearns, M.~L. Littman, and S.~Singh.
\newblock Graphical models for game theory.
\newblock In \emph{Proceedings of the Conference on Uncertainty in Artificial
  Intelligence (UAI)}, pages 253--260, 2001.

\bibitem[Kelso and Crawford(1982)]{KC82}
A.~S. Kelso and V.~P. Crawford.
\newblock Job matching, coalition formation, and gross substitutes.
\newblock \emph{Econometrica}, 50\penalty0 (6):\penalty0 1483--1504, 1982.

\bibitem[Khachiyan(1979)]{K79}
L.~G. Khachiyan.
\newblock A polynomial algorithm in linear programming.
\newblock \emph{Soviet Mathematics Doklady}, 20\penalty0 (1):\penalty0
  191--194, 1979.

\bibitem[Kjeldsen(2001)]{K01}
T.~H. Kjeldsen.
\newblock John von {N}eumann's conception of the {M}inimax theorem: A journey
  through different mathematical contexts.
\newblock \emph{Archive for History of Exact Sciences}, 56:\penalty0 39--68,
  2001.

\bibitem[Koller and Milch(2003)]{KM03}
D.~Koller and B.~Milch.
\newblock Multi-agent influence diagrams for representing and solving games.
\newblock \emph{Games and Economic Behavior}, 45\penalty0 (1):\penalty0
  181--221, 2003.

\bibitem[Kopparty et~al.(2017)Kopparty, Meir, Ron-Zewi, and Saraf]{KMRS17}
S.~Kopparty, O.~Meir, N.~Ron-Zewi, and S.~Saraf.
\newblock High-rate locally correctable and locally testable codes with
  sub-polynomial query complexity.
\newblock \emph{Journal of the ACM}, 64\penalty0 (2):\penalty0 11, 2017.

\bibitem[Koutsoupias and Papadimitriou(1999)]{KP99}
E.~Koutsoupias and C.~H. Papadimitriou.
\newblock Worst-case equilibria.
\newblock In \emph{Proceedings of the 16th Annual Conference on Theoretical
  Aspects of Computer Science (STACS)}, pages 404--413, 1999.

\bibitem[Kushilevitz and Nisan(1996)]{KN96}
E.~Kushilevitz and N.~Nisan.
\newblock \emph{Communication Complexity}.
\newblock Cambridge University Press, 1996.

\bibitem[Lautemann(1983)]{L83}
C.~Lautemann.
\newblock {BPP} and the polynomial hierarchy.
\newblock \emph{Information Processing Letters}, 17\penalty0 (4):\penalty0
  215--217, 1983.

\bibitem[Lehmann et~al.(2006)Lehmann, Lehmann, and Nisan]{LLN06}
B.~Lehmann, D.~Lehmann, and N.~Nisan.
\newblock Combinatorial auctions with decreasing marginal utilities.
\newblock \emph{Games and Economic Behavior}, 55\penalty0 (2):\penalty0
  270--296, 2006.

\bibitem[Lemke and Howson(1964)]{LH64}
C.~E. Lemke and J.~T. Howson, Jr.
\newblock Equilibrium points of bimatrix games.
\newblock \emph{SIAM Journal}, 12\penalty0 (2):\penalty0 413--423, 1964.

\bibitem[Leyton-Brown et~al.(2017)Leyton-Brown, Milgrom, and Segal]{LMS17}
K.~Leyton-Brown, P.~Milgrom, and I.~Segal.
\newblock Economics and computer science of a radio spectrum reallocation.
\newblock \emph{Proceedings of the National Academy of Sciences (PNAS)},
  114\penalty0 (28):\penalty0 7202--7209, 2017.

\bibitem[Lipton and Young(1994)]{LY94}
R.~J. Lipton and N.~E. Young.
\newblock Simple strategies for large zero-sum games with applications to
  complexity theory.
\newblock In \emph{Proceedings of the 26th Annual {ACM} Symposium on Theory of
  Computing (STOC)}, pages 734--740, 1994.

\bibitem[Lipton et~al.(2003)Lipton, Markakis, and Mehta]{LMM03}
R.~J. Lipton, E.~Markakis, and A.~Mehta.
\newblock Playing large games using simple strategies.
\newblock In \emph{Proceedings of the 4th ACM Conference on Electronic Commerce
  (EC)}, pages 36--41, 2003.

\bibitem[Littlestone and Warmuth(1994)]{LW94}
N.~Littlestone and M.~K. Warmuth.
\newblock The weighted majority algorithm.
\newblock \emph{Information and Computation}, 108\penalty0 (2):\penalty0
  212--261, 1994.

\bibitem[Maskin and Riley(1984)]{MR84}
E.~Maskin and J.~Riley.
\newblock Optimal auctions with risk-adverse buyers.
\newblock \emph{Econometrica}, 52\penalty0 (6):\penalty0 1473--1518, 1984.

\bibitem[Matthews(1984)]{M84}
S.~A. Matthews.
\newblock On the implementability of reduced form auctions.
\newblock \emph{Econometrica}, 52\penalty0 (6):\penalty0 1519--1522, 1984.

\bibitem[{McLennan}(2018)]{M15}
A.~{McLennan}.
\newblock \emph{Advanced Fixed Point Theory for Economics}.
\newblock Springer, 2018.

\bibitem[{McLennan} and Tourky(2006)]{MT06}
A.~{McLennan} and R.~Tourky.
\newblock From imitation games to {K}akutani.
\newblock Unpublished manuscript, 2006.

\bibitem[Megiddo and Papadimitriou(1991)]{MP91}
N.~Megiddo and C.~H. Papadimitriou.
\newblock On total functions, existence theorems and computational complexity.
\newblock \emph{Theoretical Computer Science}, 81\penalty0 (2):\penalty0
  317--324, 1991.

\bibitem[Milgrom(2000)]{M00}
P.~Milgrom.
\newblock Putting auction theory to work: The simultaneous ascending auction.
\newblock \emph{Journal of Political Economy}, 108\penalty0 (2):\penalty0
  245--272, 2000.

\bibitem[Milgrom(2004)]{Milgrom:2004jx}
P.~Milgrom.
\newblock \emph{Putting Auction Theory to Work}.
\newblock Churchill Lectures in Economics. Cambridge University Press, 2004.

\bibitem[Milgrom and Segal(2020)]{MS20}
P.~Milgrom and I.~Segal.
\newblock Clock auctions and radio spectrum reallocation.
\newblock \emph{Journal of Political Economy}, 2020.
\newblock To appear.

\bibitem[{Morris, Jr.}(1994)]{M94}
W.~D. {Morris, Jr.}
\newblock Lemke paths on simple polytopes.
\newblock \emph{Mathematics of Operations Research}, 19\penalty0 (4):\penalty0
  780--789, 1994.

\bibitem[Moulin and Vial(1978)]{MV78}
H.~Moulin and J.~P. Vial.
\newblock Strategically zero-sum games: The class of games whose completely
  mixed equilibria cannot be improved upon.
\newblock \emph{International Journal of Game Theory}, 7\penalty0
  (3--4):\penalty0 201--221, 1978.

\bibitem[Myerson(1981)]{myerson}
R.~Myerson.
\newblock Optimal auction design.
\newblock \emph{Mathematics of Operations Research}, 6\penalty0 (1):\penalty0
  58--73, 1981.

\bibitem[Nasar(1998)]{nasar}
S.~Nasar.
\newblock \emph{A Beautiful Mind: a Biography of John Forbes Nash, Jr., Winner
  of the Nobel Prize in Economics, 1994}.
\newblock Simon {\&} Schuster, 1998.

\bibitem[Nash(1950)]{Nas50}
J.~F. Nash, Jr.
\newblock Equilibrium points in {$N$}-person games.
\newblock \emph{Proceedings of the National Academy of Sciences}, 36\penalty0
  (1):\penalty0 48--49, 1950.

\bibitem[Nash(1951)]{Nas51}
J.~F. Nash, Jr.
\newblock Non-cooperative games.
\newblock \emph{Annals of Mathematics}, 54\penalty0 (2):\penalty0 286--295,
  1951.

\bibitem[Nisan(2002)]{Nis02}
N.~Nisan.
\newblock The communication complexity of approximate set packing and covering.
\newblock In \emph{Proceedings of the 29th International Colloquium on
  Automata, Languages and Programming (ICALP)}, pages 868--875, 2002.

\bibitem[Nisan and Segal(2006)]{NS06}
N.~Nisan and I.~Segal.
\newblock The communication requirements of efficient allocations and
  supporting prices.
\newblock \emph{Journal of Economic Theory}, 129\penalty0 (1):\penalty0
  192--224, 2006.

\bibitem[Papadimitriou(1994)]{P94}
C.~H. Papadimitriou.
\newblock On the complexity of the parity argument and other inefficient proofs
  of existence.
\newblock \emph{Journal of Computer and System Sciences}, 48\penalty0
  (3):\penalty0 498--532, 1994.

\bibitem[Papadimitriou(2007)]{P07}
C.~H. Papadimitriou.
\newblock The complexity of finding {N}ash equilibria.
\newblock In N.~Nisan, T.~Roughgarden, {\'E}.~Tardos, and V.~V. Vazirani,
  editors, \emph{Algorithmic Game Theory}, chapter~2, pages 29--51. Cambridge,
  2007.

\bibitem[Papadimitriou and Roughgarden(2008)]{PR08}
C.~H. Papadimitriou and T.~Roughgarden.
\newblock Computing correlated equilibria in multi-player games.
\newblock \emph{Journal of the ACM}, 55\penalty0 (3):\penalty0 14, 2008.

\bibitem[Pass and Venkitasubramaniam(2019)]{PV19}
R.~Pass and M.~Venkitasubramaniam.
\newblock A round-collapse theorem for computationally-sound protocols; or,
  {TFNP} is hard (on average) in {P}essiland.
\newblock arXiv:1906.10837, 2019.

\bibitem[Raz and McKenzie(1999)]{DBLP:journals/combinatorica/RazM99}
R.~Raz and P.~McKenzie.
\newblock Separation of the monotone {NC} hierarchy.
\newblock \emph{Combinatorica}, 19\penalty0 (3):\penalty0 403--435, 1999.

\bibitem[Raz and Wigderson(1994)]{RW90}
R.~Raz and A.~Wigderson.
\newblock Monotone circuits for matching require linear depth.
\newblock \emph{Journal of the ACM}, 39\penalty0 (3):\penalty0 736--744, 1994.

\bibitem[Razborov(1992)]{R92}
A.~A. Razborov.
\newblock On the distributional complexity of disjointness.
\newblock \emph{Theoretical Computer Science}, 106\penalty0 (2):\penalty0
  385--390, 1992.

\bibitem[Robinson(1951)]{Rob51}
J.~Robinson.
\newblock An iterative method of solving a game.
\newblock \emph{Annals of Mathematics}, pages 296--301, 1951.

\bibitem[Rosen et~al.(2017)Rosen, Segev, and Shahaf]{RSS17}
A.~Rosen, G.~Segev, and I.~Shahaf.
\newblock Can {PPAD} hardness be based on standard cryptographic assumptions?
\newblock In \emph{Proceedings of the 15th International Conference on Theory
  of Cryptography (TCC)}, pages 173--205, 2017.

\bibitem[Roughgarden(2005)]{book}
T.~Roughgarden.
\newblock \emph{Selfish Routing and the Price of Anarchy}.
\newblock MIT Press, 2005.

\bibitem[Roughgarden(2010)]{et}
T.~Roughgarden.
\newblock Computing equilibria: A computational complexity perspective.
\newblock \emph{Economic Theory}, 42\penalty0 (1):\penalty0 193--236, 2010.

\bibitem[Roughgarden(2014{\natexlab{a}})]{R14}
T.~Roughgarden.
\newblock Barriers to near-optimal equilibria.
\newblock In \emph{Proceedings of the 55th Annual Symposium on Foundations of
  Computer Science (FOCS)}, pages 71--80, 2014{\natexlab{a}}.

\bibitem[Roughgarden(2014{\natexlab{b}})]{w14}
T.~Roughgarden.
\newblock {CS364B} lecture notes.
\newblock Stanford University, 2014{\natexlab{b}}.

\bibitem[Roughgarden(2015)]{robust}
T.~Roughgarden.
\newblock Intrinsic robustness of the price of anarchy.
\newblock \emph{Journal of the ACM}, 62\penalty0 (5):\penalty0 32, 2015.

\bibitem[Roughgarden(2016{\natexlab{a}})]{f13}
T.~Roughgarden.
\newblock \emph{Twenty Lectures on Algorithmic Game Theory}.
\newblock Cambridge University Press, 2016{\natexlab{a}}.

\bibitem[Roughgarden(2016{\natexlab{b}})]{w15}
T.~Roughgarden.
\newblock Communication complexity (for algorithm designers).
\newblock \emph{Foundations and Trends in Theoretical Computer Science},
  11\penalty0 (3-4):\penalty0 217--404, 2016{\natexlab{b}}.

\bibitem[Roughgarden and {Talgam-Cohen}(2015)]{priceeq}
T.~Roughgarden and I.~{Talgam-Cohen}.
\newblock Why prices need algorithms.
\newblock In \emph{Proceedings of the 16th Annual ACM Conference on Economics
  and Computation (EC)}, pages 19--36, 2015.

\bibitem[Roughgarden and Tardos(2002)]{RT00}
T.~Roughgarden and {\'{E}}.~Tardos.
\newblock How bad is selfish routing?
\newblock \emph{Journal of the ACM}, 49\penalty0 (2):\penalty0 236--259, 2002.

\bibitem[Roughgarden and Weinstein(2016)]{RW16}
T.~Roughgarden and O.~Weinstein.
\newblock On the communication complexity of approximate fixed points.
\newblock In \emph{Proceedings of the 57th Annual Symposium on Foundations of
  Computer Science (FOCS)}, pages 229--238, 2016.

\bibitem[Roughgarden et~al.(2017)Roughgarden, Syrgkanis, and Tardos]{RST17}
T.~Roughgarden, V.~Syrgkanis, and {\'E}.~Tardos.
\newblock The price of anarchy in auctions.
\newblock \emph{Journal of Artificial Intelligence Research}, 59:\penalty0
  59--101, 2017.

\bibitem[Rubinstein(2016)]{R16}
A.~Rubinstein.
\newblock Settling the complexity of computing approximate two-player {N}ash
  equilibria.
\newblock In \emph{Proceedings of the 57th Annual IEEE Symposium on Foundations
  of Computer Science}, pages 258--265, 2016.

\bibitem[Savani and von Stengel(2006)]{SvS04}
R.~Savani and B.~von Stengel.
\newblock Hard-to-solve bimatrix games.
\newblock \emph{Econometrica}, 74\penalty0 (2):\penalty0 397--429, 2006.

\bibitem[Schrijver(1986)]{S86}
A.~Schrijver.
\newblock \emph{Theory of Linear and Integer Programming}.
\newblock Wiley, 1986.

\bibitem[Shapley(1964)]{S64}
L.~S. Shapley.
\newblock Some topics in two-person games.
\newblock In M.~Dresher, L.~S. Shapley, and A.~W. Tucker, editors,
  \emph{Advances in Game Theory}, pages 1--28. Princeton University Press,
  1964.

\bibitem[Solan and Vohra(2002)]{SV02}
E.~Solan and R.~Vohra.
\newblock Correlated equilibrium payoffs and public signalling in absorbing
  games.
\newblock \emph{International Journal of Game Theory}, 31\penalty0
  (1):\penalty0 91--121, 2002.

\bibitem[Sperner(1928)]{sperner}
E.~Sperner.
\newblock {Neuer Beweis f{\"u}r die Invarianz der Dimensionszahl und des
  Gebietes}.
\newblock \emph{Abhandlungen aus dem Mathematischen Seminar der Universit{\"a}t
  Hamburg}, 6\penalty0 (1):\penalty0 265--272, 1928.

\bibitem[Spielman(1997)]{S97}
D.~A. Spielman.
\newblock The complexity of error-correcting codes.
\newblock In \emph{Proceedings of the 11th International Symposium on
  Fundamentals of Computation Theory}, pages 67--84, 1997.

\bibitem[Spielman and Teng(2004)]{ST04}
D.~A. Spielman and S.-H. Teng.
\newblock Smoothed analysis: Why the simplex algorithm usually takes polynomial
  time.
\newblock \emph{Journal of the ACM}, 51\penalty0 (3):\penalty0 385--463, 2004.

\bibitem[Sun and Yang(2006)]{SY06}
N.~Sun and Z.~Yang.
\newblock Equilibria and indivisibilities: Gross substitutes and complements.
\newblock \emph{Econometrica}, 74\penalty0 (5):\penalty0 1385--1402, 2006.

\bibitem[Syrgkanis and Tardos(2013)]{ST13}
V.~Syrgkanis and {\'E}.~Tardos.
\newblock Composable and efficient mechanisms.
\newblock In \emph{Proceedings of the 45th ACM Symposium on Theory of Computing
  (STOC)}, pages 211--220, 2013.

\bibitem[Toda(1991)]{toda}
S.~Toda.
\newblock {PP} is as hard as the polynomial-time hierarchy.
\newblock \emph{SIAM Journal on Computing}, 20\penalty0 (5):\penalty0 865--877,
  1991.

\bibitem[Vetta(2002)]{V02}
A.~Vetta.
\newblock Nash equilibria in competitive societies, with applications to
  facility location, traffic routing and auctions.
\newblock In \emph{Proceedings of the 43rd Annual Symposium on Foundations of
  Computer Science (FOCS)}, pages 416--425, 2002.

\bibitem[Vickrey(1961)]{V61}
W.~Vickrey.
\newblock Counterspeculation, auctions, and competitive sealed tenders.
\newblock \emph{Journal of Finance}, 16\penalty0 (1):\penalty0 8--37, 1961.

\bibitem[Ville(1938)]{V38}
J.~Ville.
\newblock Sur la theorie g{\'e}n{\'e}rale des jeux ou intervient l'habilet{\'e}
  des joueurs.
\newblock Fascicule 2 in Volume 4 of {\'E.} Borel, \emph{Trait{\'e} du Calcul
  des probabilit{\'e}s et de ses applications}, pages 105--113.
  Gauthier-Villars, 1938.

\bibitem[von Neumann(1928)]{vN28}
J.~von Neumann.
\newblock {Zur Theorie der Gesellschaftsspiele}.
\newblock \emph{Mathematische Annalen}, 100:\penalty0 295--320, 1928.

\bibitem[von Neumann and Morgenstern(1944)]{vNM44}
J.~von Neumann and O.~Morgenstern.
\newblock \emph{Theory of Games and Economic Behavior}.
\newblock Princeton University Press, 1944.

\bibitem[von Stengel(2007)]{vS07}
B.~von Stengel.
\newblock Equilibrium computation for two-player games in strategic and
  extensive form.
\newblock In N.~Nisan, T.~Roughgarden, {\'E}.~Tardos, and V.~Vazirani, editors,
  \emph{Algorithmic Game Theory}, chapter~3, pages 53--78. Cambridge University
  Press, 2007.

\end{thebibliography}
\end{document}